\documentclass{article}

\usepackage[a4paper,margin=2.5cm]{geometry}

\usepackage{authblk}

\usepackage[utf8]{inputenc} 
\usepackage[T1]{fontenc}    
\usepackage{lmodern}
\usepackage[hidelinks]{hyperref}       
\usepackage{url}            
\usepackage{booktabs}       
\usepackage{amsfonts}       
\usepackage{nicefrac}       
\usepackage{microtype}      
\usepackage{xcolor}         
\usepackage{amsmath}
\usepackage{amssymb}
\usepackage{mathtools}
\usepackage{amsthm}
\usepackage{optidef}
\usepackage{dsfont}
\usepackage[linesnumbered,ruled]{algorithm2e}

\usepackage{subcaption}
\usepackage{caption}
\usepackage{comment}

\makeatletter
\newcommand{\pushright}[1]{\ifmeasuring@#1\else\omit\hfill$\displaystyle#1$\fi\ignorespaces}
\newcommand{\pushleft}[1]{\ifmeasuring@#1\else\omit$\displaystyle#1$\hfill\fi\ignorespaces}
\makeatother

\usepackage[inline]{enumitem}

\usepackage{tikz}
\usetikzlibrary{arrows.meta}

\usetikzlibrary{patterns,patterns.meta}
\usetikzlibrary{calc}
\tikzset{
	prefix after node/.style={prefix after command=\pgfextra{#1}},
	/semifill/ang/.initial=45,
	/semifill/upper/.initial=none,
	/semifill/lower/.initial=none,
	semifill/.style={
		circle, draw=none,
		prefix after node={
			\pgfqkeys{/semifill}{#1}
			\path let \p1 = (\tikzlastnode.north), \p2 = (\tikzlastnode.center),
			\n1 = {\y1-\y2} in [radius=\n1]
			(\tikzlastnode.\pgfkeysvalueof{/semifill/ang}) 
			edge[
			draw=none,
			fill=\pgfkeysvalueof{/semifill/upper},
			to path={
				arc[start angle=\pgfkeysvalueof{/semifill/ang}, delta angle=180]
				-- cycle}] ()
			(\tikzlastnode.\pgfkeysvalueof{/semifill/ang}) 
			edge[
			draw=none,
			fill=\pgfkeysvalueof{/semifill/lower},
			to path={
				arc[start angle=\pgfkeysvalueof{/semifill/ang}, delta angle=-180]
				-- cycle}] ();}}}
\usepackage{fp} 

\theoremstyle{plain}
\newtheorem{theorem}{Theorem}
\newtheorem{lemma}[theorem]{Lemma}
\newtheorem{observation}[theorem]{Observation}
\newtheorem{corollary}[theorem]{Corollary}

\theoremstyle{definition}
\newtheorem{definition}[theorem]{Definition}

\theoremstyle{remark}

\newcommand{\criticalBox}{
    \begin{figure}[htbp]
        \centering
        \begin{tikzpicture}
        \draw[thick] (0,0) rectangle (2.4,2.4);

        \draw[->, red, line width=1pt] (0,0) -- (0,0.4);
        \draw[->, blue, line width=1pt] (0,0) -- (0.4,0);
        \draw[->, green!80!black, line width=1pt] (0,0) -- (0.4,0.4);
        \fill (0,0) circle (1.5pt); 
        \node[below] at (0,0) {\scriptsize $x$};

        \draw[->, red, line width=1pt] (0,3.2) -- (0,2.8);
        \draw[->, red, line width=1pt] (0,2.4) -- (0,2.8);
        \fill (0,2.4) circle (1.5pt); 
        \node[below left] at (0,2.4) {\scriptsize $x + (0,h)$};

        \draw[->, blue, line width=1pt] (3.2,0) -- (2.8,0);
        \draw[->, blue, line width=1pt] (2.4,0) -- (2.8,0);
        \fill (2.4,0) circle (1.5pt); 
        \node[below] at (2.4,0) {\scriptsize $x + (w,0)$};
     
        \draw[->, green!80!black, line width=1pt] (-0.8,-0.8) -- (-1.2,-1.2);

        \draw[gray, dotted, line width=0.5pt] (1.4, 2.6) -- (2.4, 2.6);
        \draw[gray, dotted, line width=0.5pt] (1, 2.6) -- (0, 2.6);
        \node[gray] at (1.2,2.6) {\scriptsize $w$};

        \draw[gray, dotted, line width=0.5pt] (-0.2, 1.4) -- (-0.2, 2.4);
        \draw[gray, dotted, line width=0.5pt] (-0.2, 1) -- (-0.2, 0);
        \node[gray] at (-0.2, 1.2) {\scriptsize $h$};

        \foreach \x in {0,1,...,5} {
            \foreach \y in {0,1,...,5} {
                \fill[color=black] (-0.8 + 0.8*\y, -0.8 + 0.8*\x) circle (1pt);
            }
        }

        \end{tikzpicture} 
        \caption{A critical box.}
        \label{fig:critical_box}
    \end{figure}
}

\newcommand{\greenBoundary}{
    \begin{figure}[htbp]
    \centering
    \begin{minipage}{0.3\textwidth}
        \centering
        \begin{tikzpicture}
        \draw (2,0) -- (2,2);
        \draw (0,2) -- (2,2);
        \begin{scope}[very thick, every node/.style={sloped,allow upside down}]
        \draw[thick, green!80!black, line width=2pt] (0,0) -- (2,0);
        \draw[thick, green!80!black, line width=2pt] (0,2) -- (0,0);
        \draw [green!80!black, -{Latex[length=3mm]}] (0,0) -- (1.1,0);
        \draw [green!80!black, -{Latex[length=3mm]}] (0,2) -- (0,1.1);
        \draw [green!80!black, -{Latex[length=3mm]}] (1,0) -- (2,0);
        \draw [green!80!black, -{Latex[length=3mm]}] (0,1) -- (0,0);
        
        \end{scope}
        \node[green!80!black] at (-0.6,0.7) {$x_1 = 0$};
        \node[green!80!black] at (1,-0.2) {$x_2 = 0$};
        \foreach \x in {0,1,...,4} {
            \foreach \y in {0,1,...,4} {
                \fill (0.5*\y, 0.5*\x) circle (1pt);  
                 \draw[->, green!80!black, line width=1pt]  (0 + 0.5*\y, 0 + 0.5*\x) --  (0.2 + 0.5*\y, 0.2 + 0.5*\x);
            }
        }
        \end{tikzpicture}
    \end{minipage}
    \begin{minipage}{0.3\textwidth}
        \centering
        \begin{tikzpicture}
        \draw (0,0) rectangle (2,2);
        \draw[thick, green!80!black, line width=2pt] (0,2) -- (0,1.5) -- (1, 1.5) --  (1, 0.5) -- (1.5, 0.5) -- (1.5, 0) -- (2,0);
        \draw [green!80!black, -{Latex[length=3mm]}] (0,2) -- (0,1.6);
        \draw [green!80!black, -{Latex[length=3mm]}] (0,1.5) -- (1,1.5);
        \draw [green!80!black, -{Latex[length=3mm]}] (1,1.5) -- (1,0.6);
	    \draw [green!80!black, -{Latex[length=3mm]}] (1.5, 0) -- (2, 0);
    
	\draw[->, green!80!black, line width=1pt]  ( 1, 0) --  (0.8, -0.2);
	\draw[->, green!80!black, line width=1pt]  ( 1, 1) --  (1.2,1.2);
	\draw[->, green!80!black, line width=1pt]  ( 1, 0.5) --  (1.2,0.7);

        \foreach \x in {0,1,...,4} {
            \foreach \y in {0,1,...,4} {
                \fill (0.5*\y, 0.5*\x) circle (1pt); 
                \ifthenelse {\x>2 \OR \y>2}{
                        \draw[->, green!80!black, line width=1pt]  ( 0.5*\x, 0.5*\y) --  (0.2 + 0.5*\x, 0.2 + 0.5*\y);
                        \draw[->, green!80!black, line width=1pt]  ( 0.5*\x, 0.5*\y) --  (0.2 + 0.5*\x, 0.2 + 0.5*\y);
            	}	
        }
        }
        
       \foreach \x in {0,1,2} {
                \draw[->, green!80!black, line width=1pt]  ( 0, 0.5*\x) --  (-0.2, -0.2 + 0.5*\x);
                \draw[->, green!80!black, line width=1pt]  ( 0.5, 0.5*\x) --  (0.3, -0.2 + 0.5*\x);
          }	
        
        \end{tikzpicture}
    \end{minipage}
    \begin{minipage}{0.3\textwidth}
        \centering
        \begin{tikzpicture}
        \draw (0,0) rectangle (2,2);
        \draw[thick, green!80!black, line width=2pt] (0.5,2) -- (0.5,1.5) --  (1.5,1.5) -- (1.5, 0.5) -- (2, 0.5);

        \draw [green!80!black, -{Latex[length=3mm]}] (0.5,2) -- (0.5,1.5);
        \draw [green!80!black, -{Latex[length=3mm]}] (0.5,1.5) -- (1.5, 1.5);
        \draw [green!80!black, -{Latex[length=3mm]}] (1.5, 1) -- (1.5, 0.5);
        \draw [green!80!black, -{Latex[length=3mm]}] (1.5, 0.5) -- (2, 0.5);
	
	\draw[->, green!80!black, line width=1pt]  ( 1, 0) --  (0.8, -0.2);
	\draw[->, green!80!black, line width=1pt]  ( 1, 1) --  (0.8,0.8);
	\draw[->, green!80!black, line width=1pt]  ( 1, 0.5) --  (0.8,0.3);

        \foreach \x in {0,1,...,4} {
            \foreach \y in {0,1,...,4} {
                \fill (0.5*\y, 0.5*\x) circle (1pt); 
                \draw[->, green!80!black, line width=1pt]  ( 0.5*\x, 0) --  (-0.2 + 0.5*\x, -0.2);
                \draw[->, green!80!black, line width=1pt]  ( 0, 0.5*\x) --  (-0.2, -0.2 + 0.5*\x);
                \ifthenelse {\x>2}{
                    \ifthenelse {\y>0}{
                         \draw[->, green!80!black, line width=1pt]  ( 0.5*\x, 0.5*\y) --  (0.2 + 0.5*\x, 0.2 + 0.5*\y);
                }{}}{}
                \ifthenelse {\y>2}{
                    \ifthenelse {\x>0}{
                    \draw[->, green!80!black, line width=1pt]  ( 0.5*\x, 0.5*\y) --  (0.2 + 0.5*\x, 0.2 + 0.5*\y);
                }{}}{}	
        }
        }
        
       \foreach \x in {0,1,2} {
                \draw[->, green!80!black, line width=1pt]  ( 0.5, 0.5*\x) --  (0.3, -0.2 + 0.5*\x);
          }	
        
        \end{tikzpicture}
    \end{minipage}
    \caption{Possible configurations of set $\mathcal{B}$ defined on a 2-dimensional slice.}
    \label{fig:green_boundary}
    \end{figure}
}

\newcommand{\dBox}{
	\begin{figure}
	\centering
	\begin{tikzpicture}[scale=0.5]
		\draw 	(0,0) -- (0,1)
		(0,0) -- (2,0)
		(0,1) -- (5,6)
		(2,0) -- (7,5)
		(5,6) -- (7,6)
		(7,5) -- (7,6);
		
		\node[below] at (2,0) {\scriptsize $(y_1, y_2)$};
		\fill (2,0) circle (2pt);
		\node[left] at (0,1) {\scriptsize $(x_1,x_2)$};
		\fill (0,1) circle (2pt);
		\node[left] at (5,6) {\scriptsize $(x_1 + l, x_2 + l)$};
		\fill (5,6) circle (2pt);
		\node[right] at (7,5) {\scriptsize $(y_1 + l,y_2 + l)$};
		\fill (7,5) circle (2pt);
		
		\node[above] at (4,3) {\scriptsize $z$};
		\fill (4,3) circle (2pt);
		
		\draw[dotted] (2.5,0) -- (7.5,5);	
		\node[color=gray, below] at (5.5,3) {\scriptsize $l$};	
	\end{tikzpicture}
	\caption{A $\DBox(x,y,l)$.}
	\label{fig:dbox}
	\end{figure}
}

\newcommand{\gridStates}{
	\begin{figure}[htbp]
		\centering
		\vspace{0.5cm}
		\begin{minipage}{0.25\textwidth}
			\centering
			\vspace{1cm}
			\begin{tikzpicture}
				\draw (0,2) rectangle (2,3);
				\draw (-1,2.55) rectangle (0,2.75);
				\node[below] at (-1,2.55) {\scriptsize $y$};
				\node[below] at (0,2) {\scriptsize $x$};
				\fill (-1,2.55) circle (1pt); 
				\fill (0,2) circle (1pt); 
				
				\draw[gray, dotted, line width=0.5pt] (-1.2, 2.55) -- (-1.2,2.75);
				\draw[gray, dotted, line width=0.5pt] (-1, 2.95) -- (0, 2.95);
				\node[gray] at (-1.4, 2.65) {\scriptsize $u$};
				\node[gray] at (-0.5, 3.05) {\scriptsize $v$};
				
				\draw[gray, dotted, line width=0.5pt] (0, 3.2) -- (2,3.2);
				\draw[gray, dotted, line width=0.5pt] (2.1, 2) -- (2.1, 3);
				\node[gray] at (1, 3.4) {\scriptsize $w$};
				\node[gray] at (2.3, 2.5) {\scriptsize $h$};
				
			\end{tikzpicture}
			\caption*{(a) Left-lobe state.}
		\end{minipage}
		\begin{minipage}{0.25\textwidth}
			\centering
			\vspace{-1.25cm}
			\begin{tikzpicture}
				\draw (2,0) rectangle (3,2);
				\draw (2.55,-1) rectangle (2.75,0);
				\node[left] at (2,2) {\scriptsize $x$};
				\node[left] at (2.55,-1) {\scriptsize $y$};
				\fill (2,2) circle (1pt); 
				\fill (2.55,-1) circle (1pt); 
				
				\draw[gray, dotted, line width=0.5pt] (2.95, -1) -- (2.95,0);
				\draw[gray, dotted, line width=0.5pt] (3.2, 0) -- (3.2, 2);
				\node[gray] at (3.05, -0.5) {\scriptsize $v$};
				\node[gray] at (3.4, 1) {\scriptsize $w$};
				
				\draw[gray, dotted, line width=0.5pt] (2, 2.2) -- (3,2.2); 
				\draw[gray, dotted, line width=0.5pt] (2.55, -1.2) -- (2.75, -1.2);
				\node[gray] at (2.5, 2.4) {\scriptsize $h$};
				\node[gray] at (2.65, -1.4) {\scriptsize $u$};
				
			\end{tikzpicture}
			\caption*{(b) Bottom-lobe state.}
		\end{minipage}
		\begin{minipage}{0.25\textwidth}
			\centering
			\vspace{-0.65cm}
			\begin{tikzpicture}[scale=0.4]
				\draw 	(0,0) -- (0,1)
							(0,0) -- (2,0)
							(0,1) -- (5,6)
							(2,0) -- (7,5)
							(5,6) -- (7,6)
							(7,5) -- (7,6);
							
				\node[below] at (2,0) {\scriptsize $y$};
				\fill (2,0) circle (2pt);
				\node[left] at (0,1) {\scriptsize $x$};
				\fill (0,1) circle (2pt);
								
				\draw[dotted] (2.5,0) -- (7.5,5);	
				\node[color=gray, below] at (5.5,3) {\scriptsize $l$};	
					
				\draw (5.5,6) -- (7.5,8);  
				\draw (6.5,6) -- (8.5,8);  
				\draw (7.5,8) -- (8.5,8);  
				
				\node[below] at (5.5,6) {\scriptsize $a$};
				\fill (5.5,6) circle (2pt);
				\node[below] at (6.5,6) {\scriptsize $b$};
				\fill (6.5,6) circle (2pt);
				
				\draw[dotted] (6.8,6) -- (8.8,8);	
				\node[color=gray, below] at (8.2,7.2) {\scriptsize $m$};	
				
			\end{tikzpicture}
			\caption*{(c) Diagonal-lobe state.}
		\end{minipage}\\

		\begin{minipage}{0.8\textwidth}
			\centering
			\vspace{0.2cm}
			\begin{tikzpicture}[scale=0.7]
				\draw (2,2) rectangle (3,3);
				\draw (0,2) rectangle (2,3);
				\draw (-1,2.55) rectangle (0,2.75);
				\draw (2,0) rectangle (3,2);
				\draw (2.55,-1) rectangle (2.75,0);
				\draw 	(2,3) -- (3.5,4.5) 					
							(3,2) -- (4.5, 3.5)					
							(3.5,4.5) -- (4.5,4.5)			
							(4.5,3.5) -- (4.5,4.5) ;		 
				\draw 	(3.8, 4.5) -- (4.6, 5.3) 	
							(4.2, 4.5) -- (5, 5.3)
							(4.6, 5.3) -- (5, 5.3) ;		
								
				\node[below left] at (0,2) {\scriptsize $x^l$};
				\fill (0,2) circle (2pt);
				\node[below left] at (2,2) {\scriptsize $x$};
				\fill (2,2) circle (2pt);
				\node[below left] at (2,0) {\scriptsize $x^r$};
				\fill (2,0) circle (2pt);
				\node[above left] at (2,3) {\scriptsize $x^d$};
				\fill (2,3) circle (2pt);
			\end{tikzpicture}
			\caption*{(d) Full state.}
		\end{minipage}
		\caption{Grid states.}
		\label{fig:grid_states}
	\end{figure}
}

\newcommand{\leftLobeMain}{
    \begin{figure}[htbp]
        \centering
        \begin{minipage}{0.45\textwidth}
        \centering
        \begin{tikzpicture}[scale=0.6]
            \draw (-2,0) rectangle (6,2);
            \draw (-4,1.2) rectangle (-2,1.8);
            \node[below] at (2, 0) {$	q$};
            \fill (2, 0) circle (2pt);
            \node[below] at (-2, 0) {$x$};
            \fill (-2, 0) circle (2pt);
            \node[below] at (-4, 1.2) {$y$};
            \fill (-4, 1.2) circle (2pt);
            \draw[->, blue, line width=1.5pt] (2,0) -- (1.5,0);
            \draw[black, dotted, line width=1pt] (4,2) -- (4, 0);
            \fill[pattern=north east lines, pattern color=gray, opacity=0.25] (4,0) -- (6,0) 	-- (6,2) -- (4,2)  -- cycle; 
        \end{tikzpicture}
        \caption*{(a)}
        \end{minipage}
        \begin{minipage}{0.45\textwidth}
        \centering
        \begin{tikzpicture}[scale=0.6]
            \draw (-2,0) rectangle (6,2);
            \draw (-4,1.2) rectangle (-2,1.8);
            \node[below] at (2, 0) {$q $};
            \fill (2, 0) circle (2pt);
            \node[below] at (-2, 0) {$x$};
            \fill (-2, 0) circle (2pt);
            \node[below] at (-4, 1.2) {$y$};
            \fill (-4, 1.2) circle (2pt);
            \node[right] at (0, 1) {\scriptsize $q^{\text{mid}}$};
            \fill (0,1) circle (2pt);
            \draw[->, red, line width=1.5pt] (0,1) -- (0,0.5);
            \draw[->, green!80!black, line width=1.5pt] (0,1) -- (0.4,1.3);
            \draw[->, blue, line width=1.5pt] (2,0) -- (2.5,0);
            \draw[black, dotted, line width=1pt] (0,2) -- (0, 0);
            \draw[black, dotted, line width=1pt] (0,1) -- (-2, 1);

            \draw[gray, dotted, line width=1pt] (0,-0.4) -- (0.6, -0.4);
            \draw[gray, dotted, line width=1pt] (1.4,-0.4) -- (2, -0.4);
            \node[gray] at (1, -0.4) {\scriptsize $h+1$};

            \fill[pattern=north east lines, pattern color=gray, opacity=0.25] (-4,1.2) -- (-4, 1.8) -- (-2, 1.8) -- (-2, 1.2) -- cycle; 
             \fill[pattern=north east lines, pattern color=gray, opacity=0.25] (-2, 1) -- (-2, 2) -- (0,2) -- (0,1)  -- cycle; 
        \end{tikzpicture}
        \caption*{(b)}
        \end{minipage} 
        \begin{minipage}{0.45\textwidth}
        \centering
        \begin{tikzpicture}[scale=0.6]
            \draw (-2,0) rectangle (6,2);
            \draw (-4,1.2) rectangle (-2,1.8);
            \node[below] at (2, 0) {$q $};
            \fill (2, 0) circle (2pt);
            \node[below] at (-2, 0) {$x$};
            \fill (-2, 0) circle (2pt);
            \node[below] at (-4, 1.2) {$y$};
            \fill (-4, 1.2) circle (2pt);
            \node[right] at (0, 1) {\scriptsize $q^{\text{mid}}$};
            \fill (0,1) circle (2pt);
            \draw[->, red, line width=1.5pt] (0,1) -- (0,1.5);
            \draw[->, green!80!black, line width=1.5pt] (0,1) -- (-0.4,0.7);
            \draw[->, blue, line width=1.5pt] (2,0) -- (2.5,0);
            \draw[black, dotted, line width=1pt] (0,2) -- (0, 0);
            \draw[black, dotted, line width=1pt] (0,1) -- (-2, 1);

            \draw[gray, dotted, line width=1pt] (0,-0.4) -- (0.6, -0.4);
            \draw[gray, dotted, line width=1pt] (1.4,-0.4) -- (2, -0.4);
            \node[gray] at (1, -0.4) {\scriptsize $h+1$};

            \fill[pattern=north east lines, pattern color=gray, opacity=0.25] (-2,1) -- (0,1) -- (0,0) -- (-2,0)  -- cycle; 
        \end{tikzpicture}
        \caption*{(c)}
        \end{minipage}
         \begin{minipage}{0.45\textwidth}
        \centering
        \begin{tikzpicture}[scale=0.6]
            \draw (-2,0) rectangle (6,2);
            \draw (-4,1.2) rectangle (-2,1.8);
            \node[below] at (2, 0) {$q $};
            \fill (2, 0) circle (2pt);
            \node[below] at (-2, 0) {$x$};
            \fill (-2, 0) circle (2pt);
            \node[below] at (-4, 1.2) {$y$};
            \fill (-4, 1.2) circle (2pt);
            \node[right] at (0, 1) {\scriptsize $q^{\text{mid}}$};
            \fill (0,1) circle (2pt);
            \draw[->, red, line width=1.5pt] (0,1) -- (0,0.5);
            \draw[->, green!80!black, line width=1.5pt] (0,1) -- (-0.4,0.7);
            \draw[->, blue, line width=1.5pt] (2,0) -- (2.5,0);
            \draw[black, dotted, line width=1pt] (0,2) -- (0, 0);
            \draw[black, dotted, line width=1pt] (0,1) -- (-2, 1);

            \draw[gray, dotted, line width=1pt] (0,-0.4) -- (0.6, -0.4);
            \draw[gray, dotted, line width=1pt] (1.4,-0.4) -- (2, -0.4);
            \node[gray] at (1, -0.4) {\scriptsize $h+1$};

            \fill[pattern=north east lines, pattern color=gray, opacity=0.25] (-4,1.2) -- (-4, 1.8) -- (-2, 1.8) -- (-2, 1.2) -- cycle;  
             \fill[pattern=north east lines, pattern color=gray, opacity=0.25] (-2, 0) -- (-2, 2) -- (0,2) -- (0,0)  -- cycle; 
        \end{tikzpicture}
        \caption*{(d)}
        \end{minipage} \\
        
        \caption{Step 1 - Cases 1 and 2.1-2.3}
        \label{fig:leftLobeMain}
    \end{figure}
}

\newcommand{\leftLobeSub}{
    \begin{figure}[htbp]
    	\centering
        \begin{minipage}{0.32\textwidth}
        \centering
        \begin{tikzpicture}[scale=0.8]
            \draw (-2,-0.5) rectangle (0,2.5);
            \draw (-5,1.5) rectangle (-2,2.1);
            \node[below] at (-2, -0.5) {$x$};
            \fill (-2, -0.5) circle (2pt);
            \node[below] at (-5, 1.5) {$y$};
            \fill (-5, 1.5) circle (2pt);
            \fill (0,1) circle (1pt);
            
            \draw[->, red, line width=1.5pt] (0, 1.5) -- (0, 1.8);
            \draw[->, red, line width=1.5pt] (0, 2.1) -- (0, 2.4);
            
             \draw[->, green!80!black, line width=1.5pt] (0, 1.5) -- (-0.3, 1.4);
            \draw[->, green!80!black, line width=1.5pt] (0, 2.1) -- (-0.3, 2);
            
             \node[right] at (0.2, 2.1) {\scriptsize $q^{top}$};
            \fill (0, 2.1) circle (2pt);
            \node[right] at (0.2,1.5) {\scriptsize $q^{bot}$};
            \fill (0, 1.5) circle (2pt);
            
            \node at (-1, 2.3) {\scriptsize sub'};
            \draw[thick] (-2,2.1) rectangle (0,2.5);
                   
             \draw[black, dotted, line width=1pt] (0,1) -- (-2, 1);
             \draw[black, dotted, line width=1pt] (0,1.5) -- (-2, 1.5);
             \draw[black, dotted, line width=1pt] (0,2.1) -- (-2, 2.1);

             \fill[pattern=north east lines, pattern color=gray, opacity=0.25] (-2,1) -- (0,1) -- (0,-0.5) -- (-2,-0.5)  -- cycle; 
              
              \fill[pattern=north east lines, pattern color=gray, opacity=0.25] (0,1) -- (0, 2.1) -- (-2, 2.1) -- (-2, 1)  -- cycle; 
             
             \fill[pattern=north east lines, pattern color=gray, opacity=0.25] (-5,1.5) -- (-5, 2.1) -- (-2, 2.1) -- (-2, 1.5) -- cycle; 
        \end{tikzpicture}
        \caption*{(a)}
        \end{minipage} 
        \begin{minipage}{0.32\textwidth}
        \centering
        \begin{tikzpicture}[scale=0.8]
            \draw (-2,-0.5) rectangle (0,2.5);
            \draw (-5,1.5) rectangle (-2,2.1);
            \node[below] at (-2, -0.5) {$x$};
            \fill (-2, -0.5) circle (2pt);
            \node[below] at (-5, 1.5) {$y$};
            \fill (-5, 1.5) circle (2pt);
            \fill (0,1) circle (1pt);
            
            \draw[->, red, line width=1.5pt] (0, 1.5) -- (0, 1.2);
            \draw[->, red, line width=1.5pt] (0, 2.1) -- (0, 1.8);
            
             \draw[->, green!80!black, line width=1.5pt] (0, 1.5) -- (0.3, 1.6);
            \draw[->, green!80!black, line width=1.5pt] (0, 2.1) -- (0.3, 2.2);
            
             \node[right] at (0.2, 2.1) {\scriptsize $q^{top}$};
            \fill (0, 2.1) circle (2pt);
            \node[right] at (0.2,1.5) {\scriptsize $q^{bot}$};
            \fill (0, 1.5) circle (2pt);
            
            \node at (-1, 1.3) {\scriptsize sub'};
            \draw[thick] (-2,1) rectangle (0,1.5);
                   
             \draw[black, dotted, line width=1pt] (0,1) -- (-2, 1);
             \draw[black, dotted, line width=1pt] (0,1.5) -- (-2, 1.5);
             \draw[black, dotted, line width=1pt] (0,2.1) -- (-2, 2.1);

             \fill[pattern=north east lines, pattern color=gray, opacity=0.25] (-2,1) -- (0,1) -- (0,-0.5) -- (-2,-0.5)  -- cycle; 

              \fill[pattern=north east lines, pattern color=gray, opacity=0.25] (0,1.5) -- (-2, 1.5) -- (-2, 2.5) -- (0, 2.5)  -- cycle; 
             
             \fill[pattern=north east lines, pattern color=gray, opacity=0.25] (-5,1.5) -- (-5, 2.1) -- (-2, 2.1) -- (-2, 1.5) -- cycle; 
        \end{tikzpicture}
        \caption*{(b)}
        \end{minipage} 
        \begin{minipage}{0.32\textwidth}
        \centering
        \begin{tikzpicture}[scale=0.8]
            \draw (-2,-0.5) rectangle (0,2.5);
            \draw (-5,1.5) rectangle (-2,2.1);
            \node[below] at (-2, -0.5) {$x$};
            \fill (-2, -0.5) circle (2pt);
            \node[below] at (-5, 1.5) {$y$};
            \fill (-5, 1.5) circle (2pt);
            \fill (0,1) circle (1pt);
            
            \draw[->, red, line width=1.5pt] (0, 1.5) -- (0, 1.8);
            \draw[->, red, line width=1.5pt] (0, 2.1) -- (0, 1.8);
            
             \draw[->, green!80!black, line width=1.5pt] (0, 1.5) -- (-0.3, 1.4);
            \draw[->, green!80!black, line width=1.5pt] (0, 2.1) -- (0.3, 2.2);
            
             \node[right] at (0.2, 2.1) {\scriptsize $q^{\text{top}}$};
            \fill (0, 2.1) circle (2pt);
            \node[right] at (0.2,1.5) {\scriptsize $q^{\text{bot}}$};
            \fill (0, 1.5) circle (2pt);
            
            \node at (-3, 1.8) {\scriptsize sub'};
            \draw[thick] (-5,1.5) rectangle (0,2.1);
                   
             \draw[black, dotted, line width=1pt] (0,1) -- (-2, 1);

             \fill[pattern=north east lines, pattern color=gray, opacity=0.25] (-2,1.5) -- (0,1.5) -- (0,-0.5) -- (-2,-0.5)  -- cycle; 
             
             \fill[pattern=north east lines, pattern color=gray, opacity=0.25] (0,2.1) -- (-2, 2.1) -- (-2, 2.5) -- (0, 2.5)  -- cycle; 
        \end{tikzpicture}
        \caption*{(c)}
        \end{minipage} 

        \caption{Step 2 - Cases 1,2 and 5}
        \label{fig:leftLobeSub}
    \end{figure}
}

\newcommand{\leftLobeSubSub}{
    \begin{figure}[htbp]
        \centering
        \begin{minipage}{0.4\textwidth}
        \centering
        \begin{tikzpicture}[scale=0.8]            
            \draw (-5,0) rectangle (0,2);
            \node[below] at (-5, 0) {$\bar{y}$};
            \fill (-5, 0) circle (2pt);
           
            \draw[->, red, line width=1.5pt] (0, 0) -- (0, 0.3);
            \draw[->, red, line width=1.5pt] (0, 1) -- (0, 0.7);
            \draw[->, red, line width=1.5pt] (0, 2) -- (0, 1.7);
            
            \draw[->, green!80!black, line width=1.5pt] (0, 0) -- (-0.3, -0.3);
            \draw[->, green!80!black, line width=1.5pt] (0, 1) -- (0.3, 1.2);
            \draw[->, green!80!black, line width=1.5pt] (0, 2) -- (0.3, 2.2);
            
             \node[right] at (0.2, 2) {\scriptsize $q^{top}$};
            \fill (0, 2) circle (2pt);
            \node[right] at (0.2, 1) {\scriptsize $q^{cen}$};
            \fill (0, 1) circle (2pt);
            \node[right] at (0.2,0) {\scriptsize $q^{bot}$};
            \fill (0, 0) circle (2pt);
           
            \node at (-2.5, 0.5) {\scriptsize sub'};
            \draw[thick] (-5,0) rectangle (0, 1);
                   
             \draw[black, dotted, line width=1pt] (-5,1) -- (0, 1);
	
	    \fill[pattern=north east lines, pattern color=gray, opacity=0.25] (-5,1) -- (0,1) -- (0,2) -- (-5,2)  -- cycle; 
        \end{tikzpicture}
        \caption*{(a)}
        \end{minipage} 
        \begin{minipage}{0.4\textwidth}
        \centering
        \begin{tikzpicture}[scale=0.8]            
            \draw (-5,0) rectangle (0,2);
            \node[below] at (-5, 0) {$\bar{y}$};
            \fill (-5, 0) circle (2pt);
           
            \draw[->, red, line width=1.5pt] (0, 0) -- (0, 0.3);
            \draw[->, red, line width=1.5pt] (0, 1) -- (0, 1.3);
            \draw[->, red, line width=1.5pt] (0, 2) -- (0, 1.7);
            
            \draw[->, green!80!black, line width=1.5pt] (0, 0) -- (-0.3, -0.3);
            \draw[->, green!80!black, line width=1.5pt] (0, 1) -- (-0.3, 0.8);
            \draw[->, green!80!black, line width=1.5pt] (0, 2) -- (0.3, 2.2);
            
             \node[right] at (0.2, 2) {\scriptsize $q^{top}$};
            \fill (0, 2) circle (2pt);
            \node[right] at (0.2, 1) {\scriptsize $q^{cen}$};
            \fill (0, 1) circle (2pt);
            \node[right] at (0.2,0) {\scriptsize $q^{bot}$};
            \fill (0, 0) circle (2pt);

            \node at (-2.5, 1.5) {\scriptsize sub'};
            \draw[thick] (-5,1) rectangle (0, 2);
           
             \draw[black, dotted, line width=1pt] (-5,1) -- (0, 1);
	
	    \fill[pattern=north east lines, pattern color=gray, opacity=0.25] (-5,0) -- (0,0) -- (0,1) -- (-5,1)  -- cycle; 
        \end{tikzpicture}
        \caption*{(b)}
        \end{minipage} 
        \caption{Sub-Algorithm - Cases 2.5.a and 2.5.b}
        \label{fig:leftLobeSubSub}
    \end{figure}
}

\newcommand{\upRegions}{
    \begin{figure}[htbp]
        \centering
        \begin{tikzpicture}
        \draw[thick] (0,0) rectangle (2,2);

        \draw[->, red, line width=1pt] (0,2.25) -- (0,2);
        \fill (0,2.25) circle (1pt); 
        \node[right] at (0,2.25) {\scriptsize $x + (0, h+1)$};
        \draw[dotted] (0,2.25) -- (-3,2.25);
        \draw[dotted] (0,2.25) -- (2.25, 4.5);
        \draw[dotted, black] (0,2.25) -- (-2.25, 4.5);
        \node at (0,3) {(C2)};
        \node at (-2, 3) {(M2)};

        \draw[->, blue, line width=1pt] (2.25,0) -- (2,0);
        \fill (2.25,0) circle (1pt);
        \node[below left] at (2.25, 0) {\scriptsize $x + (w+1, 0)$};
        \draw[dotted] (2.25,0) -- (2.25,-3);
        \draw[dotted] (2.25,0) -- (4.5,2.25);
        \draw[dotted, black] (4.5, - 2.25) -- (2.25,0);
        \node at (3,0) {(C1)};
        \node at (3,-2) {(M1)};
        
        \draw[->, green!80!black, line width=1pt] (-0.15,-0.15) -- (-0.35,-0.35);
        \fill (-0.15, -0.15) circle (1pt);
        \node[above left] at (-0.15,-0.15) {\scriptsize $x - (1,1)$};
        \draw[dotted] (-0.15,-0.15) -- (-3,-0.15);
        \draw[dotted] (-0.15,-0.15) -- (-0.15,-3);
        \node at (-2,-2) {(M3)};

        \node at (-1.5,1) {\scriptsize $Left(x,h,w)$};
        \node at (1,-1.5) {\scriptsize $Bottom(x,h,w)$};
        \node at (3,3) {\scriptsize $Diag(x,h,w)$};

        \fill[pattern=north east lines, pattern color=gray, opacity=0.25] (-3,2.25) -- (0,2.25) -- (2.25, 4.5) -- (-3, 4.5)  -- cycle;

        \fill[pattern=north east lines, pattern color=gray, opacity=0.25] (4.5, 2.25) -- (4.5,-3) -- (2.25, -3) -- (2.25, 0)  -- cycle;

        \fill[pattern=north east lines, pattern color=gray, opacity=0.25] (-3,-0.15) -- (-3,-3) -- (-0.15, -3) -- (-0.15, -0.15)  -- cycle;
        \end{tikzpicture} 
        \caption{Regions of exclusion of points in $Up(f)$.}
        \label{fig:up_regions}
    \end{figure}
}

\newcommand{\cbDirections}{
    \begin{figure}[htbp]
        \centering
        \begin{tikzpicture}
	        \draw[thick] (0,0) rectangle (3,3);
	        \draw (1.25, 1.25) rectangle (1.75,1.75);
	        
	        \draw[->, red, line width=1.5pt] (1.25, 1.75) -- (1.25, 2.05);
	        \draw[->, red, line width=1.5pt] (1.25, 2.35) -- (1.25, 2.05);
	        
	        \draw[->, blue, line width=1.5pt] (1.75, 1.25) -- (2.05,1.25);
	        \draw[->, blue, line width=1.5pt] (2.35,1.25) -- (2.05,1.25);
	        
	        \draw[->, green!80!black, line width=1.5pt] (1.15, 1.15) -- (0.95, 0.95);
	        \draw[->, green!80!black, line width=1.5pt] (1.85, 1.85) -- (2.05, 2.05);
        \end{tikzpicture} 
        \caption{Direction functions evaluated at a critical box located within grid square.}
        \label{fig:cb_directions}
    \end{figure}
}

\newcommand{\configFromCb}{
	\begin{figure}[htbp]
		\centering
		\begin{tikzpicture}
			\draw[thick] (0,0) rectangle (3,3);
			\draw (1.25, 1.25) rectangle (1.75,1.75);
			
			\draw[->, red, line width=1.5pt] (1.25, 1.75) -- (1.25, 2.05);
			\draw[->, red, line width=1.5pt] (1.25, 2.35) -- (1.25, 2.05);
			\draw[->, red, line width=1.5pt] (0,3) -- (0, 2.7);
			\draw[->, red, line width=1.5pt] (3,0) -- (3, 0.3);

            \fill (1.25, 2.35) circle (1.5pt);
            \node[above right] at (1.25, 2.35) {\scriptsize $y + (0, h+1)$};
	
			\draw[->, blue, line width=1.5pt] (1.75, 1.25) -- (2.05,1.25);
			\draw[->, blue, line width=1.5pt] (2.35,1.25) -- (2.05,1.25);
			\draw[->, blue, line width=1.5pt] (3,0) -- (2.7,0);
			\draw[->, blue, line width=1.5pt] (0,3) -- (0.3,3);
   
            \fill (2.35,1.25) circle (1.5pt);
            \node[right] at (2.35,1.25) {\scriptsize $y + (w+1, 0)$};
			
			\draw[->, green!80!black, line width=1.5pt] (1.15, 1.15) -- (0.95, 0.95);
			\draw[->, green!80!black, line width=1.5pt] (1.85, 1.85) -- (2.05, 2.05);
			\draw[->, green!80!black, line width=1.5pt] (0,0) -- (-0.2, -0.2);
			\draw[->, green!80!black, line width=1.5pt] (3,3) -- (3.2, 3.2);

            \fill (1.15, 1.15) circle (1.5pt);
            \node[above left] at (1.15, 1.15) {\scriptsize $y - (1, 1)$};
			
			\draw[dotted] (1.25, 2.35) -- (2.25,3.35);
			\draw[dotted] (1.25, 2.35) -- (-0.25, 2.35);
			\draw[dotted] (2.35, 1.25) -- (3.35, 2.25);
			\draw[dotted] (2.35, 1.25) -- (2.35, -0.25);
			\draw[dotted] (1.15, 1.15) -- (-0.25, 1.15);
			\draw[dotted] (1.15, 1.15) -- (1.15, -0.25);
			
			\fill[pattern=north east lines, pattern color=gray, opacity=0.25] (2.35, 1.25) -- (3.35, 2.25) -- (3.35, -0.25) -- (2.35, -0.25)  -- cycle; 
			
			\fill[pattern=north east lines, pattern color=gray, opacity=0.25] (1.25, 2.35) -- (2.25, 3.35) -- (-0.25, 3.35) -- (-0.25, 2.35)  -- cycle; 
			
			\fill[pattern=north east lines, pattern color=gray, opacity=0.25] (1.15, 1.15) -- (-0.25, 1.15) -- (-0.25, -0.25) -- (1.15, -0.25)  -- cycle; 
		\end{tikzpicture} 
        \caption{Initial configurations of grid square arising from containment of critical box.}
		\label{fig:config_from_cb}
	\end{figure}
}

\newcommand{\config}{
    \begin{figure}[htbp]
    \centering
    \begin{minipage}{0.2\textwidth}
        \centering
        \begin{tikzpicture}
            \draw[thick] (0,0) rectangle (2,2);
    
            \draw[->, red, line width=1.5pt] (0,0) -- (0,0.5);
            \draw[->, red, line width=1.5pt] (0,2) -- (0,1.5);
            
            \draw[->, blue, line width=1.5pt] (0,0) -- (0.5,0);
            \draw[->, blue, line width=1.5pt] (2,0) -- (1.5,0);
            
            \draw[->, green!80!black, line width=1.5pt] (0,0) -- (-0.25,-0.25);
            \draw[->, green!80!black, line width=1.5pt] (2,2) -- (2.25,2.25);

            \foreach \x in {0,1,2} {
            \foreach \y in {0,1,2} {
                \fill[color=black] (\x, \y) circle (1.5pt);
            }}
        \end{tikzpicture} 
        \caption*{(a)}
    \end{minipage}
    \hfill
    \begin{minipage}{0.2\textwidth}
        \centering
        \begin{tikzpicture}
            \draw[thick] (0,0) rectangle (2,2);
            
            \draw[->, red, line width=1.5pt] (0,0) -- (0,0.5);
            \draw[->, red, line width=1.5pt] (0,2) -- (0,1.5);
            
            \draw[->, blue, line width=1.5pt] (0,2) -- (0.5,2);
            \draw[->, blue, line width=1.5pt] (2,2) -- (1.5,2);
            
            \draw[->, green!80!black, line width=1.5pt] (0,0) -- (-0.25,-0.25);
            \draw[->, green!80!black, line width=1.5pt] (2,2) -- (2.25,2.25);

            \foreach \x in {0,1,2} {
            \foreach \y in {0,1,2} {
                \fill[color=black] (\x, \y) circle (1.5pt);
            }}
        \end{tikzpicture} 
        \caption*{(b)}
    \end{minipage}
    \hfill
    \begin{minipage}{0.2\textwidth}
        \centering
        \begin{tikzpicture}
            \draw[thick] (0,0) rectangle (2,2);
            
            \draw[->, red, line width=1.5pt] (2,0) -- (2,0.5);
            \draw[->, red, line width=1.5pt] (2,2) -- (2,1.5);
            
            \draw[->, blue, line width=1.5pt] (0,0) -- (0.5,0);
            \draw[->, blue, line width=1.5pt] (2,0) -- (1.5,0);
            
            \draw[->, green!80!black, line width=1.5pt] (0,0) -- (-0.25,-0.25);
            \draw[->, green!80!black, line width=1.5pt] (2,2) -- (2.25,2.25);

            \foreach \x in {0,1,2} {
            \foreach \y in {0,1,2} {
                \fill[color=black] (\x, \y) circle (1.5pt);
            }}
        \end{tikzpicture} 
        \caption*{(c)}
    \end{minipage}
    \hfill
    \begin{minipage}{0.2\textwidth}
        \centering
        \begin{tikzpicture}
            \draw[thick] (0,0) rectangle (2,2);
            
            \draw[->, red, line width=1.5pt] (2,0) -- (2,0.5);
            \draw[->, red, line width=1.5pt] (2,2) -- (2,1.5);
            
            \draw[->, blue, line width=1.5pt] (0,2) -- (0.5,2);
            \draw[->, blue, line width=1.5pt] (2,2) -- (1.5,2);
            
            \draw[->, green!80!black, line width=1.5pt] (0,0) -- (-0.25,-0.25);
            \draw[->, green!80!black, line width=1.5pt] (2,2) -- (2.25,2.25);

            \foreach \x in {0,1,2} {
            \foreach \y in {0,1,2} {
                \fill[color=black] (\x, \y) circle (1.5pt);
            }}
        \end{tikzpicture} 
        \caption*{(d)}
    \end{minipage}
    \caption{CB-Config square orientations.}
    \label{fig:config}
    \end{figure}
}

\newcommand{\configExcluded}{
    \begin{figure}[htbp]
    \centering
    \begin{minipage}{0.45\textwidth}
        \centering
        \begin{tikzpicture}
            \draw[thick] (0,0) rectangle (2,2);
    
            \draw[->, red, line width=1.5pt] (0,0) -- (0,0.5);
            \draw[->, red, line width=1.5pt] (0,2) -- (0,1.5);
            \draw[dotted] (0,2) -- (-1,2);
            \draw[dotted] (0,2) -- (1, 3);
            
            \draw[->, blue, line width=1.5pt] (0,0) -- (0.5,0);
            \draw[->, blue, line width=1.5pt] (2,0) -- (1.5,0);
            \draw[dotted] (2,0) -- (3,1);
            \draw[dotted] (2,0) -- (2,-1);
            
            \draw[->, green!80!black, line width=1.5pt] (0,0) -- (-0.25,-0.25);
            \draw[->, green!80!black, line width=1.5pt] (2,2) -- (2.25,2.25);
            \draw[dotted] (0,0) -- (-1,0);
            \draw[dotted] (0,0) -- (0,-1);

            \fill[pattern=north east lines, pattern color=gray, opacity=0.25] (-1,2) -- (0,2) -- (1, 3) -- (-1, 3)  -- cycle; 

            \fill[pattern=north east lines, pattern color=gray, opacity=0.25] (2,0) -- (3,1) -- (3, -1) -- (2, -1)  -- cycle; 

            \fill[pattern=north east lines, pattern color=gray, opacity=0.25] (0,0) -- (-1, 0) -- (-1,-1) -- (0, -1)  -- cycle; 

            \foreach \x in {0,1,2} {
            \foreach \y in {0,1,2} {
                \fill[color=black] (\x, \y) circle (1.5pt);
            }}
        \end{tikzpicture} 
        \caption*{(a)}
    \end{minipage}
    \hfill
    \begin{minipage}{0.45\textwidth}
        \centering
        \begin{tikzpicture}
            \draw[thick] (0,0) rectangle (2,2);

            \draw[->, red, line width=1.5pt] (0,0) -- (0,0.5);
            \draw[->, red, line width=1.5pt] (0,2) -- (0,1.5);
            \draw[dotted] (0,2) -- (-1,2);
            \draw[dotted] (0,2) -- (1, 3);
            
            \draw[->, blue, line width=1.5pt] (0,2) -- (0.5,2);
            \draw[->, blue, line width=1.5pt] (2,2) -- (1.5,2);
            \draw[dotted] (2,2) -- (3,3);
            \draw[dotted] (2,0) -- (2,-1);
            
            \draw[->, green!80!black, line width=1.5pt] (0,0) -- (-0.25,-0.25);
            \draw[->, green!80!black, line width=1.5pt] (2,2) -- (2.25,2.25);
            \draw[dotted] (0,0) -- (-1,0);
            \draw[dotted] (0,0) -- (0,-1);

            \fill[pattern=north east lines, pattern color=gray, opacity=0.25] (-1,2) -- (0,2) -- (1, 3) -- (-1, 3)  -- cycle; 

            \fill[pattern=north east lines, pattern color=gray, opacity=0.25] (2,2) -- (3,3) -- (3, -1) -- (2, -1)  -- cycle; 

            \fill[pattern=north east lines, pattern color=gray, opacity=0.25] (0,0) -- (-1,0) -- (-1,-1) -- (0, -1)  -- cycle; 

            \foreach \x in {0,1,2} {
            \foreach \y in {0,1,2} {
                \fill[color=black] (\x, \y) circle (1.5pt);
            }}
        \end{tikzpicture} 
        \caption*{(b)}
    \end{minipage}
    \vspace{1cm}\\
    \begin{minipage}{0.45\textwidth}
        \centering
        \begin{tikzpicture}
            \draw[thick] (0,0) rectangle (2,2);
            
            \draw[->, red, line width=1.5pt] (2,0) -- (2,0.5);
            \draw[->, red, line width=1.5pt] (2,2) -- (2,1.5);
            \draw[dotted] (2,2) -- (-1,2);
            \draw[dotted] (2,2) -- (3, 3);
            
            \draw[->, blue, line width=1.5pt] (0,0) -- (0.5,0);
            \draw[->, blue, line width=1.5pt] (2,0) -- (1.5,0);
            \draw[dotted] (2,0) -- (3,1);
            \draw[dotted] (2,0) -- (2,-1);
            
            \draw[->, green!80!black, line width=1.5pt] (0,0) -- (-0.25,-0.25);
            \draw[->, green!80!black, line width=1.5pt] (2,2) -- (2.25,2.25);
            \draw[dotted] (0,0) -- (-1,0);
            \draw[dotted] (0,0) -- (0,-1);

            \fill[pattern=north east lines, pattern color=gray, opacity=0.25] (-1,2) -- (2,2) -- (3, 3) -- (-1, 3)  -- cycle; 

            \fill[pattern=north east lines, pattern color=gray, opacity=0.25] (2,0) -- (3,1) -- (3, -1) -- (2, -1)  -- cycle; 

            \fill[pattern=north east lines, pattern color=gray, opacity=0.25] (0,0) -- (-1, 0) -- (-1,-1) -- (0, -1)  -- cycle; 

            \foreach \x in {0,1,2} {
            \foreach \y in {0,1,2} {
                \fill[color=black] (\x, \y) circle (1.5pt);
            }}
        \end{tikzpicture} 
        \caption*{(c)}
    \end{minipage}
    \hfill
    \begin{minipage}{0.45\textwidth}
        \centering
        \begin{tikzpicture}
            \draw[thick] (0,0) rectangle (2,2);
            
            \draw[->, red, line width=1.5pt] (2,0) -- (2,0.5);
            \draw[->, red, line width=1.5pt] (2,2) -- (2,1.5);
            \draw[dotted] (2,2) -- (-1,2);
            \draw[dotted] (2,2) -- (3, 3);
            
            \draw[->, blue, line width=1.5pt] (0,2) -- (0.5,2);
            \draw[->, blue, line width=1.5pt] (2,2) -- (1.5,2);
            \draw[dotted] (2,0) -- (2,-1);
            
            \draw[->, green!80!black, line width=1.5pt] (0,0) -- (-0.25,-0.25);
            \draw[->, green!80!black, line width=1.5pt] (2,2) -- (2.25,2.25);
            \draw[dotted] (0,0) -- (-1,0);
            \draw[dotted] (0,0) -- (0,-1);

            \fill[pattern=north east lines, pattern color=gray, opacity=0.25] (-1,2) -- (2,2) -- (3, 3) -- (-1, 3)  -- cycle; 

            \fill[pattern=north east lines, pattern color=gray, opacity=0.25] (2,2) -- (3,3) -- (3, -1) -- (2, -1)  -- cycle; 

            \fill[pattern=north east lines, pattern color=gray, opacity=0.25] (0,0) -- (-1, 0) -- (-1,-1) -- (0, -1)  -- cycle; 

            \foreach \x in {0,1,2} {
            \foreach \y in {0,1,2} {
                \fill[color=black] (\x, \y) circle (1.5pt);
            }}
        \end{tikzpicture} 
        \caption*{(d)}
    \end{minipage}
    \caption{Excluded regions from up set search based on config square.}
    \label{fig:config_excluded}
    \end{figure}
}

\newcommand{\regions}{
    \begin{figure}
    \centering
    \begin{tikzpicture}[
        scale=0.3,
        xy/.style={circle,fill=black,inner sep=0pt,minimum size=0.2cm},
        elem/.style={circle,inner sep=0pt,minimum size=0.2cm},
        corner/.style={circle,draw=none,inner sep=0pt,minimum size=0.2cm,semifill={lower=red,upper=blue,ang=145}}]
    \coordinate (x) at (10, 10);
    \coordinate (y) at (20, 15);
    \draw[step=1cm,gray,very thin] (0, 0) grid (30, 28);
    \draw[thick,black] (1, 1) rectangle (19, 19)
                 (11, 6) rectangle (29, 24);
    \draw[thick,black] (x) rectangle (y);
    \path (x) node [xy,label=below left:$y$] {}
          (y) node [xy,label=above right:$x$] {};
    \path (10, 15) node [label=below left:{$\sublattice(y,x)$}] {};
    \path (1, 19) node [label=above right:{$B(y,k-1)$}] {};
    \path (11, 24) node [label=above right:{$B(x,k-1)$}] {};
    \foreach \i in {12,13,...,18}
        \foreach \j in {11,12,...,14}
        {
            \path (\i,\j) node[elem,fill=black] {};
        }
    \foreach \i in {12,13,...,18}
        \path (\i,15) node[elem,fill=blue] {}
              (\i,10) node[elem,fill=blue] {};
    \foreach \j in {11,12,...,14}
        \path (11,\j) node[elem,fill=red] {}
              (19,\j) node[elem,fill=red] {};
    \path (11, 10) node[corner] {}
          (11, 15) node[corner] {}
          (19, 10) node[corner] {}
          (19, 15) node[corner] {};
    \end{tikzpicture}
    \caption{An illustration of the region $B(y,k-1) \cap B(x,k-1) \cap \sublattice(y,x)$ 
	from Lemma~\ref{lem:fixed-point-between}. Red points are on the boundary of
	$B(y,k-1)\cap B(x,k-1)$, blue points are on the boundary of $\sublattice(y,x)$ and points that are
	both blue and red are on the boundary of both.}
    \label{fig:regions}
    \end{figure}
}

\newcommand{\regionsNonExpansion}{
	\begin{figure}
		\centering
		\begin{tikzpicture}[
			scale=0.3,
			xy/.style={circle,fill=black,inner sep=0pt,minimum size=0.2cm},
			elem/.style={circle,inner sep=0pt,minimum size=0.2cm},
			corner/.style={circle,draw=none,inner sep=0pt,minimum size=0.2cm,semifill={lower=red,upper=blue,ang=145}}]

			\coordinate (y) at (10, 10);
			\coordinate (x) at (15, 15);

			\filldraw[black] (x) circle (6pt) node[above right] {$x$};
			\filldraw[black] (y) circle (6pt) node[below left] {$y$};
			
			\draw[step=1cm,gray,very thin] (2,3) grid (22,22);
			
			\draw[thick,black] 
			($(y) + (-4,-4)$) rectangle ($(y) + (4,4)$);   
			\draw[thick,black] 
			($(x) + (-4,-4)$) rectangle ($(x) + (4,4)$);  
			
			\path (10, 16) node [label=below left:{$\sublattice(y,x)$}] {};
			\draw[thick,black] (x) rectangle (y);
			
			\path (3,4) node [label=above right:{$B(y,k-1)$}] {};
			\path (15, 18) node [label=above right:{$B(x,k-1)$}] {};
			
			\foreach \i in {12,13}
			\foreach \j in {12,13}
			{
				\path (\i,\j) node[elem,fill=black] {};
			}
			
			\foreach \j in {11,12,13,14}
			\path (11,\j) node[elem,fill=red] {}
			(\j,11) node[elem,fill=red] {}
			(\j,14) node[elem,fill=red] {}
			(14,\j) node[elem,fill=red] {};	
		\end{tikzpicture}
	    \caption{An illustration of the region $B(y,k-1) \cap B(x,k-1) \cap \sublattice(y,x)$ 
		from Lemma~\ref{lem:fixed-point-between} when $y,x$ are at equal distance along all dimensions. Red points correspond to the boundary of the intersection.}
		\label{fig:regions_non_expansion}
	\end{figure}
}

\newcommand{\verification}{
    \begin{figure}[h!]
    \centering
    \begin{tikzpicture}[scale=0.5,
        vtx/.style={inner sep=0pt,circle,fill=black,minimum size=0.15cm}]
    \coordinate (000) at (0, 0);
    \coordinate (100) at (4, 0);
    \coordinate (010) at (0, 4);
    \coordinate (001) at (2, 1);
    \coordinate (110) at (4, 4);
    \coordinate (111) at (6, 5);
    \coordinate (011) at (2, 5);
    \coordinate (101) at (6, 1);
    \draw (000) -- (010)
          (000) -- (100)
          (010) -- (110)
          (010) -- (011)
          (110) -- (111)
          (101) -- (111)
          (100) -- (101)
          (100) -- (110)
          (011) -- (111);
    \draw[dashed] (000) -- (001)
                  (001) -- (011)
                  (001) -- (101);
    
    \path[draw=black,thick,-{Stealth}] (000) -- (100);
    \path[draw=black,thick,-{Stealth}] (100) -- (110);
    \path[draw=black,thick,-{Stealth}] (110) -- (111);
    \path[draw=gray,thick,-{Stealth}] (001) -- ($(001)!.8!(101)$);
    \path[draw=gray,thick,-{Stealth}] (010) -- ($(010)!.8!(110)$);
    \path[draw=gray,thick,-{Stealth}] (011) -- ($(011)!.8!(111)$);
    \path[draw=gray,thick,-{Stealth}] (101) -- ($(101)!.8!(111)$);
    \path (000) node[vtx,label=below left:{$y=v^0$}] {};
    \path (111) node[vtx,label=above right:{$v^3=x$}] {};
    \path (100) node[vtx,label=below right:{$v^1$}] {};
    \path (110) node[vtx,label=below right:{$v^2$}] {};
    \end{tikzpicture}    
    \caption{A verification sequence for $x$, where $(v^i \coloneqq \VSeq(x)_i)_{i=0}^k$. The black arrows are the witnesses to each of the verification sequence points not being fixed, and the gray arrows are the directions of the displacements that must exist at every point on the cube that are implied by the verification sequence displacements.}
    \label{fig:verification sequence}
    \end{figure}
}

\newcommand{\cube}{
    \begin{figure}[h!]
    \centering
    \begin{tikzpicture}[scale=0.5,
        vtx/.style={inner sep=0pt,circle,fill=black,minimum size=0.1cm}]
    \coordinate (000) at (0, 0);
    \coordinate (100) at (4, 0);
    \coordinate (010) at (0, 4);
    \coordinate (001) at (2, 1);
    \coordinate (110) at (4, 4);
    \coordinate (111) at (6, 5);
    \coordinate (011) at (2, 5);
    \coordinate (101) at (6, 1);
    \draw (000) -- (010)
          (000) -- (100)
          (010) -- (110)
          (010) -- (011)
          (110) -- (111)
          (101) -- (111)
          (100) -- (101)
          (100) -- (110)
          (011) -- (111);
    \draw[dashed] (000) -- (001)
                  (001) -- (011)
                  (001) -- (101);
    
    \path[draw=black,thick,-{Stealth}] (001) -- (000);
	\path (001) node[vtx,label={above left:$y$}] {};
    \path (111) node[vtx,label=above right:$x$] {};
	\path (000) node[vtx,label={left = 3mm of node:{$z=y-e_i$}}]{};
    \end{tikzpicture}    
    \caption{An illustration of Lemma~\ref{lem:lower-fixed-point-for-LFP-hereditary}.}
    \label{fig:cube}
    \end{figure}
}

\newcommand{\verificationSequence}{
    \begin{figure}[h!]
    \centering
    \begin{subfigure}{0.4\textwidth}
    \begin{tikzpicture}[scale=0.5,
        vtx/.style={inner sep=0pt,circle,fill=black,minimum size=0.15cm}]
    \coordinate (000) at (0, 0);
    \coordinate (100) at (4, 0);
    \coordinate (010) at (0, 4);
    \coordinate (001) at (2, 1);
    \coordinate (110) at (4, 4);
    \coordinate (111) at (6, 5);
    \coordinate (011) at (2, 5);
    \coordinate (101) at (6, 1);
    \draw (000) -- (010)
          (000) -- (100)
          (010) -- (110)
          (010) -- (011)
          (110) -- (111)
          (101) -- (111)
          (100) -- (101)
          (100) -- (110)
          (011) -- (111);
    \draw[dashed] (000) -- (001)
                  (001) -- (011)
                  (001) -- (101);
    
    \path[draw=black,thick,-{Stealth}] (000) -- (100);
    \path[draw=black,thick,-{Stealth}] (100) -- (110);
    \path[draw=black,thick,-{Stealth}] (110) -- (111);
    \path (000) node[vtx,label=above left:{$v^0$}] {};
    \path (111) node[vtx,label=above right:{$v^3$}] {};
    \path (100) node[vtx,label=above left:{$v^1$}] {};
    \path (110) node[vtx,label=below right:{$v^2$}] {};
    \end{tikzpicture}    
    \caption{A verification sequence, and}
    \end{subfigure}
    \begin{subfigure}{0.4\textwidth}
    \begin{tikzpicture}[scale=0.5,
        vtx/.style={inner sep=0pt,circle,fill=red,minimum size=0.15cm}]
    \coordinate (000) at (0, 0);
    \coordinate (100) at (4, 0);
    \coordinate (010) at (0, 4);
    \coordinate (001) at (2, 1);
    \coordinate (110) at (4, 4);
    \coordinate (111) at (6, 5);
    \coordinate (011) at (2, 5);
    \coordinate (101) at (6, 1);
    \coordinate (020) at (-2, 4);
    \draw (000) -- (010)
          (000) -- (100)
          (010) -- (110)
          (010) -- (011)
          (110) -- (111)
          (101) -- (111)
          (100) -- (101)
          (100) -- (110)
          (011) -- (111);
    \draw[dashed] (000) -- (001)
                  (001) -- (011)
                  (001) -- (101);
    
    \draw[pattern={north west lines},pattern color={blue}] (010) -- (110) -- (111) -- (011) -- cycle;
    
    \path[draw=red,thick,-{Stealth}] (010) -- (110);
    \path[draw=red,thick,-{Stealth}] (110) -- (111);
    ]
    \path (010) node[vtx,label=below left:{$w^0$}] {};
    \path (110) node[vtx,label=below right:{$w^1$}] {};
    \path (111) node[vtx,label=above right:{$w^2$}] {};
    
     \node[above left] at (-2,4) {$y$};
     \draw[gray, dashed] (-2,4)--(0,4);
     \fill (-2, 4) circle (4pt);

    \end{tikzpicture}   
    \caption{its projection to the top face.}
    \end{subfigure}
    \caption{A verification sequence in (a) 
    and its projection in (b), 
    as used in the proof of Lemma~\ref{lem:ov1-comparable-far}.  
	Here $w^0$ has a positive displacement towards 
	$w^1$ because $v^0$ goes to $v^1$ and $w^0$ is directly above $v^0$.
	Note that $v^1$ and $v^2$ are both projected to $w^1$.}
    \label{fig:verif-seq-proj}
    \end{figure}
}

\newcommand{\llnewSlice}{
    \begin{figure}[htbp]
        \centering
        \begin{minipage}{0.48\textwidth}
        \centering
        \begin{tikzpicture}[scale=0.7]
            \draw (2,0) rectangle (6,2);
            \draw (-4,1.2) rectangle (2,1.8);
            \node[above] at (2, 2) {$	q$};
            \fill (2, 2) circle (2pt);
            \node[below] at (2, 0) {$x$};
            \fill (2, 0) circle (2pt);
            \node[below] at (-4, 1.2) {$y$};
            \fill (-4, 1.2) circle (2pt);
            \draw[->, blue, line width=1.5pt] (2,2) -- (1.5,2);
            \fill[pattern=north east lines, pattern color=gray, opacity=0.25] (2,0) -- (6,0) -- (6,2) -- (2,2)  -- cycle; 
        \end{tikzpicture}
        \caption*{(a)}
        \end{minipage}     
        \begin{minipage}{0.48\textwidth}
        \centering
        \begin{tikzpicture}[scale=0.7]
            \draw (2,0) rectangle (6,2);
            \draw (-4,1.2) rectangle (2,1.8);
            \node[above] at (2, 2) {$q$};
            \fill (2, 2) circle (2pt);
            \node[below] at (2, 0) {$x$};
            \fill (2, 0) circle (2pt);
            \node[below] at (-4, 1.2) {$y$};
            \fill (-4, 1.2) circle (2pt);
            \draw[->, blue, line width=1.5pt] (2,2) -- (2.5,2);

            \draw[black, dotted, line width=1pt] (-3,0) -- (2, 0);
            \draw[black, dotted, line width=1pt] (-3,2) -- (2, 2);
            \draw[black, dotted, line width=1pt] (-3,0) -- (-3, 2);
            
            \draw[blue, dotted, line width=1pt] (2,2) -- (0, 0);
           	\node[below, gray] at (1, 0) {\scriptsize $h$};
            \node[below left, gray] at (-0.8, 0) {\scriptsize $h+1$};

        \end{tikzpicture}
        \caption*{(b)}
        \end{minipage}
        
    \caption{A depiction of the scheme for left lobe states.}
    \label{fig:llnewslice}
    \end{figure}
}

\newcommand{\halfspaceFirst}{
    \begin{figure}[htbp]
        \centering
        \begin{minipage}{0.23\textwidth}
        \centering
        \begin{tikzpicture}[scale=0.7]
            \draw (0,0) rectangle (4,4);
            \node[below right] at (2, 2) {$	q$};
            \fill (2, 2) circle (2pt);          
            \draw[blue, opacity=0.3, line width=1pt] (2,2) -- (0, 0);
		    \draw[blue, line width=1pt] (2,2) -- (2, 4);
		    \draw[->, red, line width=1.5pt] (2,2) -- (2,2.5);
            \draw[red, line width=1pt] (2,2) -- (4,2);
            \draw[->, blue, line width=1.5pt] (2,2) -- (2.5,2);
	    \draw[red, opacity=0.3, line width=1pt] (2,2) -- (0,0);
            \fill[pattern=north east lines, pattern color=gray, opacity=0.25] (0,0) -- (4,0) -- (4,2) -- (2,2) -- (2,4) -- (0,4) -- cycle; 
        \end{tikzpicture}
        \caption*{(a)}
        \end{minipage}
	   \begin{minipage}{0.23\textwidth}
        \centering
        \begin{tikzpicture}[scale=0.7]
            \draw (0,0) rectangle (4,4);
            \node[below right] at (2, 2) {$	q$};
            \fill (2, 2) circle (2pt);
            \draw[->, blue, line width=1.5pt] (2,2) -- (1.5,2);
            \draw[blue, line width=1pt] (2,2) -- (4,4);
		    \draw[blue, line width=1pt] (2,2) -- (2, 0);
		    \draw[->, red, line width=1.5pt] (2,2) -- (2,2.5);
            \draw[red, line width=1pt] (2,2) -- (4,2);
	    	\draw[red, line width=1pt] (2,2) -- (0,0);
            \fill[pattern=north east lines, pattern color=gray, opacity=0.25] (0,0) -- (4,4) -- (4,0) -- (0,0) -- cycle; 
        \end{tikzpicture}
        \caption*{(b)}
        \end{minipage}
        \begin{minipage}{0.23\textwidth}
        \centering
        \begin{tikzpicture}[scale=0.7]
            \draw (0,0) rectangle (4,4);
            \node[below right] at (2, 2) {$	q$};
            \fill (2, 2) circle (2pt);
            \draw[->, blue, line width=1.5pt] (2,2) -- (2.5,2);
            \draw[blue, line width=1pt] (2,2) -- (2,4);
		    \draw[blue, line width=1pt] (2,2) -- (0, 0);
		    \draw[->, red, line width=1.5pt] (2,2) -- (2,1.5);
            \draw[red, line width=1pt] (2,2) -- (0,2);
	    	\draw[red, line width=1pt] (2,2) -- (4,4);
            \fill[pattern=north east lines, pattern color=gray, opacity=0.25] (0,0) -- (4,4) -- (0,4) -- (0,0) -- cycle; 
        \end{tikzpicture}
        \caption*{(c)}
        \end{minipage}   
        \begin{minipage}{0.23\textwidth}
        \centering
        \begin{tikzpicture}[scale=0.7]
            \draw (0,0) rectangle (4,4);
            \node[below right] at (2, 2) {$	q$};
            \fill (2, 2) circle (2pt);
            \draw[blue, opacity=0.3, line width=1pt] (2,2) -- (4,4);
		    \draw[blue, line width=1pt] (2,2) -- (2, 0);
		    \draw[->, red, line width=1.5pt] (2,2) -- (2,1.5);
            \draw[red, line width=1pt] (2,2) -- (0,2);
	    	\draw[red, opacity=0.3, line width=1pt] (2,2) -- (4,4);
	    	\draw[->, blue, line width=1.5pt] (2,2) -- (1.5,2);
            \fill[pattern=north east lines, pattern color=gray, opacity=0.25] (0,4) -- (4,4) --(4,0) -- (2,0) -- (2,2) -- (0,2) -- cycle; 
        \end{tikzpicture}
        \caption*{(d)}
        \end{minipage}         
    \caption{Cases 1 through 4 of the proof of Lemma 68.}
    \label{fig:case1-4}
    \end{figure}
}

\newcommand{\halfspaceSecond}{
    \begin{figure}[htbp]
        \centering
        \begin{minipage}{0.3\textwidth}
        \centering
        \begin{tikzpicture}[scale=0.7]
            \draw (0,0) rectangle (4,4);
            \node[below right] at (2, 2) {$	q$};
            \fill (2, 2) circle (2pt);
            \draw[->, blue, line width=1.5pt] (2,2) -- (2.5,2);
            \draw[blue, line width=1pt] (2,2) -- (0, 0);
		    \draw[blue, opacity=0.3, line width=1pt] (2,2) -- (2, 4);
		    \draw[->, green, line width=1.5pt] (2,2) -- (2.3,2.3);
            \draw[green, line width=1pt] (2,2) -- (2,4);
		    \draw[green, opacity=0.3, line width=1pt] (2,2) -- (4,2);
		    \draw[black, dotted] (0,2) -- (2,2);
            \fill[pattern=north east lines, pattern color=gray, opacity=0.25] (0,4) -- (0,2) -- (4,2) -- (4,4) -- cycle; 
            \node[below right] at (2, 2) {$	q$};
        \end{tikzpicture}
        \caption*{(a)}
        \end{minipage}
        \begin{minipage}{0.3\textwidth}
        \centering
        \begin{tikzpicture}[scale=0.7]
            \draw (0,0) rectangle (4,4);
            \fill (2, 2) circle (2pt);
            \draw[->, blue, line width=1.5pt] (2,2) -- (2.5,2);
            \draw[blue, line width=1pt] (2,2) -- (0, 0);
		    \draw[blue, line width=1pt] (2,2) -- (2, 4);
		    \draw[->, green, line width=1.5pt] (2,2) -- (1.7,1.7);
            \draw[green, line width=1pt] (2,2) -- (0,2);
	    	\draw[green, line width=1pt] (2,2) -- (2,0);
            \fill[pattern=north east lines, pattern color=gray, opacity=0.25] (0,0) -- (0,4) -- (2,4) -- (2,0) -- cycle; 
            \node[below right] at (2, 2) {$	q$};
        \end{tikzpicture}
        \caption*{(b)}
        \end{minipage}
       \begin{minipage}{0.3\textwidth}
        \centering
        \begin{tikzpicture}[scale=0.7]
            \draw (0,0) rectangle (4,4);
            \fill (2, 2) circle (2pt);
            \draw[->, blue, line width=1.5pt] (2,2) -- (2.5,2);
            \draw[blue, line width=1pt] (2,2) -- (0, 0);
		    \draw[blue, line width=1pt] (2,2) -- (2, 4);
		    \draw[black, dotted] (2,2) -- (2,0);
		    \node[below right] at (1, 0.65) {\scriptsize $x$};
            \fill (1.2, 0.8) circle (2pt);
            \node[below right] at (2, 0.8) {\scriptsize $x'$};
            \fill (2, 0.8) circle (2pt);
            \draw[->, black] (1.2, 0.8) -- (1.6,0.8);
            \draw[->, black] (1.6, 0.8) -- (2,0.8);
            \draw[->, black] (2, 0.8) -- (2.4,0.8);
             \draw[->, black] (2.4, 0.8) -- (2.8,0.8);
            \draw[->, black] (2.8, 0.8) -- (2.8,1.2);
		    \draw[->, black] (2.8, 1.2) -- (2.8,1.6);
		    \draw[->, black] (2.8, 1.6) -- (2.8,2);
		    \draw[->, black] (2.8, 2) -- (2.8,2.4);
		    \draw[->, black] (2.8, 2.4) -- (2.8,2.8);
		    \draw[->, black] (2.8, 2.8) -- (3.2,2.8);
		    \node[below right] at (3.2, 2.8) {\scriptsize $g$};
	        \fill (3.2, 2.8) circle (2pt);
	        \node[below right] at (2, 2) {$	q$};
		\end{tikzpicture}
        \caption*{(c)}
        \end{minipage}
    \caption{Cases 5 through 7 of the proof of Lemma 68.}
    \label{fig:case5-7}
    \end{figure}
}

\newcommand{\halfspaceThird}{
    \begin{figure}[htbp]
        \centering
        \begin{minipage}{0.3\textwidth}
        \centering
        \begin{tikzpicture}[scale=0.7]
            \draw (0,0) rectangle (4,4);
            \node[below right] at (2, 2) {$	q$};
            \fill (2, 2) circle (2pt);
	    	\draw[->, blue, line width=1.5pt] (2,2) -- (1.5,2);
            \draw[blue, line width=1pt] (2,2) -- (4,4);
	    	\draw[blue, line width=1pt] (2,2) -- (2, 0);
	    	\draw[->, green, line width=1.5pt] (2,2) -- (2.3,2.3);
            \draw[green, line width=1pt] (2,2) -- (2,4);
	    	\draw[green, line width=1pt] (2,2) -- (4,2);
            \fill[pattern=north east lines, pattern color=gray, opacity=0.25] (2,4) -- (4,4) -- (4,0) -- (2,0) -- cycle; 
        \end{tikzpicture}
        \caption*{(a)}
        \end{minipage}
        \begin{minipage}{0.3\textwidth}
        \centering
        \begin{tikzpicture}[scale=0.7]
            \draw (0,0) rectangle (4,4);
            \node[below right] at (2, 2) {$	q$};
            \fill (2, 2) circle (2pt);
	    	\draw[->, blue, line width=1.5pt] (2,2) -- (1.5,2);
            \draw[blue, line width=1pt] (2,2) -- (4,4);
	    	\draw[blue, opacity=0.3, line width=1pt] (2,2) -- (2, 0);
	    	\draw[->, green, line width=1.5pt] (2,2) -- (1.7,1.7);
            \draw[green, line width=1pt] (2,2) -- (0,2);
	    	\draw[green, opacity=0.3, line width=1pt] (2,2) -- (2,0);
	    	\draw[black, dotted] (2,2) -- (4,2);
            \fill[pattern=north east lines, pattern color=gray, opacity=0.25] (0,2) -- (4,2) -- (4,0) -- (0,0) -- cycle; 
        \end{tikzpicture}
        \caption*{(b)}
        \end{minipage}
         \begin{minipage}{0.3\textwidth}
        \centering
        \begin{tikzpicture}[scale=0.7]
            \draw (0,0) rectangle (4,4);
            \node[below right] at (2, 2) {$	q$};
            \fill (2, 2) circle (2pt);
		    \draw[->, blue, line width=1.5pt] (2,2) -- (1.5,2);
            \draw[blue, line width=1pt] (2,2) -- (4,4);
		    \draw[blue, line width=1pt] (2,2) -- (2, 0);
		    \draw[black, dotted] (2,2) -- (2,4);
		    \node[above right] at (2.8, 3.2) {\scriptsize $x$};
            \fill (2.8, 3.2) circle (1.5pt);
           \draw[->, black] (2.8, 3.2) -- (2.4,3.2);
           \draw[->, black] (2.4, 3.2) -- (2,3.2);
           \node[above right] at (2, 3.2) {\scriptsize $x'$};
           \fill (2, 3.2) circle (1.5pt);
           \draw[->, black] (2, 3.2) -- (1.6,3.2);
		   \draw[->, black] (1.6, 3.2) -- (1.2,2.8);
	 	   \draw[->, black] (1.2,2.8) -- (1.2,2.4);
		   \draw[->, black] (1.2,2.4) -- (1.2,2);
		   \draw[->, black] (1.2,2) -- (0.8,1.6);
		   \draw[->, black] (0.8,1.6) -- (0.8,1.2);
		   \node[below right] at (0.8,1.2) {\scriptsize $l$};
           \fill (0.8,1.2) circle (1.5pt);
        \end{tikzpicture}
        \caption*{(c)}
        \end{minipage}
    \caption{Cases 8 through 10 of the proof of Lemma 68.}
    \label{fig:case8-10}
    \end{figure}
}

\newcommand{\halfspaceFourth}{
	\begin{figure}[htbp]
		\centering
		\begin{minipage}{0.4\textwidth}
			\centering
			\begin{tikzpicture}[scale=0.7]
				\draw (0,0) rectangle (4,4);
				\node[below right] at (2, 2) {$	q$};
				\fill (2, 2) circle (2pt);
				\draw[->, green, line width=1.5pt] (2,2) -- (2.3,2.3);
				\draw[green, line width=1pt] (2,2) -- (2,4);
				\draw[green, line width=1pt] (2,2) -- (4,2);
				\draw[black, dotted] (2,2) -- (0,2);
				\node[right] at (2.5, 1.15) {\scriptsize $x$};
				\fill (2.5, 1.2) circle (1.5pt);
				\node[right] at (1.6, 1.2) {\scriptsize $x'$};
				\fill (1.6, 1.2) circle (1.5pt);
				\node[right] at (1.6, -0.8) {$l$};
				\fill (1.6, -0.8) circle (1.5pt);
				\draw[->, black] (1.6,-0.8) -- (1.6,-0.4);
				\draw[->, black] (1.6,-0.4) -- (1.6,0);
				\draw[->, black] (1.6,0) -- (1.6,0.4);
				\draw[->, black] (1.6,0.4) -- (1.6,0.8);
				\draw[->, black] (1.6,0.8) -- (1.6,1.2);
				\draw[->, black] (1.6,1.2) -- (1.6,1.7);
				\draw[->, black] (1.6,1.7) -- (2,2);
				\fill[pattern=north east lines, pattern color=gray, opacity=0.25] (2,4) -- (4,4) -- (4,0) -- (2,0) -- cycle; 
			\end{tikzpicture}
			\caption*{(a)}
		\end{minipage}
		\begin{minipage}{0.4\textwidth}
		\centering
		\begin{tikzpicture}[scale=0.7]
			\draw (0,0) rectangle (4,4);
			\node[below right] at (2, 2) {$	q$};
			\fill (2, 2) circle (2pt);
	    	\draw[->, green, line width=1.5pt] (2,2) -- (1.7,1.7);
			\draw[green, line width=1pt] (2,2) -- (0,2);
			\draw[green, line width=1pt] (2,2) -- (2,0);
			\draw[black, dotted] (2,2) -- (0,2);
			\node[right] at (1.5, 2.74) {\scriptsize $x$};
			\fill (1.5, 2.8) circle (1.5pt);
			\node[right] at (2.4, 2.8) {\scriptsize $x'$};
			\fill (2.4,2.8) circle (1.5pt);
			\node[right] at (2.6,4.2) {$g$};
			\fill (2.6,4.2) circle (1.5pt);
			\draw[->, black] (2,2) -- (2.2,2.2);
			\draw[->, black] (2.2,2.2) -- (2.4,2.4);
			\draw[->, black] (2.4,2.4) -- (2.4,2.8);
			\draw[->, black] (2.4,2.8) -- (2.4,3.2);
			\draw[->, black] (2.4,3.2) -- (2.6,3.4);
			\draw[->, black] (2.6,3.4) -- (2.6,3.8);
			\draw[->, black] (2.6,3.8) -- (2.6,4.2);
			\draw[black, dotted] (2,4) -- (2,2);
			\fill[pattern=north east lines, pattern color=gray, opacity=0.25] (0,0) -- (2,0) -- (2,4) -- (0,4) -- cycle; 
		\end{tikzpicture}
		\vspace{0.5cm}
		\caption*{(b)}
		\end{minipage}
		\caption{Cases 11 and 12 of the proof of Lemma 68.}
		\label{fig:case11-12}
	\end{figure}
}

\newcommand{\shapes}{
	\begin{figure}[htbp]
		\centering
			\begin{minipage}{0.4\textwidth}
			\centering
			\begin{tikzpicture}[scale=0.7]
				\draw (0,0) rectangle (4,4);
				\node[below left] at (0,0) {$x$};
				\fill (0,0) circle (2pt);
				\node[below right] at (2, 2) {$q$};
				\fill (2, 2) circle (2pt);
				\node[right] at (0, 2) {$\Floor{\frac{h}{2}}$};
				\fill (0, 2) circle (1pt);
				\node[below] at (2, 0) {$\Floor{\frac{w}{2}}$};
				\fill (2,0) circle (1pt);
			\end{tikzpicture}
			\caption*{(a)}
			\end{minipage}
			\begin{minipage}{0.4\textwidth}
			\centering
			\begin{tikzpicture}[scale=0.7]
				\draw (0,0) -- (4,0);
				\draw (4,0) -- (4,4);
				\draw (0,0) -- (4,4);
				\node[below left] at (0,0) {$x$};
				\fill (0,0) circle (2pt);
				\draw[dotted] (1,0) -- (4,3);
				\fill (1,0) circle (1pt);
				\node[below] at (1,0) {$\Floor{\frac{k}{4}}$};				
				\fill (2,0) circle (1pt);
				\node[below] at (2,0) {$\Floor{\frac{k}{2}}$};		
				\node[below right] at (2,1) {$q$};
				\fill (2,1) circle (2pt);
			\end{tikzpicture}
			\caption*{(b)}
		\end{minipage}
	
		\begin{minipage}{0.4\textwidth}
		\centering
		\begin{tikzpicture}[scale=0.7]
			\draw (0,0) -- (0,4);
			\draw (0,4) -- (4,4);
			\draw (0,0) -- (4,4);
			\node[below left] at (0,0) {$x$};
			\fill (0,0) circle (2pt);
			\draw[dotted] (0,1) -- (3,4);
			\fill (0,1) circle (1pt);
			\node[left] at (0,1) {$\Floor{\frac{k}{4}}$};				
			\fill (2,4) circle (1pt);
			\node[above] at (2,4) {$\Floor{\frac{k}{2}}$};		
			\node[below right] at (2,3) {$q$};
			\fill (2,3) circle (2pt);
		\end{tikzpicture}
		\caption*{(c)}
		\end{minipage}
		\begin{minipage}{0.4\textwidth}
		\centering
		\begin{tikzpicture}[scale=0.7]
			\draw (0,0) -- (4,4);
			\draw (3,0) -- (7,4);
			\draw (0,0) -- (3,0);
			\draw (4,4) -- (7,4);
			\node[below left] at (0,0) {$x$};
			\fill (0,0) circle (2pt);
			\fill (1.5,0) circle (1pt);
			\node[below] at (1.5,0) {$\Floor{\frac{w}{2}}$};	
			\draw[dotted] (1.5,0) -- (3.5,2);			
			\fill (5,2) circle (1pt);
			\node[right] at (5,2) {$\Floor{\frac{h}{2}}$};		
			\node[below right] at (3.5,2) {$q$};
			\fill (3.5,2) circle (2pt);
			\draw[->, dotted]  (5.4,4.2) -- (4,4.2);
			\node at (5.6,4.2) {$w$};
			\draw[->, dotted] (5.8,4.2) -- (7,4.2);	
			\draw[->, dotted] (7.2,2.2) -- (7.2,4);
			\node at (7.2,2) {$h$};
			\draw[->, dotted] (7.2,1.8) -- (7.2,0);	
		\end{tikzpicture}
		\caption*{(d)}
		\end{minipage}

		\caption{The four basic shapes considered.}
		\label{fig:shapes}
	\end{figure}
}

\newcommand{\shapedecomp}{
	\begin{figure}[htbp]
	\centering
	\begin{minipage}{0.4\textwidth}
		\centering
		\vspace{1.25cm}
		\begin{tikzpicture}[scale=0.6]
			\draw (0,1) rectangle (4,3);
			\draw (0,0) rectangle (3,1);
			\draw (1,3) rectangle (4,4);
			\draw (0,3) -- (1,4);
			\draw (3,0) -- (4,1);
		\end{tikzpicture}
		\caption*{(a)}
	\end{minipage}
	\begin{minipage}{0.4\textwidth}
	\centering
	\begin{tikzpicture}[scale=0.6]
		\draw (0,0) rectangle (2,2);
		\draw (4,4) rectangle (6,6);
		\draw (0,2) -- (4,6);
		\draw (2,0) -- (6,4);
		\draw (2,4) -- (4,4);
		\draw (2,2) -- (4,2);
	\end{tikzpicture}
	\caption*{(b)}
	\end{minipage}

	\caption{Decomposition of valid regions into at most 5 basic shapes.}
	\label{fig:shapedecomp}
	\end{figure}
}

\newcommand{\pgrams}{
	\begin{figure}[htbp]
		\centering
		\begin{minipage}{0.4\textwidth}
			\centering
			\vspace{1.6cm}
			\begin{tikzpicture}[scale=0.6]
				\draw (0,0) -- (2,2);
				\draw (2,2) -- (6,2);
				\draw (0,0) -- (4,0);
				\draw (4,0) -- (6,2);
				\draw[dotted] (2,0) -- (2,2);
				\draw[dotted] (4,0) -- (4,2);
				\fill[pattern=north east lines, pattern color=gray, opacity=0.25] (0,0) -- (2,2) -- (2,0) -- (0,0) -- cycle; 
				\fill[pattern=north east lines, pattern color=gray, opacity=0.25] (4,0) -- (4,2) -- (6,2) -- (4,0) -- cycle; 
			\end{tikzpicture}
			\caption*{(a)}
		\end{minipage}
		\begin{minipage}{0.4\textwidth}
			\centering
			\begin{tikzpicture}[scale=0.7]
				\draw (0,0) -- (4,4);
				\draw (2,0) -- (6,4);
				\draw (4,4) -- (6,4);
				\draw (0,0) -- (2,0);
				\draw[dotted] (2,2) -- (2,0);
				\draw[dotted] (4,4) -- (4,2);
				\fill[pattern=north east lines, pattern color=gray, opacity=0.25] (0,0) -- (2,2) -- (2,0) -- (0,0) -- cycle; 
				\fill[pattern=north east lines, pattern color=gray, opacity=0.25] (4,4) -- (4,2) -- (6,4) -- (4,4) -- cycle; 
			\end{tikzpicture}
			\caption*{(b)}
		\end{minipage}
		\caption{Removing two triangles from a parallelogram.}
		\label{fig:pgrams}
	\end{figure}
}

\newcommand{\diagLobeStepOne}{
	\begin{figure}[htbp]
		\centering
		\begin{minipage}{0.4\textwidth}
			\centering
			\begin{tikzpicture}[scale=0.6]
				\draw (0,0) -- (0,1);
				\draw (0,0) -- (2,0);
				\draw (0,1) -- (7,8);
				\draw (2,0) -- (9,7);
				\fill (5,3) circle (1pt);
				
				\node[below] at (2,0) {$y$};
				\fill (2,0) circle (2pt);
				\node[left] at (0,1) {$x$};
				\fill (0,1) circle (2pt);
				
				\node[right] at (5,3) {\scriptsize $y + (\Floor{\frac{l}{2}},\Floor{\frac{l}{2}})$};
				\fill (3,4) circle (1pt);
				
				\node[left] at (3,4) {\scriptsize $x + (\Floor{\frac{l}{2}},\Floor{\frac{l}{2}})$};
				
				\fill (5,4) circle (2pt);
				\node[right] at (5,4) {$q$};
				
				\draw[->, green, line width=1.5pt] (5,4) -- (4.7, 3.7);	
				
				\draw[green, dotted, line width=1.5pt] (3,4) -- (5,4);
				\draw[green, dotted, line width=1.5pt] (5,3) -- (5,4);
					
				\draw[<->, dotted] (6,7) -- (9,7);	
				\node[color=gray, below] at (7.5,7) {\scriptsize $w$};

				\draw (7.5,8) -- (9.5,10);  
				\draw (8.5,8) -- (10.5,10);  
				\draw (9.5,10) -- (10.5,10);  
				
				\draw[thick, line width = 1.3pt] 	(3,3) -- (5,3)
								(3,3) -- (3,4)
								(3,4) -- (7,8)
								(5,3) -- (9,7)
								(9,8) -- (9,7)
								(7,8) -- (9,8);
								
				\node at (6.5,6) {main};

				\fill[pattern=north east lines, pattern color=gray, opacity=0.25] (0,0) -- (0,1) -- (3,4) -- (5,4) -- (5,3) -- (2,0) -- (0,0) -- cycle; 
			\end{tikzpicture}
			\caption*{(a)}
			\end{minipage}
			\begin{minipage}{0.4\textwidth}
			\centering
			\begin{tikzpicture}[scale=0.6]
				\draw (0,0) -- (0,1);
				\draw (0,0) -- (2,0);
				\draw (0,1) -- (7,8);
				\draw (2,0) -- (9,7);
				\draw (7,8) -- (9,8);
				\draw (9,8) -- (9,7);

				\node[below] at (2,0) {$y$};
				\fill (2,0) circle (2pt);
				\node[left] at (0,1) {$x$};
				\fill (0,1) circle (2pt);
				
				\fill (5,4) circle (2pt);
				\node[below right] at (5,4) {$q$};
				\draw[->, green, line width=1.5pt] (5,4) -- (5.3, 4.3);	
				\draw[green, dotted, line width=1.5pt] (5,4) -- (5,6);
				\draw[green, dotted, line width=1.5pt] (5,4) -- (6,4);
				
				\node[right] at (7,5) {\scriptsize $y + (\Floor{\frac{l}{2}} + 2w,\Floor{\frac{l}{2}} + 2w)$};
				\fill (7,5) circle (1pt);
				
				\node[below left] at (5.2,7) {\scriptsize $x + (\Floor{\frac{l}{2}} + 2w,\Floor{\frac{l}{2}} + 2w)$};
				\fill (5,6) circle (1pt);
				
				\draw[thick, line width = 1.3pt] 	(0,0) -- (0,1)
							(0,0) -- (2,0)
							(2,0) -- (7,5)
							(0,1) -- (5,6)
							(5,6) -- (7,6) 	
							(7,5) -- (7,6);	
							
				\draw[dotted] 	(5.5,6) -- (7.5,8)
								(6.5, 6) -- (8.5,8);
				
				\fill (6.7,6) circle (2pt);
				\node[below left] at (6.7,6) {$q^{\text{mid}}$};
		
				\draw (7.5,8) -- (9.5,10);  
				\draw (8.5,8) -- (10.5,10);  
				\draw (9.5,10) -- (10.5,10);  
				
				\draw[->, red, line width=1.5pt] (6.7,6) -- (6.7, 6.5);
				\draw[->, blue, line width=1.5pt] (6.7,6) -- (6.2, 6);
				
				\draw[blue, dotted, line width=1.5pt] 	(6.7,6) -- (10.7,10)
														(6.7,6) -- (6.7,3);

				\fill[pattern=north east lines, pattern color=gray, opacity=0.25] (6.7, 4.7) -- (6.7, 6) -- (8.7,8) -- (9,8) -- (9,7) -- cycle; 
				
				\node at (3.5,3) {main};

			\end{tikzpicture}
			\caption*{(b)}
			\end{minipage}
			\begin{minipage}{0.4\textwidth}
				\begin{tikzpicture}[scale=0.6]
					\draw (0,0) -- (0,1);
					\draw (0,0) -- (2,0);
					\draw (0,1) -- (7,8);
					\draw (2,0) -- (9,7);
					\draw (7,8) -- (9,8);
					\draw (9,8) -- (9,7);

					\node[below] at (2,0) {$y$};
					\fill (2,0) circle (2pt);
					\node[left] at (0,1) {$x$};
					\fill (0,1) circle (2pt);
					
					\fill (5,4) circle (2pt);
					\node[below right] at (5,4) {$q$};
					\draw[->, green, line width=1.5pt] (5,4) -- (5.3, 4.3);	
					\draw[green, dotted, line width=1.5pt] (5,4) -- (5,6);
					\draw[green, dotted, line width=1.5pt] (5,4) -- (6,4);
					
					\node[right] at (7,5) {\scriptsize $y + (\Floor{\frac{l}{2}} + 2w,\Floor{\frac{l}{2}} + 2w)$};
					\fill (7,5) circle (1pt);
					
					\node[below left] at (5.2,7) {\scriptsize $x + (\Floor{\frac{l}{2}} + 2w,\Floor{\frac{l}{2}} + 2w)$};
					\fill (5,6) circle (1pt);
					
					\draw[thick, line width = 1.3pt] 	(0,0) -- (0,1)
					(0,0) -- (2,0)
					(2,0) -- (7,5)
					(0,1) -- (5,6)
					(5,6) -- (7,6) 	
					(7,5) -- (7,6);	
					
					\draw[dotted] 	(5.5,6) -- (7.5,8)
					(6.5, 6) -- (8.5,8);
					
					\fill (6.7,6) circle (2pt);
					\node[below left] at (6.7,6) {$q^{\text{mid}}$};
					
					\draw (7.5,8) -- (9.5,10);  
					\draw (8.5,8) -- (10.5,10);  
					\draw (9.5,10) -- (10.5,10);  
					
					\draw[->, red, line width=1.5pt] (6.7,6) -- (6.7, 5.5);
					\draw[->, blue, line width=1.5pt] (6.7,6) -- (7.2, 6);
					
					\draw[red, dotted, line width=1.5pt] 	(6.7,6) -- (10.7,10)
					(6.7,6) -- (3,6);
					
					\fill[pattern=north east lines, pattern color=gray, opacity=0.25] (5,6) -- (6.7,6) -- (8.7, 8) -- (8.5,8) -- (10.5,10) -- (9.5,10) -- (7.5,8) -- (7,8)  -- cycle; 
					
					\node at (3.5,3) {main};

				\end{tikzpicture}
			\caption*{(c)}
			\end{minipage}
		\begin{minipage}{0.4\textwidth}
			\begin{tikzpicture}[scale=0.6]
				\draw (0,0) -- (0,1);
				\draw (0,0) -- (2,0);
				\draw (0,1) -- (7,8);
				\draw (2,0) -- (9,7);
				\draw (7,8) -- (9,8);
				\draw (9,8) -- (9,7);

				\node[below] at (2,0) {$y$};
				\fill (2,0) circle (2pt);
				\node[left] at (0,1) {$x$};
				\fill (0,1) circle (2pt);
				
				\fill (5,4) circle (2pt);
				\node[below right] at (5,4) {$q$};
				\draw[->, green, line width=1.5pt] (5,4) -- (5.3, 4.3);	
				\draw[green, dotted, line width=1.5pt] (5,4) -- (5,6);
				\draw[green, dotted, line width=1.5pt] (5,4) -- (6,4);
				
				\node[right] at (7,5) {\scriptsize $y + (\Floor{\frac{l}{2}} + 2w,\Floor{\frac{l}{2}} + 2w)$};
				\fill (7,5) circle (1pt);
				
				\node[below left] at (5.2,7) {\scriptsize $x + (\Floor{\frac{l}{2}} + 2w,\Floor{\frac{l}{2}} + 2w)$};
				\fill (5,6) circle (1pt);
				
				\draw[thick, line width = 1.3pt] 	(0,0) -- (0,1)
				(0,0) -- (2,0)
				(2,0) -- (7,5)
				(0,1) -- (5,6)
				(5,6) -- (7,6) 	
				(7,5) -- (7,6);	
				
				\draw[dotted] 	(5.5,6) -- (7.5,8)
				(6.5, 6) -- (8.5,8);
				
				\fill (6.7,6) circle (2pt);
				\node[below left] at (6.7,6) {$q^{\text{mid}}$};
				
				\draw (7.5,8) -- (9.5,10);  
				\draw (8.5,8) -- (10.5,10);  
				\draw (9.5,10) -- (10.5,10);  
				
				\draw[->, red, line width=1.5pt] (6.7,6) -- (6.7, 5.5);
				\draw[->, blue, line width=1.5pt] (6.7,6) -- (6.2, 6);
				
				\draw[red, dotted, opacity=0.3, line width=1.5pt] 	(6.7,6) -- (10.7,10);
				\draw[red, dotted, line width=1.5pt] (6.7,6) -- (3,6);
				\draw[blue, dotted, opacity=0.3, line width=1.5pt] 	(6.7,6) -- (10.7,10);
				\draw[blue, dotted, line width=1.5pt](6.7,6) -- (6.7,3);
				
				\fill[pattern=north east lines, pattern color=gray, opacity=0.25] (5,6) -- (6.7,6) -- (8.7, 8) -- (8.5,8) -- (10.5,10) -- (9.5,10) -- (7.5,8) -- (7,8)  -- cycle; 
				\fill[pattern=north east lines, pattern color=gray, opacity=0.25] (6.7, 4.7) -- (6.7, 6) -- (8.7,8) -- (9,8) -- (9,7) -- cycle; 
				
				\node at (3.5,3) {main};
			\end{tikzpicture}
			\caption*{(d)}
		\end{minipage}
		\caption{Diagonal Lobe State: Step 1 - Cases 1 and 2.1-2.3.}
		\label{fig:diagLobeStep1}
	\end{figure}
}

\newcommand{\diagLobeStepTwo}{
	\begin{figure}[htbp]
		\centering
		\begin{subfigure}{0.3\textwidth}
			\centering
			\begin{tikzpicture}[scale=0.7]
				\draw (0,1) -- (3,4);
				\draw (2,0) -- (5,3);
				\draw (3,4) -- (5,4);
				\draw (5,4) -- (5,3);
				\draw[dotted] 	(1,2) -- (3,2)
								(3,2) -- (3,1)
								(1.5,2) -- (3.5,4)
								(2.5,2) -- (4.5,4);
									
				\draw (3.5,4) -- (5.5,6);  
				\draw (4.5,4) -- (6.5,6);  
				\draw (5.5,6) -- (6.5,6);  
				
				\fill (1.5,2) circle (2pt);
				\node[below left] at (1.5,2) {\scriptsize $q^{\text{left}}$};
				\fill (2.5,2) circle (2pt);
				\node[below] at (2.5,2) {\scriptsize $q^{\text{right}}$};
				
				\draw[->, red, line width=1.5pt] (1.5,2) -- (1.5, 1.7);
				\draw[red, dotted, line width=1.5pt] 	(1.5,2) -- (6,6.5)
														(1.5,2) -- (0.5,2);
														
				\fill[pattern=north east lines, pattern color=gray, opacity=0.25] (1,2) -- (1.5,2) -- (3.5,4) -- (3,4) -- cycle; 
			\end{tikzpicture}
			\caption*{(a)}
		\end{subfigure}
		\begin{subfigure}{0.3\textwidth}
			\centering
			\begin{tikzpicture}[scale=0.7]
				\draw (0,1) -- (3,4);
				\draw (2,0) -- (5,3);
				\draw (3,4) -- (5,4);
				\draw (5,4) -- (5,3);
				\draw[dotted] 	(1,2) -- (3,2)
								(3,2) -- (3,1)
								(1.5,2) -- (3.5,4)
								(2.5,2) -- (4.5,4);
				
				\draw (3.5,4) -- (5.5,6);  
				\draw (4.5,4) -- (6.5,6);  
				\draw (5.5,6) -- (6.5,6);  
				
				\fill (1.5,2) circle (2pt);
				\node[below left] at (1.5,2) {\scriptsize $q^{\text{left}}$};
				\fill (2.5,2) circle (2pt);
				\node[below left] at (2.5,2) {\scriptsize $q^{\text{right}}$};
				
				\draw[->, blue, line width=1.5pt] (2.5,2) -- (2.2,2);
				\draw[blue, dotted, line width=1.5pt] 	(2.5,2) -- (7,6.5)
														(2.5,2) -- (2.5,0);

				\fill[pattern=north east lines, pattern color=gray, opacity=0.25] (2.5,0.5) -- (2.5,2) -- (4.5,4) -- (5,4) -- (5,3) -- cycle; 
			\end{tikzpicture}
			\caption*{(b)}
		\end{subfigure}
			\begin{subfigure}{0.3\textwidth}
			\centering
			\begin{tikzpicture}[scale=0.7]
				\draw (0,1) -- (3,4);
				\draw (2,0) -- (5,3);
				\draw (3,4) -- (5,4);
				\draw (5,4) -- (5,3);
				\draw[dotted] 	(1,2) -- (3,2)
				(3,2) -- (3,1)
				(1.5,2) -- (3.5,4)
				(2.5,2) -- (4.5,4);
				
				\draw (3.5,4) -- (5.5,6);  
				\draw (4.5,4) -- (6.5,6);  
				\draw (5.5,6) -- (6.5,6);  
				
				\fill (1.5,2) circle (2pt);
				\node[below left] at (1.5,2) {\scriptsize $q^{\text{left}}$};
				\fill (2.5,2) circle (2pt);
				\node[below left] at (2.5,2) {\scriptsize $q^{\text{right}}$};
				
				\draw[->, blue, line width=1.5pt] (2.5,2) -- (2.2,2);
				\draw[->, blue, line width=1.5pt] (1.5,2) -- (1.8,2);
				\draw[blue, dotted, line width=1.5pt] 	(2.5,2) -- (7,6.5)
				(2.5,2) -- (2.5,0);
				
				\draw[->, red, line width=1.5pt] (1.5,2) -- (1.5, 1.7);
				\draw[->, red, line width=1.5pt] (2.5,2) -- (2.5,2.3);
				\draw[red, dotted, line width=1.5pt] 	(1.5,2) -- (6,6.5)
				(1.5,2) -- (0.5,2);
				
				\fill[pattern=north east lines, pattern color=gray, opacity=0.25] (2.5,0.5) -- (2.5,2) -- (4.5,4) -- (5,4) -- (5,3) -- cycle; 
				
				\fill[pattern=north east lines, pattern color=gray, opacity=0.25] (1,2) -- (1.5,2) -- (3.5,4) -- (3,4) -- cycle; 
			\end{tikzpicture}
			\caption*{(c)}
		\end{subfigure}
	
		\caption{Diagonal Lobe State: Step 2 - Cases 1,2 and 5.}
		\label{fig:diagLobeStep2}
		\end{figure}
}

\newcommand{\diagLobeStepThree}{
	\begin{figure}[htbp]
		\centering
		\begin{subfigure}{0.4\textwidth}
			\centering
			\begin{tikzpicture}[scale=0.8]
				\draw (0,0) -- (3,3);
				\draw (2,0) -- (5,3);
				\draw (3,3) -- (5,3);
				
				\fill (0,0) circle (2pt);
				\node[below left] at (0,0) {\scriptsize $q^{\text{left}}$};
				\fill (2,0) circle (2pt);
				\node[below] at (2,0) {\scriptsize $q^{\text{right}}$};
				\fill (1,0) circle (2pt);
				\node[below left] at (1,0) {\scriptsize $q^{\text{cen}}$};
				
				\draw[->, red, line width=1.5pt] (1,0) -- (1, -0.7);
				\draw[red, dotted, line width=1.5pt] 	(1,0) -- (4.3,3.3)
				(1,0) -- (-0.4,0);
				
				\fill[pattern=north east lines, pattern color=gray, opacity=0.25] (0,0) -- (3,3) -- (4,3) -- (1,0) -- cycle; 
			\end{tikzpicture}
			\caption*{(a)}
		\end{subfigure}
		\begin{subfigure}{0.4\textwidth}
			\centering
			\begin{tikzpicture}[scale=0.8]
				\draw (0,0) -- (3,3);
				\draw (2,0) -- (5,3);
				\draw (3,3) -- (5,3);
				
				\fill (0,0) circle (2pt);
				\node[below left] at (0,0) {\scriptsize $q^{\text{left}}$};
				\fill (2,0) circle (2pt);
				\node[below] at (2,0) {\scriptsize $q^{\text{right}}$};
				\fill (1,0) circle (2pt);
				\node[below left] at (1,0) {\scriptsize $q^{\text{cen}}$};
				
				\draw[->, blue, line width=1.5pt] (1,0) -- (0.6, 0);
				\draw[blue, dotted, line width=1.5pt] 	(1,0) -- (4.3,3.3)
				(1,0) -- (2.4,0);
				
				\fill[pattern=north east lines, pattern color=gray, opacity=0.25] (2,0) -- (5,3) -- (4,3) -- (1,0) -- cycle; 
			\end{tikzpicture}
			\caption*{(b)}
		\end{subfigure}
		
\caption{Diagonal Lobe State: Step 3 - Cases 1 and 2.}
\label{fig:diagLobeStep3}
\end{figure}
}

\newcommand{\twod}{
\begin{figure}[htbp]
		\centering
				
		\begin{subfigure}[t]{.4\textwidth}
			\centering
			\begin{tikzpicture}[scale=0.6]
				\draw (0,0) rectangle (4,4);
				
				\draw[red, line width=1pt] 	(0,1) -- (1,2)
											(1,2) -- (2.5,2)
											(2.5,2) -- (3.5,3)
											(3.5,3) -- (4,3);
				\draw[blue, line width=1pt] (1.5,0) -- (1.5,0.5)
											(1.5,0.5) -- (2.5,1.5)
											(2.5,1.5) -- (2.5,3)
											(2.5,3) -- (3.5,4);							

			\end{tikzpicture}
		\end{subfigure}
		
\caption{Surfaces in a two-dimensional \DMAC instance.}
\label{fig:twod}
\end{figure}
}

\newcommand{\bsidea}{
\begin{figure}[h!]
	\centering
	\begin{minipage}{0.4\textwidth}
		\centering
		\vspace{1.4cm}
			\begin{tikzpicture}[scale=0.6]
				\draw (0,0) rectangle (4,4);
				
				\draw[red, line width=1pt] 	(0,1) -- (1,1)
											(1,1) -- (3,3)
											(3,3) -- (4,3);
				\draw[blue, line width=1pt] (1.2,0) -- (1.2,0.5)
											(1.2,0.5) -- (1.8,1.1)
											(1.8,1.1) -- (1.8,1.4)
											(1.8,1.4) -- (2.8,2.4)
											(2.8,2.4) -- (2.8,4);	
																	
				\draw[line width=1pt] (1.4, 0.4) -- (1.4, 1.8);
				\draw[line width=1pt] (2.4, 1.4) -- (2.4, 2.8);
				\draw[dotted] 	(1.4, 0.4) -- (2.4, 1.4);
				\draw[dotted] 	(1.4, 1.8) -- (2.4, 2.8);
				
				\fill (1.4, 0.4)  circle (1pt);
				\node[right] at (1.4, 0.4) {\scriptsize $a_i$};
				\fill (1.4, 1.8)  circle (1pt);
				\node[above left] at (1.4, 1.8) {\scriptsize $b_i$};
				\fill (2.4, 1.4)  circle (1pt);
				\node[below right] at (2.4, 1.4) {\scriptsize $a_j$};
				\fill (2.4, 2.8)  circle (1pt);
				\node[above left] at (2.4, 2.8) {\scriptsize $b_j$};
				\fill (1.4, 1.1)  circle (1.5pt);
				\node[below left] at (1.4, 1.1) {\scriptsize $p$};
			\end{tikzpicture}
			\caption*{(a)}
		\end{minipage}
	\begin{minipage}{0.4\textwidth}
			\centering
			\begin{tikzpicture}[scale=0.6]
				\draw (0,0) rectangle (4,4);
				\draw[dotted] (0,-1.2) rectangle (4,0);
				\draw[dotted] (0,4) rectangle (4,5.2);			
				
				\draw[red, line width=1pt] 	(0,1) -- (1,1)
											(1,1) -- (3,3)
											(3,3) -- (4,3);
				\draw[blue, line width=1pt] (1.2,0) -- (1.2,0.5)
											(1.2,0.5) -- (1.8,1.1)
											(1.8,1.1) -- (1.8,1.4)
											(1.8,1.4) -- (2.8,2.4)
											(2.8,2.4) -- (2.8,4);	
				\draw[blue, line width=0.5pt]	(1.2,0) -- (0,-1.2)
																(2.8,4) -- (4,5.2);													
			\end{tikzpicture}
			\caption*{(b)}
		\end{minipage}
\caption{Carrying over information from one inner binary search to the next.}
\label{fig:bsidea}
\end{figure}
}

\newcommand{\twoDslice}{
\begin{figure}[htbp]
		\centering
			\begin{subfigure}[t]{.4\textwidth}
			\centering
			\begin{tikzpicture}[scale=0.8]
				\draw (0,0) rectangle (4,4);
				
				\draw[red, line width=1pt] 	(0,2.5) -- (2,2.5)
											(2,2.5) -- (3,3.5)
											(3,3.5) -- (4,3.5);
				\draw[blue, line width=1pt]  (2.2,0) -- (2.2,1.8)
											(2.2,1.8) -- (3.4,3)
											(3.4,3) -- (3.4,4);
				\draw[green, line width=1pt] 	(0,2) -- (0.5,2)
										(0.5,2) -- (0.5,1.5)
										(0.5,1.5) -- (1.2,1.5)
										(1.2,1.5) -- (1.2,1) 
										(1.2,1) -- (1.8,1)
										(1.8,1) -- (1.8,0);
										
				\node at (1.5,2) {\scriptsize U};												

			\end{tikzpicture}
			\caption*{(a)}
		\end{subfigure}
		\begin{subfigure}[t]{.4\textwidth}
			\centering
			\begin{tikzpicture}[scale=0.8]
				\draw (0,0) rectangle (4,4);
				
				\draw[red, line width=1pt] 	(0,2.5) -- (2,2.5)
											(2,2.5) -- (3,3.5)
											(3,3.5) -- (4,3.5);
				\draw[blue, line width=1pt]  (2.2,0) -- (2.2,1.8)
											(2.2,1.8) -- (3.4,3)
											(3.4,3) -- (3.4,4);
				\draw[green, line width=1pt] 	(0,2) -- (0.5,2)
										(0.5,2) -- (0.5,1.5)
										(0.5,1.5) -- (1.2,1.5)
										(1.2,1.5) -- (1.2,1) 
										(1.2,1) -- (1.8,1)
										(1.8,1) -- (1.8,0);
										
				\draw (1.2,1.5) rectangle (2.2, 2.5);
				\node at (1.7,2) {\scriptsize C};	
				
				\node at (0.2,2.25) {\scriptsize L};	
				\node at (1.5,0.5) {\scriptsize B};	
				\node at (2.7,2.7) {\scriptsize D};

				\draw[dotted] 	(1.2,1.5) -- (0,1.5)
							(1.2,2.5) -- (0,2.5)	
							(1.2,2.5) -- (2.7,4)
							(1.2,1.5) -- (1.2,0)
							(2.2,1.5) -- (4,3.3);								

			\end{tikzpicture}
			\caption*{(b)}
		\end{subfigure}

		\caption{The structure of $\Up(f)$ in a two-dimensional slice.}
\label{fig:2dslice}
\end{figure}
}

\newcommand{\cbConfigIntro}{
\begin{figure}[htbp]
		\centering
		\begin{subfigure}[t]{.4\textwidth}
			\centering
			\begin{tikzpicture}[scale=0.6]
     				\draw (0,0) rectangle (6,6);
				\draw[line width=1pt, opacity=0.5] (2,2) rectangle (4,4);
				\draw[->,red, line width=1pt] 	(2,2) -- (2,2.4);
				\draw[->,red, line width=1pt] 	(2,4) -- (2,3.6);
				\draw[->,blue, line width=1pt] 	(2,2) -- (2.4,2);
				\draw[->,blue, line width=1pt]	(4,2) -- (3.6,2);	
				\draw[->,green, line width=1pt] 	(2,2) -- (1.7,1.7);
				\draw[->,green, line width=1pt]	(4,4) -- (4.3,4.3);											

        				\foreach \x in {1,...,5} {
            				\foreach \y in {1,...,5} {
                					\fill[color=black] (1*\y, 1*\x) circle (1pt);
            				}
        				}
			\end{tikzpicture}
		\caption*{(a)}
		\end{subfigure}
			\begin{subfigure}[t]{.4\textwidth}
			\centering
			\begin{tikzpicture}[scale=0.6]
				\draw (0,0) rectangle (6.3,6);
				\draw (1.8,2.8) rectangle (3,4);			
				\draw (2.2,3.2) rectangle (2.8,3.8);		
				\draw[->, red] 	(2.2,3.2) -- (2.2,3.4);
				\draw[->, red] 	(2.2,3.8) -- (2.2,3.6);
				\draw[->, blue] 	(2.2,3.2) -- (2.4,3.2);
				\draw[->, blue] 	(2.8,3.2) -- (2.6,3.2);
				\draw[->, green] (2.15,3.15) -- (2,3);
				\draw[->, green] (2.85,3.85) -- (3,4);
				\draw[dotted] 	(1.8,2.8) -- (1.8,0);
				\draw[dotted]	(1.8,2.8) -- (0,2.8);
				\draw[dotted]	(1.8,4) -- (0,4);
				\draw[dotted]	(3,2.8) -- (3,0);
				\draw[dotted]	(1.8,4) -- (3.8,6);
				\draw[dotted]	(3,2.8) -- (6.2,6);
				\fill[pattern=north east lines, pattern color=gray, opacity=0.25] (0,0) -- (0,2.8) -- (1.8,2.8) -- (1.8,0) -- cycle; 
				\fill[pattern=north east lines, pattern color=gray, opacity=0.25] (0,6) -- (0,4) -- (1.8,4) -- (3.8,6) -- cycle; 
				\fill[pattern=north east lines, pattern color=gray, opacity=0.25] (3,0) -- (3,2.8) -- (6.2,6) -- (6.3,6) -- (6.3,0) -- cycle; 
			\end{tikzpicture}
			\caption*{(b)}
		\end{subfigure}

\caption{A CB-config, and the corresponding region that is ruled out.}
\label{fig:cbconfig}
\end{figure}
}

\newcommand{\inlobe}{
	\begin{figure}[htbp]
		\centering
			\begin{tikzpicture}[scale=0.4]
				\draw (2,2) rectangle (4,4);
				\draw (-2,2) rectangle (2,4);
				\draw (2,-2) rectangle (4,2);
				\draw	(2,4) -- (5,7)
						(4,2) -- (7,5)
						(5,7) -- (7,5);
				\draw (-0.4,3.8) -- (-0.9,3.3) -- (-1.7, 3.3) -- (-1.7,3) -- (-0.9, 3) -- (-0.9, 2.6) -- (-0.5,2.6) -- (-0.5,2.2) -- (-0.3,2.2) -- (-0.3,3) -- (0.1,3.4) -- (-0.4,3.8);
				\node at (-0.6,3) {\scriptsize U};
				\node[above left] at (2,4) {\scriptsize p};
				\node[below right] at (4,2) {\scriptsize q};
				\node[above right] at (4,4) {\scriptsize r};
				\fill (2,4) circle (2pt);
				\fill (4,2) circle (2pt);
				\fill (4,4) circle (2pt);
				\draw[->,blue, line width=1pt] (2,4) -- (1.4, 4); 
			\end{tikzpicture} 		
		\caption{An example in which $\Up(f) \cap s_i$ lies in a lobe.}
		\label{fig:inlobe}
	\end{figure}
}

\newcommand{\leftlobe}{
	\begin{figure}[htbp]
		\begin{subfigure}[t]{.45\textwidth}
		\centering
		\begin{tikzpicture}[scale=0.6]
			\draw (0,0) rectangle (12,2);
			\node[above] at (6,0) {\scriptsize q};
			\fill (6,0) circle (2pt);
			\draw[dotted] (6,0) -- (8,2);
			\draw (8,2) -- (8,0);
			\draw[->,blue, line width=1pt] (6,0) -- (5.5, 0);
			\draw[<->,gray] (6,-0.2) -- (8, -0.2);
			\node at (7,-0.5) {\scriptsize h};
			\fill[pattern=north east lines, pattern color=gray, opacity=0.25] (8,0) -- (12,0) -- (12,2) -- (8,2) -- cycle; 
		\end{tikzpicture} 		
		\caption*{(a)}
		\end{subfigure}
	\begin{subfigure}[t]{.45\textwidth}
		\centering
		\begin{tikzpicture}[scale=0.6]
			\draw (0,0) rectangle (12,2);
			\node[above] at (6,0) {\scriptsize q};
			\fill (6,0) circle (2pt);
			\draw[dotted] (6,0) -- (4,2);
			\draw[->,blue, line width=1pt] (6,0) -- (6.5, 0);
			\draw[<->,gray] (6,-0.2) -- (4, -0.2);
			\node at (5,-0.5) {\scriptsize h};
			\node[right] at (1,1) {\scriptsize p};
			\fill (1,1) circle (2pt);
			\draw (1,2) -- (1,0);
			\draw[->,red, line width=1pt] (1,1) -- (1,0.5);
			\draw[<->,gray] (1,-0.2) -- (4, -0.2); 
			\node at (2.2,-0.5) {\scriptsize {h+1}};
			\draw (0,1) -- (1,1);
			\node at (1,2.2) {\scriptsize L};
			\fill[pattern=north east lines, pattern color=gray, opacity=0.25] (0,1) -- (1,1) -- (1,2) -- (0,2) -- cycle; 
		\end{tikzpicture} 		
		\caption*{(b)}
	\end{subfigure}
	\begin{subfigure}[t]{.9\textwidth}
		\centering
		\begin{tikzpicture}[scale=0.7]
			\draw (0,0) rectangle (8,2);
			\draw (-4,1.2) rectangle (0,2);
			\node[above] at (4,0) {\scriptsize q};
			\fill (4,0) circle (2pt);
			\fill (0,1.6) circle (2pt);
			\draw[dotted] (4,0) -- (6,2);
			\draw[dotted] (-4,1.6) -- (0,1.6);
			\draw (6,2) -- (6,0);
			\draw[->,blue, line width=1pt] (4,0) -- (3.5, 0);
			\draw[->,green, line width=1pt] (0,1.6) -- (0,1.2);
			\fill[pattern=north east lines, pattern color=gray, opacity=0.25] (6,0) -- (8,0) -- (8,2) -- (6,2) -- cycle; 
			\fill[pattern=north east lines, pattern color=gray, opacity=0.25] (-4,1.2) -- (-4,1.6) -- (0,1.6) -- (0,1.2) -- cycle;
			\node at (-2, 1.8) {\scriptsize sub};
			\node at (3, 1) {\scriptsize main};
		\end{tikzpicture} 		
		\caption*{(c)}
	\end{subfigure}
	\caption{Handling left lobes.}
	\label{fig:leftlobe}
	\end{figure}
}

\newcommand{\usbound}{
	\begin{figure}[htbp]
		\centering
		\begin{tikzpicture}[scale=0.5]
			\draw (2,2) rectangle (5,5);
			\draw (-2,3) rectangle (2,4);
			\draw (3,-2) rectangle (4,2);
			\draw	(4,5) -- (6,7)
					(5,4) -- (7,6)
					(6,7) -- (7,6); 
			\draw	(6.3,6.7) -- (8.6,9) 
					(6.7,6.3) -- (9,8.6)
					(9,8.6) -- (8.6,9); 
			\draw (-4,3.3) rectangle (-2,3.6); 
			\draw (3.3,-4) rectangle (3.6,-2); 
		\end{tikzpicture} 		
		\caption{The shape of the regions that we consider.}
		\label{fig:usbound}
	\end{figure}
}

\newcommand{\fsnewslice}{
	\begin{figure}[htbp]
		\centering
		\begin{tikzpicture}[scale=0.5]
			\draw (2,2) rectangle (5,5);
			\draw (-2,3) rectangle (2,4);
			\draw (3,-2) rectangle (4,2);
			\draw	(3.5,5) -- (7,8.5)
			(5,4.5) -- (8.5,8)
			(7,8.5) -- (8.5,8.5)
			(8.5,8.5) -- (8.5,8); 
			\draw	(7.3,8.5) -- (8.8,10) 
			(7.8,8.5) -- (9.3,10)
			(8.8,10) -- (9.3,10); 
			\draw (-4,3.3) rectangle (-2,3.6); 
			\draw (3.3,-4) rectangle (3.6,-2); 
			\node[above] at (2,5) {\scriptsize $q^{\text{tl}}$};
			\fill (2,5) circle (2pt);
			\draw[->,blue, line width=1pt] (2,5) -- (2.5,5);
			\node[above left] at (5,5) {\scriptsize $q^{\text{tr}}$};
			\fill (5,5) circle (2pt);
			\draw[->,green, line width=1pt] (5,5) -- (5.3,5.3);
			\node[below] at (5,2) {\scriptsize $q^{\text{br}}$};
			\fill (5,2) circle (2pt);
			\draw[->,red, line width=1pt] (5,2) -- (5,2.5);
		\end{tikzpicture} 		
		\caption{The shape of the regions that we consider.}
		\label{fig:fsnewslice}
	\end{figure}
}

\newcommand{\contractionViolation}{
	\begin{figure}[htbp]
		\centering
			\begin{tikzpicture}[scale=0.75]
				\foreach \x in {0,...,2} {
					\node at (0, \x) {};
					\node at (2, \x) {};
					\fill[color=black] (0, \x) circle (2pt);
					\fill[color=black] (2, \x) circle (2pt);
				}
				\node[left] at (-1,2) {\scriptsize Case 1:};
				\node[left] at (-1,1) {\scriptsize Case 2:};
				\node[left] at (-1,0) {\scriptsize Case 3:};
				\draw[->] (2,2) -- (2.5,2);
				\draw[->] (0,1) -- (-0.5,1);
				\draw[->] (2,0) -- (2.5,0);
				\draw[->] (0,0) -- (-0.5,0);
			\end{tikzpicture}
	\caption{Cases consituting contraction violations.}
	\label{fig:contractionviolation}
	\end{figure}
}

\newcommand{\interpolationSlices}{
	\begin{figure}[htbp]
		\centering
		\begin{minipage}{0.3\textwidth}
		\centering
		\begin{tikzpicture}[scale=0.3]
			\foreach \x in {1,2,3} {
				\foreach \y in {1,2,3} {
					\fill (\x, \x+\y) circle (2pt);  
				}
			}
			\draw (0,0) -- (0,4);
			\draw (0,4) -- (4,8);
			\draw (4,8) -- (4,4);
			\draw (4,4) -- (0,0);
			
			\draw[red, thick] (0,0)--(4,4);
		\end{tikzpicture}
		\caption*{\small \textit{Case (1):} $\{0, \dots, 1\}$}
		\end{minipage}
		\begin{minipage}{0.3\textwidth}
			\centering
			\begin{tikzpicture}[scale=0.3]
				\foreach \x in {1,2,3} {
					\foreach \y in {1,2,3} {
						\fill (\x, \x+\y) circle (2pt);  
					}
				}
				\draw (0,0) -- (0,4);
				\draw (0,4) -- (4,8);
				\draw (4,8) -- (4,4);
				\draw (4,4) -- (0,0);
				
				\draw[red, thick] (0,0)--(4,8);
			\end{tikzpicture}
			\caption*{\small \textit{Case (4):} $\{0, \dots, 1\}$}
		\end{minipage}
		
		\begin{minipage}{0.15\textwidth}
			\centering
			\begin{tikzpicture}[scale=0.3]
				\foreach \x in {1,2,3} {
					\foreach \y in {1,2,3} {
						\fill (\x, \x+\y) circle (2pt);  
					}
				}
				\draw (0,0) -- (0,4);
				\draw (0,4) -- (4,8);
				\draw (4,8) -- (4,4);
				\draw (4,4) -- (0,0);
				
				\draw[red, thick] (0,0)--(4,4);
			\end{tikzpicture}
			\caption*{\small \textit{Case (2):} $\{0\}$}
		\end{minipage}
		\begin{minipage}{0.15\textwidth}
			\centering
			\begin{tikzpicture}[scale=0.3]
				\foreach \x in {1,2,3} {
					\foreach \y in {1,2,3} {
						\fill (\x, \x+\y) circle (2pt);  
					}
				}
				\draw (0,0) -- (0,4);
				\draw (0,4) -- (4,8);
				\draw (4,8) -- (4,4);
				\draw (4,4) -- (0,0);
				
				\draw[red, thick] (0,1)--(4,5);
			\end{tikzpicture}
			\caption*{\small $\{1/4\}$}
		\end{minipage}
		\begin{minipage}{0.15\textwidth}
			\centering
			\begin{tikzpicture}[scale=0.3]
				\foreach \x in {1,2,3} {
					\foreach \y in {1,2,3} {
						\fill (\x, \x+\y) circle (2pt);  
					}
				}
				\draw (0,0) -- (0,4);
				\draw (0,4) -- (4,8);
				\draw (4,8) -- (4,4);
				\draw (4,4) -- (0,0);
				
				\draw[red, thick] (0,2)--(4,6);
			\end{tikzpicture}
			\caption*{\small $\{1/2\}$}
		\end{minipage}
		\begin{minipage}{0.15\textwidth}
			\centering
			\begin{tikzpicture}[scale=0.3]
				\foreach \x in {1,2,3} {
					\foreach \y in {1,2,3} {
						\fill (\x, \x+\y) circle (2pt);  
					}
				}
				\draw (0,0) -- (0,4);
				\draw (0,4) -- (4,8);
				\draw (4,8) -- (4,4);
				\draw (4,4) -- (0,0);
				
				\draw[red, thick] (0,3)--(4,7);
			\end{tikzpicture}
			\caption*{\small $\{3/4\}$}
		\end{minipage}
		\begin{minipage}{0.15\textwidth}
			\centering
			\begin{tikzpicture}[scale=0.3]
				\foreach \x in {1,2,3} {
					\foreach \y in {1,2,3} {
						\fill (\x, \x+\y) circle (2pt);  
					}
				}
				\draw (0,0) -- (0,4);
				\draw (0,4) -- (4,8);
				\draw (4,8) -- (4,4);
				\draw (4,4) -- (0,0);
				
				\draw[red, thick] (0,4)--(4,8);
			\end{tikzpicture}
			\caption*{\small $\{1\}$}
		\end{minipage}

		\begin{minipage}{0.15\textwidth}
			\centering
			\begin{tikzpicture}[scale=0.3]
				\foreach \x in {1,2,3} {
					\foreach \y in {1,2,3} {
						\fill (\x, \x+\y) circle (2pt);  
					}
				}
				\draw (0,0) -- (0,4);
				\draw (0,4) -- (4,8);
				\draw (4,8) -- (4,4);
				\draw (4,4) -- (0,0);
				
				\draw[red, thick] (0,0)--(4,8);
			\end{tikzpicture}
			\caption*{\small \textit{Case (3):} $\{0\}$}
		\end{minipage}
		\begin{minipage}{0.15\textwidth}
			\centering
			\begin{tikzpicture}[scale=0.3]
				\foreach \x in {1,2,3} {
					\foreach \y in {1,2,3} {
						\fill (\x, \x+\y) circle (2pt);  
					}
				}
				\draw (0,0) -- (0,4);
				\draw (0,4) -- (4,8);
				\draw (4,8) -- (4,4);
				\draw (4,4) -- (0,0);
				
				\draw[red, thick] (1,2)--(4,8);
				\draw[red, thick] (0,1)--(1,2);
			\end{tikzpicture}
			\caption*{\small $\{1/4\}$}
		\end{minipage}
		\begin{minipage}{0.15\textwidth}
			\centering
			\begin{tikzpicture}[scale=0.3]
				\foreach \x in {1,2,3} {
					\foreach \y in {1,2,3} {
						\fill (\x, \x+\y) circle (2pt);  
					}
				}
				\draw (0,0) -- (0,4);
				\draw (0,4) -- (4,8);
				\draw (4,8) -- (4,4);
				\draw (4,4) -- (0,0);
				
				\draw[red, thick] (0,2)--(2,4);
				\draw[red, thick] (2,4)--(4,8);
			\end{tikzpicture}
			\caption*{\small $\{1/2\}$}
		\end{minipage}
		\begin{minipage}{0.15\textwidth}
			\centering
			\begin{tikzpicture}[scale=0.3]
				\foreach \x in {1,2,3} {
					\foreach \y in {1,2,3} {
						\fill (\x, \x+\y) circle (2pt);  
					}
				}
				\draw (0,0) -- (0,4);
				\draw (0,4) -- (4,8);
				\draw (4,8) -- (4,4);
				\draw (4,4) -- (0,0);
				
				\draw[red, thick] (0,3)--(3,6);
				\draw[red, thick] (3,6)--(4,8);
			\end{tikzpicture}
			\caption*{\small $\{3/4\}$}
		\end{minipage}
		\begin{minipage}{0.15\textwidth}
			\centering
			\begin{tikzpicture}[scale=0.3]
				\foreach \x in {1,2,3} {
					\foreach \y in {1,2,3} {
						\fill (\x, \x+\y) circle (2pt);  
					}
				}
				\draw (0,0) -- (0,4);
				\draw (0,4) -- (4,8);
				\draw (4,8) -- (4,4);
				\draw (4,4) -- (0,0);
				
				\draw[red, thick] (0,4)--(4,8);
			\end{tikzpicture}
			\caption*{\small $\{1\}$}
		\end{minipage}

		\begin{minipage}{0.15\textwidth}
			\centering
			\begin{tikzpicture}[scale=0.3]
				\foreach \x in {1,2,3} {
					\foreach \y in {1,2,3} {
						\fill (\x, \x+\y) circle (2pt);  
					}
				}
				\draw (0,0) -- (0,4);
				\draw (0,4) -- (4,8);
				\draw (4,8) -- (4,4);
				\draw (4,4) -- (0,0);
				
				\draw[red, thick] (0,0)--(4,4);
			\end{tikzpicture}
			\caption*{\small \textit{Case (5):} $\{0\}$}
		\end{minipage}
		\begin{minipage}{0.15\textwidth}
			\centering
			\begin{tikzpicture}[scale=0.3]
				\foreach \x in {1,2,3} {
					\foreach \y in {1,2,3} {
						\fill (\x, \x+\y) circle (2pt);  
					}
				}
				\draw (0,0) -- (0,4);
				\draw (0,4) -- (4,8);
				\draw (4,8) -- (4,4);
				\draw (4,4) -- (0,0);
				
				\draw[red, thick] (0,0)--(1,2);
				\draw[red, thick] (1,2)--(4,5);
			\end{tikzpicture}
			\caption*{\small $\{1/4\}$}
		\end{minipage}
		\begin{minipage}{0.15\textwidth}
			\centering
			\begin{tikzpicture}[scale=0.3]
				\foreach \x in {1,2,3} {
					\foreach \y in {1,2,3} {
						\fill (\x, \x+\y) circle (2pt);  
					}
				}
				\draw (0,0) -- (0,4);
				\draw (0,4) -- (4,8);
				\draw (4,8) -- (4,4);
				\draw (4,4) -- (0,0);
				
				\draw[red, thick] (0,0)--(2,4);
				\draw[red, thick] (2,4)--(4,6);
			\end{tikzpicture}
			\caption*{\small $\{1/2\}$}
		\end{minipage}
		\begin{minipage}{0.15\textwidth}
			\centering
			\begin{tikzpicture}[scale=0.3]
				\foreach \x in {1,2,3} {
					\foreach \y in {1,2,3} {
						\fill (\x, \x+\y) circle (2pt);  
					}
				}
				\draw (0,0) -- (0,4);
				\draw (0,4) -- (4,8);
				\draw (4,8) -- (4,4);
				\draw (4,4) -- (0,0);
				
				\draw[red, thick] (0,0)--(3,6);
				\draw[red, thick] (3,6)--(4,7);
			\end{tikzpicture}
			\caption*{\small $\{3/4\}$}
		\end{minipage}
		\begin{minipage}{0.15\textwidth}
			\centering
			\begin{tikzpicture}[scale=0.3]
				\foreach \x in {1,2,3} {
					\foreach \y in {1,2,3} {
						\fill (\x, \x+\y) circle (2pt);  
					}
				}
				\draw (0,0) -- (0,4);
				\draw (0,4) -- (4,8);
				\draw (4,8) -- (4,4);
				\draw (4,4) -- (0,0);
				
				\draw[red, thick] (0,0)--(4,8);
			\end{tikzpicture}
			\caption*{\small $\{1\}$}
			\end{minipage}
		
		\caption{Interpolation scheme based on the configuration of the cube.}
		\label{fig:interpolation}
\end{figure}
}

\newcommand{\cubeFixedPointCases}{
	\begin{figure}[h!]
		\centering
		\vspace{0.5cm}
		\begin{minipage}{0.3\textwidth}
			\centering
			\begin{tikzpicture}[scale=0.4]
				\draw (0, 0) -- (0, 4)
				(0, 0) -- (4, 0)
				(0, 4) -- (4, 4)
				(0, 4) -- (2, 5)
				(4, 4) -- (6, 5)
				(6, 1) -- (6, 5)
				(4, 0) -- (6, 1)
				(4, 0) -- (4, 4)
				(2, 5) -- (6, 5)
				(0, 0) -- (2, 1)
				(2, 1) -- (2, 5)
				(2, 1) -- (6, 1);
				
				\fill (0,0) circle (3pt);
				\fill (4,0) circle (3pt);
				\fill (2,1) circle (3pt);
				\fill (6,1) circle (3pt);
				
				\draw[->, red, thick,-{Stealth}] (0,-2)--(0,0);
				\draw[->, red, thick,-{Stealth}] (0,2)--(0,0);
								
				\draw[->, red, thick,-{Stealth}] (4,-2)--(4,0);
				\draw[->, red, thick,-{Stealth}] (4,2)--(4,0);
				
				\draw[->, red, thick,-{Stealth}] (2,3)--(2,1);
				\draw[->, red, thick,-{Stealth}] (2,-1)--(2,1);

				\draw[->, red, thick,-{Stealth}] (6,3)--(6,1);
				\draw[->, red, thick,-{Stealth}] (6,-1)--(6,1);
			\end{tikzpicture}  
		\caption*{Case (1)}
		\end{minipage}
		\begin{minipage}{0.3\textwidth}
			\centering
			\vspace{-0.8cm}	
			\begin{tikzpicture}[scale=0.4]
				\draw (0, 0) -- (0, 4)
				(0, 0) -- (4, 0)
				(0, 4) -- (4, 4)
				(0, 4) -- (2, 5)
				(4, 4) -- (6, 5)
				(6, 1) -- (6, 5)
				(4, 0) -- (6, 1)
				(4, 0) -- (4, 4)
				(2, 5) -- (6, 5)
				(0, 0) -- (2, 1)
				(2, 1) -- (2, 5)
				(2, 1) -- (6, 1);
				
				\fill (0,0) circle (3pt);
				\fill (4,4) circle (3pt);
				\fill (2,1) circle (3pt);
				\fill (6,5) circle (3pt);
				
				\draw[->, red, thick,-{Stealth}] (0,-2)--(0,0);
				\draw[->, red, thick,-{Stealth}] (0,2)--(0,0);
				
				\draw[->, red, thick,-{Stealth}] (2,3)--(2,1);
				\draw[->, red, thick,-{Stealth}] (2,-1)--(2,1);
				
				\draw[->, red, thick,-{Stealth}] (4,6)--(4,4);
				\draw[->, red, thick,-{Stealth}] (4,2)--(4,4);
				
				\draw[->, red, thick,-{Stealth}] (6,7)--(6,5);
				\draw[->, red, thick,-{Stealth}] (6,3)--(6,5);
			\end{tikzpicture}
			\caption*{Case (2)}   
		\end{minipage}
		\begin{minipage}{0.3\textwidth}
			\centering
			\vspace{-0.8cm}
			\begin{tikzpicture}[scale=0.4]
				\draw (0, 0) -- (0, 4)
				(0, 0) -- (4, 0)
				(0, 4) -- (4, 4)
				(0, 4) -- (2, 5)
				(4, 4) -- (6, 5)
				(6, 1) -- (6, 5)
				(4, 0) -- (6, 1)
				(4, 0) -- (4, 4)
				(2, 5) -- (6, 5)
				(0, 0) -- (2, 1)
				(2, 1) -- (2, 5)
				(2, 1) -- (6, 1);
				
				\fill (0,0) circle (3pt);
				\fill (4,4) circle (3pt);
				\fill (2,5) circle (3pt);
				\fill (6,5) circle (3pt);
				
				\draw[->, red, thick,-{Stealth}] (0,-2)--(0,0);
				\draw[->, red, thick,-{Stealth}] (0,2)--(0,0);
				
				\draw[->, red, thick,-{Stealth}] (2,7)--(2,5);
				\draw[->, red, thick,-{Stealth}] (2,3)--(2,5);
				
				\draw[->, red, thick,-{Stealth}] (4,6)--(4,4);
				\draw[->, red, thick,-{Stealth}] (4,2)--(4,4);
				
				\draw[->, red, thick,-{Stealth}] (6,7)--(6,5);
				\draw[->, red, thick,-{Stealth}] (6,3)--(6,5);
			\end{tikzpicture}  
			\caption*{Case (3)}
		\end{minipage}
		\begin{minipage}{0.3\textwidth}
			\centering
			\begin{tikzpicture}[scale=0.4]
				\draw (0, 0) -- (0, 4)
				(0, 0) -- (4, 0)
				(0, 4) -- (4, 4)
				(0, 4) -- (2, 5)
				(4, 4) -- (6, 5)
				(6, 1) -- (6, 5)
				(4, 0) -- (6, 1)
				(4, 0) -- (4, 4)
				(2, 5) -- (6, 5)
				(0, 0) -- (2, 1)
				(2, 1) -- (2, 5)
				(2, 1) -- (6, 1);
				
				\fill (0,0) circle (3pt);
				\fill (4,0) circle (3pt);
				\fill (2,5) circle (3pt);
				\fill (6,5) circle (3pt);

				\draw[->, red, thick,-{Stealth}] (0,-2)--(0,0);
				\draw[->, red, thick,-{Stealth}] (0,2)--(0,0);
				
				\draw[->, red, thick,-{Stealth}] (2,7)--(2,5);
				\draw[->, red, thick,-{Stealth}] (2,3)--(2,5);
				
				\draw[->, red, thick,-{Stealth}] (4,-2)--(4,0);
				\draw[->, red, thick,-{Stealth}] (4,2)--(4,0);
				
				\draw[->, red, thick,-{Stealth}] (6,7)--(6,5);
				\draw[->, red, thick,-{Stealth}] (6,3)--(6,5);
			\end{tikzpicture}    
			\caption*{Case (4)}
		\end{minipage}
		\begin{minipage}{0.5\textwidth}
			\centering
			\begin{tikzpicture}[scale=0.4]
				\draw (0, 0) -- (0, 4)
				(0, 0) -- (4, 0)
				(0, 4) -- (4, 4)
				(0, 4) -- (2, 5)
				(4, 4) -- (6, 5)
				(6, 1) -- (6, 5)
				(4, 0) -- (6, 1)
				(4, 0) -- (4, 4)
				(2, 5) -- (6, 5)
				(0, 0) -- (2, 1)
				(2, 1) -- (2, 5)
				(2, 1) -- (6, 1);
				
				\fill (0,0) circle (3pt);
				\fill (4,0) circle (3pt);
				\fill (2,1) circle (3pt);
				\fill (6,5) circle (3pt);
				
				\draw[->, red, thick,-{Stealth}] (0,-2)--(0,0);
				\draw[->, red, thick,-{Stealth}] (0,2)--(0,0);
				
				\draw[->, red, thick,-{Stealth}] (2,-1)--(2,1);
				\draw[->, red, thick,-{Stealth}] (2,3)--(2,1);
				
				\draw[->, red, thick,-{Stealth}] (4,-2)--(4,0);
				\draw[->, red, thick,-{Stealth}] (4,2)--(4,0);
				
				\draw[->, red, thick,-{Stealth}] (6,7)--(6,5);
				\draw[->, red, thick,-{Stealth}] (6,3)--(6,5);
			\end{tikzpicture}    
			\caption*{Case (5)}
		\end{minipage}
	
		\caption{The possible configurations of a surface intersecting the cube
$C$.}
		\label{fig:1d_unique_fp_location}
	\end{figure}
}

\newcommand{\neighboringSlice}{
	\begin{figure}[htbp]
		\centering
		\begin{tikzpicture}[scale=0.6]
			\draw (0,0) -- (0,4);
			\draw (1,0) -- (1,4);
			
			\foreach \y in {0.5, 1.5, 2.5, 3.5} {
				\fill (0,\y) circle (2pt);
				\fill (1,\y) circle (2pt);
			}
			
			\draw[dotted, thick] (0,0.5) -- (1,0.5);
			\draw[dotted, thick] (0,1.5) -- (1,1.5);
			\draw[dotted, thick] (0,2.5) -- (1,2.5);
			
			\node at (-0.4,1.5) {$x$};
			\node at (1.4,0.5) {$a$};
			\node at (1.4,1.5) {$b$};
			\node at (1.4,2.5) {$c$};
			\node at (1.9,3.5) {$d \neq y$};
			
			\node at (-0.2,-0.5) {$s$};
			\node at (2,-0.5) {$s' = s + e_j$};
		\end{tikzpicture}
		\caption{Distance between LFPs of neighboring slices. The LFP $y$ of $s'$ 
		could be at $a$, $b$, or $c$, but cannot be a $d$ as it is too far away from $x$.}
		\label{fig:neighboring-lfps}
	\end{figure}
}

\newcommand{\dlnewslice}{
	\begin{figure}[htbp]
		\centering
		\begin{minipage}{0.3\textwidth}
		\vspace{0.5cm}
		\begin{tikzpicture}[scale=0.5]
			\draw 	(0,0) -- (0,1)
			(0,0) -- (2,0)
			(0,1) -- (5,6)
			(2,0) -- (7,5)
			(5,6) -- (7,6)
			(7,5) -- (7,6);
			
			\node[below] at (2,0) {$y$};
			\fill (2,0) circle (2pt);
			\node[left] at (0,1) {$x$};
			\fill (0,1) circle (2pt);
			\node[right] at (7,6) {$q$};
			\fill (7,6) circle (2pt);
			
			\draw[dotted] (0,1.5) -- (4.5,6);	
			\node[color=gray, below] at (1.5,4) {\scriptsize $l$};	
			
			\draw (5.5,6) -- (8.5,9);  
			\draw (6.5,6) -- (9.5,9);  
			\draw (8.5,9) -- (9.5,9);  
			
			\draw[->, green] (7,6) -- (6.5,5.5);
			\draw[green, thick, dotted] (3,6) -- (7,6);
			\draw[green, thick, dotted] (7,2) -- (7,6);

			\fill[pattern=north east lines, pattern color=gray, opacity=0.25] (0,0) -- (0,1) -- (5,6) -- (7,6) -- (7,5) -- (2,0) -- (0,0) -- cycle; 
		\end{tikzpicture}
		\caption*{(a)}
		\end{minipage}
	\begin{minipage}{0.3\textwidth}
		\begin{tikzpicture}[scale=0.5]
			\draw 	(0,0) -- (0,1)
			(0,0) -- (2,0)
			(0,1) -- (5,6)
			(2,0) -- (7,5)
			(5,6) -- (7,6)
			(7,5) -- (7,6);
			
			\node[below] at (2,0) {$y$};
			\fill (2,0) circle (2pt);
			\node[left] at (0,1) {$x$};
			\fill (0,1) circle (2pt);
			\node[right] at (7,6) {$q$};
			\fill (7,6) circle (2pt);
			
			\draw[black, dotted, line width=1pt] 	(5,6)--(7,8)
							(7,5) -- (9,7)
							(7,8) -- (9,8)
							(9,7) -- (9,8);
			
			\draw[dotted] (0,1.5) -- (6.5,8);	
			\node[color=gray, below] at (2.5,5.5) {\scriptsize $l + 2w$};	
			
			\draw (5.5,6) -- (8.5,9);  
			\draw (6.5,6) -- (9.5,9);  
			\draw (8.5,9) -- (9.5,9);  
			
			\draw[->, green] (7,6) -- (7.5,6.5);
			\draw[green, thick, dotted] (11,6) -- (7,6);
			\draw[green, thick, dotted] (7,10) -- (7,6);
		
		\end{tikzpicture}
		\caption*{(b)}
		\end{minipage}
	\caption{A depiction of the scheme for diagonal-lobe states.}
	\label{fig:dl_new_slice}
\end{figure}
}

\newcommand{\twodunique}{
	\begin{figure}[htbp]
		\centering
		\begin{minipage}{0.45\textwidth}
			\centering
			\vspace{2.4cm}
			\begin{tikzpicture}[scale=0.8]
				\foreach \x in {0, 1, 2, 3}
				\foreach \y in {0, 1, 2, 3}
					\fill (\x, \y) circle (1pt);
				
				\draw[blue,opacity=0.8] (1, 0) -- (1, 1) -- (2, 2) -- (2,3);
				\draw[red,opacity=0.8] (0, 1) -- (1, 1) -- (2, 2) -- (3,2);
			\end{tikzpicture}
			\caption*{(a)}
		\end{minipage}
		\begin{minipage}{0.45\textwidth}
			\centering
			\begin{tikzpicture}[scale=0.8]
				\foreach \x in {0, 1, 2, 3, 4, 5, 6}
				\foreach \y in {0, 1, 2, 3, 4, 5, 6}
					\fill (\x, \y) circle (1pt);
				
				\draw[blue,opacity=0.8] (2,0) -- (2, 2) -- (4,4) -- (4, 6);
				\draw[red,opacity=0.8] (0,1) -- (2, 1) -- (4,3) -- (6, 3);
			\end{tikzpicture}
			\caption*{(b)}
		\end{minipage}
	\caption{A reduction to ensure that a two-dimensional instance has a unique fixed point.}
	\label{fig:twodunique}
	\end{figure}
}

\newcommand{\reals}{\mathbb{R}}

\title{Monotone Contractions}
\author[1]{Eleni Batziou}
\author[1]{John Fearnley}
\author[1]{Spencer Gordon}
\author[2]{Ruta Mehta}
\author[1,3]{Rahul Savani}

\affil[1]{University of Liverpool\\
	\texttt{\{eleni.batziou, john.fearnley, rahul.savani\}@liverpool.ac.uk}}
\affil[2]{University of Illinois at Urbana-Champaign\\
	\texttt{rutameht@cs.illinois.edu}}
\affil[3]{Alan Turing Institute}

\def\Floor#1{\left\lfloor #1 \right\rfloor}
\def\Ceil#1{\left\lceil #1 \right\rceil}

\def\Set#1{\left\{ #1 \right\}}
\def\Abs#1{\left| #1 \right|}

\def\Norm#1{\left\| #1 \right\|}
\def\Brack#1{\left[ #1 \right]}

\def\Natural{\mathbb{N}}

\def\One{\mathbf{1}}

\DeclareMathOperator{\poly}{poly}

\def\sublattice{L}

\DeclareMathOperator{\proj}{\Pi}

\def\eps{\varepsilon}
\def\beps{\hat{\varepsilon}}
\def\hatx{\hat{x}}
\def\haty{\hat{y}}

\DeclareMathOperator{\Stretch}{\mathsf{stretch}}

\DeclareMathOperator{\sgne}{\mathsf{clamp}_{\beps}}

\DeclareMathOperator{\FindLesserFixpoint}{\mathsf{FindLesserFixedpoint}}

\def\up{\mathsf{up}}
\def\down{\mathsf{down}}
\def\zero{\mathsf{zero}}

\def\OPDC{\ensuremath{\mathsf{OPDC}}\xspace}
\def\OnePermutationDiscreteContraction{\mathsf{OnePermutationDiscreteContraction}}
\def\UEOPL{\ensuremath{\mathsf{UEOPL}}\xspace}
\def\MonotoneContraction{\ensuremath{\mathsf{MonotoneContraction}}\xspace}

\def\DMAC{\ensuremath{\mathsf{DMAC}}\xspace}
\def\1DUniqueDMAC{\ensuremath{\mathsf{1DUniqueDMAC}}\xspace}
\def\2DUniqueDMAC{\ensuremath{\mathsf{2D-UNIQUE-DMAC}}\xspace}
\def\MonoGrad1Surface{\ensuremath{\mathsf{Monotone-Gradient-1-Surfaces}}\xspace}
\DeclareMathOperator{\slice}{\mathsf{slice}}
\DeclareMathOperator{\Slice}{\mathsf{Slice}}
\def\star{\ast}

\DeclareMathOperator{\IsLFP}{\mathsf{IsLFP}}

\DeclareMathOperator{\VSeq}{\mathsf{VSeq}}
\DeclareMathOperator{\free}{\mathsf{free}}
\DeclareMathOperator{\fixed}{\mathsf{nonfree}}
\DeclareMathOperator{\bind}{\mathsf{bind}}

\DeclareMathOperator{\Up}{Up}
\DeclareMathOperator{\Down}{Down}

\DeclareMathOperator{\CBox}{Box}
\DeclareMathOperator{\DBox}{DBox}

\DeclareMathOperator{\leftl}{Left}
\DeclareMathOperator{\downl}{Bott}
\DeclareMathOperator{\diagl}{Diag}
\DeclareMathOperator{\UC}{UC}
\DeclareMathOperator{\DC}{DC}
\DeclareMathOperator{\hei}{Height}
\DeclareMathOperator{\area}{Area}
\DeclareMathOperator{\points}{Points}

\newcommand{\linf}{\ensuremath{\ell_\infty}\xspace}

\newcommand{\Normi}[1]{\Norm{#1}_{\infty}}

\begin{document}
	
\maketitle

\begin{abstract}
We study functions $f : [0, 1]^d \rightarrow [0, 1]^d$ that are both monotone
and contracting, and we consider the problem of finding an $\eps$-approximate fixed
point of $f$. We show that the problem lies in the complexity class \UEOPL.
We give an algorithm that finds an $\eps$-approximate fixed point of a
three-dimensional monotone contraction using $O(\log
(1/\eps))$ queries to $f$. We also give a decomposition theorem
that allows us to use this result to obtain 
an algorithm that finds an $\eps$-approximate fixed point of a $d$-dimensional
monotone contraction using $O((c \cdot \log
(1/\eps))^{\lceil d / 3 \rceil})$ queries to $f$ for some constant $c$.
Moreover, each step of both of our algorithms takes time that is
polynomial in the representation of $f$. These results are strictly better than
the best-known results for functions that are only monotone, or
only contracting.

All of our results also apply to Shapley stochastic games, which are known to
be reducible to the monotone contraction problem. Thus we put Shapley games in
\UEOPL, and we give a faster algorithm for approximating the value of a Shapley
game.
\end{abstract}

\setcounter{tocdepth}{2}

\newpage
\tableofcontents
\newpage

\section{Introduction}
\label{sec:intro}
In this paper we study monotone contractions. These are functions $f : [0, 1]^d
\rightarrow [0, 1]^d$ that are simultaneously \emph{monotone}, meaning that
$f(x) \le f(y)$ whenever $x \le y$,  and \emph{contracting} in the
$\ell_\infty$-norm, meaning that $\|
f(x) - f(y) \|_\infty \le \lambda \cdot \| x - y \|_\infty$ for some $\lambda
\in [0, 1)$.

Every function $f$ that is a monotone contraction has a fixed point, which is a
point $x$ such that $x = f(x)$. This can be proved in two different ways.
Tarski's fixed point theorem (also known as the Knaster-Tarski theorem), shows
that any monotone function has a fixed point~\cite{Tarski55}, and moreover the
set of fixed points of a monotone function forms a complete lattice. 
Banach's fixed point theorem states that any contraction map has a
\emph{unique} fixed point~\cite{Banach1922}. Both of these theorems apply
to a monotone contraction, so a monotone contraction also has a unique fixed
point.

We are interested in the problem of finding a fixed point of a monotone contraction.
Unfortunately however, there exist monotone contractions in which the unique
fixed point has irrational coordinates~\cite{Shapley53}. So instead, we study the
problem of finding an \emph{approximate} fixed point of a monotone contraction
$f$, which means that we seek a point $x$ such that $\| x - f(x) \|_\infty \le
\eps$ for a given approximation parameter $\eps$.

While the combination of monotonicity and contraction may seem artificial at
first, there are in fact important problems that naturally reduce to the problem of
finding an approximate fixed point of a monotone contraction. Specifically, the
problem of approximating the value of a Shapley stochastic game, and the problem of
solving a simple-stochastic game can both be formulated as the problem of finding an
approximate fixed point of a function that is both monotone and contracting in
the $\ell_\infty$-norm~\cite{EPRY20}.


\paragraph{\bf Shapley games.}

We are particularly interested in Shapley games in this paper, because the
current best known algorithms for approximating the value of a Shapley game
arise from the fact that they can be formulated as a monotone function, while
ignoring the fact that the function is also contracting. 

Shapley stochastic games~\cite{Shapley53}
are a fundamental two-player zero-sum game model with both sequential and concurrent play.
They have a central role in the Game Theory literature, see e.g.~\cite{stochasticgames}, as well as in the 
Computer Science literature where they are a model of formal verification of reactive systems with concurrency~\cite{PRISM3}.
The computational complexity of approximating the value of a Shapley stochastic game
was shown by Etessami and Yannakakis~\cite{EtessamiY10} to be in PPAD\footnote{They also showed that the 
problem of finding an exact solution is in FIXP, but our focus in this paper is on approximate solutions.}.
It has more recently been observed that the problem also lies in PLS~\cite{EPRY20},
which we now know means that the problem lies in CLS, since CLS = PPAD $\cap$
PLS~\cite{FearnleyGHS23}.
Hansen et al.~\cite{HansenKLMT11,abs-1202-3898} provided exponential time
algorithms for solving Shapley games exactly.\footnote{Here solutions may be
irrational and the task is to find a symbolic representation of the optimal
strategies as algebraic numbers.} Those results have been improved by
Oliu-Barton~\cite{Oliu-Barton21}.

The best known algorithm for finding an $\eps$-approximation of the value of a
Shapley game $G$ with $d$ states is $O((c \cdot \log (1/\eps))^{\lceil d/2 \rceil
+1} \cdot \poly(|G|))$, where $c$ is some constant, and $|G|$ is the size of
the representation of~$G$. This result arises from applying the current
best-known algorithm for monotone functions~\cite{ChenL22} to the
monotone contraction that is defined by the game~\cite{EPRY20}.
%

\paragraph{\bf Monotone functions.}

There has recently been considerable research interest in the computational
complexity of finding a fixed point of a monotone function.
This work has focused on monotone functions $f : G \rightarrow G$ on the
grid-lattice $G = \{1, 2, \dots, n\}^d$.

Dang et al.~\cite{DQY20} provided an algorithm that finds a fixed point using
$\log^d n$ queries to $f$. Etessami et al.~\cite{EPRY20} placed the problem in
the complexity classes PPAD and PLS, 
which also implies that the problem lies in CLS~\cite{FearnleyGHS23}.
They also showed that $\Omega(\log^2 N)$ queries are necessary for
two-dimensional instances, thus showing that the algorithm of Dang et al. is
tight in two-dimensions.

Both Dang et al. and Etessami et al. conjectured that the
$\log^d n$-query algorithm was tight for all dimensions $d$.
However, Fearnley et al. disproved this by providing an algorithm that needs only $O(\log^2 N)$ 
queries in dimension three~\cite{FPS22}.
They also introduced a \emph{decomposition theorem}, which showed that if the
problem can be solved using $q_1$ queries in dimension $d_1$, and $q_2$ queries
in dimension $d_2$, then it can be solved using $q_1 \cdot q_2$ queries in
dimension $d_1 + d_2$.
Applying this decomposition theorem to their three-dimensional algorithm gives
an algorithm that solves $d$ dimensional instances in 
$O((c \cdot \log n)^{2\Ceil{d/3}})$ queries\footnote{We note that~\cite{FPS22}
claim that the algorithm uses 
$O(\log^{2\Ceil{d/3}} n)$ queries, but this is not quite correct. They apply
the decomposition theorem to an algorithm that uses $O(\log^2 n)$ queries to
solve a three-dimensional instance, and when they do so, the constant that was
hidden by the big-oh notation also gets powered by 
$2\Ceil{d/3}$. This is a relatively minor distinction, however, since $c \cdot
\log x = \log_{c'} x$ for some constant $c'$.}, for some constant $c$.


%
%
%

%
%
The current best-known algorithm for monotone functions was given by 
Chen et al.~\cite{ChenL22}.
Internally, the algorithm of 
Fearnley et al.\ works by solving a relaxed problem for two-dimensional
instances, and
Chen et al.~\cite{ChenL22} call this problem \textsc{Tarski*}.
They gave a decomposition theorem for 
\textsc{Tarski*}, and applying this to the original algorithm of 
Fearnley et al.\ gives an algorithm that solves the problem using
$O((c \cdot \log n)^{\Ceil{(d+1)/2}})$ queries, for some constant $c$.

Recently, Chen at al.~\cite{ChenLY23} showed that the query complexity of finding Tarski fixed points is the 
same as the query complexity of finding fixed points when the Tarski instance is promised to have a unique 
solution, and it is also promised that all sub-instances have unique solutions. 


%
%
%

\paragraph{\bf Contracting functions.}

There has also been considerable interest in the problem of finding the fixed
point of a contracting function. 
Our focus in this paper is on functions that are contracting with respect to the infinity norm,
since it is this type of contraction that arises from Shapley and Simple stochastic
games~\cite{EPRY20}.
As mentioned previously, there are infinity norm contraction maps whose fixed point uses
irrational numbers, so the focus has been on finding an $\eps$-approximate fixed point
of a contraction map $f : [0, 1]^d \rightarrow [0, 1]^d$. 

The problem naturally lies in PPAD, because it is a special case of finding a
Brouwer fixed point, and also PLS, because $\| x - f(x) \|_\infty$ gives a
natural potential function for the problem. 
The problem therefore lies in CLS~\cite{DP11,FearnleyGHS23}.

Fearnley et al.~\cite{FGMS20} studied contractions defined by piecewise
linear arithmetic circuits (PL-contraction), which capture the problem of simple stochastic
games, but not Shapley games. They showed that 
finding an exact fixed point of a PL-contraction 
lies in the
complexity class \UEOPL, which is a subclass of CLS that contains problems that
have unique solutions. 
It remains open whether the contraction problem itself lies in \UEOPL.


Shellman and Sikorski proved that the query complexity of
finding an $\eps$-fixed point of a contraction map in dimension two is
$\log(\frac{1}{\eps})$~\cite{ShellmanS02,ShellmanS03}; they left it as an open problem to 
extend their algorithm for 2 dimensions to more dimensions.
In other work,
they game a polynomial-time algorithm for general dimension $d$ with query
complexity~$\log^d(\frac{1}{\eps})$~\cite{ShellmanS03a,ShellmanS05}.
%
In a striking recent result, Chen et al.~\cite{Chen0Y24}
showed that the query complexity of finding an $\eps$-fixed-point of a contraction map
under the infinity norm is actually $O(d^2 \log(\frac{1}{\eps}))$, so the query
complexity of contraction is polynomial in $d$. However this result does come with a drawback: the time
complexity of the algorithm is not polynomial, meaning that although the
algorithm makes polynomially many queries, it may spend exponential time in
between each of those queries. 

While our focus here is on the $\ell_\infty$-norm, due to the importance of this
setting for game theory, we note that finding fixed points of contractions or
non-expansions under other norms has also received considerable attention.
Sikorski provides a survey of results for both the $\ell_\infty$ and $\ell_2$
(Euclidean) norms in~\cite{Sik09}. 
In the Euclidean norm, it is known that a fixed point of a contraction or non-expansion can be found in
polynomial time even in non-constant dimension using algorithms based on the ellipsoid
method~\cite{BoonyasiriwatSX07,HuangKS99}.
Fearnley et al.~\cite{FearnleyGMS19} gave a polynomial-time algorithm for finding an
approximate fixed point of any contraction in the $\ell_p$-norm with $p < \infty$ in any constant dimension.

The problem of finding an approximate fixed point 
of a contracting function has also been studied in the case where the function is contracting
not with respect to the metric induced by a fixed norm, but where a metric or meta-metric
is itself input to the  problem. In those cases the problem is known 
to be CLS-complete~\cite{DTZ18,FGMS17}.


\paragraph{\bf Our contribution.}

While both monotone and contracting functions have received a considerable amount
of attention recently, monotone contractions have not yet been studied at all.
This paper provides complexity results and algorithms for
monotone contractions (for the $\ell_\infty$-norm), and in doing so we improve the best-known
results for approximating the value of a Shapley stochastic game.

Our main results can be stated as follows.

\begin{enumerate}
\item Finding an approximate fixed point of a monotone contraction lies in
\UEOPL. 

\item There is an algorithm that finds an $\eps$-approximate fixed point of a
three-dimensional monotone contraction $f$ using $O(\log (1/\eps))$ queries
to $f$, where each step runs in time that is polynomial in the representation of
$f$. 

\item There is an algorithm that finds an $\eps$-approximate fixed point of a
$d$-dimensional monotone contraction~$f$ using $O((c \cdot \log
(1/\eps))^{\lceil d/3 \rceil})$ queries to $f$, where each step runs in time that is polynomial in the representation of $f$. 
\end{enumerate}
Since Shapley stochastic games can be formulated as monotone contractions
for the $\ell_\infty$-norm, these results
apply to Shapley games as well. So we place the problem of approximating the
value of a Shapley games in \UEOPL, which was not previously known, and we also
provide an algorithm that finds an $\eps$-approximation of the value of a game
$G$ that has $d$ states in time $O((c \cdot \log (1/\eps))^{\lceil d/3 \rceil} \cdot \poly(|G|))$,
which is faster than the previous best algorithm.

The main theme of our results is that functions that are monotone and
contracting appear to be easier than functions that are just monotone, or just
contracting. From a complexity point of view we show that monotone contractions
lie in \UEOPL, which is not known for only-monotone functions, and it is not
known for only-contracting functions, except for the special case of
PL-contraction. Our algorithm for monotone contractions is also faster than the
best known algorithm for only-monotone functions, with the exponent decreasing
from $\lceil d/2 \rceil + 1$ to $\lceil d/3 \rceil$, where for comparisons
sake, if we applied the algorithms for monotone functions in a continuous
setting, we would set $n = \Omega (1/\eps)$. 

We do not beat the query complexity of the best-known algorithm for contraction
under the infinity norm, because, as mentioned, Chen et al.~\cite{Chen0Y24}
gave a $O(d^2 \log(\frac{1}{\eps}))$-query algorithm for the problem. However,
their algorithm does not run in polynomial time. The best known polynomial-time
algorithm for contraction is makes $O(\log^d(1/\eps))$ queries to the
function~\cite{ShellmanS03}. Our results show that if the function is also
monotone, then this can be improved to $O((c \cdot \log (1/\eps))^{\lceil d/3
\rceil})$ queries.


\paragraph{\bf Techniques.}

We show many interesting properties of monotone contractions on the way to
proving our main results. 

The first step that we take is to discretize the problem. We define the
\emph{discrete monotone approximate contraction} (\DMAC) problem, which, like
monotone functions, is defined by a function $f : G \rightarrow G$ over the
discrete grid $G = \{1, 2, \dots, n\}^d$. The \DMAC problem asks us to either
find a fixed point of $f$, or a violation of monotonicity of $f$, or a strict
violation of contraction (or more precisely, a violation of non-expansion) in
$f$. So \DMAC is a discrete problem that naturally captures both monotonicity
and contraction. We show that $\eps$ approximating a $d$-dimensional
monotone contraction reduces to $d$-dimensional \DMAC with $n = \Omega (1/\eps)$. 

We then show that \DMAC can be placed in \UEOPL. Here the main challenge is that
problems in \UEOPL are required to have a unique solution, while \DMAC instances
can have many fixed points, each of which can intuitively be mapped back to an
approximate fixed point of a monotone contraction. However, we can use
monotonicity to overcome this, because Tarski's theorem states that every
monotone function has a least fixed point, so we can select this to be our
unique solution.

We show that given a fixed point of a \DMAC instance, we can verify in
polynomial time whether it is a least fixed point. This is another distinction
between monotone functions and monotone contractions, because
Etessami et al.~\cite{EPRY20} have shown that verifying whether a fixed point
is the least fixed point of a monotone function is coNP-complete, even in
one-dimensional instances. We circumvent this by using the extra contraction
properties of a monotone contraction in our verification algorithm. 

We then use this verification algorithm to implement a reduction from \DMAC to
the \OPDC problem, which is known to be \UEOPL-complete~\cite{FGMS20}. This reduction
also handles violations, and all violations of the \DMAC instance are mapped on
to violation in the \OPDC instance. This means that we have a promise-preserving
reduction, as defined in~\cite{FGMS20}.

We then prove that a decomposition theorem can be shown for
monotone contractions, similar to the decomposition theorems that have been
shown for monotone functions. But this comes with a catch: our decomposition
theorem only works for algorithms that find the least fixed point of the \DMAC
instance, rather than just any fixed point. This has the potential to be
problematic, because our algorithm for three-dimensional instances does not
guarantee that it finds a least fixed point.

We are able to overcome this by showing that given a three-dimensional \DMAC instance $f$, we can
produce a three-dimensional \DMAC instance $f'$ in time that is polynomial in the representation
of $f$ such that $f'$ has a unique fixed point, and that fixed point can be
mapped back to the least fixed point of $f$. Thus, if we apply any \DMAC
algorithm to $f'$, we can recover the least fixed point of $f$.

Finally, we give an algorithm that solves a three-dimensional \DMAC instance
$f$ in $O(\log n)$ queries and each step runs in time that is polynomial in the
representation of $f$. This is another distinction between monotone functions
and monotone contractions. For monotone functions, the query complexity of
finding a fixed point of a three-dimensional instance is $\Theta(\log^2 n)$~\cite{EPRY20,FPS22}, but we show that if
the function is also contracting then this can be improved to $O(\log n)$. The
algorithm is by far the most technically complex part of the paper, and we give
a full overview of the ideas used in it in
Section~\ref{sec:technical_overview}. This algorithm combined with our
decomposition theorem then gives us our two main algorithmic results stated
above.

\section{MonotoneContraction and DMAC: Problem Definitions}
\label{sec:definitions}

In this paper, we study the following problem. 

\begin{definition}[$\MonotoneContraction$]
\label{def:mc_def}
Given $g:[0,1]^d \to [0,1]^d$ represented as a polynomial-time Turing
machine\footnote{Formally, we are given a Turing machine $\mathcal{M}$ and a
polynomial $p$. If $|x|$ denotes the number of bits
needed to represent the rational vector $x \in [0, 1]^d$, then to compute $g(x)$ we
run $\mathcal{M}$ on input $x$ for $p(|x|)$ steps. If $\mathcal{M}$ terminates
and outputs a vector $y \in [0, 1]^d$ then we set $g(x) = y$. If $\mathcal{M}$ terminates without
outputting a vector,
or does not terminate within the required number of steps, then we set $g(x) =
0^d$.} and constants $\lambda \in [0, 1)$ and $\eps \in
(0, 1)$, find one of the following. 
\begin{enumerate}[label=(M\arabic*)]
\item An $\eps$-approximate fixed point, i.e., $x$ such that
	$\Norm{g(x) - x}_{\infty}\leq \eps$. \label{M1} \hfill [Approximate fixed point]
\end{enumerate}
\begin{enumerate}[label=(MV\arabic*)]
	\item A pair of points $x,y$ such that $x \leq y$ but $g(x) \not\leq g(y)$. \label{MV1} \hfill [Monotonicity violation]
	\item A pair of points $x,y$ such that $\Norm{g(x) - g(y)}_{\infty} > \lambda \Norm{x-y}_{\infty}$. \label{MV2} \hfill [Contraction violation]
\end{enumerate}
\end{definition}

Solutions of type \ref{M1} are the approximate fixed points of $g$. Note, however, that given a function $g$, there is no easy way to
determine whether $g$ is monotone or contracting. Instead, we define the
problem for all functions $g$, and we allow the algorithm to return a \emph{violation
solution} if a violation of monotonicity \ref{MV1} or 
a violation of contraction \ref{MV2} is detected. This ensures that the
problem is \emph{total}, meaning it always has a solution.


The first step towards showing our results will be to discretize the problem.
Specifically, we will show that
\MonotoneContraction can be reduced in polynomial time to the following
problem.

\begin{definition}[Discrete Monotone Approximate Contraction ($\DMAC$)]
\label{def:dmac_def}
Let $G = \{1, 2, \dots, n\}^d$ be a $d$-dimensional grid on length $n$.
Given a function $f: G \rightarrow G$ represented by a polynomial-time Turing
machine\footnote{This is defined analagously to the previous footnote.}, find one of the following. 
\begin{enumerate}[label=(D\arabic*)]
\item A point $x \in G$ such that $f(x) = x$. 
\label{D1} \hfill [Fixed point]
\end{enumerate}
\begin{enumerate}[label=(DV\arabic*)]
	\item A pair of points $x,y \in G$ such that $x \leq y$ but $f(x) \not\leq f(y)$. \label{DV1} \hfill [Monotonicity violation]
\item A pair of points $x,y \in G$ such that $\Norm{f(x) - f(y)}_{\infty} > \Norm{x-y}_{\infty}$. \label{DV2} \hfill [Non-expansion violation]
\end{enumerate}
\end{definition}

Solutions of type \ref{D1} encode fixed points of $f$, which under our reduction from \MonotoneContraction
will correspond to approximate fixed points of $g$. 
Violations of type \ref{DV1} encode violations of monotonicity in $f$, while violations of type
\ref{DV2} encode pairs of points at which $f$ is \emph{strictly} non-contracting, i.e., expanding.
Because non-expansion is allowed in \DMAC, unlike in \MonotoneContraction, there is no $\lambda$ in 
the violations of type \ref{DV2}.

\section{Technical Overview}
\label{sec:technical_overview}

In this section we give a high-level overview of our results and the techniques
used to prove them. Full technical details can be found in subsequent sections.

Our results focus on the $\DMAC$ problem, for which we provide a \UEOPL
containment result and algorithms. The first step
is to show that $\eps$-approximating a $\MonotoneContraction$ $g : [0, 1]^d \rightarrow [0, 1]^d$
can be reduced to $\DMAC$ instance $f : \{1, 2, \dots, n\}^d \rightarrow \{1, 2, \dots, n\}^d$ 
where $n = \Omega (1/\eps)$. For this, we first discretize the space $[0, 1]^d$ using a grid of
length $n = \Omega (1/\eps)$. Then for each point $x$ in the grid and each dimension
$i$, we inspect the
displacement $g_i(x) - x_i$: 
\begin{itemize}
\item If $|g_i(x) - x_i| < \eps$ then we set $f_i(x) = x_i$.
\item If $g_i(x) \le x_i - \eps$ then we set $f_i(x) = x_i - 1$.
\item If $g_i(x) \ge x_i + \eps$ then we set $f_i(x) = x_i + 1$.
\end{itemize}
In other words, if the displacement of $g$ in dimension $i$ is at least 
$\eps$, then we round the displacement in that dimension down to the
nearest adjacent grid point of $G$, while if the displacement of~$g$ in
dimension $i$ is strictly smaller than $\eps$, we round the displacement in that
dimension down to zero. 

It is clear that any fixed point of $f$ is an $\eps$-approximate fixed point of
$g$ by construction. We also show in Section~\ref{sec:monotone_to_dmac} that
any \DMAC violation for $f$ can be mapped back to a \MonotoneContraction
violation in $g$.

\subsection{UEOPL containment of DMAC} 
\label{sec:UEOPLoverview}

Our first main result is to show that \DMAC is contained in the complexity
class \UEOPL. The main barrier to showing this is that \DMAC captures the
approximate fixed points of a monotone contraction map, and while Banach's theorem
ensures that there is always a unique exact fixed point of a contraction map, there can be many
approximate fixed points. \UEOPL captures problems that have unique solutions,
so to show that $\DMAC$ is contained in it, we need some way of declaring one
of the approximate fixed points to be the unique solution. 

To do this, we use the fact that the instance is also monotone. This means that
Tarski's theorem ensures that there is a \emph{least} approximate fixed point,
which is unique. However, for monotone functions, it is known that verifying
whether a fixed point is the least fixed point is NP-hard, even in
one-dimensional instances~\cite{EPRY20}. Fortunately, monotone contraction maps
have extra properties that allow us to circumvent this and efficiently verify
least approximate fixed points.

\paragraph{\bf Verifiable least fixed points.}

We show the following property: if $x$ is a fixed point in a \DMAC instance that is not a least
fixed point, then there exists a fixed point $y$ that is below $x$ and at $\ell_\infty$ distance 1
from it, i.e., $y \le x$ and $\| x - y \|_\infty = 1$.
This gives us a relatively straightforward
way of testing whether a given fixed point $x$ is a least fixed point:
\begin{itemize}
\item Start at the point $y = x - \One$. 
\item If $y$ is a fixed point then $x$ is not a least fixed point.
\item Otherwise, there is some dimension $i$ such that $y_i \ne f_i(y)$. If
$y_i > f_i(y)$ then we have a strict violation of contraction with $x$. If 
$y_i < f_i(y)$ then monotonicity implies that all points $z \ge y$ with $z_i =
y_i$ should also satisfy $z_i < f_i(z)$, so none of them are fixed. Hence we
can eliminate all such points, and repeat the process from 
$y' = x - \One + e_i$.
\end{itemize}
This process rules out one dimension in each step, so after $d$ steps we have
either found a fixed point less than $x$, or we have verified that $x$ is the
least fixed point. We call the path of points visited by this process the
\emph{verification sequence} of $x$. 

This property allows us to use the least fixed point as the unique solution.
However, to reduce to \OPDC, we also need to be able to find unique solutions
in each \emph{slice} of the instance, where a slice of a \DMAC instance is a
sub-instance in which some of the coordinates have been fixed. We show that the
least fixed point of the slice can be used for this, and the correctness of
this crucially relies on the fact that least fixed points in \DMAC
instances are \emph{hereditary} in the following sense: we show that if $x$ is
the least fixed point of a slice $s$, and $s'$ is a sub-slice of $s$ in which
one or more dimensions of $s$ have been fixed, then $x$ is also a least fixed
point of $s'$. 

These properties are sufficient to show a reduction from promise \DMAC to the
promise version of \UEOPL, where in both cases the promise is that neither
instance contains a violation. 

\paragraph{\bf Handling violations.}

However, to show \UEOPL containment for \DMAC, we must also handle instances
that contain violations. In particular, we show a \emph{promise-preserving}
reduction from \DMAC to \OPDC (which is in \UEOPL).
This means that all violations of the \DMAC instance must be mapped onto a violation of the \OPDC
instance. 

Note that the correctness of our verification algorithm relied on both
contraction and monotonicity, so when we consider an instance that contains
violations, it is possible that our algorithm might falsely declare $x$ to be a
least fixed point when it is in fact not. 
This will lead to violations in our
\OPDC instance caused by two distinct fixed points $x$ and $y$, both of which have 
purportedly been verified to be the least fixed point by our verification
algorithm. In this case, we show that
a \DMAC violation can be found by inspecting the
verification sequences of $x$ and $y$, and making at most polynomially many
extra queries. This allows us to map any \OPDC violation back to a \DMAC
violation in polynomial time.

\paragraph{\bf Paths of fixed points.}

Before moving on, we note one interesting property of monotone contraction maps
that arises from our reduction that will turn out to be useful later. Our verification algorithm used
the fact that if $x$ is not a least fixed point, then there is a fixed point
$y \le x$ with $\| x - y \|_\infty = 1$. Flipping all dimensions of the
instance and applying the same reasoning implies that if $x$ is not a greatest
fixed point, then there is a fixed point $y \ge x$ with 
$\| x - y \|_\infty = 1$. 
If $l$ and $g$ are the least and greatest fixed points, respectively, then
for any fixed point $x$, we can repeatedly apply
these properties to generate
\begin{itemize}
\item a path of fixed points $l = p^1, p^2, \dots, p^k = x$ with $p^i \le
p^{i+1}$ and $\|p^i - p^{i+1}\|_\infty = 1$, 
and
\item a path of fixed points $x = p^1, p^2, \dots, p^k = g$ with $p^i \le
p^{i+1}$ and $\|p^i - p^{i+1}\|_\infty = 1$.
\end{itemize}
In particular, this means that there is such a path of fixed points from the
least fixed point to the greatest fixed point. We will use this fact crucially in
the terminal phase of our algorithm for finding a fixed point of a
three-dimensional \DMAC instance.


\subsection{Algorithmic Results for DMAC}

Our two other main results give an algorithms for finding a fixed point of a
\DMAC instance, and for our algorithms we restrict ourselves to \DMAC instances
that are promised to be violation-free.
We restrict ourselves to the
promise problem here for two reasons. Firstly, the algorithms are already very
complex, and finding violations in each of the steps would make the
algorithms even more complex. Secondly, we envision that the algorithm would
mainly be applied to \DMAC instances that arise from problems that are known up front to be monotone and contracting,
such as Shapley games. If the algorithm is applied to a function that has
violations, then it may find a fixed point, or it may crash.
If the algorithm does crash, then the user can be certain that the function has
violations, and the only downside is that the algorithm will not actually
produce an explicit violation.

\subsubsection{1DUniqueDMAC and Surfaces}

Before we describe the algorithm, we first introduce some useful concepts
that we will use repeatedly. 

\paragraph{\bf 1DUniqueDMAC.}

We introduce a restriction on \DMAC instances that we call \1DUniqueDMAC. This
restriction requires that a \DMAC instance $f : G \rightarrow G$ satisfies the
following properties.
\begin{itemize}
\item The instance has displacements that are at most unit length in the
$\ell_\infty$-norm. That is, we
require that $\| x - f(x) \|_\infty \le 1$ for all $x \in G$.

\item Every one-dimensional slice of the instance has a unique one-dimensional
fixed point. That is, if $x_i = f_i(x)$ for some dimension $i$, then for all
non-zero $c$ we have that $x_i + c \ne f_i(x + c \cdot e_i)$. 
\end{itemize}
The first property is a technical requirement that simplifies our proofs. In
fact we already show that this property can be assumed without loss of
generality when we present our \UEOPL reduction, because we give a
polynomial-time reduction that transforms an arbitrary \DMAC instance $f$ into a \DMAC
instance $f'$ with $\| x - f'(x) \|_\infty \le 1$ for all $x$.

The second property is the more important one, and it is what gives the name to
the \1DUniqueDMAC problem. We show that a \DMAC instance $f$ can be reduced to
a \1DUniqueDMAC instance in time that is polynomial in the representation of
$f$. Note that this does not imply anything about higher-dimensional fixed
points. Indeed, the instance may still have many distinct fixed points, even
though all one-dimensional slices have a unique one-dimensional fixed point.

\paragraph{\bf Surfaces.}

\twod


Every one-dimensional slice of a \1DUniqueDMAC has a unique one-dimensional
fixed point, and we think of these fixed points as a \emph{surface} that spans
the instance. Figure~\ref{fig:twod} shows the two surfaces that appear in a
two-dimensional instance. The red line shows the unique one-dimensional fixed
point in each slice that spans the up/down dimension, while the blue line shows
the unique one-dimensional fixed point in each slice that spans the left/right
dimension. Any point at which these two lines cross is therefore a
two-dimensional fixed point.

More formally, the surface in dimension $i$ as a function $s_i$ that maps
$x_{-i}$ to a height in dimension $i$, where $x_{-i}$ is the point $x$ with
dimension $i$ removed. We show that the surfaces in a violation-free \DMAC
instance obey the following properties.

\begin{itemize}
\item The surface is monotone, meaning that $s_i(x_{-i}) \le s_i(y_{-i})$
whenever $x_{-i} \le y_{-i}$. 

\item The surface has gradient at most one, meaning that $
\| s_i(x_{-i}) - s_i(y_{-i}) \|_\infty \le \| x_{-i} - y_{-i} \|_\infty$.
\end{itemize}

This can be seen in Figure~\ref{fig:twod}. Here the red line moves weakly
monotonically upward with gradient at most one, while the blue line moves
weakly monotonically rightward with gradient at most one.

We also show that the converse holds. If all surfaces in a \1DUniqueDMAC
instance are monotone with gradient at most one, and if one more technical
condition holds on the one-dimensional slices of the instance (see
Lemma~\ref{lem:surfacesback} for details), then the instance is violation-free.
We will use this property in our proofs when we are reducing \DMAC to \DMAC,
because it is usually much easier to verify that the surfaces are monotone with
gradient at most one, than it is to directly prove that the instance is
violation-free.

\subsubsection{The Decomposition Theorem}

We say that an algorithm \emph{LFP-solves} a \DMAC instance if it finds a least
fixed point of the instance or a violation. Our decomposition theorem states
that if $d_1$-dimensional \DMAC can be LFP-solved using $q_1$ queries, and
$d_2$-dimensional \DMAC can be LFP-solved using $q_2$ queries, then $(d_1 +
d_2)$-dimensional \DMAC can be LFP-solved using $q_1 \cdot q_2$ queries. The
decomposition theorem also applies to violation-free instances, in which case
the algorithm is only required to find a least fixed point, and it is not
required to produce a violation.

\twodunique

Our algorithm for three-dimensional instances does not necessarily find a least fixed point. We show that this
is not a problem, because given an arbitrary \1DUniqueDMAC instance $f$, we can
produce another \1DUniqueDMAC instance $f'$ in time that is polynomial in the
representation of $f$ such that $f'$ has a unique fixed point, and that fixed
point can be mapped back to the least fixed point of $f$.

Here we give an overview of this reduction for two-dimensions. 
The full reduction for three-dimensional instances can be found in
Section~\ref{sec:uniqueness}. We illustrate the reduction using
Figure~\ref{fig:twodunique}. 
Recall that we have a fixed point whenever all surfaces cross. 
In Figure~\ref{fig:twodunique} (a) we have a
two-dimensional instance with two fixed points, and this is because the two
surfaces track each other diagonally. The fact that the surfaces are monotone
with gradient
at most one ensures that this must always be the case: in a two-dimensional
instance the only way for both surfaces to cross at multiple points is for both
surfaces to track each other exactly diagonally between those points. 

We transform the instance to create $f'$ using the following two steps.
\begin{enumerate}
\item We refine the grid to $G' = \{1, \; 3/2, \; 2, \; 5/2, \dots, n\}^2$, where $G'$
is the extension of $G$ that contains points at half integers as well as whole
integers, and then we interpolate the surfaces
over $G'$. 

\item We then shift the red surface downward by $1/2$ everywhere.
\end{enumerate}
The result of applying these two steps is shown in Figure~\ref{fig:twodunique}
(b). 

The key property is that now the red and blue surfaces move along
different diagonals in $G'$, and so they cannot track each other diagonally.
Thus they cross at most once, so $f'$ has a unique fixed point. We also
show that since we shifted the surface downward, this unique fixed point lies
in a small cube of width at most $1.5$ below the least fixed point of $f$, and
so if we find the unique fixed point of $f'$, we can search through constantly many
points to recover the least fixed point of $f$.

The full reduction for three-dimensions
(which is specified formally in Section~\ref{sec:uniqueness}) 
 follows essentially the same approach.
This time we blow up the grid so that we have points at integers $i$ and $i +
1/4$, $i + 1/2$, and $i + 3/4$, and we interpolate in three dimensions, which
requires a custom construction to ensure that the interpolated surfaces are
monotone and have gradient at most one. Then we shift one surface down by $1/4$,
another down by $1/2$, and we leave the final one unchanged. For
three-dimensions, it is not the case that all three surfaces must track each
other diagonally whenever there are multiple fixed points, but instead we show
that if there are multiple fixed points, then there are at least two surfaces that
track each other diagonally. Since the reduction ensures that all three
surfaces use different diagonals, the resulting instance must have a unique fixed
point, and we show that the fixed point must lie near the least fixed point of $f$, as
before.

We should stress that surfaces are not actually used to implement the reduction
from $f$ to $f'$. This is because computing the value of the surface
$s_i(x_{-i})$ requires a binary search in general, and making $\log n$ queries
to $f$ to answer a single query to $f'$ is too expensive. Instead we implement the reduction
directly, while ensuring that the surfaces of $f'$ behave as we have described here. 
Once we have constructed $f'$, we can then use the fact that the surfaces of
$f'$ are monotone and have gradient at most one to prove that $f'$ is
violation-free.

\subsubsection{The Algorithm}

We show that a three-dimensional \DMAC instance $f$ can be solved in $O(\log n)$
queries, where each step runs in time that is polynomial in the representation
of $f$.

\paragraph{\bf An algorithm for two-dimensions.}

\bsidea

As a warm-up, we will describe our ideas in a two-dimensional instance. We
start by describing an algorithm that makes $O(\log^2 n)$ queries, which follows the
approach of~\cite{FPS22}. We define the \emph{up-set} to be $\Up(f) = \{ x \; :
\; x \le f(x) \}$, which is 
the set of points at which $f$ moves weakly upward, and the \emph{down-set} to
be $\Down(f) = 
\{ x \; : \; x \ge f(x) \}$, which is the set of points at which $f$ moves weakly
downward. 

The idea is to carry out an \emph{outer} binary search on the first coordinate.
For a fixed slice $s_i = \{x \; : \; x_1 = i\}$, the algorithm uses an
\emph{inner} binary search to find some point in $\Up(f) \cap s_i$ or $\Down(f)
\cap s_i$. If we find a point $x \in \Up(f) \cap s$ then we can eliminate all
points $y$ with $y_1 < x_1$, while if we find a 
point $x \in \Down(f) \cap s$ then 
we can eliminate all points $y$ with $y_1 > x_1$. The justification for this
can be found in~\cite{FPS22}. We will not dwell on this point here, because it
turns out we will actually need to use a different justification for our
algorithm. Since the algorithm uses two levels of binary search, it makes
$O(\log^2 n)$ queries overall.

Our key idea is that we can use the fact that our surfaces are monotone with
gradient at most one to speed up this algorithm. Consider the instance shown in
Figure~\ref{fig:bsidea} (a). Recall that each surface represents the
one-dimensional fixed point of each one-dimensional slice. Moreover, since the
instance is contracting, we must have that all points beneath the surface move
upward under $f$, while all points above the surface move downward under $f$,
because otherwise we would have a strict violation of contraction.

Our inner binary searches are searching
for a point in $\Up(f)$ or $\Down(f)$, which therefore correspond to the points
that lie between the red and blue lines. Specifically, whenever the red line is
above the blue line, the region in between is $\Up(f)$, while whenever the blue
line lies above the red line, the region in between is $\Down(f)$. 

Suppose, as shown in 
Figure~\ref{fig:bsidea} (a), that we were carrying out a binary search on
slice $s_i$. In the last iteration of this search, we had eliminated all but
the region $[a_i, b_i]$, and we found point $p \in \Up(f)$. Suppose further
that the outer binary search moves to the slice $j > i$. If $a_j = a_i + (j
- i)$ and $b_j = b_i + (j - 1)$, then we prove that
$\Up(f) \cap s_j$ must lie in the range $[a_j, b_j]$, as shown in Figure~\ref{fig:bsidea} (a).

This is due to the properties we have for surfaces. Since the blue line
can move rightward with gradient at most one, it can get no closer to the line
between $a_i$ and $a_j$ as we move rightward. Likewise, since the red line can
move upward with gradient at most one, it can get no closer to the line between
$b_i$ and $b_j$. The only way for the red or blue lines to leave this region is
for them to cross and exit the other side, but if the lines have crossed to the
left of a slice $j$, then we have
that $\Up(f) \cap s_j$ is empty so our claim that 
$\Up(f) \cap s_j$ is contained in the region $[a_j, b_j]$ holds trivially.

The key point here is that we can now initialize the binary search in slice $j$
with the bounds $[a_j, b_j]$, and thereby
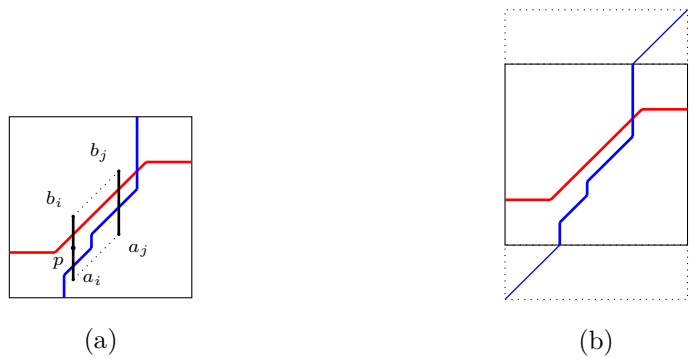
carry over progress from one inner binary
search to the next. So rather than doing $\log n$ separate binary searches, we
do a single binary search that occasionally transfers itself to a new slice. 
So our query complexity is reduced from $O(\log^2 n)$ to $O(\log n)$.

We should remark that the algorithm needs to be slightly changed to accommodate 
this. The original algorithm used a single inner binary search that looked for
a point either in $\Up(f)$ or $\Down(f)$. Here we run two independent binary
searches: one binary search looks for a point in 
$\Up(f)$ and maintains a bound on $\Up(f)$, and the other looks for a point
in $\Down(f)$ and maintains a bound on $\Down(f)$. Indeed, everything we have
described holds symmetrically when we are seeking a point in $\Down(f)$.
Specifically, if we have a bound $[a_i, b_i]$ on $\Down(f) \cap s_i$, then when
we move to a slice $j < i$ we can set $a_j = a_i - (i - j)$ and $b_j = b_i - (i
- j)$ to translate the bound diagonally downward to the slice $j$.

At a high level, our algorithm is as follows.

\begin{enumerate}
\item We first extend the instance vertically
as shown in Figure~\ref{fig:bsidea} (b), to ensure that the blue line touches
both the left and right boundaries of the instance. This enables us to ensure that
our initial bounds entirely contain the up- and down-sets, while at most
tripling the height of the instance. 

\item Initialize $l = 1$ and $u = n$ to be the lower and upper bounds on the
outer binary search, respectively. Let $a$ and $b$ be the minimum and maximum
$y$-values of our extended instance.
Initialize $b^{\text{up}}_1 = (a, b)$ and
$b^{\text{dn}}_n = (a, b)$ to be the trivial bounds on the up and down sets in slices $1$
and $n$ respectively.
\item 
\label{itm:start}
Choose a new slice $i = \Floor{\frac{u - l}{2}}$. Set $b^{\text{up}}_i =
b^{\text{up}}_l + (i - l) \cdot (1, 1)$, and set 
$b^{\text{dn}}_i = b^{\text{dn}}_u + (u - i) \cdot (-1, -1)$, where here we are
applying our technique for translating bounds from one slice to the next.

\item 
\label{itm:bs}
Carry out two inner binary searches in parallel: one looking for a point in
$\Up(f)$ in slice $i$
initialized with the bounds $b^{\text{up}}_i$, and a second looking for
$\Down(f)$ in slice $i$ 
initialized with the bounds $b^{\text{dn}}_i$. 

\item If the up-set binary search succeeds, then we set $l = i$ and
$b^{\text{up}}_i$ to be the last bounds considered by the binary search, and
we move back to Step~\ref{itm:start}.
\item If the down-set inner algorithm succeeds, then we set $u = i$ and
$b^{\text{dn}}_i$ to be the last bounds considered by the binary search, and
we move back to Step~\ref{itm:start}.
\end{enumerate}
In particular, it is important that we run the two algorithms in
Step~\ref{itm:bs} in lockstep, since we only consider the results of one of
them, and throw the progress made by the other one away. So while we
potentially waste the progress made by one binary search, this does not matter,
because the other inner algorithm makes progress that we keep for the rest of the
algorithm. So the overall number of steps made by the algorithm is still $O(\log n)$. 

There are obviously many things that we need to prove to show that this
approach is correct. Firstly we show that for every slice $s_i$ we have $\Up(f)
\cap s_i \ne \emptyset$ or $\Down(f) \cap s_i \ne \emptyset$, so one of the two
binary searches always succeeds. 
Secondly we need to show that once the outer
binary search has converged, meaning we have found a slice $s_i$ with $\Up(f)
\cap s_i \ne \emptyset$ and $\Down(f) \cap s_i \ne \emptyset$ then we can find
a fixed point. For the two-dimensional case we can simply use one further
binary search in the slice $j$ to find a fixed point. For our three-dimensional
algorithm we need an entirely separate \emph{terminal algorithm} to do this job.

\paragraph{\bf Bounding the up-set in a three-dimensional instance.}

We use the same setup 
in 
our three-dimensional algorithm. We carry out an outer binary search on
dimension 3, and so each slice $s_i$ that the algorithm visits is now a
two-dimensional slice, rather than a one-dimensional slice.

For this to work, we need to be able to maintain succinct bounds on $\Up(f)
\cap s_i$ and $\Down(f) \cap s_i$, and unfortunately the simple binary search
bounds we used in the two-dimensional case are no longer sufficient. Here we
describe how a succinct bound on $\Up(f) \cap s_i$ can be formulated for a
two-dimensional slice $s_i$. Our bound on $\Down(f) \cap s_i$ can be obtained
by flipping both dimensions and applying the same ideas.

\twoDslice

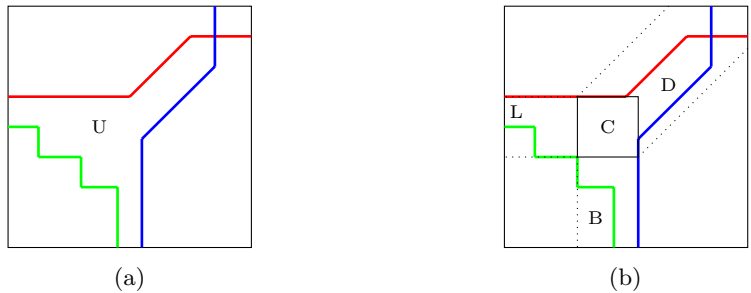
The structure of the up-set in a two-dimensional slice can be seen in
Figure~\ref{fig:2dslice} (a). As before we have red and blue boundaries for
dimensions 1 and 2. We also have the green line that represents the boundary
between the points $x$ with $x_3 > f_3(x)$, which lie below the green line, and
the points $y$ with $y_3 \le f_3(y)$, which lie above it. Note that
monotonicity implies that if $x$ satisfies $x_3 > f_3(x)$, then all points $p$ with
$p \le x$ and $p_3 = x_3$ also satisfy $p_3 > f_3(p)$. So we can always draw
the green boundary as a line that starts on the left or top boundaries of the
instance, and then moves weakly monotonically down, and weakly monotonically
left until it reaches the right or bottom boundaries.

The region that we are interested in is labelled $U$ in the diagram. This is
the region that lies below the red line, to the left of the blue line, and
above the green line, which is precisely the set of points that move weakly
upward in all three dimensions, and therefore $U = \Up(f) \cap s_i$. As can be seen, this region is not necessarily
simple to describe, since there is no bound on the number of turns that any of
the boundaries of $U$ make.
Nevertheless, we show that we can always find a
relatively succinct upper bound on $U$.

We say that a box-shaped region is a \emph{critical box} (defined formally
in Definition~\ref{def:cb}) if it satisfies the following conditions. 
\begin{itemize}
\item It only contains points in the up-set. 
\item It touches all three boundaries of the up-set in the sense that
\begin{itemize}
\item The point directly to the right of the bottom-right corner of the box
moves strictly down in dimension 1.
\item The point directly above the top-left corner of the box moves strictly down in
dimension 2.
\item The point diagonally below the bottom-left corner of the box moves
strictly down in dimension~3.
\end{itemize}
\end{itemize}
The region $C$ is a critical box in Figure~\ref{fig:2dslice} (b).
For each critical box we define
three \emph{lobes}: a left lobe that contains all points to the left of the
box, a bottom lobe that contains all points beneath the box, and a diagonal
lobe that contains all points diagonally above the box. These are shown as $L$,
$B$, and $D$, respectively in the figure. 

The key point is that, given any critical box, we have that $\Up(f)
\cap s_i$ is contained in the union of the box and its three lobes. Hence we
can bound the up-set using at most three box-shaped regions and one diagonal
region. Our algorithm will use bounds of this format.

Another crucial point that we will prove is that 
if $\Up(f) \cap s_i$ is non-empty, then an \emph{almost square} critical box
always exists. Specifically, this means that if the box has height $h$ and
width $w$, then $| h - w | \le 1$. The proof of this fact is actually quite
involved, and will be covered in Section~\ref{sec:almostsquare}.

\paragraph{\bf Refining a square-shaped region.}

So our bound on $\Up(f) \cap s_i$ will take the form of a critical box and its
three lobes. In the two-dimensional algorithm we had a bound $[a, b]$ on 
$\Up(f) \cap s_i$, and we used binary search to reduce the size of this bound
by one-half in each step. Analogously, in the three-dimensional algorithm, we have a
critical box and its lobes, and we must make constantly many queries to $f$ and
then reduce the area covered by the bound by a constant fraction. This turns
out to be quite a complex task, with several different cases that must be
considered.

We start by considering the case where
we have a square-shaped region that we know contains 
$\Up(f) \cap s_i$, which would have, for example, at the start of the algorithm,
where we can use the entire slice as a trivial bound on the up-set. We want to refine
this square-shaped region to reduce the area of our bound
$\Up(f) \cap s_i$ by a constant fraction.

\cbConfigIntro

We use a grid-search to achieve this. That is, we pick some constant $k$, we
overlay a $k \times k$ grid over the square, and we then query each point of
the grid. This
clearly uses constantly many queries. If we find a point in $\Up(f)$ in our
grid then we are done, and we can move to a new slice in the outer algorithm.

On the other hand, we know that if $\Up(f) \cap s_i$ is non-empty, then an
almost square critical box exists. If our grid search failed to find a point in
$\Up(f)$, then we know that all almost square critical boxes have height and
width less than or equal to $k$, because every point in a critical box lies in
$\Up(f)$ and we did not find a point in $\Up(f)$ in our grid. This means that the up-set must be contained in a
much smaller region.

We proceed by searching for a \emph{CB-config} (defined formally in
Definition~\ref{def:config_square}) in our grid, which is 
a $2 \times 2$ grid square satisfying the
following properties.
\begin{itemize}
\item The lower-left corner of the square moves strictly downward in dimension
3, while the upper-right corner moves weakly upward.
\item Either the top or bottom edge of the square has points that move toward
each other in dimension 1, with the left point moving weakly right and the
right point moving strictly left.
\item Either the left or right edge of the square has points that move toward
each other in dimension 2, with the lower point moving weakly upward and the
upper point moving strictly downward.
\end{itemize}
Figure~\ref{fig:cbconfig} (a) shows an example of a CB-config. 

We prove that if an almost-square critical box exists, then the grid will
contain a CB-config. So if we fail to find a CB-config we can declare that 
$\Up(f) \cap s_i$ is empty.
In this case we will prove that if $\Up(f) \cap s_i$ is empty, then $\Down(f)
\cap s_i$ is non-empty, so 
the outer algorithm can proceed as if the down-set
inner algorithm found a point in 
$\Down(f) \cap s_i$.  

If we do find a CB-config, then we prove that all almost square critical boxes
must lie in a $7k \times 7k$ square region $S$ surrounding the CB-config (which
is specified formally in Lemma~\ref{lem:config_cb}). Hence, if we take the
union of the left,
bottom, and diagonal lobes of $S$, then this union  is guaranteed to contain all almost
square critical
boxes and all of their lobes, which themselves contain $\Up(f)$. So this region is a
new bound on $\Up(f) \cap s_i$. As shown in Figure~\ref{fig:cbconfig} (b), so long
as $k$ was chosen to be large enough, this
rules out a significant fraction of the area of the original square. 

We can now proceed by performing a more refined grid search on the square $S$.
Since $S$ contains all of the almost square critical boxes, and since an almost
square critical box always exists so long as $\Up(f) \cap s_i$ is non-empty,
there is guaranteed to be a point in $S$ that lies in $\Up(f)$ whenever $\Up(f)
\cap s_i$ is non-empty. So it is valid for us to proceed in this way.

The procedure continues until we find a point\footnote{
Or until we discover that 
$\Up(f) \cap s_i$ is empty, or
until the down-set inner
algorithm, which is running parallel, finds a down-set point. In both of these
cases we would throw
away any progress we made here. We proceed in this section assuming that we do
need to carry over the progress of the up-set inner algorithm.}
in 
$\Up(f) \cap s_i$, at which point we jump to a new slice.

\paragraph{\bf Jumping to a new slice.}

When we jump to a new slice $j > i$, we transpose our existing bound on $\Up(f)
\cap s_i$ to the new slice. Our bound so far has been defined by a central box
and its three lobes, and to transpose this bound to a new slice we 
simply translate all of the regions by $(j - i) \cdot (1, 1, 1)$. We will show that
this gives us a valid bound on $\Up(f)$ in the new slice.

\inlobe

However, we cannot proceed as before by refining the central square of this bound
using a grid search. Figure~\ref{fig:inlobe} shows an example of why this is
the case: in the example the entire up-set (labelled as $U$ in the figure) is
contained in one of the lobes, so we lose the invariant that all almost square
critical boxes lie in the central square when we jump to a new slice.

So we must now use a different procedure to reestablish that invariant.
For the example given in 
Figure~\ref{fig:inlobe}, we can actually make a single query at the point $p$
shown in the figure to
determine that the up-set lies in the left lobe: if we find that $p_1 >
f_1(p)$, then monotonicity implies that all points that lie directly below $p$
also move strictly to the left, meaning that if $\Up(f) \cap s_i$ exists, it
must lie strictly within the left lobe. We can likewise make queries at the
points $q$ and $r$ shown in Figure~\ref{fig:inlobe} to determine if $\Up(f) \cap s_i$ lies entirely within the bottom or diagonal lobes. 

In the case where the queries at $p$, $q$, and $r$ do not tell us that 
$\Up(f) \cap s_i$ lies entirely within one of the lobes, we prove that every
almost square critical box must lie within a square region that encompasses the
central square, whose area is at most a constant times the area of the central
square. So we reestablish the invariant for this new larger square. The fact
that we have increased the area in this case does not matter, because we
perform this step only when we jump to a new slice, and we jump to a new slice
at most $\log n$ times. Moreover, since the area is increased by a constant
factor, constantly many reduction steps can cancel out this increase. So our
overall $O(\log n)$ query complexity is not affected.

\paragraph{\bf Dealing with non-square regions.}

If the jumping-slices procedure finds a square region, then we can proceed with
a grid-search over this region as before. However, sometimes the jumping-slices
procedure determines that the up-set lies in a lobe. If, for example, the
up-set is contained in a left lobe, then it may be the case that the width of
the lobe is more than a constant times the height, which means we cannot place
a constant sized grid of squares over the lobe. So we cannot proceed via
grid-search in this case, and we must instead use a different procedure. 

\leftlobe

Here we describe the procedure for a left lobe with width $w$ and height $h$.
The procedure for a bottom lobe is symmetric, and can be obtained by exchanging
dimensions 1 and 2. The procedure for a diagonal lobe follows essentially the
same approach, but has the added annoyance of working with diagonal regions.

The high-level idea is to use binary search to try to reduce the width of the
lobe, while maintaining the invariant that the remaining space contains all
almost square critical boxes. We keep doing this until $w \le c \cdot h$ for
some constant $c$. At which point we can continue via grid search, since we can
now place a constant-sized grid of squares over the space.

To process a left lobe, we make a query at the point $q$ shown in
Figure~\ref{fig:leftlobe}, which is the point that lies half-way along the
bottom of the lobe. 

In Figure~\ref{fig:leftlobe} (a) we have $q_1 > f_1(q)$. As has been observed
in prior work on contraction maps~\cite{Chen0Y24}, the point $q$ induces the
following \emph{cone constraint}: any point $x$ in the cone defined $x \ge q$
and $|x_2 - q_2| \le x_1 - q_1$ must also satisfy $x_1 > f_1(x)$, because
otherwise we would have a violation of contraction. Overlaying
this cone onto the lobe shows us that once we move $h$ units to the right of
$q$, all points in the lobe will move strictly rightward. Thus we can eliminate
all points to the right of $q + h \cdot e_1$, because none of these points can
lie in $\Up(f)$, and since we assume $h \ll w$, this roughly halves the area of
the lobe.

The case where $q_1 \le f_1(q)$ is more problematic, however. Here,
as shown in Figure~\ref{fig:leftlobe} (b), the cone-constraints applied to $q$
imply that all points $p$ in the lobe with $p_1 \le q_1 - h$ satisfy $p_1 \le f_1(p)$.
Moreover, any almost square critical box must have height at most $h$, and so
width at most $h + 1$. Since a critical box is required to touch a point $p$
with $p_1 > f_1(p)$ on its right boundary, there cannot exist any almost square critical box whose
lower-left corner $x$ satisfies $x_1 \le q_1 - 2h - 1$. 

So if we consider the line $L$ defined by points $x$ with $x_1 = q_1 - 2h -
1$, we know that the area to the right of the line satisfies the invariant that
the region contains all almost square critical boxes. However, we cannot remove
the area to the left of $L$, because we are required to maintain a bound on the
entire up-set, and the up-set may still exist to the left of $L$. 

To solve this we make one further query at the point $p$ shown in
Figure~\ref{fig:leftlobe} (b), which lies half-way along the line $L$. We show
that the response to this query always allows us to rule out half of the
remaining area. In the example given in the picture, the region above and to
the left of $p$ is excluded, for example. 

We could attempt to perform binary search along the line $L$ to rule out the
entire space to the left of $L$, but this would cost us $\log n$ queries, which
is too expensive. Instead, we will slowly reduce this remaining space over the
subsequent iterations of the algorithm.

Specifically, we split the lobe into a \emph{main lobe}, which is the area on
the right of $L$, which always contains all of the almost square critical
boxes, and a \emph{sub-lobe}, which lies to the left of $L$ and contains the
rest of the up-set. 
In each step we reduce the size of the main-lobe by
a constant fraction, and then if we create a new-sub lobe, we make some extra
queries to merge it with the existing sub-lobe. Finally we make one extra query
to reduce the size of the sub-lobe by a constant fraction.
An example of this is shown in Figure~\ref{fig:leftlobe} (c). Here, like in 
Figure~\ref{fig:leftlobe} (a) we have eliminated a constant fraction of the
main lobe, and we then make one additional query at the point $r$ to rule out a
constant fraction of the sub-lobe.

\paragraph{\bf The actual regions considered by the algorithm.}

Eventually we will reduce the size of the main-lobe enough so that grid-search
can once again be applied. However, we may still have a sub-lobe at this point,
so our grid-search algorithm also needs to keep reducing the size
of the sub-lobe in each step to ensure that our overall bound on the up-set
keeps decreasing by a constant fraction.

\usbound

Overall, the shape of our bound on the up-set can always be described by the
shape shown in Figure~\ref{fig:usbound} (which is defined formally in
Definition~\ref{def:state}). This shape has a central region that is 
roughly square shaped, meaning that the height and width are constant multiples
of one another. The central region is
surrounded by a left, bottom, and diagonal lobe, each of which have sub-lobes
attached. Of course, in most scenarios we do not use the full generality of the
shape. For example, when we are processing a left lobe, the central region and
the bottom and diagonal lobes are both empty. 

The main technical result that enables the algorithm is that, if we are given a
region of the form shown in Figure~\ref{fig:usbound}, then we can make
constantly many queries and either find a point in $\Up(f) \cap s_i$, correctly
declare that $\Up(f) \cap s_i = \emptyset$, or find a
new region of the same form in which the area has been reduced by a constant
fraction. This then gives us our inner up-set algorithm that can be used with
the outer binary search described earlier.

\paragraph{\bf The terminal phase of the algorithm.}

The outer algorithm continues until it finds a slice $s_i$ that contains both
an up-set point and a down-set point. Now we must find a fixed point in this
slice. 

The first thing we must do is to prove that $s_i$ actually contains a fixed
point\footnote{The 
prior algorithm for monotone functions given in~\cite{FPS22} ensured that the remaining
space always lies between the up-set and down-set points that are found by the
outer algorithm, which ensures that $s_i$ must contain a fixed point. We are unable to maintain this invariant
here, because if we try to reduce the remaining space by excluding all points
that do not lie above some up-set point $x$, then we may change the set of
almost-square critical boxes, and
invalidate our bounds. So this extra step is necessary for our setting.}. 
This might not be the case if $f$ was just a monotone function. However,
using the extra properties of a \DMAC instance, we can show that this is indeed
the case. Specifically, we can use the path of fixed points $l = p^1, p^2,
\dots, p^k = g$ 
with $\| p^i - p^{i+1} \|_\infty = 1$
from the least fixed point $l$ to the greatest fixed point $g$
whose existence we proved when we put
$\DMAC$ in \UEOPL. Since slice $s_i$ contains an up-set point, by Tarski's
theorem the greatest fixed point either lies in $s_i$ or above it. Likewise,
since $s_i$ contains a down-set point, the least fixed point either lies in
$s_i$ or below it. If neither the greatest or least fixed points lie in $s_i$,
then the path must pass through $s_i$, so there is indeed a fixed point in
$s_i$.

Next we must actually find this fixed point. We design a specialized terminal phase
algorithm for this task. We show that any query to the slice $s_i$ either
returns a fixed point, or rules out a vertical half-space, a horizontal
half-space, or diagonal half-space that is aligned with the vector $(1, 1)$. We
use this fact to build a $O(\log n)$ query algorithm that finds a fixed point
in the slice.

\paragraph{\bf Running time.}

In total, the algorithm makes $O(\log n)$ queries and returns a fixed point of
$f$. Moreover, all of the steps of our algorithm run in time that is polynomial
in the size of the representation of $f$.

\section{Reducing MonotoneContraction to DMAC}
\label{sec:monotone_to_dmac}

In this section, we present a promise-preserving polynomial-time reduction from
\MonotoneContraction to \DMAC.

Prior work has defined the concept of a promise-preserving reduction from problem A to problem B, which ensures that if problem A has no violations, the corresponding instance of problem B must also remain free of violations~\cite{FGMS20}. 
In other words, we must never map a non-violation solution of A onto a violation solution of B.

\subsection{Creating a DMAC instance}
 
Let $g : [0, 1]^d \rightarrow [0, 1]^d$ be a (purported) monotone contraction map.
Our reduction discretizes the continuous domain $[0,1]^d$ of $g$ and then applies rounding 
to the displacements of $g$ at grid points, in order to define a \DMAC instance~$f$ on the 
discrete domain $G = \{0,1,2,\dots, n\}^d$, which differs only from the 
Definition \ref{def:dmac_def} in that we start from $0$ rather that $1$, which is for 
convenience in this section. This is otherwise inconsequential as it is just a relabelling.
First, we set 
$$\beps \coloneqq 1/\Ceil{1/\eps}.$$
This choice gives two key properties:
$\frac{1}{\hat{\eps}} \in \mathbb{N}$ and $\hat{\eps} \le \eps$.
%
Then we discretize the $[0, 1]^d$ domain using a grid of width $\beps$ in each dimension, which
defines a discrete grid:
\begin{equation}
\label{def:hatG}
\hat{G} = \{0, \beps, 2\beps, \dots, 1\}^d.
\end{equation}
So $G$ has grid width 1 and $\hat{G}$ has grid width $\beps$.
Every point in our \DMAC domain $G$ will be in one-to-one correspondence with points
in $\hat{G}$ according to the bijection $\Stretch: \hat{G} \mapsto G$ with 
$\Stretch(\hat{x}) = \frac{1}{\beps} \cdot \hat{x}$. We can define the inverse of the transformation as $\Stretch^{-1}: G \mapsto \hat{G}$ as $\Stretch^{-1}(x) = \beps \cdot x = \hat{x}$.
To create $f: G \mapsto G$ we will first \emph{clamp} displacements of $g$ in 
each dimension 
using the following function $\sgne: [0,1] \mapsto \Set{\pm \beps,0}$:
\begin{equation*}
\sgne(\hat{z}) =\begin{cases} 
	\phantom{-}\beps &\quad\text{if $\hat{z} \geq \beps$,}\\ 
	\phantom{-}0&\quad\text{if $\Abs{\hat{z}} < \beps$,}\\ 
	-\beps&\quad\text{if $\hat{z}\leq -\beps$.} 
\end{cases} 
\end{equation*}
This function rounds all values in the range $(-\beps, \beps)$ to 0, while
clamping values outside of the range to either $-\beps$ or $\beps$.
We then stretch these values so that $-\beps$ becomes $-1$ and $\beps$ becomes $1$ using $\Stretch$.
We then define $f$ via a dimension-wise displacement function $h$ as follows.
\begin{definition}
\label{def:h}
For all $y \in G$, and any dimension $i \in [d]$, we define:
\begin{equation}
\label{def:h}
h_i(y) = \Stretch\left(\sgne(g_i(\hat{y}) - \hat{y}_i)\right), \text{ with } \hat{y} = \Stretch^{-1}(y). 
\end{equation}
To define $h_i$ at a point $y \in G$, we first compute the displacement of $g$ in dimension $i$ for the point $\hat{y} \in \hat{G}$ corresponding to $y$ in the domain $\hat{G}$. Then, we apply the $\sgne$ function to this displacement, producing a value in $\{\pm \beps, 0\}$, which is finally mapped to displacement values $\{\pm 1, 0\}$ corresponding to the domain $G$ using function $\Stretch$. 
The overall displacement $h$ for a point $y \in G$ is then defined as:
\begin{equation*}
h(y) \coloneqq (h_1(y),h_2(y),\dotsc,h_d(y)).
\end{equation*}
Finally, we define $f: G\to G$ for the resulting \DMAC instance as: $$f(y) \coloneqq y + h(y)$$
\end{definition}
This completes a definition of how our reduction creates in polynomial time a \DMAC instance from a
\MonotoneContraction instance.
What is left to do is show that the solutions of the resulting \DMAC instance can be mapped
in polynomial time to solutions of the original \MonotoneContraction instance.
First we explain why we use both $G$ and $\hat{G}$, and introduce some convenient notation to deal
with this.

\paragraph{Different grid widths.}
We intentionally chose to make the grid width in \DMAC equal to 1, rather than $\beps$. 
This makes our containment proofs of \DMAC in \UEOPL and our complicated algorithm for 3d \DMAC
cleaner.
It does make the proofs in this section, showing the correctness of the reduction from \MonotoneContraction
to \DMAC, slightly more involved though.
For convenience, we introduce functions $\hat{f}$ and $\hat{h}$ that are analogous to $f$ and $h$ but 
that keep the grid width as $\beps$ rather than 1. 
%
%
%
We also define $\hat{x}$ as the point in $\hat{G}$ that
corresponds to $x \in G$, as we did in Definition~\ref{def:h}. 

\begin{definition}
For all $\hat{y} \in \hat{G}$, and any dimension $i \in [d]$, we define:
\begin{equation}
\label{def:hhat}
\hat{h}_i(\hat{y}) = \sgne(g_i(\hat{y}) - \hat{y}_i).
\end{equation}
The overall displacement $\hat{h}$ at $\hat{y}$ is then $\hat{h}(\hat{y}) \coloneqq
(\hat{h}_1(y),\hat{h}_2(y),\dotsc,\hat{h}_d(y))$.
Then $\hat{f}: \hat{G} \to \hat{G}$ is defined as: $\hat{f}(\hat{y}) \coloneqq \hat{y} + \hat{h}(\hat{y})$. Finally, given a point $y \in G$, 
we define $\hat{y} := \Stretch^{-1}(y) = \beps \cdot y \in \hat{G}$. 

\end{definition}

As stated in the following observation,
we can easily relate the values of $\ell_\infty$ norms of interest within $G$ and $\hat{G}$, as we have
simply multiplied all terms used within the norm by $\beps$ when going from $f$ to $\hat{f}$ and from
$x, y$ to $\hat{x}, \hat{y}$, respectively.
\begin{observation}
\label{obs:gversusgprime}
We have that: $\beps \cdot \Normi{x-y} = \Normi{\hat{x} - \hat{y}}$ and 
$\beps \cdot \Normi{f(y) - f(x)} = \Normi{\hat{f}(\hat{y}) - \hat{f}(\hat{x})}$.
\end{observation}

\subsection{Mapping solutions back}

We now show how each type of solution of the \DMAC instance produced by our reduction can be mapped
back, in polynomial time, to a solution of the original \MonotoneContraction instance.

First, we will deal with actual fixed point solutions.
We show how a \ref{D1} solution for $f$ trivially maps back to an \ref{M1} solution.
\begin{lemma}
A point $y \in G$ is a \ref{D1} solution of the \DMAC instance arising from our reduction only if
$\hat{y} = \beps \cdot y \in \hat{G}$ is an \ref{M1} solution of the original \MonotoneContraction instance.
\end{lemma}
\begin{proof}
Suppose we have a \ref{D1} solution, $y \in G$ such that $f(y) = y$.
By definition of $f$, this means that $h(y) = 0 \implies \hat{h}(\hat{y}) < |\beps|$.
By definition of $\hat{h}$, this means that $\Normi{g(\hat{y}) - \hat{y}} < \beps \le \eps$, 
which completes the proof.
\end{proof}





Now we deal with violation solutions. 
The following lemma states that if we can find a violation of monotonicity for $f$, 
then we can efficiently turn that into a violation of monotonicity for $g$. 
This means that if $g$ is monotone, then $f$ is also monotone. 

\begin{lemma} 
\label{lem:g monotone violation}
If there exists a violation of type \ref{DV1} for $f$, then we can find a violation of type
\ref{MV1} for~$g$ in polynomial time. 
\end{lemma}

\begin{proof}
Let $x, y \in G$ witness a violation of type \ref{DV1}, i.e., $x \leq y$ and
$f(x) \not\leq f(y)$. 
Then there exists an $i\in [d]$ such that $f_i(y) < f_i(x)$. 
We already know that $x_i \leq y_i$. 
%
Suppose $y_i \ge x_i + 2$. 
%
Then since $|h_i(y)| \le 1$ and $|h_i(x)| \le 1$ we have that $f_i(y) \geq f_i(x)$, which is a contradiction,
so $y_i \ngeq x_i + 2$. 
%
Thus, we have $x_i \le y_i \leq x_i + 1$. 
%
%
There are two cases to consider. 
The first is where $y_i = x_i$.
In this case, we have $h_i(y) < h_i(x)$, which is equivalent to 
\[ \sgne(g_i(\hat{y}) - \hat{y}_i) < \sgne(g_i(\hat{x}) - \hat{x}_i) = \sgne(g_i(\hat{x}) - \hat{y}_i) \implies
g_i(\hat{y}) < g_i(\hat{x}), \] so $g(\hat{x}) \nleq g(\hat{y})$, giving us a violation of type \ref{MV1} for the pair $\hat{x},\hat{y}$.

The second case is where $y_i = x_i + 1$. 
Then we must have $h_i(x) = 1$ and $h_i(y) = -1$, which implies that 
$$g_i(\hat{x}) \geq \hat{x}_i + \beps  = \hat{y}_i \ge g_i(\hat{y}) + \beps > g_i(\hat{y}),$$ 
so $g(\hat{x}) \nleq g(\hat{y})$, and we again have a violation of type \ref{MV1} for the pair $\hat{x},\hat{y}$.
\end{proof}

Next we deal with violations of type \ref{DV2}, which occur whenever $f$ is
strictly non-contracting, i.e. expanding. 
In this case, we can efficiently map these violations onto \ref{MV2} violations of contraction in $g$. 
This means that if $g$ is contracting, then $f$ is non-expansive. 
Given a pair of points $x, y$ witnessing a violation of type \ref{DV2} for $f$, we will show that $\hatx,\haty$ witness a violation
of type \ref{MV2} for~$g$.
To that end, we first show in Lemma~\ref{lem:boundfgap} that there exists a dimension $i \in [d]$
where a violation of strict contraction is witnessed for the pair of points $x, y$ and in which the displacement of
$\hat{f}$ cannot be more than $\beps$ smaller than the displacement of $g$ at $\hatx, \haty$.

\contractionViolation

\begin{lemma} 
\label{lem:boundfgap}
Given a violation $x,y$ of type \ref{DV2} for $f$, then for some $i \in [d]$ we have that:
\begin{equation}
\label{eq:lemma7assumption1}
|f_i(y) - f_i(x)| > \Normi{y-x} \ge 1.
\end{equation}
and 
\begin{equation}
\label{eq:boundfgap}
|g_i(\hat{y}) - g_i(\hat{x})| \ge |\hat{f}_i(\hat{y}) - \hat{f}_i(\hat{x})| - \hat{\eps}.
\end{equation}
\end{lemma}
\begin{proof}
If $x,y\in G$ are a violation of type \ref{DV2}, i.e., $\Normi{f(y) - f(x)} > \Normi{y-x}$, then
there exists a dimension, say $i\in [d]$, such that \eqref{eq:lemma7assumption1} holds.
Recall that $f_i(y) = y_i + h_i(y)$.
We use a case analysis on the values of $h_i(x)$ and $h_i(y)$.

First assume that $y_i > x_i$. 
%
Since $h_i(x), h_i(y) \in \{-1,0,1\}$, there are nine possibilities
for the values of $(h_i(x),h_i(y))$.
First we show that, under our assumption $y_i > x_i$, condition~\eqref{eq:lemma7assumption1} implies 
that $h_i(y) > h_i(x)$.
Towards a contradiction assume that $h_i(y) - h_i(x) \le 0$.
Since $|y_i - x_i| \le \Normi{y - x}$, we require
$$|f_i(y) - f_i(x)| \ge |y_i - x_i| + 1.$$ 
But, as $y_i \ge x_i + 1$ and $h_i(y) - h_i(x) \in \{-2, -1, 0\}$, this is not possible, since
\begin{align*}
|f_i(y) - f_i(x)| -  |y_i - x_i| = \ & |y_i + h_i(y) - x_i - h_i(x)| - |y_i - x_i| \\
= \ & |(y_i - x_i) + (h_i(y) - h_i(x))| - |y_i - x_i| \\
\le \ & 0.
\end{align*}
Thus we have $h_i(y) > h_i(x)$, as claimed, i.e., the only three relevant cases,
illustrated in Figure~\ref{fig:contractionviolation}, are:
$$(h_i(x),h_i(y)) \in \{(0,1),(-1, 0),(-1,1)\}.$$ 

Assume that $(h_i(x),h_i(y)) = (0,1)$, then the smallest that 
$|\hat{f}_i(\hat{y}) - \hat{f}_i(\hat{x})|  - |g_i(\hat{y}) - g_i(\hat{x})|$ can be is when
$g_i(\hat{x}) - \hat{x}$ is just slightly less than $\hat{\eps}$ and $g_i(\hat{y}) = \hat{y} + \hat{\eps}$.
Then $|\hat{f}_i(\hat{y}) - \hat{f}_i(\hat{x})| = 2\cdot \hat{\eps}$ and 
$|g_i(\hat{y}) - g_i(\hat{x})| > \hat{\eps}$, which implies~\eqref{eq:boundfgap}.

The argument for the case that  $(h_i(x),h_i(y))=(-1,0)$ is symmetric
to the case for $(h_i(x),h_i(y))=(0,1)$, with the smallest 
gap $|\hat{f}_i(\hat{y}) - \hat{f}_i(\hat{x})|  - |g_i(\hat{y}) - g_i(\hat{x})|$ occurring when
$g_i(\hat{y}) = \hat{y} - \beps$ and when $g(\hat{x})- \hat{x}$ is just slightly less than $-\beps$.

If $(h_i(x),h_i(y)) = (-1,1)$, then the clamping done by $h$ has (weakly) reduced both displacements for $x$ and $y$, 
so in the $\hat{G}$ domain we have the following inequality, which is stronger than~\eqref{eq:boundfgap}:
\begin{equation}
\label{eq:stronger}
|g_i(\hat{y}) - g_i(\hat{x})| \ge |\hat{f}_i(\hat{y}) - \hat{f}_i(\hat{x})|.
\end{equation}
This completes the proof of the theorem under the assumption that $y_i > x_i$.
The proof for the case of $y_i < x_i$ is symmetric.

Finally, assume that $y_i = x_i$.
In this case, only two of the nine cases are relevant, i.e., are consistent with~\eqref{eq:lemma7assumption1},
namely where $(h_i(x),h_i(y)) \in \{(1,-1),(-1,1)\}$.
Intuitively, now because $x_i = y_i$ and dimension $i$ 
witnesses $|f_i(y) - f_i(x)| \ge 2$, the
displacements for $x_i$ and $y_i$ must go in strictly opposite directions. 
Formally, since we require $|f_i(y) - f_i(x)| = |y_i + h_i(y) - x_i - h_i(x)| = |h_i(y) - h_i(x)| \ge 2$,
and $h_i(y), h_i(x) \in \{-1,0,1\}$, we must have $|h_i(y) - h_i(x)|  = 2$.
Now, since the absolute value of both displacements in $\hat{G}$ was at least $\beps$, the clamping
has weakly reduced the gap, and we have the stronger inequality~\eqref{eq:stronger} again, which completes
the proof.
\end{proof}

We can now prove the following.

\begin{lemma} If there exists a violation of type \ref{DV2} for $f$, then we
can find a violation of type \ref{MV2} for~$g$ in polynomial time.
\end{lemma}

\begin{proof}
Let $x, y$ witness the \ref{DV2} violation of $f$.
We will use Lemma~\ref{lem:boundfgap} to show that $\hatx,\haty$ is our desired \ref{MV2} violation 
for~$g$.
First we choose $i \in [d]$ so that~\eqref{eq:lemma7assumption1} and~\eqref{eq:boundfgap}
from Lemma~\ref{lem:boundfgap} both hold.

Note that we have $\Normi{\hat{y}-\hat{x}} + \hat{\eps} \le |\hat{f}_i(\hat{y}) -
\hat{f}_i(\hat{x})|$, which is due to \eqref{eq:lemma7assumption1} and the fact that the grid width
is $\hat{\eps}$.
Putting this together with~\eqref{eq:boundfgap}, we have: 
\begin{align*}
\Normi{\hat{y}-\hat{x}} + \hat{\eps} \le \ & |\hat{f}_i(\hat{y}) - \hat{f}_i(\hat{x})| \\
\le \ & |g(\hat{y}_i) - g_i(\hat{x})| + \hat{\eps} \\
\le \ & \Normi{g(\hat{y}) - g(\hat{x})} + \hat{\eps},
\end{align*}
so we have
$$
\Normi{\hat{y}-\hat{x}} \le \Normi{g(\hat{y}) - g(\hat{x})},
$$
which shows that we get an \ref{MV2} violation of strict contraction for points  $\hat{x},\hat{y}$,
and which completes the proof.
\end{proof}

To conclude, we have shown that each solution type of $\DMAC$ for $f$ can be
mapped onto a solution of $\MonotoneContraction$ for $g$. Furthermore, it is clear that $f$
can be constructed from $g$ in polynomial time, establishing the correctness of our reduction.
Our reduction maps all approximate fixed points of $g$ to fixed
points of $f$, so we have proved the following theorem. 

\begin{theorem} 
\label{thm:mctodmac}
There is a promise-preserving polynomial-time reduction from
$\MonotoneContraction$ to $\DMAC$.
\end{theorem}

\section{DMAC is in UEOPL}
\label{sec:dmac_in_ueopl}

In Section~\ref{sec:UEOPLoverview} we briefly outlined key features of our proof of containment 
of \DMAC in \UEOPL, which goes via a reduction to the problem \OPDC. 
In this section, we give full details of our reduction from \DMAC to \OPDC.

Our very first step, in Section~\ref{sec:unitdisplacement}, is to show that we can reduce \DMAC to a 
variant of it where all displacements have length at most $1$ in the $\ell_\infty$-norm.
Our reduction from \MonotoneContraction to \DMAC had this property, and here we show that 
any \DMAC instance can be efficiently transformed into another instance with this property.
It will make a number of subsequent proofs in this section simpler, and will also be used
by our algorithm for 3d instance later in the paper.
After Section~\ref{sec:unitdisplacement}, in the rest of the paper we will always assume
that our \DMAC instances have this property.
 
For a \DMAC instance, the domain of the function $f$ is a discrete grid, which is a complete
lattice under the pointwise-minimum ordering). 
If the \DMAC instance is violation-free then $f$ is \emph{monotone}, which is also known as
\emph{order preserving}.
In Section~\ref{sec:verifyingLFPs}, we first recall the Knaster-Tarski Theorem (henceforth just
Tarski's Theorem), which implies that an order-preserving function that maps a complete lattice to
itself has a \emph{least fixed point}.
Then, we describe a number of properties of least fixed points of monotone functions that are also
\emph{non-expansive}, as $f$ will be in a \DMAC instance with no violations.
The properties will be used in our reduction from \DMAC to \OPDC.
The reduction itself appears in Section~\ref{sec:DMACtoOPDC}.

\subsection{DMAC with Unit Displacements}
\label{sec:unitdisplacement}

Here we show that $\DMAC$ can be reduced, in polynomial time, to $\DMAC$ in
which all displacements have length at most $1$ in the $\ell_\infty$-norm. The
reduction creates a new function $f' : G \rightarrow G$ defined as follows.
\begin{equation*}
f_i'(x) = \begin{cases}
x_i - 1& \text{if $f_i(x) < x_i$,}\\
x_i & \text{if $f_i(x) = x_i$,}\\
x_i + 1& \text{if $f_i(x) > x_i$.}\\
\end{cases}
\end{equation*}
This function simply clamps the displacement of $f$ so that it moves distance
at most 1 in each dimension, in much the same way as we did when we reduced
$\MonotoneContraction$ to $\DMAC$. It is therefore clear that we have
$\| f'(x) - x \|_\infty \le 1$ for all $x \in G$.

The next lemma proves that this is a correct promise-preserving reduction. 
\begin{lemma}
\label{lem:unitdisplacement}
Every fixed point of $f'$ is a fixed point of $f$. Furthermore, if $f$ is
violation-free, if $x,y$ is a \ref{DV1} or \ref{DV2} violation for $f'$, then it 
is also, respectively, the same type of violation for $f$.
\end{lemma}
\begin{proof}
It is clear from the definition of $f'$ that $f'(x) = x$ if and only if $f(x) =
x$. So every fixed point of $f'$ is a fixed point of $f$.

Next we deal with violations of type \ref{DV1}. Suppose we have points 
$x, y \in G$ with $x \le y$ and $f'(x) \not \le f'(y)$. 
Let $i$ be a dimension in which $f'_i(x) > f'_i(y)$. 
Observe that we cannot have $f'_i(x) = x_i - 1$, since this would imply that
$f'_i(y) < f'_i(x) = x_i - 1 \le y_i - 1$, which contradicts the definition of
$f'$.  So we have $f'_i(x) \ge x_i$, and the definition of $f'$ then implies $f_i(x) \ge
f'_i(x)$, because each displacement in $f'$ was generated by clamping a
displacement in $f$. 
Applying the same argument symmetrically implies that we cannot have $f'_i(y) = y_i + 1$, 
since this would imply $f'_i(x) > f'_i(y) = y_i + 1 \ge x_i + 1$, which also contradicts the
definition of $f'$. Therefore we have $f_i(y) \le f'_i(y)$. 
Hence we have $f_i(x) \ge f'_i(x) > f'_i(y) \ge f_i(y)$, which implies that $x$
and $y$ are a violation of monotonicity in $f$ as claimed.

Finally, we deal with violations of type \ref{DV2}. Suppose 
that $x$ and $y$ are a pair of points satisfying 
$\Norm{f'(x) - f'(y)}_{\infty} > \Norm{x-y}_{\infty}$. 
Let $i$ be a dimension in which $|f'_i(x) - f'_i(y) | > \| x - y\|_\infty$, and
suppose without loss of generality that $x_i \le y_i$. 

\begin{itemize}
\item If $x_i < y_i$, then note that we cannot have $f'_i(x) = x_i + 1$,
because  we have $y_i - 1 \le f'_i(y) \le y_i + 1$, and so this would imply
that $|f'_i(x) - f'_i(y)| \le \| x - y \|_\infty$. Therefore we have $f'_i(x)
\le x_i$, and then the definition of $f'$ implies that $f_i(x) \le f'_i(x)$.
Applying the same argument symmetrically allows us to conclude that
$f'_i(y) \ge y_i$, and therefore $f_i(y) \ge f'_i(y)$. So we have $|f_i(y) -
f_i(x)| \ge | f'_i(y) - f'_i(x) | > \| x - y \|_\infty$, which implies that
$x$ and $y$ violate non-expansion in $f$.

\item If $x_i = y_i$, then note we have
$|f'_i(x) - f'_i(y) | > \| x - y\|_\infty \ge 1$, where the second inequality
holds because $x$ and $y$ are distinct. Since all displacements in $f'$ are at
most unit length by construction, the only way for this to be true is for
$f'_i(x)$ and $f'_i(y)$ to move in opposite directions, and for both to move
strictly away from each other. So let us assume without 
loss of generality that $f'_i(x) <
x_i$ and $f'_i(y) > y_i$. The definition of $f'$ then tells us that $f_i(x) \le
f'(x)$, and $f_i(y) \ge f'_i(y)$. So as before we have $|f_i(y) -
f_i(x)| \ge | f'_i(y) - f'_i(x) | > \| x - y \|_\infty$, which implies that
$x$ and $y$ violate non-expansion in $f$.
\end{itemize}

\end{proof}

\subsection{Verifiable Least Fixed Points}
\label{sec:verifyingLFPs}
 
Formally, order preservation is defined as follows.
\begin{definition}
	A function $f: L \rightarrow L$ defined over a complete lattice $(L, \preceq)$ is called
	\emph{order-preserving} w.r.t. $\preceq$ if, for all $x,y \in L$, we have that $x \preceq y$
	implies that $f(x) \preceq f(y)$.
\end{definition}
Tarski theorem is the following.
\begin{theorem}[Tarski~\cite{Tarski55}]
	\label{thm:tarski1}
	Let $(L, \preceq)$ be a complete lattice, and let $f: L \rightarrow L$ by a function that is 
	order-preserving w.r.t. $\preceq$. 
	The set of fixed points of $f$ form a complete lattice.
\end{theorem}

A complete lattice cannot be empty, so this theorem guarantees that at least one fixed point 
of $f$ exists.
It also guarantees that, w.r.t. $\preceq$, a least and a greatest fixed point exist, which may be
the same point.


In order to place the $\DMAC$ problem in $\UEOPL$, we need
to have some notion of a unique solution for the problem. While a monotone contraction has a single
fixed point, discretizing the problem into $\DMAC$ may result in multiple fixed points, each
representing an approximate fixed point of the original function. Therefore, we require a way to
distinguish one of these fixed points in order to designate it as the unique solution.

Initially, we will assume that we are given a $\DMAC$ instance defined by $f$ that
has no violations, and we will deal with violating instances at the very end of
the reduction. 
When there are no violations of monotonicity, Theorem~\ref{thm:tarski1} says that $f$ has a
least (and greatest) fixed point. 
We would like to declare the least fixed point to be our distinguished unique solution. 

However, prior work has shown that it is NP-hard to decide whether a given
fixed point of a monotone function is the least fixed point\footnote{Actually, their stated result
is that the least fixed point is NP-hard to compute, but their proof technique also applies to the
verification problem as well.}, even for one-dimensional instances~\cite{EPRY20}.
We show that this is not the case for $\DMAC$ instances: we can leverage
the non-expansion properties to build a polynomial-time algorithm that verifies
whether a given fixed point is the least fixed point. 
In this section we provide that algorithm. 

\paragraph{\bf Contiguity of fixed points.}
A key property that allows for a polynomial-time verification algorithm is contiguity: if a given
fixed point $x \in G$ is not the least fixed point, there must exist a lesser fixed point at 
distance at most $1$ from $x$.
%
%
To show the existence of these other fixed points, we will employ Tarski's theorem on suitable
sub-lattices of the grid domain $G$. 
We use the following notation to refer to sublattices of the grid, defined between two points $a$
and $b$, and we also introduce notation that is useful for working with dimension-wise displacements
of functions.

\begin{definition}[Sublattices and projections] For any two points
	$a,b \in G$ with $a \leq b$, let $$\sublattice(a,b) \coloneqq \Set{z\in G: a \leq z \leq b}.$$
	Sometimes we will write a sublattice $\sublattice(a,b)$ as just $L$, leaving its corners $a$ and $b$ implicit.
	For any $i \in [d]$, the projection along dimension in a sublattice is defined as
	$$
	\proj_i(\sublattice(a,b)) \coloneqq \{\, z_i \mid z \in \sublattice(a,b) \text{ and } a_i \leq z_i \leq b_i \,\},
	$$
	where $z_i$ denotes the $i$-th component of $z \in G$.
\end{definition}

With this notation in hand, the following immediate observation states a condition that a
function $f$ must possess in order for $f$ to map the sublattice to itself, which we will need
to be able to apply Tarski's theorem on the sublattice.
\begin{observation}[Closedness of Tarski function on a sublattice] 
	\label{obs:sublat-verif}
	Suppose $\sublattice(a,b) \subseteq G$ is a sublattice. 
	Then, if for every point $z\in \sublattice(a,b)$, we have $f_i(z) \in \proj_i(\sublattice)$, 
	then for all $z$ we have $f(z) \in \sublattice(a,b)$. 
\end{observation}

\regionsNonExpansion
\regions

To show the property about lesser fixed points, we need one more bit of notation:
For points in $G$, we define $\ell_\infty$ boxes as follows.
\begin{definition}
For any point $x \in G$ and $r \in \mathbb{N}$, we define the region centered at 
$x$ of $\ell_\infty$-radius $r$ as:
$$B(x,r) := \{ y \in G: \Norm{x-y}_{\infty} \leq r\}.$$
\end{definition}

We next prove the following lemma, which states that between any
two fixed points~$x$ and~$y$, with $y \leq x$ that are at distance more than $1$ from one
another, there is another fixed point that is strictly between $y$ and $x$.
Our overview of the proof technique will use Figures~\ref{fig:regions_non_expansion} 
and~\ref{fig:regions}
We draw an $\ell_\infty$ box around each of the points $x$ and $y$, and intersect these two boxes
also with the sublattice $L(y,x)$.
The resulting region is itself a sublattice.
 
Because our \DMAC instances have displacements with length at most 1, points in the interior
of this region stay within this region when mapped by the Tarski function.
We argue that every point on the boundary of this region points weakly inwards so these points
stay in the region too.
The proof deals with two cases.
 
Points on the boundaries of one of the two $\ell_\infty$ boxes (shown in red in
Figures~\ref{fig:regions_non_expansion} and~\ref{fig:regions}) 
must weakly point inward because otherwise we would violate the non-expansion property with
either $x$ or $y$.
This argument works when the \linf distance between the boundary points and one of 
the points $x$ and $y$ is only one less than the \linf distance between $x$ and $y$ themselves.
For Figure~\ref{fig:regions_non_expansion}, this is true for all boundary points; for
Figure~\ref{fig:regions}, this is true only for the left and right boundary points.

Boundary points for which this is not true are those that are also on the 
boundary the of the sublattice between $y$ and $x$ (shown in blue in Figure~\ref{fig:regions}).
For these points the argument that they point weakly inwards is based on monotonicity.
In Figure~\ref{fig:regions}, points that are directly to the right of $y$ must point weakly
upwards, while the points that are directly to the left of $x$ must point weakly downward, because
otherwise we would violate monotonicity with the two fixed points $y$ and $x$, respectively. 

These arguments show that $f$ maps this sublattice region to itself.
Furthermore, since $f$ is monotone everywhere, it is also monotone
within this sublattice, so we can apply Tarksi's theorem to prove that there is a
fixed point in this region.
We now prove this formally.

\begin{lemma}
\label{lem:fixed-point-between}
Suppose $f$ is a violation-free $\DMAC$ instance.
%
Let $x$ and $y$ be fixed points of $f$ with $y \leq x$ and $\Norm{x-y}_{\infty} > 1$. 
There exists a fixed point $z\in G$ of $f$ in $\sublattice(y,x)$ with
$\Norm{z-y}_{\infty} < \Norm{x-y}_{\infty}$ and $\Norm{x-z}_{\infty} < \Norm{x-y}_{\infty}$.
\end{lemma}

\begin{proof}
Let $k \coloneqq \Norm{x-y}_{\infty}$. 
We consider the region \[R \coloneqq B(y,k-1) \cap B(x,k-1) \cap \sublattice(y,x).\] 
Because $\sublattice(y,x) \subseteq G$, we have $R \subseteq G$.
By construction, $R$ is non-empty and contains neither $x$ nor $y$. 
We will show that there exists a fixed point of $f$ in $R$, proving the lemma. 

To do this, we first argue that for each $z \in R$ we have $f(z) \in R$. 
Consider any point $z \in R$. 
For each dimension $i\in [d]$, if $z_i$ is in the interior of $\proj_i(R)$, then $f_i(z)$ is
in $\proj_i(R)$ as well as by assumption it does not move by more than 1. 
If $z_i$ is not in the interior then there are two (potentially overlapping) cases to consider: 
Either $z_i$ is on the boundary of $\proj_i(\sublattice(y,x))$ or 
$z_i$ is on the boundary of $\proj_i(B(y,k-1)\cap B(x,k-1))$. 
Figure~\ref{fig:regions} illustrates the two cases and the region in question.


Suppose that $z_i$ is on the boundary of $\proj_i(B(y,k-1)\cap B(x,k-1))$. 
Then either $z_i - y_i = k - 1$ or $x_i - z_i = k - 1$. 
Suppose towards a contradiction that $f_i(z) < z_i$ and $x_i - z_i = k - 1$. 
Then 
\begin{align*}
	\Norm{f(x) - f(z)}_{\infty} = \Norm{x - f(z)}_{\infty}
	= k > k - 1 = \Norm{x-z},
\end{align*}       
which contradicts non-expansion of $f$.
The same contradiction can be reached if we start from $f_i(z) > z_i$ and $z_i - y_i = k - 1$.
We conclude that $f_i(z) \in \proj_i(R)$. 


Suppose $z_i$ is on the boundary of $\proj_i(\sublattice(y,x))$. 
Then we have $y_i \leq z_i \leq x_i$ with at least one equality. 
Suppose that $y_i = z_i \leq x_i$. 
Then $f_i(z) \geq y_i$ by monotonicity. 
If instead $y_i \leq z_i = x_i$, then $f_i(z) \leq x_i$ by monotonicity. 
%
We conclude that $f_i(z) \in \proj_i(R)$.

We have thus shown that for all points $z$ in the sublattice $R$ and all dimensions
$i \in [d]$, we have $f_i(z) \in \proj_i(R)$. 
So, by Observation~\ref{obs:sublat-verif}, we have that $f$ maps $R$ to itself.
Now, since $f$ is a monotone function over $R$, and since $f$ maps $R$
to itself, we can apply Theorem~\ref{thm:tarski1} 
to conclude that $f$ has a fixed point in $R$, which completes the proof.
\end{proof}

With Lemma~\ref{lem:fixed-point-between} in hand, it is now straightforward to prove our desired property, 
namely that if we find a fixed point $x \in G$ that is not the least fixed
point, then there is a fixed point $y \le x$ such that $\| x - y\|_\infty = 1$, i.e., in 
the grid cell with $x$ as its top right corner.
We denote the vector of all ones by $\One$, where the dimension will be clear from the context.

\begin{corollary} 
	\label{cor:neighboring-lesser-fp}
	Suppose $f$ is a violation-free $\DMAC$ instance.
	If $x\in G$ is a fixed point of $f$ other than the least fixed point of $f$, there exists a
	fixed point other than $x$ in $\sublattice(x-\One,x)$. 
\end{corollary}
\begin{proof} 
Let $x \in G$ be a fixed point of $f$ that is not the least fixed point.
Suppose for the sake of contradiction that $y \in G$ is the nearest lesser fixed point to $x$ 
and $\Norm{y-x}_{\infty} > 1$. 
Then we can apply Lemma~\ref{lem:fixed-point-between} to conclude that there exists a fixed
point $z \in G$ with $y \leq z \leq x$ and $\Norm{z-x}_{\infty} < \Norm{y-x}_{\infty}$, which is a
contradiction to $y$ being the nearest lesser fixed point to $x$.
\end{proof}

\paragraph{\bf The verification algorithm.}



We have so far proved that if $x$ is not a least fixed point, then there is
another fixed point $y$ that is in the $d$-dimensional cube below $x$ at
distance $1$. We now provide an algorithm that, in a violation-free
instance, finds $y$ if it exists. Therefore if the algorithm does not find a
lesser fixed point, it certifies that $x$ is the least fixed point. 

\verification

Let us start by considering the point $y = x - \One$, which is shown as $v^0$ in 
Figure~\ref{fig:verification sequence}.
\begin{itemize}
	\item If $y$ is a fixed point, then $y$ is a fixed point that is less than $x$,
	which proves that $x$ is not the least fixed point. So we can terminate.

	\item Otherwise, observe that we cannot have $f_i(y) < y_i$ in any dimension
	$i$, since this would be a violation of non-expansion with $x$. Since $y$
	is not fixed, it must therefore be the case that $f_i(y) > y_i$ for at least
	one dimension $i$.  In Figure~\ref{fig:verification sequence}, this is represented
	by the edge from $v^0$ to $v^1$.

	\item By monotonicity, any point $y' \ge y$ such that $y'_i = y_i$ must also
	satisfy $f_i(y') > y'_i$ and so $y'$ is also not a fixed point. 
	This allows us to rule out an entire face of the cube that lies beneath $x$,
	since no vertex of that face can contain a fixed point.

	\item Using this idea, we can walk a sequence of $d$ steps in the cube beneath $x$.
	In Figure~\ref{fig:verification sequence}, $v^1$ is the point we reach after the 
	first step of this sequence.
	When we reach a new point in the sequence, either it is fixed, and we have shown
	that $x$ is not the least fixed point, or we rule out a face, move to the new point
	and repeat. 
	If we reach $x$ then we have proven that it is indeed the least fixed point.
\end{itemize}

To recap, in Figure~\ref{fig:verification sequence}, we start at $v^0=y$, and the fact that $f(v^0)$ points
right allows us to rule out the left face of the cube. 
The directions for $f$ implied by monotonicity are shown in gray pointing away from that left face. 
We then move to $v^1$, which is the least vertex of the remaining cube.
Since $f(v^1)$ points up, we rule out the bottom face (of the right face), indicated by the black
arrow and implied gray arrow that point from bottom to top.
We are now at $v^2$, and the arrow from $v^2$ to $v^3=x$ complete the verification of $x$ as the least 
fixed point.

We now prove formally that this approach works. 
\begin{lemma} Suppose $f$ is a violation-free $\DMAC$ instance. Given a fixed
	point of $f$, $x\in G$, we can verify in polynomial time whether $x$ is the
	least fixed point. \label{lem:lfp verification in polytime}
\end{lemma}
\begin{proof}
	By Corollary~\ref{cor:neighboring-lesser-fp}, to check whether $x$ is the least fixed
	point we need only to verify whether there is a lesser fixed point in $\sublattice(x-\One,
	x)$. Let $y \coloneqq x-\One$. We will use Algorithm~\ref{alg:find lesser fixed point}.

	\begin{algorithm} \caption{Algorithm to find a lesser fixed point.} 
	\label{alg:find lesser fixed point}
		\SetKwInOut{Input}{Input}
		\SetKwInOut{Output}{Output}
		
		\underline{procedure $\FindLesserFixpoint(f, y, x)$}\;
		\Input{$f,y,x$ satisfying $f(x) = x$,\quad $y \leq x$,\quad $\Norm{y-x}_{\infty} = 1$.}
		\Output{A fixed point $z < x$ in $\sublattice(y,x)$ if one exists, or $x$ otherwise.}
		\eIf{there exists an $i$ such that $x_i > y_i$ and $f_i(y) > y_i$ \label{line:test}}
		{
			return $\FindLesserFixpoint(f, y+ e_i, x)$
		}
		{
			return $y$
		}
	\end{algorithm}
	
	This algorithm runs in time that is polynomial in the representation of $f$.
	To see this, we first note that the representation of $f$ must be of size at least $d$, 
	which is the number of dimensions, because the input to $f$ is of size $d \cdot log n$.
	The algorithm makes at most $d$ queries to $f$ in total, noting that each query gives
	back the direction of all $d$ edges at the current point in the verification sequence
	that we are building.
	We only use polynomial time to process such a point, to find one of the relevant
	edges that points away from the point if it exists, or to declare this point as the lesser
	fixed point that disproves that $x$ is the least fixed point.

	To see that the algorithm is correct, observe that if the condition on line~\ref{line:test} is
	true, then every point $z_i$ on $\sublattice(y,x)$ with $z_i = y_i$ must have $f_i(z) > z_i$ by
	monotonicity, so none of those points can be fixed. Thus, any lesser fixed point must be
	contained on the sublattice $\sublattice(y+e_i, x)$.
\end{proof}

\paragraph{\bf Verification sequences.}

We will say that a fixed point $x$ of $f$ is a \emph{verified} least fixed point (LFP) if our algorithm
from Lemma~\ref{lem:lfp verification in polytime} did not find a lesser fixed point in the cube
below $x$. 
We should note that this only implies that $x$ is the LFP if the instance has no violations. 
In an instance with violations, there may be verified LFPs that are not actually LFPs. 
We will deal with this point when we deal with the violation cases in our containment result.

As can be seen in Figure~\ref{fig:verification sequence}, the algorithm
generates a sequence of points while verifying $x$. We refer to this sequence
as $x$'s \emph{verification sequence}, which we define as follows.

\begin{definition}[Verification sequence] Given a verified LFP $x \in G$, let $\VSeq(x)$ we define the
	sequence of points encountered in the recursive calls to $\FindLesserFixpoint$ when
	verifying~$x$ against the entire domain, i.e., where all coordinates of $x$ and the initial $y$ 
	differ, as:
	$$y = \VSeq(x)_0 < \VSeq(x)_1 < \dotsb < \VSeq(x)_d = x.$$
\end{definition}

In our reduction from \DMAC to \OPDC in the next section, we will run $\FindLesserFixpoint(f,y,x)$ 
in a slice rather than the entire domain, i.e., with a number of coordinates of $y$ and $x$ being
identical. As a result, the verification sequence we get will have fewer than $d$ points.

\subsection{Reducing DMAC to OPDC}
\label{sec:DMACtoOPDC}

We will show $\UEOPL$ containment by giving a polynomial-time reduction to the
\emph{One Permutation Discrete Contraction} problem (\OPDC), which is known to
be $\UEOPL$-complete~\cite{FGMS20}. 
%
The containment result for $\DMAC$ in turn implies that $\MonotoneContraction$ is also in $\UEOPL$.

\OPDC which we define formally below, like \DMAC, uses as discrete grid as its domain.
Our reduction from \DMAC to \OPDC will preserve the domain.
The definition of \OPDC uses the concept of $i$-slices, which we describe next.
The same concept also applies to \DMAC instances, and after we have defined \OPDC, we will
revisit \DMAC and prove some properties about least fixed points and $i$-slices, before 
we present our reduction from \DMAC to \OPDC.

\subsubsection{One Permutation Discrete Contraction}
\label{sec:opdcdef}

A \emph{slice} refers to a subset of the grid domain where some of the dimensions are fixed. 
A slice $s$ is obtained by fixing some coordinates, while letting the rest vary. 
%
We denote by $\Slice_d = ([0,1] \cup \{*\})^d$ the set of all possible slices in dimension $d$.
For a slice $s \in \Slice_d$, every point $p$ within $s$ must obey $p_i = s_i$ whenever 
$s_i \neq *$.
A slice $s$ is an \emph{$i$-slice} if the first $i$-coordinates are not fixed, i.e., 
$s_1 = s_2 = \dotsb = s_i = \star$, 
and the remaining coordinates are fixed, i.e., $s_j \ne \star$ for $j > i$. 
Given a point $x \in G$, we write $\slice_i(x)$ for the unique $i$-slice containing $x$, so
$\slice_i(x) = (\star,\dotsc,\star,x_{i+1},\dotsc,x_d)$. 
We use $\free(s) \subseteq [d]$ to denote the coordinates of $s$ set to $\star$ and $\fixed(s)$
to denote the coordinates that are not set to $\star$. 
Given a slice $s$, we say that a slice $s'$ is a \emph{subslice} of $s$, denoted by $s'
\subseteq s$ if $s'_j = s_j$ for each $j \in \fixed(s)$, and $\free(s') \subseteq \free(s)$.

We are now ready to define the \OPDC problem.

\begin{definition}[$\OnePermutationDiscreteContraction$]\label{def:OPDC} 
	Given a family of direction functions $D = (D_i)_{i=1,\dotsc,d}$ over a domain
	$[k_1]\times [k_2]\times \dotsm \times [k_d]$, output one of the following:
	\begin{enumerate}[label=(O\arabic*)]
		\item A point $x$ such that $D_i(x) = \zero$ for all $i\in [d]$. \label{O1} 
	\end{enumerate}
	\begin{enumerate}[label=(OV\arabic*)]
		\item An $i$-slice $s$ and two distinct points $x,y$ on $s$ such that $D_j(x) = D_j(y) = \zero$ for all $j\leq i$. \label{OV1}
		\item An $i$-slice $s$ and two points $x,y$ on $s$ such that $x_i = y_i + 1$, $D_j(x) = D_j(y) = \zero$ for all $j < i$ and $D_i(x) = \down$, $D_i(y) = \up$. \label{OV2}
		\item An $i$-slice $s$ and a point $x$ on $s$ such that $D_j(x) = \zero$ for $j < i$ and either 
		\begin{itemize}
			\item $x_i = 0$ and $D_i(x) = \down$, or
			\item $x_i = k_i$ and $D_i(x) = \up$.
		\end{itemize} \label{OV3}
	\end{enumerate}
\end{definition}

The high-level idea of the problem is that it encodes an instance in which
every $i$-slice has a unique fixed point. Moreover, given a fixed point $x$ of an
$i$-slice, the function $D_i(x)$ should tell us the direction to move in
dimension $i$ to reach the unique fixed point of the ($i+1$)-slice that contains
$x$. 

Since every $i$-slice should have a unique fixed point,  this means that the
$d$-slice $(\star, \star, \dots, \star)$ has a unique fixed point. Solutions of
type \ref{O1} ask us to find this point. 

The violation solutions encode ways in which our promise might fail to be
satisfied. 
Violations of type \ref{OV1} encode a situation in which a single $i$-slice has
two different fixed points, meaning that the fixed point is not unique. 
Violations of type \ref{OV2} ask us to find two
points in an $i$-slice that are adjacent in dimension $i$ that point toward
each other in dimension $i$. Since both of these points are required to point
toward the fixed point of the $(i+1)$-slice, that fixed point should lie between
the two points, 
so the violation captures the scenario where that fixed point is
missing. Violations of type \ref{OV3} ask us to find $i$-slice fixed points
whose directions point outside the instance, meaning that the direction
function fails to point towards the $(i+1)$-slice's fixed point.

\subsubsection{Useful properties of DMAC $i$-slice least fixed points}
\label{sec:dmacs-slices-lfps}

Before we present the reduction, we first show some useful properties of least
fixed points in violation-free $\DMAC$ instances.

Given a $\DMAC$ instance $f$ and an $i$-slice $s$, we say that a point $x$ is an \emph{$i$-fixed
point} if $f_j(x) = x_j$ for all $j \le i$, or in other words, if $x$ is a fixed point when we
disregard all fixed dimensions of the $i$-slice.

The first observation is that if we have an $i$-fixed point $x$ of an $i$-slice, and we consider
any $j \ge i+1$, then both of the neighbouring $i$-slices, where for we increase or decrease
the value of the fixed coordinate $x_j$ by 1, have $i$-fixed points that are within \linf 
distance $1$ of $x$.
In the following lemmas, given $j$, we will denote these neighbouring slices of an $i$-slice $s$ by 
by $s' = s \pm e_{j}$, and similarly as $s \pm k\cdot e_{j}$ for integer $k$ with $|k| > 1$ 
for neighbours that are further away.
This lemma follows from Tarski's theorem, Theorem~\ref{thm:tarski1}.

\begin{observation} 
	\label{obs:neighboring-slice-fp-nearby}
	Suppose that $f$ is a violation-free $\DMAC$ instance. 
	Let $i, j \in [d]$ with $j \ge i+1$.
	If $x$ is a fixed point of $s = \slice_i(x)$, then for both of neighbouring slices 
	$s' = s \pm e_{j}$, there exists an $i$-fixed point $y$ of the neighbouring slice such that  
	$\Norm{y-x}_{\infty} = 1$. 
\end{observation}
\begin{proof} 
	The same proof works for either slice, $s' = s - e_j$ or $s' = s + e_j$. 
	Assume $s'$ is one of them, and let $R:= B(x,1) \cap s'$.
	For every point $y \in R$ we must have $f(y) \in R$, since 
	otherwise we have a \ref{DV2} violation of non-expansion.
	Since $f$ is monotone on the whole domain it is also monotone on $R$. 
	Thus, we have satisfied the conditions of Tarski's theorem for $f$ on $R$,
	which guarantees the existence of the claimed fixed point.
\end{proof}

We can show the same property for least fixed points: if $x$ is the LFP of an
$i$-slice, then the LFPs of the neighbouring $i$-slices must be within distance
$1$ of $x$. 
We can see an illustration in Figure~\ref{fig:neighboring-lfps}, where $d$ cannot 
be the LFP of $s'$ because, by Observation~\ref{obs:neighboring-slice-fp-nearby}
we know that at least one of $a, b$ or $c$, which are below $d$, must be fixed in $s'$.

\neighboringSlice

\begin{lemma}
	\label{lem:lfps-are-adjacent}
	Suppose $f$ is a violation-free $\DMAC$ instance. 
	Let $i, j \in [d]$ with $j \ge i+1$.
	If $x$ is the LFP of $s=\slice_i(x)$ and $y$ the LFP of a slice $s' = s \pm e_j$, 
	then $\Norm{y-x}_\infty = 1$. 
\end{lemma}

\begin{proof}
	By Observation~\ref{obs:neighboring-slice-fp-nearby}, there must be a fixed point 
	$z$ on $s'$ with $\Norm{x-z} = 1$ and a fixed point $z'$ on $s$ with $\Norm{y-z'} = 1$. 
	Suppose there is a coordinate $l$ such that $y_l \leq x_l - 2$. 
	Then $z' \not\geq x$, which contradicts $x$ being the LFP of $s$. 
	Suppose there is a coordinate $l$ such that $x_l \leq y_l - 2$. 
	Then $z\not\geq y$, which contradicts~$y$ being the LFP of $s'$. 
	We conclude that $\Norm{y-x} = 1$.
\end{proof}

As a corollary of this, if we have two $i$-slices that differ by $k$ in dimension $j \ge i+1$, 
then the LFPs of these slices can have distance at most $k$ from one another.

\begin{corollary}
	\label{cor:bind-dim}
	Suppose $f$ is a violation-free $\DMAC$ instance. 
	Let $i, j \in [d]$ with $j \ge i+1$.
	If $x$ is the LFP of $s = \slice_i(x)$ and $y$ the LFP of $s' = s \pm k \cdot e_j$, 
	with $k \in \Natural$, then $\Normi{y-x} \le k$. 
\end{corollary}

\begin{proof}
	The same proof works for either slice, $s' = s - k \cdot e_j$ or $s' = s + k \cdot e_j$. 
	We proceed by induction on~$k$. 
	Let $z$ be the LFP of the slice $s \pm (k-1) e_j$. 
	The base case $k=1$ holds by Lemma~\ref{lem:lfps-are-adjacent}. 
	Then $\Normi{z-x} = (k-1)$ by the inductive hypothesis and $\Normi{y-z} = 1$ by 
	Lemma~\ref{lem:lfps-are-adjacent}.
	Thus we have:
	$$\Normi{y-x} \leq \Normi{z-x} + \Normi{y-z} = (k-1) + 1 = k,$$
	which completes the induction and proves the claim.
\end{proof}

\paragraph{\bf Least fixed points are hereditary.}

The most important result for our construction will be that LFPs are
hereditary. Specifically, if $x$ is the LFP of a slice $s$, then it is the LFP
of all subslices $s' \subseteq s$. In order to prove this, we first prove the
following technical lemma.

\begin{lemma}
	\label{lem:lower-fixed-point-for-LFP-hereditary}
	Suppose that $f$ is a violation-free $\DMAC$ instance, and that $y,x$ are such that
	$y<x$ and $\Norm{y-x}_{\infty} = 1$. 
	Let $s'$ be the smallest slice that contains both $y$ and $x$, i.e., 
	$\free(s') = \{j: j \in [d], y_j \ne x_j\}$.
	Suppose that both $y$ and $x$ are fixed in $s'$, i.e. for all $j \in \free(s)$, we have $f_j(y) =y$
	and $f_j(x) = x$.
	Moreover, suppose that there exists $i \in \fixed(s')$ such that
	$x$ is fixed in dimension $i$ and $y$ is not fixed in dimension $i$,
	i.e., $f_i(x) = x_i$ and $f_i(y) \neq y_i$, 
	Then the point $z = y - e_i$ is fixed in the slice $s$ with $\free(s) = \free(s') \cup \Set{i}$.
\end{lemma}

\begin{proof}
	
\cube 

By assumption $f_i(y) \ne y_i$.
If $f_i(y) > y_i$, we would have a violation of monotonicity,
since, by assumption $f_i(x) = x_i$ and $y < x$. 
Thus, we have $f_i(y) < y_i$ (shown as the arrow in Figure~\ref{fig:cube}).
 
Now consider the point $z = y - e_i$.
We have that:
\begin{itemize}
\item
For all $j \in \free(s)$, if $f_i(z) < z_i$ we would have a violation of non-expansion with $x$,
since then we would have $|f_i(z)- f_i(x)| > \Norm{z-x}_{\infty}$. 
\item
For all $j \in \free(s)$, if $f_i(z) > z_i$, we would have a violation of monotonicity with $y$, 
since $z < y$ and the fact that $f_i(y) < y_i$, which we have already shown.
\end{itemize}
Thus $f_j(z) = z_j$ for all $j \in \free(s)$, so $z$ is fixed in slice $s$ as claimed.
\end{proof}

Now we can prove that LFPs are hereditary. 

\begin{theorem} 
	\label{thm:hereditary}
	Suppose $f$ is a violation-free $\DMAC$ instance.
	If $x$ is the LFP of a slice $s$, it is also the LFP of all subslices $s'
	\subseteq s$ that contain $x$. 
\end{theorem}
\begin{proof}
	Suppose $x$ is not the LFP of a subslice $s'$ that contains $x$. 
	Then, by Corollary~\ref{cor:neighboring-lesser-fp}, there exists a fixed point $y$ of 
	$s'$ below $x$ in $s'$ with $\Norm{y-x}_{\infty} = 1$.
	Moreover, by assumption $y$ is not fixed in $s$ as this would contradict the leastness of $x$. 
	Let $i$ be a dimension of $\free(s)$ that is fixed at $x$ but not $y$. 
	Then Lemma~\ref{lem:lower-fixed-point-for-LFP-hereditary} applies and says that 
	$z= y-e_i$ is fixed in $s$. 
	But $z$ is below $x$, which contradicts that $x$ is the LFP of $s$.

\end{proof}

\subsubsection{A reduction for the DMAC promise problem}

As a warm-up, we first present a reduction from $\DMAC$ to \OPDC where we assume
that the $\DMAC$ instance has no violations, and we produce an \OPDC instance that
has no violations. 

First, we define a predicate that will be used in constructing the direction
function: 

\begin{equation*}
	\IsLFP(x,i) \coloneqq \Brack{x = \FindLesserFixpoint (g, y, x) }.
\end{equation*}
This predicate returns true whenever the $\FindLesserFixpoint$ algorithm verifies
that it is the least fixed point of its $i$-slice. 
Note that we initialize the algorithm to with a point $y = x - \One_{\free(\slice_i(x))}$ for some
$i$, so that algorithm is run within the $i$-slice $s$, and ignores all other dimensions,
$\fixed(s)$.

Like in \DMAC, the \OPDC problem, defined formally below, also uses a discrete grid as its domain,
and our reduction will not change the domain, which we will simply refer to as $G$ for both problems.

%
Starting from a $\DMAC$ with function $f$, our reduction first defines $h$ similarly as we did in
Section~\ref{sec:monotone_to_dmac}. 
\begin{definition}
	Given a \DMAC function $f: G \rightarrow G$ with all displacements of length at most 1, 
	for all~$x \in G$ and $i \in [d]$ we define:
	\begin{equation*}
		\label{def:h_from_f}
		h_i(x) = f_i(x) - x_i.
	\end{equation*}
\end{definition}

The function for the overall displacement is thus $h(x) := (h_1 (x), h_2 (x), \dots, h_d(x))$.
We proceed to define a function capturing the direction of $f$, indicated by the sign of $h$.
\begin{definition}
	Given function $h_i: G \rightarrow \Set{\pm1,0}$, for any $i \in [d]$, we define the \emph{direction function} $D_i$ as:
	\begin{align*}
		D_i(x) \coloneqq \begin{cases}
			\up&\quad\text{if $h_i(x) = 1$}\\
			\down&\quad\text{if $h_i(x) = -1$ or $h_i(x) = 0 \land \lnot \IsLFP(x, i)$}\\
			\zero&\quad\text{if $h_i(x) = 0 \land \IsLFP(x, i)$}
		\end{cases}
	\end{align*}
\end{definition}
The idea here is that the direction points up whenever $f_i(x)$ points above $x_i$, and it
points down whenever $f_i(x)$ points below $x_i$. When $f_i(x)$ points to $x_i$, we must be
more careful: here we force $D_i$ to point down, unless $x$ is the least fixed point of the
$i$-slice. This allows us to ensure that the $i$-slice has a unique fixed point, and our ability to determine whether a fixed point is a least fixed point is crucial here.

We refer to a point $x$ such that $D_j(x) = \zero$ for $j \leq i$ as an
$i$-zero of $D$. 

The following lemma formally proves that each $i$-slice has a
unique $i$-zero. This implies that there is a unique \ref{O1} solution, and
rules out any \ref{OV1} violations. Note also that an \ref{O1} solution is a
\ref{D1} solution of the $\DMAC$ instance, since we have $D_j(x) = \zero$ for all $j$ only if $h_j(x) = 0$ for all $j$. 

\begin{lemma} 
\label{lem:unique-i-zero}
If $f$ has no violation solutions, then $D$ has a unique $i$-zero in each $i$-slice.
\end{lemma}
\begin{proof} 
	Assume that $f$ has no violation solutions.
	Any violation solution in a slice translates to a violation in $f$, hence restricting $f$ to a
	slice implies that no properties are violated within each $i$-slice. 
	We fix an $i$-slice $s$. 
	Based on the monotonicity assumption, we can apply Tarski's theorem
	and get that $f$ has a unique least fixed point in $s$, which we will denote by $x$. 
	Monotonicity must additionally hold for any subslice $s' \subseteq s$ of the $i$-slice and the
	theorem can be similarly applied.
	Since $x$ is the LFP of $slice_j(x)$, for each $j \leq i$, we have that $h_j(x) = 0$ and hence
	$D_j(x) = \zero$. 
	By Theorem~\ref{thm:hereditary}, we have that $x$ is the LFP of $\slice_j(x)$
	whenever it is the LFP of $\slice_i(x)$, so we get that $x$ is an $i$-zero of $D$.
	
	Now suppose that there is another $i$-zero of $D$, $y \neq x$. 
	In order for $y$ to be an $i$-zero, we must have that $y$ is the LFP of $\slice_i(x)$, 
	which is a contradiction. 
	We conclude that $x$ is the unique $i$-zero of $D$.
\end{proof}

As an immediate corollary we get the following.

\begin{corollary}
If $f$ has no violation solutions, then $D$ has no \ref{OV1} violations.
\end{corollary}

The next lemma shows that the direction function at each $i$-zero points
towards the $(i-1)$-zero of the $(i-1)$-slice. This rules out any \ref{OV2}
solutions and crucially relies on the absence of violations in $f$.

\begin{lemma}
	Suppose that $f$ has no violations.
	Let $s'$ be an $(i-1)$-slice and $s$ an $i$-slice, with $s' \subseteq s$. 
	Let $x$ be the unique $i$-zero of $s$, and $y$ be the unique $(i-1)$-zero of $s'$. 
	Then $D_i(y) =\up $ if $y_i < x_i$, $D_i(y) = \down$ if $y_i > x_i$.
\end{lemma}

\begin{proof}
	Suppose that $f$, $x$, $y$, $s$ and $s'$ satisfy the conditions of the lemma.
	Since $x$ is the unique $i$-zero of~$s$, it must hold that $D_j(x) = \zero$ for all 
	$j \leq i$, which in turn means that $x$ is the LFP of $s'$. 
	%

	Now assume that $y_i < x_i$.
	We want to show that this implies $D_i(y) =\up $.
	Towards a contradiction, suppose $D_i(y) = \down$. 
	Note that $h_i(y) = 0$ would imply that $y$ is an $i$-zero of $s$, but that 
	contradicts the fact that $x$ is the unique $i$-zero of $s$.
	This, we must have $h_i(y) = -1$.	
	Note that since $x$ is the unique $i$-zero of $s$, was have $D_i(x) = \zero$ and thus
	$h_i(x) = 0$.
	Now we apply Corollary~\ref{cor:bind-dim} with $j=i$ and $k=|y_i - x_i|$, which tells us that
	the \linf distance between $x$ and $y$ is at most $k$.
	Thus we have
	$$f_i(x) - f_i(y) = (x_i + h_i(x)) - (y_i - h_i(y)) = (x_i - y_i) + 1 = \Norm{x - y}_{\infty} + 1.$$
	But this means that $\Norm{f(x)-f(y)}_{\infty} > \Norm{x - y}_{\infty}$, 
	which is an \ref{DV2} violation of non-expansion, contradicting the fact that $f$ is 
	violation free.

	The proof for the case that $y_i > x_i$ is symmetric, and the proof is complete.
\end{proof}

Since an $i$-slice has a unique fixed point by Lemma~\ref{lem:unique-i-zero}, and we have just shown
that the $i$-zeros of~$i-1$ slices must point this fixed point in dimension $i$,
we get as an immediate corollary that we get the following.

\begin{corollary}
If $f$ has no violation solutions, then $D$ has no \ref{OV2} violations.
\end{corollary}

The fact that there are no \ref{OV3} violations follows from the fact that $f$
maps points on the grid to other points on the grid. This means that $f$ can
never move to a point outside the grid, and so our reduction will never create
an \ref{OV3} violation.

So we have proved that, if $f$ has no violations, then the reduction will
produce an instance with a unique \ref{O1} solution that maps back to the LFP
of $f$, and that there cannot be any violation solutions. 

This establishes the following theorem. 

\begin{theorem}
	\label{thm:promise-DMAC-to-promise-opdc}
	There is a polynomial-time reduction from
	the promise version of $\DMAC$ to the promise version of $\OPDC$. 
\end{theorem}


\subsubsection{The full reduction: handling violations} 
\label{sec:full_reduction}

\bigskip

We now give a full reduction between the non-promise versions of the problems. 
The reduction itself will be identical to the reduction used in the previous
section, but we must now handle the case where the function $f$ has violations. 
We recall that now when we say ``verified LFP'' this is only a \emph{purported}
LFP, as due to violations a successful run of the verification procedure does
not mean that the inputted candidate LFP really is one.
Nonetheless we still have the (full) verification sequence from the 
``successful'' run of the verification procedure and will use this
sequences to produce violations in the original \DMAC instance in certain cases
when we get a violation in our produced \OPDC instance.

We produce that \OPDC instance in exactly the same way as we did in the last
section.
Our arguments from the promise reduction still imply any solution of type \ref{O1} can be mapped
back to a solution of type \ref{D1}, but in contrast to the previous section, the \OPDC instance may
contain violation solutions, and we must map these back to violations of the $\DMAC$ instance. 

In the rest of this section, we prove the following theorem.

\begin{theorem} 
	\label{thm:main}
	There is a promise-preserving polynomial time reduction from $\DMAC$ to $\OPDC$.
\end{theorem}

We note that \ref{OV3} violations are impossible by definition. 
This is because such a violation would have the direction function pointing outside the whole
domain, but that is not possible since $f$ maps $G$ to itself in \DMAC.
Thus we rule out \ref{OV3} violations.

We do need to deal with \ref{OV1} violations and \ref{OV2} violations. 
We consider each separately in the following two sections, and in each case we show that these
violations can be efficiently mapped back to violations of $\DMAC$. 
For both \ref{OV1} and \ref{OV2} violations, we may produce either \ref{DV1} or \ref{DV2}
violation.

\paragraph{\bf Handling \ref{OV1} violations.}

Recall that in an \ref{OV1} violation we are given two distinct $i$-zeros $x$
and $y$ of the same $i$-slice. This could not happen if the original \DMAC instance 
was violation free because it would mean that the slice has two LFPs, which is impossible.
Thus there must be some violation in the original $\DMAC$ instance. 
We will give a polynomial-time algorithm to find one such violation.
We do this by considering three cases: where $x$ and $y$ are comparable and within distance 1;
comparable but further apart; incomparable.
The following lemmas will allow us to deal with these cases.

The next lemma deals with the case that the points are within distance 1 of each other so that 
one point is within the verification cube of the other.
In this case we will produce a \ref{DV1} violation of monotonicity in polynomial time. 


\begin{lemma} Suppose we are given points $x$ and $y$ on slice $s$ satisfying:
\label{lem:ov1-comparable-close}
\begin{enumerate}    
	\item $x$ is a verified LFP of $s$;
	\item $y < x$ and $y\in \sublattice(x-\One, x)$; and 
	\item $f_i(y) - y_i \leq 0\quad \forall i\in \free(s)$.
\end{enumerate} Then we can find a violation of type \ref{DV1} in polynomial time. 
\end{lemma}

\begin{proof}
	Let $k:= |\free(s)|$.
	Since $x$ is a verified LFP, there is a verification sequence 
	$x - \One = v^{0} \leq v^1 \leq \dotsb \leq v^k = x$. 
	Let $i$ be the index such that $y \in \sublattice(v^i,x)$ but $y \notin \sublattice(v^{i+1},x)$, 
	which must exist since $y < x$ and $y\in \sublattice(x-\One, x)$.
	Let $j$ be the dimension in which $v^i$ and $v^{i+1}$ differ.
	Since the verification algorithm moved from $v^i$ to $v^{i+1}$, we have  $f_j(v^i) > v^i_j$. 
	By assumption~3, it must hold that $f_j(y) \le y_j$, as $j$ is a free dimension of the slice. 
	Since $v^i$ is the minimum element of~$\sublattice(v^i,x)$ we have $v^i \le y$ by definition
	of $i$
	However, $f_j(v^i) \not \le f_j(y)$, which corresponds to violation of type \ref{DV1}.
\end{proof}

The next case is where $x$ and $y$ are comparable, with  $y \le x$, but they are not within 
distance 1 of each other
In this case, we will produce in polynomial time either a \ref{DV1} monotonicity violation, or 
a \ref{DV2} violation of non-expansion.

\begin{lemma} 
\label{lem:ov1-comparable-far}
Suppose we are given points $x$ and $y$ on slice $s$ satisfying:
\begin{enumerate}    
	\item $x$ is a verified LFP of $s$;
	\item $y < x$ and $y\notin \sublattice(x-\One, x)$; and 
	\item $f_i(y) - y_i \leq 0\quad \forall i\in \free(s)$.
\end{enumerate} Then we can find a violation of type \ref{DV1} or \ref{DV2} in polynomial time. 
\end{lemma}

\verificationSequence

\begin{proof}
	Suppose $y\leq x$ and $y \notin \sublattice(x-\One,x)$. 
	Let $s'$ be the unique minimal slice containing $\sublattice(y,x)$, i.e., where 
	$\free(s') = \{j : j\in [d], x_j \ne y_j\}$.
	Let $m:= |\free(s)|$.
	We have a verification sequence for $x$ on slice~$s$, say $v^0, v^1,\dotsc,v^m$.
	We do an orthogonal projection of this sequence onto $s'$, i.e., 
	onto the hyperplane determined by the fixed coordinates of $s'$. 
	Figure~\ref{fig:verif-seq-proj} provides an illustration.

	Observe that if $v^{i+1}-v^i$ was normal to the hyperplane passing through $s'$, 
	two points of the original verification sequence might be projected to the same point on $s'$.
	We define the sequence $w^0, w^1, \dots, w^{k}$ as the distinct points of the projected 
	sequence, where $k \le |\free(s')|$.

	For $i \in [k]$, let $d_i \in \free(s')$ denote the dimension travelled to get from $w^{i-1}$ to
	$w^{i}$, so $w^{i} - w^{i-1} = e_{d_i}$. 
	Every transition in the original verification sequence points upwards by
	definition, and if this does not persist after the projection, i.e., if there is 
	any $i\in [k]$ such that $f_{d_i} (w^{i-1}) - w^{i-1}_{d_i} \not > 0$, then
	$w^{i-1}$ together with its unique analogue in the original $v$ sequence that went up in 
	dimension $d_i$ witness a \ref{DV1} violation.

	Thus, we now assume that $f_{d_i} (w^{i-1}) - w^{i-1}_{d_i} > 0$ for all $i\in [k]$.
	%
	 
	Given two points $a, b\in G$, for $z=a-b$, let $\bind(z)$ be the dimensions with largest absolute value: 
	$\bind(z) \coloneqq \Set{i:\Normi{z} = \Abs{z_i}}$. 
	%
	%
	Since $\Normi{y-x} \ge 1$, there must be a step in the projected verification that increases
	that distance from $y$ by 1, i.e., there exists a $j \in [k]$ such that:
	$$\Normi{y-w^{j}} = \Normi{y-w^{j-1}} + 1.$$
	For the binding dimension $d_j \in \bind(w^{j-1} - y)$, we know that 
	$f_{d_j} (w^{j-1}) - w^{j-1}_{d_j} > 0$. 
	However, by assumption $f_{d_j} (y) - y_{d_j}  \leq 0$, which leads to 
	\[\ |f_{d_j}(w^{j-1}) - f_{d_j}(y)| >  |w^{j-1}_{d_j} - y_{d_j}|\] 
	which means that $w^{j-1}$ and $y$ constitute a violation of type \ref{DV2}.
\end{proof}

The final case is for when $x$ and $y$ are incomparable. Here we will generate a new point $z$ that 
is comparable (and beneath) both $x$ and $y$. 
We then show that if we do not get a violation of monotonicity when we consider each relevant 
dimension in turn, then we can use either the pair $x$ and $z$ or the pair $y$ and $z$ with 
one of the two Lemmas~\ref{lem:ov1-comparable-close} and~\ref{lem:ov1-comparable-far} to produce
the desired violation.

\begin{lemma} 
	\label{lem:ov1-incomparable}
	If $x$ and $y$ are both verfied LFPs of an $i$-slice $s$, and $x$ and $y$ are incomparable, 
	then we can find a violation of type~\ref{DV1} or~\ref{DV2} in polynomial time.
\end{lemma}

\begin{proof}
	We construct a point $z$ such that, for each dimension $i \in [d]$, $z_i = \min (x_i, y_i)$. 
	Since $x$ and $y$ are incomparable, we have $z < x$ and $z < y$. 
	We consider each dimension $i \in \free(s)$ in turn.
	If $z_i = x_i$ and $f_i(z) - z_i > 0$, we have a violation of type 
	\ref{DV1} between $z$ and $x$, since $z < x$ and $f_i(z) > f_i(x)$. 
	If $z_i = y_i$ and $f_i(z) - z_i > 0$, we have a violation of type 
	\ref{DV1} between $z$ and $y$.
	If we do not find a violation through this process, then we have $f_i(z) - z_i \le 0$ for all
	dimensions $i \in \free(s)$.
	
	Now we can pick the pair $x$ and $z$ or the pair $y$ and $z$; the argument is the same is 
	completely analogous in both cases.
	Let us take $x$ and $z$.
	Now if $z \in \sublattice(x-\One, x)$ we can apply Lemma~\ref{lem:ov1-comparable-close}
	to get the desired \ref{DV1} or \ref{DV2} violation;
	else $z \not \in \sublattice(x-\One, x)$ and we can apply Lemma~\ref{lem:ov1-comparable-far}
	to get the desired violation.
\end{proof}

With the preceding three lemmas we have covered all possible cases and thereby shown the following.

\begin{lemma} \label{lem:ov1_mapping}
	If there exists a violation of type \ref{OV1}, then we can find a violation of type \ref{DV1}
	or \ref{DV2} in polynomial time.
\end{lemma}

\paragraph{\bf Handling \ref{OV2} violations.} 

We now turn to \ref{OV2} violations, where we are given $i$-zeros from adjacent
slices, that point towards each other in dimension $i+1$. 

As a first step, we show that, if we have a fixed point of an $i$-slice $s$,
then for the slices $s'$ that are adjacent to $s$, we can either find a fixed
point of $s'$, or find a violation of the $\DMAC$ instance.  

\begin{lemma} 
	\label{lem:ov2-nbor-slice-poly-time}
	Suppose that $x$ is an $i$-fixed point of some $i$-slice $s$.
	Let $s'$ be an adjacent $i$-slice containing $y = x \pm e_{i+1}$. 
	Then we can in polynomial time find either:
	\begin{enumerate}
	\item an $i$-fixed point $z$ of $s'$ such that $\Normi{z - y} \le 1$; or
	\item a \ref{DV1} or \ref{DV2} violation.
	\end{enumerate}
\end{lemma}
\begin{proof}
	Let $Z \coloneqq B(x,1) \cap s'$, and let $z^0$ be the bottom corner of $Z$.
	Starting from $z^0$, we iterate $f$, which, for each $t = 1, 2, \dots$, gives a new point 
	$z^t = f(z^{t-1}) = f^t(z^0)$.
	We conclude the process when one of the following conditions is met:
	\begin{enumerate}
		\item $z^t \ngeq z^{t-1}$;
		\item $z^t \notin Z$;
		\item $z^t$ is an $i$-fixed point of $s'$.
	\end{enumerate}

	First note that either $z^0 \le f(z^0)$, or we have a \ref{DV2} violation of non-expansion
	between $z^0$ and $x$, which is fixed in $s$.
	So, suppose we have $z^0 \le f(z^0) = z^1$.
	Now, for $t > 1$, if the first condition is ever met, we have that $f(z^{t-1}) \ngeq
	f(z^{t-2})$ but $z^{t-1} \geq z^{t-2}$, which contradicts monotonicity of $f$ and is a
	violation of type~\ref{DV1}. 

	Next, note that if the second condition is ever met, we have a non-expansion violation of type
	\ref{DV2} between $z^{t-1}$ and $x$ since $\Normi{f(x) - f(z^{t-1})} > 1$ while 
	$\Normi{x - z^{t-1}} = 1$. 
	
	Finally, we observe that at most $2d$ iterations can be performed until the third condition is 
	met, assuming that neither the first or second condition was met during the process.
	To see this note that in each iteration we have $z^t > z^{t-1}$, which implies that at least
	one coordinate of $z^t$ has increased by $1$ from~$z^{t-1}$. 
	But this can only happen at most $2d$ times since $B(x,1)$ has width 2 in each dimension.
	The $i$-fixed point that we find in this case is the $z$ that we are looking for, where the 
	condition $\Normi{z - y} \le 1$ follows from the definition of $Z$. 
	This completes the proof.
\end{proof}

We now use this lemma to show that we can handle all \ref{OV2} violations. 

\begin{lemma} 
\label{lem:ov2_mapping}
	If there exists a violation of type \ref{OV2}, then we can find a violation of type \ref{DV1}
	or \ref{DV2} in polynomial time.
\end{lemma}

\begin{proof}
	Let $x$ and $y$ be the witness of the \ref{OV1} violation, so both are $(i-1)$-zeros of an $i$-slice
	$s$, we have $x_i = y_i + 1$, and $D_i(x) = \down$, $D_i(y) = \up$.
	Since $x_i = y_i + 1$, we have that $\Normi{x-y} \geq 1$.
	We consider two cases: 
	\begin{enumerate}
		\item There exists a $j \in \free(s)$ such that $|x_j-y_j| \ge 2$.
		\item For all $j \in \free(s)$ we have $|x_j-y_j| \le 1$. 
	\end{enumerate}

	We start with the first case, where there exists a $j \in \free(s)$ such that $|x_j-y_j| \ge 2$.
	Suppose that $y_j \le x_j -2$ (the case $x_j \le y_j -2$ is symmetric).
	We use Lemma~\ref{lem:ov2-nbor-slice-poly-time} to either find a \DMAC violation, in which 
	case we are done, or to find a fixed point $z$ of the $(i-1)$-slice that contains $y+e_i$ such 
	that:
	$$
	\Normi{z - (y + e_i)} \le 1.
	$$
	Observe that $z_j \le x_j - 1$, which is implied by $\Normi{z - (y + e_i)} \le 1$ and $y_j \le
	x_j -2$.
	Now consider the verification sequence of $x$, which we denote by $v^0, v^1,\ldots, v^k$.
	There must be a $l \in [k]$ such that $v^l = v^{l-1} + e_j$, but now $v^{l-1}$ and $z$
	are a \ref{DV2} violation of non-expansion since $f_j(z) = z_j$ as $z$ is an $(i-1)$-fixed point of $s$, and $j$ belongs to $\free(s)$ while also being the 
	binding dimension for $\Normi{v^{l-1} - z}$.
	That completes the proof for the first case.

	We now consider the second case where $|x_j - y_j| \le 1$ for all $j \in \free(s)$.
	We consider two sub-cases.
	\begin{enumerate}
		\item We have $x_j = y_j$ for all $j \in \free(s)$. In this case $x$ and $y$ are right
			next to each other and differ only in dimension $i$, so $y < x$.
			But then the fact that $D_i(x) = \down$ and $D_i(y) = \up$, together with the fact that 
			$x$ and $y$ are $(i-1)$-zeros, means that $f(x) < f(y)$,
			so we have a \ref{DV1} violation of monotonicity and we are done with this sub-case.
		\item There exists a $j \in \free(s)$ such that $|x_j - y_j| = 1$.
	\end{enumerate}
	For the second sub-case, we analyse two further sub-cases.

	First, suppose that there exists a $j \in \free(s)$ such that $x_j = y_j - 1$. 
	Now consider the point $z:= x-e_i$.
	Let $v^0,v^1,\ldots,v^k$ be the verification sequence for $y$. 
	There must have been a point $v^l$ that implies that $z$ points up in dimension $j$, i.e., 
	$f_j(v^l) > v_j^l$.
	Note that $v^l \le z$.
	If $v^l \ne z$, then we now query $f$ at $z$.
	If $f_j(z) \le z_j$, which is only possible if $v^l \ne z$, then $z$ and $v^l$ are a \ref{DV1}
	violation of monotonicity.
	%
	If $f_j(z) > z_j$, then $z$ and $x$ are a \ref{DV1} violation of monotonicity, since $z \le x$
	by definition, but $f(z) \not \le f(x)$ since $x$ is fixed in dimension $j$ and $f_j(z) > z_j$.

	Second, suppose that $y_j \le x_j - 1$ for all $j \in \free(s)$.
	In this case $y \le x$, and the fact that $D_i(x) = \down$ and $D_i(y) = \up$, means 
	that $f(y) \not \le f(x)$, so $x$ and $y$ are themselves a \ref{DV1} violation of monotonicity.
	This completes the proof of this final sub-case, and the proof overall.
\end{proof}

By combining Theorem~\ref{thm:promise-DMAC-to-promise-opdc} with Lemmas~\ref{lem:ov1_mapping}
and~\ref{lem:ov2_mapping}, we have now proven Theorem~\ref{thm:main}, which in turn gives us the
following.

\begin{theorem}
	\label{thm:dmac_lfp_ueopl}
	The problem of finding the least fixed point in a $\DMAC$ instance is in $\UEOPL$.
\end{theorem}


\section{1DUniqueDMAC and Surfaces}

For the rest of the paper, which is dedicated towards deriving algorithms for
\DMAC, we restrict ourselves to \emph{violation-free} \DMAC instances, which are
those instances that do not contain any violation solutions. We do this because
it significantly simplifies our presentation and proofs, and also because we
are interested in algorithms for problems like Shapley games that always
produce violation-free \DMAC instances.

In this section we formalize the idea of a \emph{surface} in a \DMAC instance.
We do this in two steps. We first introduce the \1DUniqueDMAC problem,
which is a restriction of \DMAC in which it is promised that all one-dimensional
slices have a unique one-dimensional fixed point. This then allows us to
formally define the notion of a surface. 

We show that the surfaces in a violation-free \1DUniqueDMAC instance are always
monotone and have gradient at most one. Importantly, we show that the converse
also holds: if the surfaces in a \1DUniqueDMAC are monotone and have gradient at
most one, and if we have one other technical requirement that we will define
later, then the instance is violation-free. 

This gives us a useful theoretical tool that we will use throughout the rest of
the paper: whenever we build a \1DUniqueDMAC instance, we do not need to
laboriously prove that the instance contains no violations of monotonicity or
strict contraction, because instead we can prove that the surfaces of the
instance are monotone and have gradient one, which is typically much easier to
verify.

\subsection{Reducing DMAC to 1DUniqueDMAC}

In this section we introduce a restriction of the \DMAC problem in which it is
promised the one-dimensional slices have a unique one-dimensional fixed
point.

Given a point $x$ and a dimension $i$, we define the set 
$S_i(x) = \{ y \in G \; : \; y = x + c \cdot e_i \text{ for some $c \in
\mathbb{Z}$} \}$ to be
the set of points in the one-dimensional slice in dimension $i$ that contains $x$.
We now define \1DUniqueDMAC, which is a \DMAC instance in which all one
dimensional slices have a unique one-dimensional fixed point.

\begin{definition}
We say that a \DMAC instance $f : G \rightarrow G$ is a \1DUniqueDMAC instance
if, $\| x - f(x) \|_\infty \le 1$ for all $x \in G$, and if
for every $x \in G$ and every $i$, there is exactly one point $y \in S_i(x)$
such that $f_i(y) = y_i$.
\end{definition}

Note that we have also required that $f$ has displacements that are at most
unit length. We have already shown that we can make this assumption 
without loss of generality in Lemma~\ref{lem:unitdisplacement}.

In this section we will show the following lemma, which shows that we can
reduce from a \DMAC instance to a \1DUniqueDMAC instance in polynomial time. We
also show that this reduction preserves the least fixed point of the instance,
which will be important later when we apply our decomposition theorem. 

\begin{lemma}
\label{lem:1dunique}
There is a polynomial-time reduction from violation-free \DMAC to
violation-free \1DUniqueDMAC that preserves
the least fixed point.
\end{lemma}

We first prove the lemma for a single dimension $i$. We give a reduction that
ensures that all one-dimensional slices in dimension $i$ have unique fixed
points. We can then simply apply this algorithm to each dimension $i$
independently to prove Lemma~\ref{lem:1dunique}.

Given a \DMAC instance $f: G \rightarrow G$, we first apply
Lemma~\ref{lem:unitdisplacement} to ensure that all displacements of $f$ have
length at most $1$ in the $\ell_\infty$-norm, and then we define $f': G
\rightarrow G$ in the following way. For each $x \in G$ we define $f'_j(x) =
f_j(x)$ for all $j \ne i$, and for dimension $i$ we define
\begin{equation*}
f'_i(x) = \begin{cases}
x_i - 1 & \text{if $f_i(x) = x_i$ and $f_i(x-e_i) = x_i-1$,} \\
f_i(x) & \text{otherwise.}
\end{cases}
\end{equation*}
This definition simply follows $f$ unless we are at a point $x$ that is a
one-dimensional fixed point in dimension $i$, where $x$ is also directly above another
one-dimensional fixed point in dimension $i$. At any such point we alter $f$ so
that $x$ moves one unit downwards in dimension $i$, which prevents $x$ from
being a one-dimensional fixed point in dimension $i$. It is clear that we can
build $f'$ in time that is polynomial in the representation of $f$.

It is also clear that $f'$ has displacements that are at most unit length, because
we applied Lemma~\ref{lem:unitdisplacement} to $f$ to ensure that it has this
property, and then $f'$ either copies the displacements of $f$, or introduces a
unit displacement downward.
The next lemma shows that the reduction correctly ensures that all
one-dimensional slices in dimension $i$ have a unique one-dimensional fixed
point.

\begin{lemma}
\label{lem:1dured}
If $f$ is violation-free, then
every one-dimensional slice in dimension $i$ in $f'$ has a unique fixed point.
\end{lemma}
\begin{proof}
First suppose that there is no point $y \in S_i(x)$ such that $f'_i(y) = y_i$.
Since $f'$ only removes one-dimensional fixed points, it must be the case that
$f_i(y) \ne y_i$ for all $y \in S_i(x)$ as well.
Let $l$ be the least element of 
$S_i(x)$ and observe that we must have $f_i(l) \ge l_i$. Likewise, if
$g$ is the greatest element of $S_i(x)$ then we have
$f_i(g) \le g_i$. So we can apply binary search starting at $l$ and $g$ to find two adjacent points
$a$, and $b = a + e_i$ such that $f_i(a) \ge a_i$ and $f_i(b) \le b_i$. If
either of those inequalities is weak, then we contradict the assumption that
there is no point $y \in S_i(x)$ such that $f_i(y) = y_i$. If they are both
strict, then $x$ and $y$ are a violation of monotonicity in $f$, which is also
a contradiction.

Now suppose that there are two distinct points $x, y \in S_i(x)$ such that
$f'_i(x) = x_i$ and $f'_i(y) = y_i$, and suppose without loss of generality that
$x \le y$. 
Note that we cannot have $y_i = x_i + e_i$, because then the
definition of $f'$ would have set $f'_i(y) = y_i - 1$. Hence, the point $y' = y
- e_i$ is distinct from both $x$ and $y$. The definition of $f'$ tells us
  that $f_i(x) = x_i$ and $f_i(y) = y_i$, since $f'$ never introduces a
new one-dimensional fixed point, and it also tells us we do not have $f_i(y') = y'_i$
because if we did then $f'_i(y)$ would be set to $y_i - 1$. But now we can see
that if $f_i(y') > y'_i$
then $y'$ and $x$ are a strict violation of contraction, while if 
$f_i(y') < y'_i$ then $y'$ and $y$ are a strict violation of contraction, both of
which are contradictions. 
\end{proof}

The next lemma shows that the reduction is a correct 
reduction from violation-free \DMAC to violation-free \DMAC.

\begin{lemma}
Every fixed point of $f'$ is a fixed point of $f$. Furthermore, if $f$ is
violation-free, then $f'$ is violation-free.
\end{lemma}
\begin{proof}
It is clear that any fixed point of $f'$ is a fixed point of $f$, since $f'$
modifies $f$ only by introducing strict displacements. 

Next we consider violations of type \ref{DV1}. 
Suppose that we have two points $x, y \in G$ such that $x \le y$ and $f'(x)
\not \le f'(y)$. 
If $f'(x) = f(x)$ and $f'(y) = f(y)$ then $x$ and $y$ witness a violation
of monotonicity in~$f$, and so we are done. 
If $f'_j(x) > f'_j(y)$ for some $j \ne i$, then since we have not changed
dimension $j$, we have that 
$x$ and $y$ also witness a violation of monotonicity in $f$.
So we can assume that the violation of monotonicity has $f'_i(x) > f'_i(y)$. 

We now consider the following cases.
\begin{itemize}
\item If $f'_i(x) \ne f_i(x)$, then we have $f_i'(x) = x_i - 1$ from the
definition of $f'$. But since all displacements are at most unit length, we
also have $f'_i(y) \ge y_i -1$, so we have $f_i'(x) = x_i - 1 \le y_i - 1 \le
f'_i(y)$, which contradicts the assumption that $x$ and $y$ violated
monotonicity. 

\item If $f'_i(x) = f_i(x)$ and $f'_i(y) \ne f_i(y)$ then note that we must
have $x_i \ge y_i - 1$ otherwise it would be impossible for us to have $f'_i(x)
> f'_i(y)$ in an instance where all displacements are at most unit length. 
If we define $y' = y - e_i$, then note that $f_i(y') = y'_i$ because we had
$f'(y) \ne f(y)$. 
We now consider the following cases.

\begin{itemize}

\item If $f'_i(x) = x_i - 1$ then we cannot have 
$f'_i(x) > f'_i(y)$ since $f'_i(y) \ge y_i - 1 \ge x_i - 1 = f'_i(x)$. So $x$ and
$y$ are not a violation of monotonicity, which is a contradiction with our
assumptions.

\item If $f'_i(x) = x_i + 1$ and $x_i = y_i - 1$ then we have $f_i(x) = f'_i(x)
= x_i + 1 = y_i = y'_i + 1 = f_i(y') + 1$. Hence $x$ and $y'$ are a violation
of monotonicity in $f$.

\item If $f'_i(x) = x_i + 1$ and $x_i \ge y_i$ then we have 
$f_i(x) = f'_i(x) = x_i + 1 \ge y_i + 1 > f_i(y) + 1$, where the final
inequality holds because $f'(y) \ne f(y)$. Hence $x$ and
$y$ are a violation of monotonicity in $f$.

\item If $f'_i(x) = x_i$ and $x_i = y_i - 1$ then we have $f'_i(x) =
x_i = y_i - 1 \le f'_i(y)$, which contradicts our assumption.

\item If $f'_i(x) = x_i$ and $x_i \ge y_i$ then we also consider the point $x'
= x - e_i$. Note that we cannot have $f_i(x') = x'_i$ because $f'(x) = f(x)$,
and we cannot have $f_i(x') < x'_i$ because this would be a violation of
contraction between $x$ and $x'$. Thus we have $f_i(x') > x'_i$. So we have
$f_i(x') > x'_i = x_i - 1 \ge y_i - 1 = y'_i = f_i(y')$. So $x'$ and $y'$ are a
violation of monotonicity for $f$.
\end{itemize}
\end{itemize}

Next we consider violations of type \ref{DV2}. Suppose that we have two points
$x, y \in G$ such that $\| f'(x) - f'(y) \|_\infty > \| x - y \|_\infty$. 
If there exists a dimension $j \ne i$ such that $|f'_j(x) - f'_j(y) | > \| x
- y \|_\infty$, then because we have not altered dimension $j$, the points $x$
and $y$ must also violate contraction in $f$. So we can assume that we have
$|f'_i(x) - f'_i(y) | > \| x - y \|_\infty$.

If $f'(x) = f(x)$ and $f'(y) = f(y)$ then clearly $x$ and $y$ are a violation
of contraction for $f$. If, on the other hand, we have $f'(x) \ne f(x)$ and
$f'(y) \ne f(y)$, then note that 
\begin{align*}
|f'_i(x) - f'_i(y) | &= | (f_i(x) - 1) - (f_i(y) - 1) | \\
&= |f_i(x) - f_i(y)|,
\end{align*}
because $f'$ only ever modifies $f$ in order to move a point that was fixed in
dimension $i$ downward. Hence $x$ and $y$ are again a violation of contraction
in $f$.

So we can assume, without loss of generality, that $f'(x) = f(x)$ and $f'(y) \ne
f(y)$.  Since $f'(y) \ne y$ we must have $f'_i(y) = y_i - 1$ due to the
definition of $f'$. Therefore we must have $x_i \ge y_i$, since otherwise it
would be impossible for us to have $|f'_i(x) - f'_i(y) | > \| x - y \|_\infty$.

Since $f'$ is an instance in which all displacements are at most unit length,
we must have $x_i - y_i \ge \| x - y \|_\infty - 1$, because otherwise it would
be impossible to have 
$|f'_i(x) - f'_i(y) | > \| x - y \|_\infty$. We consider the following cases.

\begin{itemize}
\item If $x_i - y_i = \| x - y \|_\infty - 1$ then note that we must have
$f'(x_i) = x_i + 1$, because otherwise we could not have
$|f'_i(x) - f'_i(y) | > \| x - y \|_\infty$. Let $y' = y - e_i$, and observe
that since $f'(y) \ne f(y)$, we must have $f_i(y') = y'_i$. 
Also note that $x_i - y'_i = x_i - y_i + 1 = \| x - y \|_\infty$, and therefore
we must have $\| x - y' \|_\infty = \|x - y \|_\infty$ because we only modified
dimension $i$ when we created $y'$ from $y$.
So we have
\begin{align*}
f_i(x) - f_i(y') &= x_i + 1 - y'_i  \\
&= x_i - y_i + 2 \\
&= \| x - y \|_\infty + 1 \\
&= \| x - y'\|_\infty + 1,
\end{align*}
and therefore $x$ and $y'$ are a violation of contraction in $f$.

\item If $x_i - y_i \ge \| x - y \|_\infty$ then note that we must have $f_i(y)
= y_i$ because $f'(y) \ne f(y)$. 
If $f_i(x_i) = x_i + 1$, then we immediately have that $x$ and $y$ violate
contraction in $f$. We cannot have $f_i(x_i) = x_i - 1$ because this would
contradict the fact that 
$|f'_i(x) - f'_i(y) | > \| x - y \|_\infty$.

So the remaining case is where $f_i(x_i) = x_i$. Let $x' = x - e_i$ and $y' = y
- e_i$. Observe that since $f'(x) = f(x)$ we cannot have $f_i(x') = x'_i$, and
  we also cannot have $f_i(x') = x'_i - 1$ because this would give us a
violation of contraction between $x$ and $x'$. So we have $f_i(x') = x'_i + 1$.
On the other hand, since $f'(y) \ne f(y)$, we must have $f_i(y') = y'_i$. Note
also that $\| x - y \|_\infty = \| x' - y' \|_\infty$, since $x'$ and $y'$ were
both generated by shifting $x$ and $y$ by $e_i$. So we have
\begin{align*}
f_i(x') - f_i(y') &= x'_i - y'_i + 1 \\
&= x_i - y_i + 1 \\
&\ge \| x - y \|_\infty \\
&=\| x' - y' \|_\infty,
\end{align*}
so $x'$ and $y'$ witness a violation of contraction in $f$.
\end{itemize}
\end{proof}

The reduction from $f$ to $f'$ can delete some of the fixed points of $f$, so
it is not necessarily the case that every fixed point of $f$ is also a fixed
point of $f'$. However, in the next lemma we show that the least fixed point is
preserved by the reduction.  

\begin{lemma}
If $f$ is violation-free, then the least fixed point of $f'$ is the same as the least fixed point of $f$.
\end{lemma}
\begin{proof}
Let $x$ be the least fixed point of $f$. We can apply the hereditariness
property given in Theorem~\ref{thm:hereditary} to argue that $x$ is the least
one-dimensional fixed point in $S_i(x)$. Hence, if $x' = x - e_i$, then we must
have $f_i(x') \ne x'_i$, and so the definition of $f'$ implies that $f_i(x) =
f'_i(x)$. Since $f'$ only modifies dimension $i$ this implies that $f'(x) =
f(x)$, and so $x$ is a fixed point of $f'$.

Lemma \ref{lem:1dured} implies that if there were a fixed point $y$ of $f'$
satisfying $y \le x$, then $y$ is also a fixed point of $f$ satisfying $y \le
x$, which contradicts the fact that $x$ is the least fixed point of $f$.

\end{proof}

In summary, we have shown that the reduction correctly reduces a violation-free \DMAC instance
to another violation-free \DMAC instance in which each one-dimensional slice in dimension $i$
has a unique one-dimensional fixed point, and that the reduction preserves the
least fixed point. We can then apply this reduction to each dimension
iteratively 
to ensure that all one-dimensional slices have a unique one-dimensional fixed
point. This is valid because the reduction never introduces a one-dimensional
fixed point in any dimension, so each iteration will not destroy the properties
from the previous iterations. 

Thus we arrive at a violation-free \DMAC instance in which all displacements are at most unit
length, and in which every one-dimensional slice has a unique one-dimensional
fixed point. Hence it is a \1DUniqueDMAC instance.

It is clear that this reduction can be carried out in time that is polynomial
in the representation of $f$. In particular we note that the representation of $f$ must
be at least $d$, which is the number of dimensions, because the input to $f$ is
of size $d \cdot log n$. This completes the proof of~Lemma~\ref{lem:1dunique}

\subsection{Surfaces}

We now formalize the concept of a surface, which we will use as a proof
technique throughout the paper. 

Suppose that we have a
\1DUniqueDMAC instance $f : G \rightarrow G$, 
where $G = \{1, 2, \dots, n\}^d$, 
and let $i$ be a dimension. 
For each $x \in G$ we define $x_{-i}$ to be the point $x$ with dimension $i$
removed. We likewise define $G_{-i}$ to be the grid $G$ in which dimension $i$
has been deleted. Finally, given a point $x \in G_{-i}$ and a value $a \in \{1,
2, \dots, n\}$ we define $x \circ_i a$ to be the point in $G$ obtained by
placing the value $a$ in dimension $i$ of $x$.

We use the following definition of a surface.

\begin{definition}[Surface]
Given a \1DUniqueDMAC instance $f : G \rightarrow G$, the surface in
dimension $i$ is the function $s_i : G_{-i} \rightarrow \{1, 2, \dots, n\}$ such
that $s_i(x) = a$, where $a$ is the unique value such that $f_i(x \circ_i a) =
a$.
\end{definition}

Intuitively, we view the unique one-dimensional fixed points in dimension $i$
as heights that define the surface. It is of course necessary that we have a \1DUniqueDMAC
instance here, because otherwise there would not necessarily be a unique value
$a$ such that $f_i(x \circ_i a) = a$, and the surface would not be well
defined.

We remark that the surface function $s_i$ is not necessarily easy to compute:
in general a binary search must be carried out to find the value $a$ that
satisfies $f_i(x \circ_i a) = a$, which would take $\log n$ steps. This makes
surfaces unsuitable for defining reductions, where spending $\log n$ steps to
answer a single query to~$f$ would be too expensive. However, surfaces have
several properties that will be useful in our proofs. 

\paragraph{\bf Monotonicity.}

We say that a surface $s_i$ is \emph{monotone} if, given two points $x, y \in
G_{-i}$ with $x \le y$ we have that $s_i(x) \le s_i(y)$. 
We show that every surface arising from a 
violation-free \1DUniqueDMAC instance is monotone.

\begin{lemma}
If $s_i$ is a surface of a violation-free \1DUniqueDMAC instance, then $s_i$ is
monotone.
\end{lemma}
\begin{proof}
Suppose for the sake of contradiction that we have two points $x, y \in G_{-i}$
such that $x \le y$ but $s_i(x) > s_i(y)$. Let $x' = x \circ_i s_i(x)$ and 
$y' = y \circ_i s_i(y)$ be the points on the surface at $x$ and $y$ in $G$.

Define $z = y \circ_i s_i(x)$, and observe that $x' \le z$. Since $f_i(x') =
x'_i$ we cannot have $f_i(z) < z_i$, because then we would have $f_i(z) < z_i =
x'_i = f_i(x')$, meaning that $z$ and $x'$ would witness a violation of
monotonicity of $f$. We also cannot have 
$f_i(z) = z_i$ because $y'$ is the unique one-dimensional fixed point of the
one-dimensional slice that contains $y'$ and $z$, and since $s_i(x) > s_i(y)$,
we have that $y'$ and $z$ are distinct points. Finally if we have
$f_i(z) > z_i$ then $y'$ and $z$ are a violation of contraction. So we have
arrived at a contradiction.
\end{proof}

\paragraph{\bf Gradient at most one.}

We say that a surface $s_i$ has \emph{gradient at most one} if, for every pair of
points $x, y \in G_{-i}$ we have that $| s_i(x) - s_i(y) | \le \| x - y \|_\infty$.
We show that every surface arising from a 
violation-free \1DUniqueDMAC instance has gradient at most one.

\begin{lemma}
If $s_i$ is a surface of a violation-free \1DUniqueDMAC instance, then $s_i$ 
has gradient at most one.
\end{lemma}
\begin{proof}
Suppose for the sake of contradiction that we have 
$x, y \in G_{-i}$ with $| s_i(x) - s_i(y) | > \| x - y \|_\infty$.
Let $x' = x \circ_i s_i(x)$ and 
$y' = y \circ_i s_i(y)$ be the points on the surface at $x$ and $y$ in $G$.
We will assume, without loss of generality, that $s_i(x) > s_i(y)$. 

Define the point $z = x' - e_i$ to be the point directly beneath $x$ in
dimension $i$. Note that we cannot have $f_i(z) < z_i$ because then $z$ and
$x'$ would be a violation of contraction, and we cannot have $f_i(z) = z_i$
because $x'$ is the unique one-dimensional fixed point of the one-dimensional
slice that contains $x'$ and $z$.
So we must have $f_i(z) > z_i$.

We now consider two cases.
\begin{itemize}
\item If $z_i - y'_i = \| z - y' \|_\infty$ then observe that we have 
$f_i(z) - f_i(y') > z_i - y'_i$ since $y'$ lies on the surface $s'_i$. This
means that $z$ and $y$ are a violation of contraction.

\item On the other hand, if 
$z_i - y'_i < \| z - y' \|_\infty$ then observe that $\| z - y' \|_\infty = \|
x - y \|_\infty$. Then we have
\begin{align*}
f_i(z) - f_i(y') &> z_i - y'_i \\
&= x'_i - 1 - y'_i  \\
&\ge \| x - y \|_\infty \\
& = \| z - y' \|_\infty,
\end{align*}
where the second inequality holds because we have $x'_i - y'_i > \| x - y
\|_\infty$. So $z$ and $y'$ are again a violation of contraction.
\end{itemize}
So in either case we have arrived at a contradiction.
\end{proof}

\paragraph{\bf 1D rationality.}

We say that a \1DUniqueDMAC instance $f : G \rightarrow G$ is \emph{1D
rational} if, for every $x \in G$ and every dimension $i$, we have that $f_i(x)
> x_i$ whenever $x_i < s_i(x_{-i})$, and we have that $f_i(x) < x_i$ whenever
$x_i > s_i(x_{-i})$. Intuitively, this states that $f$ always moves towards the
surface in dimension $i$, and in practical terms it rules out the existence of
violations within each one-dimensional slice. 

The next lemma states that all violation-free \1DUniqueDMAC instances are 1D
rational.

\begin{lemma}
If $f$ is a violation-free \1DUniqueDMAC instance, then $f$ is 1D rational.
\end{lemma}
\begin{proof}
Suppose for the sake of contradiction that we have a point $x$ such that $x_i <
s_i(x_{-i})$ and $f_i(x) \le x_i$. Let $y = x_{-i} \circ_i s_i(x_{-i})$ be the
point on the surface that lies above $x$ in dimension $i$, and observe that by definition we
have $f_i(y) = y_i$. We cannot have $f_i(x) = x_i$, because $y$ is the unique
fixed point of the one-dimensional slice that contains $x$ and $y$. So we must
have $f_i(x) < x_i$. This means that $x$ and $y$ are a violation of
contraction, which is a contradiction.

The same argument can be applied symmetrically in the case where
we have a point $x$ such that $x_i > s_i(x_{-i})$ and $f_i(x) \ge x_i$.
\end{proof}

\paragraph{\bf 1D rationality, monotonicity, and gradient at most one are sufficient.}

We have so far shown that if a \1DUniqueDMAC instance is violation-free, then
the instance is 1D rational, and the corresponding surfaces are monotone with
gradient at most one. We now show that the other direction also holds: if a
\1DUniqueDMAC instance is 1D rational, and all surfaces are monotone with
gradient at most one, then the instance is violation-free.

\begin{lemma}
\label{lem:surfacesback}
Let $f : G \rightarrow G$ be a \1DUniqueDMAC instance.  If $f$ is 1D rational,
and if for each dimension $i$ the surface $s_i$ is monotone and has gradient at
most one, then $f$ is violation-free.
\end{lemma}
\begin{proof}
We begin with \ref{DV1} violations. Suppose for the sake of contradiction that
we have $x, y \in G$ with $x \le y$ but $f(x) \not \le f(y)$. Let $i$ be a
dimension for which $f_i(x) > f_i(y)$. 

Since $f$ has displacements that are at most unit length, we cannot have
$f_i(x) = x_i - 1$, because this would imply that $f_i(x) = x_i - 1 \le y_i - 1
\le f_i(y)$, which contradicts our assumption. Likewise, we cannot have $f_i(y)
= y_i + 1$, because this would imply that $f_i(x) \le x_i + 1 \le y_i + 1 =
f_i(y)$, which also contradicts our assumption. So we have $f_i(x) \ge x_i$ and
$f_i(y) \le y_i$. 

Now we can apply 1D rationality to conclude that $s_i(x_{-i}) \ge x_i$ while
$s_i(y_{-i}) \le y_i$. If both of these inequalities are weak, then we have
$f_i(x) = x_i \le y_i = f_i(y)$, which contradicts our assumption. 
We then consider two cases.
\begin{itemize}
\item If $s_i(x_{-i}) > x_i$ then note that if $f_i(y) = y_i$ then we have
$s_i(y_{-i}) = y_i = f_i(y)$, while if $f_i(y) < y_i$ then we have $s_i(y_{-i})
\le y_i - 1 \le f_i(y)$. So in either case we have $s_i(y_{-i}) \le f_i(y)$. So
we have $s_i(x_{-i}) \ge x_i + 1 \ge f_i(x) > f_i(y) \ge s_i(x_{-i})$. This
implies that $s_i(x_{-i})
> s_i(y_{-i})$ despite the fact that $x_{-i} \le y_{-i}$. So $s_i$ is not
 monotone, which is a contradiction.

\item On the other hand, if 
$s_i(y_{-i}) < y_i$  then we can apply the same argument as above with all
inequalities flipped to again obtain a contradiction with the monotonicity of
$s_i$. 
\end{itemize}

We now turn our attention to violations of type \ref{DV2}. Assume for the sake
of contradiction that we have a pair of points $x$ and $y$ points satisfying
$\Norm{f(x) - f(y)}_{\infty} > \Norm{x-y}_{\infty}$.
Let $i$ be a dimension such that $| f_i(x) - f_i(y) | > \| x - y \|_\infty$. 
We assume, without loss of generality, that $x_i < y_i$. 

Note that, since the instance has displacements that are at most unit length,
we cannot have $f_i(x) = x_i + 1$, because we have $f_i(y) \le y_i + 1$ and so
we would have $| f_i(x) - f_i(y) | \le | x_i - y_i| \le \| x - y \|_\infty$,
which contradicts our assumption. We likewise cannot have $f_i(y) = y_i - 1$
for the same reason. 

So we have $f_i(x) \le x_i$ and $f_i(y) \ge y_i$. 
Note that we must therefore have $s_i(x_{-i}) \le f_i(x)$ and $s_i(y_{-i}) \ge
f_i(y)$. We now consider two cases.
\begin{itemize}
\item If $| x_i - y_i| = \| x - y \|_\infty$ then we have $| s_i(x_{-i}) -
s_i(y_{-i})| \ge | f_i(x) - f_i(y) | > |x_i - y_i| = \| x - y \|_\infty \ge \|
x_{-i} - y_{-i} \|_\infty$. So $x_{-i}$ and $y_{-i}$ witness a violation of 
the fact that $s_i$ has gradient 
at most one, which is a contradiction. 

\item If $| x_i - y_i| < \| x - y \|_\infty$ then note that we must have $|x_i
- y_i| = \| x - y \|_\infty - 1$, because otherwise it would be impossible to
have $| f_i(x) - f_i(y) | | > \| x - y \|_\infty$ in an instance with at most
unit displacements. For the same reason we must also have both $f_i(x) > x$ and
$f_i(y) < y$. So we have
$| s_i(x_{-i}) - s_i(y_{-i})| \ge | f_i(x) - f_i(y) | > |x_i - y_i| + 1 = \| x
- y \|_\infty \ge \| x_{-i} - y_{-i} \|_\infty$, and we again contradict the
  fact that $s_i$ has gradient at most one.
\end{itemize}
\end{proof}

The following theorem summarizes what we have shown in this section.
\begin{theorem}
\label{thm:surfaces}
Let $f$ be a \1DUniqueDMAC instance. We have that $f$ is violation-free
if and only if $f$ is 1D rational, and the surfaces defined by $f$ are monotone
with gradient at most one.
\end{theorem}

Theorem~\ref{thm:surfaces} gives us a useful tool for showing that a
\1DUniqueDMAC instance is violation-free. Instead of having to laboriously
rule out each possible type of violation, it is sufficient to ensure only that
the instance is 1D rational, and that the surfaces defined by the instance are
monotone and have gradient at most one. We will use this theorem repeatedly as we
implement reductions from \1DUniqueDMAC to \1DUniqueDMAC later on in the paper.

\subsection{Forcing Boundary Points to Move Inward}

We now apply the surface theorem to show 
that $\1DUniqueDMAC$ can be reduced, in polynomial time, to
$\1DUniqueDMAC$ in which all points on the boundary of the instance move
strictly inward. This is a property that we will use both in our decomposition
theorem and in our algorithm. 

\begin{lemma}
\label{lem:boundaryin}
Let $G = \{1, 2, \dots, n\}^d$, and
let $f : G \rightarrow G$ be a $\1DUniqueDMAC$ instance for all
x. There is a polynomial-time reduction that creates a 
$\1DUniqueDMAC$ instance $f'$ such that for
all dimensions $i$ we have 
$f_i(x) > x_i$ whenever $x_i = 1$ and $f_i(x) < x_i$ whenever $x_i = n$.
Moreover, the least fixed point of $f'$ can be mapped back to the least fixed
point of $f$.
\end{lemma}

We start by giving a reduction that ensures that $f_1(x) > x_1$ whenever $x_1 = 0$. 
We define the grid $G' = \{0, 1, \dots, n\} \times \{1, 2, \dots, n\}^{d-1}$,
meaning that we have extended dimension $1$ downward by one unit. 


We define $f' : G' \rightarrow G'$ in the following way.
\begin{equation*}
    f'(x) = \begin{cases}
        f(x + e_1) & \text{if $x_1 = 0$,} \\
        f(x) & \text{otherwise.}
    \end{cases}
\end{equation*}
This ensures by construction that $x_1 < f_1(x)$ for all $x$ with $x_1 = 0$.

The following lemma justifies this operation. It shows that the
transformation does not introduce any new fixed points, and if the original
instance was violation-free, then the new instance will be violation-free as
well.

\begin{lemma}
A point is a fixed point of $f'$ if and only if it is a fixed point of $f$.
Moreover, if $f$ is violation-free, then $f'$ is also violation-free.
\end{lemma}
\begin{proof}
Since $f'$ does not alter $f$ on the points in $G$, it is clear 
that any fixed
point of $f'$ in $G$ is also a fixed point of $f$ and vice versa. Moreover, for every point $x
\not \in G$, we have $f'_1(x) > x$ by definition, so $x$ cannot be a fixed point
of $f'$. 

If $f$ is violation-free, then we can apply Theorem~\ref{thm:surfaces} to argue
that the instance is 1D rational, and that all surfaces are monotone with
gradient at most one. Observe that $f'$ does not change the surfaces of~$f$,
and that since all points $x$ with $x_1 = 0$ move upward, we also have that
$f'$ is 1D rational. Hence $f'$ must also be violation-free by 
Theorem~\ref{thm:surfaces}.
\end{proof}

We can apply the same operation to each other dimension in sequence, to
ensure that our assumption holds for all dimensions. For the
requirements on the points $x$ with $x_i = n$, we can flip all dimensions, and
use the same reduction. It is clear that the reduction can be carried out in
time that is polynomial in the representation of $f$. Finally, since the
reduction does not change the set of fixed points, it is clear that the least
fixed point of $f'$ is also the least fixed point of $f$. This completes the proof
of Lemma~\ref{lem:boundaryin}.


\section{A Decomposition Theorem}
\label{sec:decomposition}

\newcommand{\CA}{{\mbox{$\mathcal A$}}}
\newcommand{\CB}{{{\mathcal B}}}
\newcommand{\tx}{{{\tilde{x}}}}
\newcommand{\ty}{{{\tilde{y}}}}
\newcommand{\lbx}{{{\underline{x}}}}
\newcommand{\lby}{{{\underline{y}}}}
\newcommand{\ubx}{{{\bar{x}}}}
\newcommand{\uby}{{{\bar{y}}}}
\def\ones{\mathbf 1}


In this section, we will prove a decomposition theorem for \DMAC, which is inspired by a
similar decomposition theorem that has recently been shown for the Tarski fixed
point problem~\cite{FPS22}.


Unlike the decomposition theorem given in~\cite{FPS22}, our decomposition
theorem only works for algorithms that find the least fixed point of the
instance. 
Specifically, 
given a \DMAC instance $f$, we say that an algorithm \emph{LFP-solves} the
instance if it returns a \ref{D1} solution that is a least fixed point of $f$,
or a \ref{DV1} violation, or a \ref{DV2} violation. Note that
Tarski's theorem guarantees that the instance will either have a least fixed
point or a \ref{DV1} violation, so the problem remains total even if we insist
that it must be LFP-solved. In Section~\ref{sec:decomposition} we show a
decomposition theorem for algorithms that LFP-solve a \DMAC instance.

Our algorithm for three-dimensional instances does not explicitly LFP-solve the
instance. This is not a problem, because
in Section~\ref{sec:uniqueness} 
we give a polynomial-time that transforms a violation-free three-dimensional
\DMAC instance $f$ into a violation-free three-dimensional \DMAC instance $f'$,
such that $f'$ has a \emph{unique} fixed point $x^*$, and $x^*$ can be mapped
back to the least fixed point of $f$. Thus applying our algorithm for
violation-free \DMAC to $f'$ will allow us to LFP-solve $f$, and thus we can
apply the decomposition theorem to our algorithm.

\subsection{A Decomposition Theorem for LFP-solving DMAC}

In this section we prove the following theorem.

\begin{theorem}
	\label{thm:decomposition}
	If a $d_1$-dimensional \DMAC instance can be LFP-solved in $q_1$ time by algorithm $\CA$ and
	a $d_2$-dimensional \DMAC instance can be LFP-solved in $q_2$ time by algorithm $\CB$, then
	there exists a $(q_1 \cdot q_2)$-time algorithm to LFP-solve
	$(d_1+d_2)$-dimensional \DMAC instances. 
\end{theorem}

Let the given \DMAC instance be $f:G\rightarrow G$ where $G = \{1, 2, \dots, n\}^d$. Let $d=d_1+d_2$ for $d_1, d_2 \ge 1$. We now define some operators that will help us to decompose this instance.
\begin{itemize}
	\item For each 
	$x = (x_1, x_2, \dots, x_{d_1}) \in \reals^{d_1}$ and
	$y = (y_1, y_2, \dots, y_{d_2}) \in \reals^{d_2}$ we define
	\begin{equation*}
		x \oplus y = (x_1, x_2, \dots, x_{d_1}, y_1, y_2, \dots, y_{d_2}) 
	\end{equation*}
	to be the point obtained by concatenating $y$ and $x$.

	\item 
	For each $x = (x_1, x_2, \dots, x_{d_1 + d_2}) \in \reals^{d_1 + d_2}$ we define $x
	\restriction d_1 = (x_{1}, x_{2}, \dots, x_{d_1})$ to be the point obtained by
	deleting all but the first $d_1$ coordinates of $x$.
	
	\item For each $x = (x_1, x_2, \dots, x_{d_1 + d_2}) \in \reals^{d_1 + d_2}$ we define $x
	\restriction d_2 = (x_{d_1+1}, x_{d_1+2}, \dots, x_{d_1+d_2})$ to be the point obtained by
	deleting the last $d_2$ coordinates of $x$.
	
\end{itemize}

At a high level, the idea is to run Algorithm $\mathcal{A}$ on an instance
$f^{\mathcal{A}}:G_1 \rightarrow G_1$ over domain $G_1=\{1,2\dots,n\}^{d_1}$. Whenever $\mathcal{A}$ queries a point
$x \in G_{1}$, we then run algorithm $\mathcal{B}$ on the instance
$f^{x}:G_2\rightarrow G_2$ over domain $G_2=\{1,2,\dots,n\}^{d_2}$, which is defined so that 
\begin{equation*}
	f^x(y) = f(x\oplus y)\restriction d_2,
\end{equation*}

In other words, $f^x$ is the {\em sub-instance} of $f$ where the first $d_1$
coordinates are fixed to be $x$. If $\mathcal{B}$ returns a fixed point (\ref{D1} solution) $y^*$ of $f^x$, where $f^x(y^*)=y^*$, then we return $f(x \oplus y^*) \restriction d_1$ to $\mathcal{A}$ and continue. Thus $f^\CA$ can be thought of as a {\em super-instance}. 


We first show that if $\mathcal A$ and $\mathcal B$ return the LFP of corresponding instances, then we get an LFP of the original instance. 
\begin{lemma}\label{lem:lfp}
	If $x^*$ is an LFP of instance $f^\CA$ returned by $\mathcal{A}$, and $y^*$ is an LFP of the instance $f^{x^*}$ returned by $\CB$, then $(x^* \oplus y^*)$ is an LFP of instance $f$. 
\end{lemma}
\begin{proof}
	Towards a contradiction, suppose $(x^*\oplus y^*)$ is not the LFP of $f$. Then there exists another fixed point of $f$ namely $(\tx \oplus \ty)\le (x^*\oplus y^*)$. So, either $\tx_i < x^*_i$ for some $i\le d_1$ contradicting $x^*$ being the LFP of $f^\CA$. Or $\tx=x^*$ and for some $j\le d_2$, $\ty_j \le y^*_j$ contradicting $y^*$ being the LFP of $f^{x^*}$.
\end{proof}

Of course, either $\mathcal{A}$ or $\mathcal{B}$ could instead return either type of the violation solutions, namely \ref{DV1} and \ref{DV2}, at any point in this procedure. We will show that these
violations can be efficiently translated into violations in $f$, meaning that
we can stop the procedure immediately if this occurs.\footnote{If algorithms
	$\CA$ or $\CB$ are only for violation-free instances, and may return a non-solution or does not stop if the violations are present, then the same will hold for our final algorithm as well.} The following lemma discusses violations from algorithm $\mathcal{B}$.

\begin{lemma}\label{lem:gx-violations}
	For some $x\in G_1$, suppose on the input instance $f^x$, algorithm $\mathcal{B}$ produces $y,z\in G_2$, which witness a violation of type \ref{DV1} or \ref{DV2}. Then $(x\oplus y)$ and $(x\oplus z)$ form a \ref{DV1} or \ref{DV2} violation for the instance $f$ respectively.    
\end{lemma}
\begin{proof}
	For the first case, suppose $y,z\in G_2$ witness a \ref{DV1} violation. That is, $y\le z$ but $f^x(y) \not\leq f^x(z)$. Then we have that $(x\oplus y) \le (x\oplus z)$, while $$f^x(y)=f(x\oplus y)\restriction d_2 \not\leq f^x(z)= f(x\oplus z)\restriction d_2 \Rightarrow f(x\oplus y) \not\leq f(x\oplus z)$$
	
	For the second case, uppose $y,z\in G_2$ witness a \ref{DV2} violation, namely $\Norm{f^x(y) -  f^x(z)}_\infty > \Norm{y-z}_\infty$. Then, we get:
	\[
	\begin{array}{lcl}
		\Norm{f(x\oplus y) - f(x \oplus z)}_\infty & \ge &  \Norm{f(x\oplus y)\restriction d_2 - f(x \oplus z)\restriction d_2}_\infty \\
		& = & \Norm{f^x(y) -  f^x(z)}_\infty \\
		& > & \Norm{y-z}_\infty = \Norm{(x\oplus y) - (x\oplus z)}_\infty
	\end{array} 
	\] 
\end{proof}

Next, we will show how to convert violations returned by algorithm $\mathcal A$ for instance $f^\CA$ to violations for the original instance $f$. This proof is more involved, since even when $f$ is a violation-free instance, in other words $f$ has neither \ref{DV1} nor \ref{DV2} violations, $f^\CA$ need not be violation-free. The main difficulty is that for a fixed $x\in G_1$, instance $f^x$ may have multiple fixed points, and hence $f^\CA(x)$ highly depends on which of these fixed points algorithm $\mathcal B$ manages to find. 

If $f$ is violation-free, then $f^x$ is violation-free for all $x\in G_1$ by Lemma \ref{lem:gx-violations}. In that case, $\mathcal B$ has to return a fixed point for any given instance $f^x$. 
Next, we show that if algorithm $\mathcal B$ always returns the LFP of instance $f^x$ for any given $x\in G_1$, then the corresponding $f^\CA$ will be violation-free as well. In other words, violations of $f^\CA$ can be mapped to violations of $f$. For this, we will have to prove certain properties of the LFPs across different $x$-slices ($f^x$ instances).  

\begin{lemma}\label{lem:subslice-fp-prop}
	Let $x, \tilde{x} \in G_1$, such that $x\le \tx$ and $\Norm{x-\tilde{x}}_\infty = 1$. Suppose, sub-instances $f^x$ and $f^\tx$ are violation-free. Then, 
	\begin{enumerate}
		\item[(a)] If $y$ is a fixed point of $f^x$, then there exists a fixed point $\ty$ of $f^\tx$, such that $y \le \ty$ and $\Norm{y-\ty} \le 1$.
		\item[(b)] If $\ty$ is a fixed point of $f^\tx$, then there exists a fixed point $y$ of $f^x$, such that $y \le \ty$ and $\Norm{y-\ty} \le 1$. 
	\end{enumerate}  
\end{lemma}
\begin{proof}
	We will prove statement (a), and the proof of statement (b) follows similarly. By monotonicity, we have $f(x \oplus y) \le f(\tilde{x} \oplus y)\Rightarrow y = f^x(y) \le f^\tx(y)$. By contraction, we have $\Norm{f(x\oplus y) - f(\tx \oplus y+\ones)}_\infty \le \Norm{(x\oplus y) - (\tx \oplus y+\ones)}_\infty =1\Rightarrow \Norm{f^x(y) - f^\tx(y+\ones)}_\infty = \Norm{y - f^\tx(y+\ones)}_\infty \le 1$. Therefore, $y\le f^\tx(y) \le f^\tx(y+\ones) \le (y +\ones)$. This together with monotonicity implies $f^\tx(f^\tx(y)) \le f^\tx(y+\ones)\le y+\ones$. 
	
	Applying this last argument repeatedly we get that $y \le {f^\tx}^k(y) \le f^\tx(y+\ones) \le y+\ones$ for all $k\ge 1$. Thus for some $1\le k \le d_2$ it must be that $\ty={f^\tx}^k(y)$ is a fixed point of $f^\tx$ satisfying both the desired properties, namely $y\le \ty$ and $\Norm{y -\ty} \le 1$. 
\end{proof}

\begin{lemma}\label{lem:subslice-lfp-prop1}
	Let $x, \tilde{x} \in G_1$,
	and $y^*$ and $\tilde{y}^*$ be LFP of instances $f^x$ and $f^{\tx}$ respectively. If $x\le \tilde{x}$ then $y^* \le \tilde{y}^*$ and $\Norm{y^* - \tilde{y}^*}_\infty \le \Norm{x-\tilde{x}}_\infty$.
\end{lemma}
\begin{proof}
	First let us prove the claim when $\Norm{x-\tilde{x}}_\infty = 1$ and then apply induction. For this case, suppose to the contrary $y^* \not\le \tilde{y}^*$ and $\ty^*_i < y^*_i$ for some $1\le i \le d_2$ . By Lemma \ref{lem:subslice-fp-prop} there exists a fixed point $y$ of $f^x$ such that $y \le \ty^*$, in which case and there exists $y_i \le \ty^*_i < y^*_i$, a contradiction to $y^*$ being the LFP of $f^x$. 
	
	Furthermore, by Lemma \ref{lem:subslice-fp-prop} it also follows that there exists fixed point $\ty$ of $f^\tx$ such that $\Norm{y^*-\ty} \le 1$. We also have that $y^* \le \ty^*\le \ty$ because $\ty^*$ is the LFP of $f^\tx$. These two together imply $\Norm{y^*-\ty^*} \le 1$.
	
	For the general case, let $s=\Norm{x-\tilde{x}}_\infty$ and $x=x^0 \le x^1\le \dots \le x^s = \tx$ be a sequence of points in $G_2$ such that consecutive points are at most distance $1$ apart, i.e., $\Norm{x^k - x^{k+1}}_\infty=1$ for all $0\le k < s$. Let ${y^k}^*$ be the LFP of $f^{x^k}$ instance. Applying the above argument repeatedly to consecutive pairs we get that for all $0\le k <s$, ${y^k}^* \le {y^{(k+1)}}^*$ and $\Norm{{y^k}^* - {y^{(k+1)}}^*}_\infty \le 1$. Tying these together, we get, $y^* \le \ty^*$ and $\Norm{y^* - \ty^*}_\infty = \Norm{{y^0}^* - {y^k}^*}_\infty \le s = \Norm{x-\tilde{x}}_\infty$. 
\end{proof}

\begin{lemma}\label{lem:subslice-lfp-prop2}
	Let $x, \tilde{x} \in G_1$ and $y^*$ and $\tilde{y}^*$ be the LFPs of instances $f^x$ and $f^{\tx}$ respectively. Then, $\Norm{y^* - \tilde{y^*}}_\infty \le \Norm{x-\tilde{x}}_\infty$.
\end{lemma}
\begin{proof}
	Consider $\lbx=(\min{x_i,\tx_i})_{i\in [d_2]}$ and $\ubx=(\max{x_i,\tx_i})_{i\in [d_2]}$. We have that $\lbx \le x, \tx \le \ubx$ and $\Norm{\lbx -\ubx}_\infty = \Norm{x - \tx}_\infty$.  
	
	Let $\lby^*$ and $\uby^*$ be the LFP of $f^\lbx$ and $f^\ubx$ respectively. Then, Lemma \ref{lem:subslice-lfp-prop1} implies $\lby^* \le y^*, \ty^* \le \uby^*$ and $\Norm{\lby^* - \uby^*}_\infty \le \Norm{\lbx -\ubx}_\infty$. Then it follows that, $\Norm{y^* - \ty^*}_\infty \le \Norm{\lby^* - \uby^*}_\infty \le \Norm{\lbx -\ubx}_\infty = \Norm{x - \tx}_\infty$.
\end{proof}

The above two lemmas provide all the necessary tools to map violations of $f^\CA$ to violations of $f$. 

\begin{lemma}\label{lem:gA-violations}
	For $x,\tx\in G_1$, let $y^*$ and $\ty^*$ be the LFPs of instances $f^x$ and $f^\tx$ respectively. If $x,\tx$ witness a \ref{DV1} or \ref{DV2} violation for the corresponding instance $f^\CA$, then $(x\oplus y^*), (\tx \oplus \ty^*)$ witness a \ref{DV1} or \ref{DV2} violation for the original instance $f$ respectively. 
\end{lemma} 
\begin{proof}
	For the \ref{DV1} violation, suppose $x\le \tx$ but $f^\CA(x) \not\le f^\CA(\tx)$. By Lemma \ref{lem:subslice-lfp-prop1}, we have that $y^* \le \ty^*$, and hence $(x \oplus y^*) \le (\tx \oplus \ty^*)$. While, since $f^\CA(x)=f(x\oplus y^*)\restriction d_1$, it follows that $g(x\oplus y^*) =(f^\CA(x) \oplus y^*) \not\le (f^\CA(\tx) \oplus \ty^*) = f(\tx \oplus \ty^*)$. Thus, $(x\oplus y^*)$ and $(\tx \oplus \ty^*)$ lead to a \ref{DV1} violation. 
	
	For the \ref{DV2} violation suppose, $\Norm{f^\CA(x) - f^\CA(\tx)}_\infty > \Norm{x -\tx}_\infty$. Then, by Lemma \ref{lem:subslice-lfp-prop2} it also follows that $\Norm{f^\CA(x) - f^\CA(\tx)}_\infty > \Norm{y^* -\ty^*}_\infty$. 
	\[
	\begin{array}{lclr}
		\Norm{f(x\oplus y^*) - f(\tx \oplus \ty^*)}_\infty & = & \Norm{(f^\CA(x) \oplus y^*) - (f^\CA(\tx) \oplus \ty^*)}_\infty & \mbox{(by the definition of $f^\CA$)} \\
		& = & \max{\Norm{f^\CA(x) - f^\CA(\tx)}_\infty, \Norm{y^* - \ty^*}}_\infty \\
		& = & \Norm{f^\CA(x) - f^\CA(\tx)}_\infty & \mbox{(by the hypothesis and Lemma \ref{lem:subslice-lfp-prop2})} \\
		& > & \Norm{x -\tx}_\infty & \mbox{(by the hypothesis)} \\
		& = & \Norm{(x\oplus y^*) - (\tx \oplus \ty^*)}_\infty & \mbox{(using Lemma \ref{lem:subslice-lfp-prop2})}
	\end{array}
	\]
	This implies $(x\oplus y^*)$ and $(\tx \oplus \ty^*)$ lead to a \ref{DV2} violation.
\end{proof}

Next we show Theorem \ref{thm:decomposition} by building on Lemmas \ref{lem:gx-violations} and \ref{lem:gA-violations}. 

\begin{proof}[Proof of Theorem \ref{thm:decomposition}]
	Lemma \ref{lem:lfp} shows that if $\mathcal{A}$ and $\mathcal{B}$ return \ref{D1} solutions $x^*\in G_1$ and $y^*\in G_2$ respectively that are LFPs of instances $f^\CA$ and $f^{x^*}$ respectively, then our algorithm will return $(x^*\oplus y^*)$, an LFP of instance $f$. If at any point either of the algorithms reports a violation solution \ref{DV1} or \ref{DV2}, then by Lemmas \ref{lem:gx-violations} and \ref{lem:gA-violations}, we immediately find a violation for the instance $f$. Finally, note that $\mathcal{A}$ makes at most $q_1$ queries, and for each of those, $\mathcal{B}$ will make at most $q_2$ queries, giving $q_1 \cdot q_2$ queries in
	total.
\end{proof}

\subsection{All 3D DMAC Algorithms Can LFP-Solve}
\label{sec:uniqueness}

In this section we show the following lemma. 

\begin{lemma}
\label{lem:uniquefp}
Let $G = \{1, 2, \dots, n\}^3$, and let $f : G \rightarrow G$ be a
violation-free
\DMAC instance. 
There is a polynomial-time reduction that produces a violation-free \1DUniqueDMAC instance $f'
: G' \rightarrow G'$ such that $f'$ has a unique fixed point $x^*$, and $x^*$
can be mapped back in polynomial time to the least fixed point of $f$.
\end{lemma}

We assume that we have already applied Lemmas~\ref{lem:1dunique}
and~\ref{lem:boundaryin} to $f$, meaning that $f$ is a \1DUniqueDMAC instance
with $f_i(x) > x_i$ whenever $x_i = 1$, and $f_i(x) < x_i$ whenever $x_i = n$.
Note also that those lemmas ensure that the least fixed point of $f$ is the
same as the least fixed point of the original instance. 

Our reduction then proceeds in two steps. 

\paragraph{\bf Step 1: Increasing the resolution.}

\newcommand{\fint}{f^\text{Int}}
\newcommand{\gint}{G^\text{Int}}

Let 
$$\gint = \left\{1, \; 1 + \frac{1}{4}, \; 1 + \frac{1}{2},\;  1 +
\frac{3}{4}, \; 2, \; 2 + \frac{1}{4}, \dots, n \right\}^3$$
be the refinement of $G$ in which we have added points at $i + 1/4$, $i + 1/2$
and $i + 3/4$ for each integer $i < n$ that was used to define $G$. We define a function
$\fint$ that interpolates $f$ over this new denser grid. 

\cubeFixedPointCases

For each point $x \in \gint$ we use the following procedure to compute
$\fint_i(x)$. 
\begin{enumerate}
\item Compute the point $p = (\lfloor x_1 \rfloor, \lfloor x_2 \rfloor, \lfloor
x_3 \rfloor)$, which is the least element of the integer cube that contains~$x$.
\item Consider the set $C = \{ q \in G \; : \; q = p + \sum_{j \in D} e_j
\text{ where $D \subseteq \{1, 2, 3\}$} \}$, which are the corners of the
integer cube that contains $x$.
\item If $f_i(y) = y_i + 1$ for all $y \in C$ with $y_i = p_i$, then set $\fint_i(x) = x_i + 1$.
\item If $f_i(y) = y_i - 1$ for all $y \in C$, then set $\fint_i(x) = x_i - 1$.
\item Otherwise, note that the surface in dimension $i$ passes through the cube
defined by $C$. Theorem~\ref{thm:surfaces} implies that this surface is
monotone and has gradient one. Hence the orientation of $C$ must be one of the
five cases depicted in Figure~\ref{fig:1d_unique_fp_location}, where in the
figure the dimension $i$ is represented by the up/down dimension.

We interpolate each of these cubes using the scheme depicted in
Figure~\ref{fig:interpolation}.
Specifically, for each case, we present 5 slices that move left to right across
the cube, and we depict the location of the surface $s_i$ within the cube. 
This specifies $s_i$ at all intermediate points within $C$.

The algorithm itself simply determines the configuration of the cube $C$, and
then finds the location of the surface $s_i(x_{-i})$ using our specified
interpolation. It then outputs $f_i(x) = x_i + 1$ if $s_i(x_{-i}) > x_i$, it
output $f_i(x) = x_i - 1$ if $s_i(x_{-i}) < x_i$, and it outputs $f_i(x) = x_i$
if $s_i(x_{-i}) = x_i$.
\end{enumerate}
\interpolationSlices

It is clear that by construction that each one-dimensional slice in dimension
$i$ has a unique fixed point, and so $\fint$ is a \1DUniqueDMAC instance. It is
also clear that $\| \fint(x) - x \|_\infty \le 1$ by construction.
The surface $s_i$ is monotone and has gradient
at most one within each cube $C$. It is also not difficult to verify that the
surfaces match on the boundary of each cube, and so each surface $s_i$ in
$\fint$ is monotone and has gradient one. It is also clear that we have
maintained 1D rationality. Thus $\fint$ is a violation-free \1DUniqueDMAC
instance by Theorem~\ref{thm:surfaces}. Moreover, since $\fint$ can be computed
by querying $f$ at $8$ points, we have that $\fint$ can be built in time that
is polynomial in the representation of $f$.

\paragraph{\bf Step 2: Ensuring that there is a unique fixed point.}

\newcommand{\fu}{f^\text{u}}

Since $\fint$ is an interpolation of $f$, we have not yet reduced the number of
fixed points, and indeed we have likely added new fixed points in the
interpolated regions. We now build a new function $\fu : \gint
\rightarrow \gint$ that will have a unique fixed point.

The idea is to shift the surfaces in $\fint$ down slightly.
Specifically, we leave $s_1$ as it is, we shift $s_2$ down by 1/4, and we shift
$s_3$ down by 1/2. This is achieved formally as follows.
Let $h_i(x) = \fint_i(x) - x_i$ be the displacement of $\fint$ in dimension $i$ at the point
$x$. For each $x \in G'$ we define
\begin{align*}
\fu_1(x) &= x_1 + h_1(x), \\
\fu_2(x) &= x_2 + h_2(x + 1/4 \cdot e_2), \\
\fu_3(x) &= x_3 + h_3(x + 1/2 \cdot e_3).
\end{align*}
If in any of the cases above we query a point that is not in $G'$, we instead
fix $\fu_i(x) = x_i - 1$ to orient that point downward. 

Note that since we have assumed that Lemma~\ref{lem:boundaryin} has been
applied to $f$, we have that $f$ always moves strictly away from the
boundary, so no point on any surface in $f$ can lie on the boundary.
Since in each case we have shifted by less than an integer, 
we still therefore have $0 < s_i(x_{-i}) < n$ after our shifts. Since all
we have done is move the surfaces in $\fint$ by a constant amount, we have not
changed the monotonicity of any surface, or the fact that it has gradient at
most one. We have also not changed the fact that $\fint$ is 1D rational. Thus
by Theorem~\ref{thm:surfaces} we have that $\fu$ is a violation-free \1DUniqueDMAC instance. Moreover,
$\fu$ can be computed by making three queries to $\fint$, and so $\fu$ can be
built in time that is polynomial in the size of $\fint$.

\paragraph{\bf Correctness.}

We must show that $\fu$ has a unique fixed point $x^*$, and that given this
fixed point, we can find the fixed point of $f$ in polynomial time.
We start by showing that $\fu$ has a unique fixed point. 

The following lemma states that if any \1DUniqueDMAC instance has at least two
fixed points $x$ and $y$, then there are at least two dimensions $i$ for which
$| x_i - y_i \| = \| x - y \|_\infty$.

\begin{lemma}
\label{lem:fpdiagonal}
Let $f : G \rightarrow G$ be a violation-free \1DUniqueDMAC instance.
If $x$ and $y$ are two distinct fixed point of $f$, then there are at least two
distinct dimensions $i$ and $j$ such that $| x_i - y_i | = | x_j - y_j | = \| x
- y \|_\infty$.
\end{lemma}
\begin{proof}
Suppose for the sake of contradiction that this is not the case, meaning that 
there is a unique dimension $i$ such that $| x_i - y_i | = \| x - y \|_\infty$,
and observe that this implies that
$| x_i - y_i | > \| x_{-i} - y_{-i} \|_\infty$. 
Since $x$ and $y$ are both fixed points, we have that $x$ and $y$ lie on the
surface $s_i$,
so $| s_i(x_{-i}) - s_i(y_{-i}) | = | x_i - y_i |$.
Therefore $s_i$ does not have
gradient at most one, which contradicts Theorem~\ref{thm:surfaces}.
\end{proof}

The next lemma shows that our interpolation has a special property: all
diagonals of the surfaces used in the instance lie on principle diagonals of
the instance, or principle diagonals of the two-dimensional slices of the
instance. This allows us to prove special properties about the points that lie
on these surfaces.
For each $x \in \gint$ let $r(x) = x - \lfloor x \rfloor$ be the function that
gives the non-integer part of $x$.

\begin{lemma}
\label{lem:surfaceround}
Let $s_k$ denote surface $k$ in $\fint$.
If $x$ and $y$ are two distinct points that lie on $s_i$, and 
if there are two distinct dimensions $i$ and $j$ such that
$| x_i - y_i | = | x_j - y_j | = \| x
- y \|_\infty$, 
then we have $|r(x_i) - r(x_j)| = |r(y_i) - r(y_j) | = 0$.
\end{lemma}
\begin{proof}
Observe that since 
$| r(x_i) - r(x_j) | = | r(y_i) - r(y_j) |$, we must be moving diagonally along the surface
as we move from $x$ to $y$. It can be verified by inspecting the construction
in Figure~\ref{fig:interpolation} that the surface that we introduce moves
diagonally only along principle diagonals of the instance, or principle
diagonals of 2d slices of the instance. More precisely, by construction, we have 
$|r(p_i) - r(p_j)|$ for all points $p$ that lie on a path that move at gradient
one along dimensions $i$ and $j$ simultaneously. Thus we have 
$|r(x_i) - r(x_j)| = |r(y_i) - r(y_j) | = 0$.
\end{proof}

Since $\fu$ was constructed from $\fint$ by shifting $s_2$ down by $1/4$ and
shifting $s_3$ down by $1/2$, we get the following immediate corollary of
Lemma~\ref{lem:surfaceround}, which intuitively states that each of the surfaces use
distinct diagonals in the grid.

\begin{corollary}
\label{cor:rounddifferent}
Let $s_i$ denote surface $i$ in $\fu$,
let $x$ and $y$ be two distinct points that lie on $s_i$, 
suppose that 
there are two distinct dimensions $i$ and $j$ such that
$| x_i - y_i | = | x_j - y_j | = \| x - y \|_\infty$.
\begin{itemize}
\item If $i = 1$ then $|r(x_i) - r(x_j)| = 0$.
\item If $i = 2$ then $|r(x_i) - r(x_j)| = 1/4$.
\item If $i = 3$ then $|r(x_i) - r(x_j)| = 1/2$.
\end{itemize}
\end{corollary}
\begin{proof}
The proof is the same as the proof of Lemma~\ref{lem:surfaceround}, noting that
shifting the surfaces downwards alters the diagonals used by our construction.
\end{proof}

We are now ready to prove that $\fu$ has a unique fixed point. At a high level,
this follows from the lemmas that we have shown so far because every fixed
point lies on all three surfaces simultaneously by definition. So if we have
two distinct fixed points, we can apply the corollary above to conclude that
the point simultaneously lies on two distinct diagonals, which is a
contradiction.

\begin{lemma}
$\fu$ has a unique fixed point.
\end{lemma}
\begin{proof}
Since $\fu$ is a violation-free \1DUniqueDMAC instance it must have at
least one fixed point, because violation-free \DMAC is a total problem due to
Tarski's theorem. 

Now suppose for the sake of contradiction that there are two distinct points
$x, y \in \gint$ that are both fixed points of $\fu$. 
Let $i$ and $j$ be the dimensions 
whose existence is asserted by Lemma~\ref{lem:fpdiagonal}. 
Let $s_i$ and $s_j$ be the surfaces for dimensions $i$ and $j$ in $\fu$,
respectively.
Since $x$ and $y$ are both fixed, they must lie on both $s_i$ and $s_j$, but
then applying
Corollary~\ref{cor:rounddifferent} to $s_i$ and $s_j$ gives us two distinct values for
$|r(x_i) - r(x_j)|$, which is a contradiction.
\end{proof}

Finally, we show that the least fixed point of $\fu$ can be mapped back to the
least fixed point of $f$ in polynomial time. We will do this using the
following lemma from~\cite{FPS22}.

\begin{lemma}[Lemma 3.2 of \cite{FPS22}]
\label{lem:updownfp}
Let $f : G \rightarrow G$ be a monotone function. If $x, y \in G$ satisfy $x
\le y$ and $x \le f(x)$ and $y \ge f(y)$, then there exists a fixed point $x^*$ of
$f$ satisfying $x \le x^* \le y$. 
\end{lemma}

We utilize the previous lemma to show the following lemma.

\begin{lemma}
\label{lem:findlfp}
Let $x^*$ be the least fixed point of $f$. 
There is a fixed point $x$ of $\fu$ such that $x \ge x^* - (1.5, 1.5, 1.5)$.
\end{lemma}
\begin{proof}
First note that, since $\fu$ was created by interpolating $f$ and then
shifting all surfaces weakly downward, we have $\fu(x^*) \le x^*$. 

Then consider the point $y = x^* - (1, 1, 1)$. Observe that we cannot have
$f_i(y) < y_i$ in any dimension $i$, since this would be a violation of strict
contraction between $y$ and $x^*$ and $\fu$ is violation-free. So we have $y \le
f(y)$. Furthermore, since $y \le x^*$ and $x^*$ is the least fixed point of $f$, there
must exist a dimension $i$ such that $y_i < f_i(y)$. 

Let $z = y - \sum_{j \ne i} e_j/2$. We argue that $z \le \fu(z)$. This can be
proved by inspecting the surfaces of $\fu$. 
\begin{itemize}
\item For each dimension $j \ne i$, we
have that $y$ lies weakly below $s_j$ in $f$, and so the distance between $z$ and
$s_j$ is at least $1/2$ in $f$. Since we moved surface $j$ downward by at most
$1/2$ when we constructed $\fu$, this implies that $z$ lies weakly below $s_j$
in $\fu$. 

\item For dimension $i$, we have that the distance between $y$ and $s_i$ is at
least $1$ in $f$, because $f$ moves $y$ strictly upward in dimension $i$. Since
the surface has gradient at most one, this implies that the distance between
$z$ and $s_i$ is at least $1/2$ in $\fint$. We then shifted $s_i$ downward by
at most $1/2$ when we created $\fu$, so we have that $z$ lies weakly below
$s_i$ in $\fu$. 
\end{itemize}
Since $z$ lies weakly below $s_j$ for all dimensions $j$, we have that 
$z \le \fu(z)$.

So we have $z \le x^*$, and $z \le \fu(z)$, and $x^* \ge \fu(x^*)$, which allows us
to invoke Lemma~\ref{lem:updownfp} to argue that there is a fixed point $x$ of $\fu$
satisfying $z \le x \le x^*$. Since 
$z \ge x^* - (1.5, 1.5, 1.5)$, we therefore have $x \ge x^* - (1.5, 1.5, 1.5)$, as required.
\end{proof}

So if we find the unique fixed point $x$ of $\fu$, then we can search all
points $y$ in $G$ that satisfy $x \le y \le x + (1.5, 1.5, 1.5)$, and
Lemma~\ref{lem:findlfp} guarantees that we will find the least fixed point of $f$.
Since there are constantly many points to search, this mapping can be carried
out in time that is polynomial in the representation of $f$.
This completes the proof of Lemma~\ref{lem:uniquefp}. 

We have shown the lemma for three dimensions here because our algorithm works
for three-dimensional instances. 
We remark that this proof
can easily be generalized to any constant dimension: we simply need to increase
the resolution of the grid so that each surface can be moved by a distinct
shift, and all shifts need to be less than or equal to 1/2. We cannot
generalize beyond constantly many dimensions due to the need to search a
sub-grid when we map the unique fixed point of $\fu$ to the least fixed point
of~$f$. This sub-grid has constantly many points when the dimension is
constant, but in general it contains exponentially many points.

\section{An Algorithm for 3-dimensional DMAC}
\label{sec:dmac_algo}





We now present an algorithm that solves three-dimensional $\DMAC$ instances.
Throughout this section we assume that we have a violation-free \DMAC instance defined
by a function $f$ over the grid $G = \{1, 2, \dots, n\}^3$. 


A fundamental definition for the algorithm is the \emph{up} and \emph{down}
sets of $f$. The up-set contains every point $x \in G$ such that $x \le f(x)$,
while the down-set contains every point $x \in G$ such that $x \ge f(x)$.

\begin{definition}[\textsc{Up- and Down- sets}]
Let $f : G \rightarrow G$ be a $\DMAC$ instance. 
We define 
\begin{align*}
\Up(f) &= \{ x \in G: f(x) \geq x \}, \\
\Down(f) &= \{ x \in G: f(x) \leq x \}.
\end{align*}
\end{definition}

This section proceeds as follows. We first define a preprocessing step that is
necessary for our algorithm to work. We then define the notion of a
\emph{critical box}, which is the crucial definition that our algorithm is
based on, and we then show that an almost-square critical box always exists.
Then we define the algorithm itself. Finally, we give a terminal phase
algorithm that will be applied once we have found a two-dimensional slice that
contains a fixed point.

\subsection{Preprocessing}



Our algorithm requires a number of assumptions on the $\DMAC$ instance.
Firstly we assume that we have
applied Lemma~\ref{lem:1dunique} to produce a \1DUniqueDMAC instance.
Then we assume that we have applied Lemma~\ref{lem:boundaryin} to ensure that
all points on the boundary of the instance move strictly inward.



Finally, we need one more assumption that we have not yet considered.
We require that the up set does not touch the lower boundary of the
instance unless it is touching the boundary of dimension 3, and the symmetric
property for the down set. This is encoded by the following statements.
\begin{itemize}
\item For each $x \in G$, if $x_3 = 1$ and either $x_1 > 1$ or $x_2 > 1$, then $x \not \le f(x)$. 
\item For each $x \in G$, if $x_3 = n$ and either $x_1 < n$ or $x_2 < n$, then $x \not \ge f(x)$. 
\end{itemize}

The first statement is analogous to the situation shown on the right in
Figure~\ref{fig:bsidea}, though in that picture we show the 2D-version in which
we have $x \not \le f(x)$ whenever $x_1 > 1$ and $x_2 = 1$. As mentioned in the
technical overview, the transformation is necessary to ensure that our initial
bounds on the up- and down-sets can be correctly translated to other slices. 
This assumption gives the three-dimensional analogue, and it is likewise used to
ensure that the translation of bounds from a slice with $x_3 = i$ to a slice
with $x_3 = j \ne i$ is correct.




In this section we 
show that this assumption can be made without loss of generality, by proving
the following lemma.

\begin{lemma}
\label{lem:updownboundary}
Let $G = \{1, 2, \dots, n\}^3$, let $G' = \{1, 2, \dots, n\}^2 \times \{1, 2,
\dots, n'\}$ and
let $f : G \rightarrow G$ be a $\1DUniqueDMAC$ instance. There is a polynomial-time reduction that creates a 
$\1DUniqueDMAC$ instance $f' : G' \rightarrow G'$ such 
that the following statements hold.
\begin{itemize}
\item For each $x \in G'$, if $x_3 = 1$ and $x_1 > 1$ or $x_2 > 1$, then $x
\not \le f'(x)$. 
\item For each $x \in G'$, if $x_3 = n'$ and $x_1 < n$ or $x_2 < n$, then $x
\not \ge f'(x)$. 
\end{itemize}
Moreover, the least fixed point of $f'$ can be mapped back to the least fixed
point of $f$.
\end{lemma}

We start by giving a reduction that produces a function $f' : G' \rightarrow
G'$ where 
$G' = \{1, 2, \dots, n\}^2 \times \{1, 2,
\dots, 2n\}$ that ensures that if $x_3 = 1$ and $x_1 > 1$ then $x \not \le
f'(x)$. We define $f'$ as follows.
\begin{itemize}
\item If $x_3 > n$ then we define $f'(x) = f(x - n \cdot e_3)$.
\item If $x_3 \le n$ then we define $f'_2(x) = f_2(x - (x_3 + 1) \cdot e_3)$
and 
$f'_3(x) = x_3 + 1$. For dimension 1 we define 
$m(x) = x - (x_3 + 1) \cdot e_3$, 
$z = x + (n + 1 - x_3) \cdot (1, 0, 1)$ and 
$h_1(x) = f_1(x) - x_1$. 
We set 
\begin{equation*}
f'_1(x) = \begin{cases}
\max(x_1 + h_1(m(z)), 0) & \text{if $z \in G'$} \\
x_1 - 1 & \text{otherwise.}
\end{cases}
\end{equation*}
\end{itemize}
The idea here is that we are carrying out the analogous reduction to the one
shown for two-dimensions in Figure~\ref{fig:bsidea}.
The top half of $G'$, meaning the points $x$ with $x_3 >
n$ exactly copies $f$ on~$G$. For the bottom-half, the definition ensures that
all points move upward in dimension 3. For dimensions 1 and 2, we try to copy
the direction of a point that lies on the bottom boundary of $G$ in
dimension~$3$ (which has now been transplanted to the slice where $x_3 = n +
1$). For dimension 2 we copy the point that is now directly above us, while for
dimension 1, we project along the vector $(1, 0, 1)$, and this projection may
leave $G'$ before it hits a point with $x_3 = n + 1$. If the vector hits a
point with $x_3 = n + 1$, then we copy the displacement of that
point, and otherwise we move downwards. 

To see that we have $x \not \le f'(x)$ whenever
$x_3 = 1$ and $x_1 > 1$ then, it suffices to note that
for any such point we will have $z \not \in G'$, and so $f'_1(x) < x_1$ by
definition. The following lemma shows that the reduction is correct.

\begin{lemma}
A point $x$ is a fixed point of $f'$ if and only if $x - n \cdot e_3$ is a fixed point of $f$.
Moreover, if $f$ is violation-free, then $f'$ is also violation-free.
\end{lemma}
\begin{proof}
For a point $x \in G'$ with $x_3 \ge n + 1$ it is clear from the definition
that $x$ is a fixed point of $f'$ if and only if $x - n \cdot e_3$ is a fixed
point of $f$. For a point $x \in G'$ with $x_3 < n + 1$, we have $f_3(x) > x_3$
by definition, so $x$ is not a fixed point of $f'$.

If $f$ is violation-free, then we can apply Theorem~\ref{thm:surfaces} to
argue that it is 1D rational, and that all surfaces are monotone with gradient
at most one. Since we copy $f$ to the points $x$ in $G'$ with $x_3 \ge n + 1$,
we have that all surfaces continue to be monotone and have gradient one in this region.

We now consider the points $x$ with $x_3 < n + 1$. 
For the surface $s_3$, note that we have defined $f'_3(x) > x_3$, and so this
ensures that each one-dimensional slice in dimension 3 has a unique fixed point
and all points are 1D rational in dimension 3. 

For dimension 2 observe that we are simply copying the dimension-2 surface from
the slice in which $x_3 = n + 1$. Thus the surface continues to be monotone and
have gradient at most one, and the instance continues to be 1D rational in
dimension 2.

For dimension 1 we are copying the surface from 
the slice in which $x_3 = n + 1$, but we are moving the entire surface 
downward with a gradient of one as we move downward in dimension 3. This
ensures that the surface is still monotone and still has gradient at most one.
The $\max$ operation with $0$ ensures that the surface does not leave the lower
boundary of the instance, which also ensures that we continue to have a
\1DUniqueDMAC instance. Finally, the downward movement that is introduced when
$z \not \in G$ ensures that the instance remains 1D rational for dimension 1.

Thus, the new instance is 1D rational, and has surfaces that are monotone and
have gradient at most~1. So we can apply Theorem~\ref{thm:surfaces} to argue
that $f'$ is violation-free.
\end{proof}

So we have shown that the reduction is correct for the case when $x_3 = 1$ and
$x_1 > n$, and we have also shown that the least fixed point of $f'$ can be
mapped back to the least fixed point of $f$. Moreover, $f'$ has displacements
that are at most unit lengths because it either copies displacements from $f'$,
which has that property by assumption, or it introduces a downward displacement
of unit length.  Finally, it is clear that the reduction can be carried out in
time that is polynomial in the representation of $f$.

We can cover the other cases in the following way. 
\begin{itemize}
\item We run the same reduction again after exchanging dimensions 1 and 2.
\item We then flip all dimensions in the instance, and run the reduction for
both dimensions 1 and 2, which ensures the symmetric property holds for the
down set. 
\end{itemize}
It is clear that running the reduction four times in this way can be carried
out in time that is polynomial in the representation of $f$. This completes the proof of
Lemma~\ref{lem:updownboundary}.

\subsection{Critical Boxes}
\label{sec:cb}


We now define the concept of a critical box. 

\paragraph{\bf Boxes and diagonal boxes.}

We start by defining a $\CBox$,
which is simply a box of points in the instance. 

\begin{definition}[$\CBox$]
Given a point $x \in G$, and two non-negative integers $h, w$, we denote the \textit{box} $\CBox(x,h,w)$ that contains all points $y \in G$ that satisfy the inequalities
\begin{align*}
x_1 \leq y_1 \leq x_1 + w \\
x_2 \leq y_2 \leq y_2 + h 
\end{align*}
\end{definition}

We also define a $\DBox$, which is diagonal region with flattened ends,
as shown in Figure~\ref{fig:dbox}. 

\begin{definition}[\textsc{DBox}]
Given two points $x, y \in G$ with $x_1 \leq y_1$ and $y_2 \leq x_2$, and a non-negative integer $l$, we denote the \textit{diagonal-box} $\DBox(x,y,l)$ that contains all points $z \in G$ that satisfy the inequalities
\begin{align*}
y_2 + (z_1 - y_1) &\le z_1 \le x_2 + (z_1 - x_1) \\
x_1 &\le z_1 \le y_1 + l \\
y_2 &\le z_2 \le x_2 + l
\end{align*}
\end{definition}

The states that we consider in our algorithm will be built from these two
shapes.

\dBox

\paragraph{\bf Critical boxes.}

We are now ready to define critical boxes. We will initially describe this
concept for the up set. As we will later see, the same concept can be applied
to the down set by simply flipping all inequalities. 

\criticalBox

Throughout this section we consider a two-dimensional slice $s$ of a
three-dimensional instance. Without loss of generality we suppose that
dimensions $1$ and $2$ are the dimensions used in the slice, and dimension $3$
is the dimension that is not used in the instance. 

Consider the set $S$ of points in the slice that are in $\Up(f)$, meaning that
$f$ points upward in all three dimensions (including the dimension that is not
involved in the slice). Intuitively, a critical box is a (not necessarily
square) box that only contains points in $S$, and cannot be made taller and
wider without including some point not in $S$.

Figure~\ref{fig:critical_box} shows an example of our formal definition of a
critical box. Specifically, the box is defined by a point $x$ in $\Up(f)$ such
that
\begin{itemize}
\item the point $x+w\cdot e_1$ moves weakly upwards in the left-right direction, while the
point $x+(w+1) \cdot e_1$ moves strictly downwards, and
\item the point $x+h\cdot e_2$ moves weakly upwards in the up-down direction, while the point
$x+(h+1) \cdot e_2$ moves strictly downwards,
\item the point diagonally beneath $x$ moves strictly downward in the third dimension,
or $x$ lies on the left or bottom boundaries of the instance (meaning that there are no points
diagonally beneath it).
\end{itemize}
We note that there will always exist a $w$ such that $x+w\cdot e_1$ moves
strictly downwards, due to our assumption that all points on the boundary of
the instance move strictly inward, and likewise there will always exist an $h$
such that $x+h\cdot e_2$ moves strictly downward. We encode this formally in the
following definition.


\begin{definition}[\textsc{Critical Box (Up Set)}]
\label{def:cb}
Given a two-dimensional slice $s$, a point $x$, and two non-negative integers $h$ and $w$,
the tuple $(x, h, w)$ is a critical box if all of the following are true.
\begin{enumerate}
\item 
\label{itm:cb1}
$x \in \Up(f)$. 
\item The points $x+w\cdot e_1$ and $x+(w+1) \cdot e_1$ obey the following.
\begin{align*}
f_1 (x + w \cdot e_1)  - (x_1 + w) &\geq 0 \\
f_1 (x + (w+1) \cdot e_1) - (x_1 + (w+1)) &< 0
\end{align*}
\item The points $x+h\cdot e_2$ and $x+(h+1) \cdot e_2$ obey the following.
\begin{align*}
f_2 (x + h \cdot e_2) - (x_2 + h) &\geq 0 \\
f_2 (x + (h+1) \cdot e_2) - (x_2 + (h+1) ) &< 0
\end{align*}

\item 
\label{itm:cb4}
Either $x_1 = 0$ or $x_2 = 0$, or the point $x - (1, 1)$ obeys the
following.
\begin{align*}
f_3(x - (e_1 + e_2)) - x_3 &< 0 
\end{align*}

\end{enumerate}
\end{definition}




It is clear from the definition that a critical box $(x, h, w)$ cannot be made
taller and wider without including some point that is not in $\Up(f)$. In
Figure~\ref{fig:critical_box} we can see that the critical box cannot be
extended rightwards, since it would then include the point $x + (w+1) \cdot e_1$,
which is not in $\Up(f)$ due to its direction in dimension 1, and it likewise
cannot be extended upwards without including the point $x + (h+1) \cdot e_2$, which
is also not in $\Up(f)$ due to its direction in dimension 2. It also cannot be
extended both downward and leftward without including the point $x - e_1 - e_2$,
which is not in $\Up(f)$ due to its direction in dimension 3. Note that we may
be able to extend the box either downward or leftward individually, but this
will turn out to not be a problem, since the points directly leftward and
downward are contained in the critical boxes lobes, which we will define
shortly.

It is however not immediately obvious that all points in a critical box are
actually contained in $\Up(f)$. We show that in the following lemma.

\begin{lemma}\label{lem:cb_up}
If $(x, h, w)$ is a critical box, then we have $\CBox(x, h, w) \subseteq
\Up(f)$. 
\end{lemma}
\begin{proof}
Since point $x$ represents the bottom left point of the critical box $(x,h,w)$, any point inside the box can be described as $y = x + a \cdot e_1 + b \cdot e_2$, where $a \leq w$ and $b \leq h$. 

From the definition of a critical box, $x \in Up(f)$ and we know that
$f_1 (x + w \cdot e_1) \geq x_1 +  w$. Therefore, for any point $x + a \cdot e_1$, since $a \leq w$, we have that
$f_1 (x + a \cdot e_1) \geq x_1 + a$ holds due to the contraction property, because if
we instead assume that $f_1
(x + a \cdot e_1) < x_1 + a$ then we obtain the following violation of strict contraction: $|| f(x + w \cdot e_1) - f(x
+ a \cdot e_1) ||_{\infty} > w-a = || (x + w \cdot e_1) -
(x + a \cdot e_1) ||_{\infty}$.

By a similar argument, for the second dimension, $f_2(x + h \cdot e_2) \geq x_2 + h$. For any point $x + b \cdot e_2$ within the box, $f_2 (x + b \cdot e_2) \geq x_2 + b$. 

Finally, monotonicity implies that $x \in Up(f) \implies f_3 (x + a \cdot e_1) \geq x_3 + a$ and $f_3 (x+ b \cdot e_2) \geq x_3+ b$.

Since $y = x  + a \cdot e_1 + b \cdot e_2$ is composed of two displacements of $x$ along dimensions 1 and 2, the inequalities for $f_1(x + a \cdot e_1), f_2 (x + b \cdot e_2)$ carry over to $y$.

Thus, $f_i(x + a \cdot e_1 + b \cdot e_2) \geq x + a \cdot e_1 + b \cdot e_2$, for all $i = \{1,2,3\}$, confirming that every point within the critical box belongs to $Up(f)$.

\end{proof}


\paragraph{\bf Lobes.}

\upRegions

The critical box will form the basis of a bound on the up set. Specifically,
we will show that the up set in a two-dimensional slice is contained within a
critical box and its three \emph{lobes}, which are shown in Figure~\ref{fig:up_regions}.

Specifically, for each critical box $(x, h, w)$, we define
\begin{itemize}
\item a \emph{left lobe} $\leftl(x, h, w)$ containing all points to the left of
the box,
\item a \emph{bottom lobe} $\downl(x, h, w)$ containing all points below the
box, and
\item a \emph{diagonal lobe} $\diagl(x, h, w)$ containing all points diagonally
above the box.
\end{itemize}
To define these lobes formally, we first define the \emph{upward} and
\emph{downward cones} that
originate from a particular point. 
\begin{definition}
Let $x \in G$ be a point, and let $i \in \{1, 2\}$ be a dimension. We define
\begin{align*}
\UC_i (x) &= \{y \in G: y_i - x_i \geq |y_j - x_j | \ \forall j \neq i \}, \\
\DC_i (x) &= \{y \in G: x_i - y_i \geq |y_j - x_j | \ \forall j \neq i \}.
\end{align*}
\end{definition}
For example, in Figure~\ref{fig:up_regions} the region labelled C1 is $\UC_1(x
+ (w + 1) \cdot e_1)$ while the region labelled C2 is $\UC_2(x + (h+1) \cdot e_2)$.
We can now formally define the three lobes.
\begin{definition}
Given a critical box $(x, h, w)$, the three \textit{lobes} originating from it
are defined as follows.
\begin{align*}
\leftl(x, h, w) &= \{ y \in G: y_1 \leq x_1,\ x_2 \leq y_2 \leq h \}, \\ 
\downl(x, h, w) &= \{ y \in G: x_1 \leq y_1 \leq w,\ y_2 \leq x_2\}, \\
\diagl(x, h, w) &= \{ y \in G: y \in \UC_2 (x + w \cdot e_1) \setminus (\UC_2(x + (h+1) \cdot e_2)
\cup \{(x,h,w)\}). 
\end{align*}

\end{definition}

We aim to prove the following lemma, which states that a critical box and its
lobes provide a bound on the up set within a two-dimensional slice.

\begin{lemma}
\label{lem:cbbound}
Let $s$ be a two-dimensional slice, let $(x, h, w)$ be a critical box within $s$, and let
$U$ be the subset of $\Up(f)$ that is contained in $s$. The
set $U$ is contained in
\begin{equation*}
\CBox(x, h, w) \cup \leftl(x, h, w) \cup \downl(x, h, w) \cup \diagl(x, h, w).
\end{equation*}
\end{lemma}

We will prove this by showing that any point not in these four sets cannot be
in $\Up(f)$. 
We first show that contraction implies that if $x$ points downwards in dimension $i$, then any point in
the dimension $i$ upward cone of $x$ must also points downwards. The above property holds for both weak and strict inequalities. 
We also show the symmetric property for the downwards cone, which we will use
later.

\begin{lemma}
\label{lem:contrdown}
Let $x \in G$ be a point. 
\begin{itemize}
\item \textit{(strict)} If $f_i(x) < x_i$, then for all $y \in \UC_i(x)$
we have $f_i(y) < y_i$.
\item \textit{(weak)} If $f_i(x) \leq x_i$, then for all $y \in \UC_i(x)$
we have $f_i(y) \leq y_i$.
\item \textit{(strict)}  If $f_i(x) > x_i$, then for all $y \in \DC_i(x)$ we have $f_i(y) > y_i$.
\item \textit{(weak)}  If $f_i(x) \geq x_i$, then for all $y \in \DC_i(x)$ we have $f_i(y) \geq y_i$.
\end{itemize}
\end{lemma}
\begin{proof}
We begin with the proof of the claim regarding upward cones.
Since $y \in \UC_i(x)$ we have that $\| x - y \|_\infty = y_i - x_i$. 

We use two different assumptions on $f_i(y)$, depending on the value of $f_i(x)$ relative to $x_i$. If $f_i(x) < x_i$, we set $f_i(y) \geq y_i$ and if $f_i(x) \leq x_i$, then $f_i(y) > y_i$. In both cases, we would have
\begin{align*}
\| f(y) - f(x) \|_\infty &\ge f_i(y) - f_i(x) \\
&> y_i - x_i \\
&= \| y - x \|_\infty,
\end{align*} 
where the strict inequalities follow from the assumption on $f_i(x), f_i(y)$. Thus we have a violation of strict contraction, which is a contradiction.

We prove the downward cone case in a symmetric manner. Since $y \in \DC_i(x)$ we have that $\| x - y \|_\infty = x_i - y_i$. 

Similarly, the assumptions are as follows: if $f_i(x) > x_i$, then we assume $f_i(y) \leq y_i$, and if $f_i(x) \geq x_i$, we set $f_i(y) < x_i$. Then, in both cases, 
\begin{align*}
\| f(x) - f(y) \|_\infty &\ge f_i(x) - f_i(y) \\
&> x_i - y_i \\
&= \| x - y \|_\infty,
\end{align*}
which is a violation of strict contraction. 

\end{proof}

This immediately implies that $\Up(f)$ cannot lie in the regions labelled C1
and C2 in Figure~\ref{fig:up_regions}. Specifically, since $x + (w+1) \cdot e_1$ is
required to move strictly leftwards in dimension $1$, Lemma~\ref{lem:contrdown}
implies that every point in C1 must also move strictly leftwards, and thus no
point in the set can lie in $\Up(f)$. The same property holds symmetrically for
C2.

Next we show that monotonicity implies that if $x$ moves downwards in dimension
$i$, either strictly or weakly, then any point directly beneath $x$ in some dimension $j \ne i$ must also
follow the same direction in dimension $i$. We also show the symmetric counterpart of this,
since it will be useful later.
\begin{lemma}
\label{lem:monodown}
Let $x \in G$ be a point.
\begin{itemize}
\item (strict) If $f_i(x) < x_i$, and $y$ is a point such that $y \le x$ and $y_i = x_i$, we have $f_i(y) < y_i$.
\item (weak) If $f_i(x) \leq x_i$, and $y$ is a point such that $y \le x$ and $y_i = x_i$, we have $f_i(y) \leq y_i$.
\item (strict) If $f_i(x) > x_i$, and $y$ is a point such that $y \ge x$ and $y_i = x_i$, we have $f_i(y) > y_i$.
\item (weak) If $f_i(x) \geq x_i$, and $y$ is a point such that $y \ge x$ and $y_i = x_i$, we have $f_i(y) \geq y_i$.
\end{itemize}
\end{lemma}
\begin{proof}
We begin with the first case. We have $y \le x$, so if we had $f_i(y)
\ge y_i$ then we would have $f_i(y) \ge y_i = x_i 
> f_i(x)$, which is a violation of monotonicity.

Similarly, for the second case, if we had $f_i(y) > y_i$, we would get $f_i(y) > y_i = x_i \geq f_i(x)$, which violates monotonicity.

The third and fourth cases are proved symmetrically. 
We have $y \ge x$, so in the third case, if we had $f_i(y)
\le y_i$ then we would have $f_i(y) \le y_i = x_i 
< f_i(x)$, which is a violation of monotonicity.

For the fourth case, assuming $f_i(y) < y_i$ we would get $f_i(y) < y_i = x_i \leq f_i (x)$, which violates monotonicity.
\end{proof}

This now rules out the remaining regions. Specifically, M1 cannot contain a
point in $\Up(f)$ because each point in M1 lies directly beneath a point in C1.
We have already proved that each point in C1 must move strictly leftwards, and
therefore by Lemma~\ref{lem:monodown} each point in M1 must also move strictly
leftwards. The same property holds symmetrically for the region M2. Finally,
for the region M3, note that $x - e_1 - e_2$ moves strictly downward in the third
dimension, and that every point in M3 lies strictly beneath $x - e_1 - e_2$.
Therefore we can apply Lemma~\ref{lem:monodown} to argue that all point in M3
move strictly downward in dimension 3. 
This completes the proof of Lemma~\ref{lem:cbbound}.

\subsection{Almost Square Critical Boxes}
\label{sec:almostsquare}

Recall that a critical box $(x, h, w)$ is not required to be square, meaning
that $h$ and $w$ can differ significantly. While any critical box provides a
bound on the up set via Lemma~\ref{lem:cbbound}, our algorithm will crucially
use the fact that an \emph{almost square}\footnote{In fact, if we were
considering continuous monotone contraction maps, rather than $\DMAC$, the
techniques that we use here can be used to show that an exactly square critical box
always exists. Our weakening to consider almost square critical boxes is
required here only because we are working with the discrete problem.} critical box
always exists, which we will show by proving the following lemma.

\begin{lemma}
\label{lem:cb_shape}
If $\Up(f) \cap s$ is non-empty then there exists a critical box $(x, h, w)$
with $|h - w| \le 1$.
\end{lemma}
Note that the precondition of 
$\Up(f) \cap s$ being non-empty is required, because if 
$\Up(f) \cap s$ is empty, then no critical boxes exist. 
The rest of this section is dedicated to proving Lemma~\ref{lem:cb_shape}.

\paragraph{\bf The boundary set $\mathcal{B}$.}

We begin by defining the set $\mathcal{B}$, which captures the boundary between
the points that move upward in the third dimension and the rest of the slice.

\greenBoundary

\begin{definition}
We define $\mathcal{B}$ to contain all points $x$ in the slice $s$ such that
$f_3(x) \ge x_3$ and either
\begin{itemize}
\item $x_1 = 0$ or $x_2 = 0$, or
\item $f_3(x - e_1 - e_2) - x_3 < 0$
\end{itemize}
\end{definition}

Observe that $\mathcal{B}$ contains exactly the set of points for which
dimension 3 moves upward, as required by point~\ref{itm:cb1} of
Definition~\ref{def:cb}, and for which point~\ref{itm:cb4} in
Definition~\ref{def:cb} is satisfied. Thus all critical boxes $(x, h, w)$ must
have $x \in \mathcal{B}$.

Furthermore, Lemma~\ref{lem:monodown} tells us that if $f_3(x) > x_3$ and if $y
\le x$ then $f_3(y) > y_3$. So the set $\mathcal{D} = \{ x \in s \; : \; f_3(x) > x_3\}$
is downward closed. Since $\mathcal{B}$ contains points not in $\mathcal{D}$ that
are diagonally above a point in $\mathcal{D}$, along with the points on the
boundary that are not in $\mathcal{D}$, it must therefore trace a path of
points that moves from the top or left side of the slice to the bottom or right
side, as shown in the three examples given in Figure~\ref{fig:green_boundary}.  

Specifically, we can represent $\mathcal{B}$ as a sequence of points $x^1, x^2,
\dots, x^k$ such that the following hold.
\begin{itemize}
\item $x^1_2 = n$, 
\item $x^k_1 = n$,  and
\item For each $i$ we have $x^{i+1} = x^i + e_1$ or 
$x^{i+1} = x^i - e_2$.
\end{itemize}

\paragraph{\bf Heights.}

For each $x \in \mathcal{B}$ we define the following \emph{height} functions.
\begin{align*}
\hei_1(x) &= \min_{w} \{ f_1(x+ w \cdot e_1) < x_1 + w \},\\
\hei_2(x) &= \min_{h} \{ f_2(x+ h \cdot e_2) < x_2 + h \}.
\end{align*}
Note that the height functions are always well defined, due to our assumption
that $f_i(x) < x_i$ whenever $x_i = n$, so the minimums are not taken over
empty sets. 

Observe that if $w = \hei_1(x) > 0$ and $h = \hei_2(x) > 0$ then $(x, h, w)$ meets
all of the requirements of a critical box. So to prove Lemma~\ref{lem:cb_shape}
it is sufficient to find a point $x$ where both heights are positive, and where
$|\hei_1(x) - \hei_2(x)| \le 1$.

The next lemma analyses how the heights change as we walk along the path
defined by $\mathcal{B}$. In particular, it proves that both height functions
change by at most 1 as we take a step along the path defined by $\mathcal{B}$. 

\begin{lemma}\label{lem:dist_adj}
Let $x$ be a point such that 
$\hei_1(x) > 0$ and 
$\hei_2(x) > 0$.
\begin{enumerate}
\item If $y = x+ e_1$, then both of the following hold.
\begin{align*}
\hei_1(y) &= \hei_1(x) - 1 \\
\hei_2(x) + 1 \geq \hei_2(y) &\geq \hei_2(x)
\end{align*}
\item If $y = x-e_2$, then both of the following hold.
\begin{align*}
\hei_1(x) - 1 \le \hei_1(y) &\leq \hei_1(x) \\
\hei_2(y) &=\hei_2(x) + 1
\end{align*}
\end{enumerate}
\end{lemma}
\begin{proof}
We start with the first case, in which 
$y = x + e_1$.
\begin{itemize}
\item We have
\begin{align*}
\hei_1(y) &= \min_{w} \{ f_1(y+ w \cdot e_1) < y_1 + w \},\\
&= \min_{w} \{ f_1(x + (w+1) \cdot e_1) < x_1 + w + 1 \},\\
&= \hei_1(x) - 1.
\end{align*}

\item Let $x' = x + \hei_2(x) \cdot e_2$, and observe that by definition we
have $f_2(x') < x'_2$. Let $y' = y + (\hei_2(x) + 1) \cdot e_2$, and note that
$y' \in \UC_2(x')$. Hence we can apply Lemma~\ref{lem:contrdown} to argue that
$f_2(y') < y'_2$, which means that $\hei_2(y) \le \hei_2(x) + 1$.

\item Suppose for the sake of contradiction that $\hei_2(y) < \hei_2(x)$. Let
$z = y + \hei_2(y) \cdot e_2$, and observe that by definition $f_2(z) < z_2$.
Consider the point $z' = x + \hei_2(y) \cdot e_2$, and observe that since $x =
y - e_1$ we have $z' = z - e_1$ and thus $z' < z$. 
Hence Lemma~\ref{lem:monodown} implies $f_2(z') < z'_2$, which then implies
$\hei_2(x) \le \hei_2(y)$, giving a contradiction.
\end{itemize}

The second case, where $y = x - e_2$, is proved essentially symmetrically.
\begin{itemize}
\item 
We have
\begin{align*}
\hei_2(y) &= \min_{w} \{ f_2(y+ w \cdot e_2) < y_2 + w \},\\
&= \min_{w} \{ f_2(x + (w-1) \cdot e_2) < x_2 + w - 1 \},\\
&= \hei_2(x) + 1.
\end{align*}

\item 
Let $y' = y + \hei_1(y) \cdot e_1$, and observe that
by definition we have $f_1(y') < y'_1$. Let $x' = x + (\hei_1(y) + 1) \cdot
e_1$, and note that $x' \in \UC_1(y')$. 
Hence we can apply Lemma~\ref{lem:contrdown} to argue that
$f_1(x') < x'_1$, which means that $\hei_1(x) \le \hei_1(y) + 1$, and therefore
$\hei_1(x) - 1 \le \hei_1(y)$.

\item 
Suppose for the sake of contradiction that $\hei_1(y) > \hei_1(x)$. Let
$z = x + \hei_1(x) \cdot e_1$, and observe that by definition $f_1(z) < z_1$.
Consider the point $z' = y + \hei_1(x) \cdot e_1$, and observe that since $x =
y + e_2$ we have $z' = z - e_2$ and thus $z' < z$. 
Hence Lemma~\ref{lem:monodown} implies $f_1(z') < z'_1$, which then implies
$\hei_1(y) \le \hei_1(x)$, giving a contradiction.
\end{itemize}
\end{proof}

\paragraph{\bf The structure of $\Up(f) \cap \mathcal{B}$.}

As mentioned, we are seeking a point in $\mathcal{B}$ that is in $\Up(f)$ whose
heights differ by at most 1. The first thing we prove is that our assumption
that $\Up(f) \cap s$ is non-empty implies that there is at least one point in
$\Up(f) \cap \mathcal{B}$.

\begin{lemma}
\label{lem:bup}
If $\Up(f) \cap s$ is non-empty, then $\Up(f) \cap \mathcal{B}$ is non-empty. 
\end{lemma}
\begin{proof}
Let $x$ be a point in $\Up(f) \cap s$. We will construct a point in 
$\Up(f) \cap \mathcal{B}$.

We first observe that $x \ge y$ for all points $y \in \mathcal{B}$. This is
because, for each non-boundary point $y \in \mathcal{B}$, the point $y - e_1 - e_2$ satisfies $f_3(y - e_1 - e_2) < y_3$, and by Lemma~\ref{lem:monodown} we can
therefore conclude that all points $y' < y$ also satisfy $f_3(y') < y'_3$,
meaning that $y' \not\in \Up(f)$. 

Since $x \ge y$ for all points $y \in \mathcal{B}$, and since $\mathcal{B}$ is
a path that spans the instance, there exists a non-negative constant $c$ such
that $x' = x - c \cdot e_1 - c \cdot e_2$ lies in $\mathcal{B}$. Observe that $x' \in
\DC_1(x)$ and $x' \in \DC_2(x)$ by construction. Thus, since $x \in \Up(f)$, we
can apply Lemma~\ref{lem:contrdown} to argue that $x' \in \Up(f)$, which
completes the proof.
\end{proof}

The next lemma shows that $\Up(f) \cap \mathcal{B}$ forms a contiguous sub-path
of the path defined by $\mathcal{B}$. In particular, this means that there is a
least index $l$ such that $x^l \in \Up(f)$ while $x^{l-1} \not\in \Up(f)$, that 
there is a greatest index $u$ such that 
$x^u \in \Up(f)$ while $x^{u+1} \not\in \Up(f)$, and that all points $x^i$ with
$l \le i \le u$ satisfy $x^i \in \Up(f)$. 

\begin{lemma}
\label{lem:bcontig}
Let $x^1, x^2, \dots, x^k$ be the path defined by $\mathcal{B}$. If there
exists a $j$ such that $x^j \in \Up(f)$,  then there exist indices $l$ and $u$
such that 
\begin{itemize}
\item If $l \le i \le u$ then $x^i \in \Up(f)$.
\item If $i < l$ or $i > u$ then $x^i \not\in \Up(f)$.
\end{itemize}
\end{lemma}
\begin{proof}
We will show that if there are indices $a < b < c$ with $x^a \in \Up(f)$ and
$x^c \in \Up(f)$, then we must have $x^b \in \Up(f)$, which implies the claim. 
Note that all points $x \in \mathcal{B}$ satisfy $x_3 \le f_3(x)$ by
definition, so we can restrict our attention to dimensions 1 and 2. 

Since for all $i$ we have $x^i = x^{i-1} + e_1$ or $x^i = x^{i-1} - e_2$, we
must have $x^b = x^a + c_1 \cdot e_1 - c_2 \cdot e_2$ for some non-negative
constants $c_1$ and $c_2$. Since $x^a \in \Up(f)$ we have $x^a_2 \le f_2(x^a)$.
Let $x' = x^a + c_1 \cdot e_1$. Since $x^a \le x'$ we can apply
Lemma~\ref{lem:monodown} to argue that $x'_2 \le f_2(x')$. Since $x' = x^b +
c_2 \cdot e_2$, this implies that $x^b_2 \le f_2(x^b)$, because otherwise we
would have a violation of strict contraction.

Likewise we must have $x^b = x^c - d_1 \cdot e_1 + d_2 \cdot e_2$ for some
non-negative constants $d_1$ and $d_2$. Let $x' = x^c + d_2 \cdot e_2$. Since
$x^c$ is in $\Up(f)$ we have $x^c_1 \le f_1(x^c)$, and we can apply
Lemma~\ref{lem:monodown} to conclude that $x'_1 \le f_1(x')$. This then implies
that $x^b_1 \le f_1(x^b)$ because otherwise we would have a violation of
contraction.

So we have shown that $x^b_i \le f_i(x^b)$ for all dimensions $i$, meaning that
$x^b \in \Up(f)$, which completes the proof of the claim.

We can now apply the claim to prove the lemma. Since there exists at least one
point $x^j \in \Up(f)$, we have that the set of indices $i$ for which $x^i \in
\Up(f)$ is a non-empty contiguous set of indices between some index $l$ and
some index $u$ such that $x^i \in \Up(f)$ if and only if $l \le i \le u$, where
the contiguity of the sequence is implied by the claim that we have proved. 
This then proves both statements of the lemma.
\end{proof}

\paragraph{\bf Finding the almost square critical box.}

We can now proceed toward proving that an almost square critical box exists.
The next lemma shows that we can find a point $x$ in $\mathcal{B}$ such that
$\hei_1(x) \ge \hei_2(x)$, and another point $y$ further along the line defined
by $\mathcal{B}$ such that $\hei_1(y) \le \hei_2(y)$. In fact, the proof shows
that $x^l$ satisfies the first property, while $x^u$ satisfies the second,
where $l$ and $u$ are the indices whose existence is asserted by
Lemma~\ref{lem:bcontig}.

\begin{lemma}\label{lem:order_dist}
Let $x^1, x^2, \dots, x^k$ be the path defined by $\mathcal{B}$.
If $\Up(f) \cap s$ is non-empty, then 
there exist indices $l$ and $u$ with $l \le u$, where $x^l \in \Up(f)$ and $x^u
\in \Up(f)$, and where the following
inequalities are satisfied.
\begin{align*}
\hei_1 (x^l) &\geq \hei_2(x^l) > 0, \\
\hei_2 (x^u) &\geq \hei_1(x^u) > 0.
\end{align*}
\end{lemma}
\begin{proof}
We begin with the first claim. Note that Lemma~\ref{lem:bup} implies that
$\Up(f) \cap \mathcal{B}$ is non-empty. So let $i$ be the least index such that
$x^i \in \Up(f)$. Note that $i$ cannot be $1$, since that would imply $x^i_2 =
n$, and we have ensured by our assumptions on the instance that we cannot have $x^i \in \Up(f)$ in this case. 

Since $x^i \in \Up(f)$ we clearly have both $\hei_1(x^i) > 0$ and $\hei_2(x^i)
> 0$. We will show that $\hei_2(x^i) = 1$, which then implies $\hei_2(x^i) \le
\hei_1(x^i)$, because the $\hei$ function is integer valued. 

Note that since $x^{i-1} \in \mathcal{B}$ we have $x^{i-1}_3 \le f_3(x^{i-1})$
by definition. There are two cases to consider.
\begin{itemize}
\item In the case where $x^{i-1} = x^i - e_1$ we use the following argument. Note that Lemma~\ref{lem:dist_adj}
implies that $\hei_1(x^{i-1}) >
\hei_1(x^{i}) > 0$, which implies $x^{i-1}_1 \le f_1(x^{i-1})$.
Since $x^{i-1} \not \in \Up(f)$ we must therefore have
$x^{i-1}_2 > f_2(x^{i-1})$. Consider the point $x' = x^i + e_2$, and note that
$x' \in \UC_2(x^{i-1})$. Hence Lemma~\ref{lem:contrdown} implies that $x'_2 >
f_2(x')$, which means that $\hei_2(x^i) = 1$ by definition.

\item In the case where $x^{i-1} = x^i + e_2$ we use the following argument. Note that since $x^i_1 \le f_1(x^i)$ we
can use Lemma~\ref{lem:monodown} to prove that $x^{i-1}_1 \le f_1(x^{i-1})$.
Since $x^{i-1} \not \in \Up(f)$ we must therefore have $x^{i-1}_2 >
f_2(x^{i-1})$. This means that $\hei_2(x^i) = 1$ by definition.
\end{itemize}
So we have shown that $\hei_1(x^i) \ge \hei_2(x^i) = 1 > 0$, so we can set $l =
i$, which completes the proof of the first claim.

For the second claim, we proceed in a symmetric manner. 
Let $i$ be the greatest index such that $x^i \in \Up(f)$. Note that $i$ cannot be
$k$, since that would imply $x^i_1 = n$, and we have ensured by preprocessing that we cannot have $x^i \in \Up(f)$ in this case. 
Since $x^i \in \Up(f)$ we clearly have both $\hei_1(x^i) > 0$ and $\hei_2(x^i)
> 0$. We will show that $\hei_1(x^i) = 1$, which then implies $\hei_1(x^i) \le
\hei_2(x^i)$, because the $\hei$ function is integer valued. 

Note that since $x^{i+1} \in \mathcal{B}$ we have $x^{i+1}_3 \le f_3(x^{i+1})$
by definition. There are two cases to consider.
\begin{itemize}
\item In the case where $x^{i+1} = x^i - e_2$ we use the following argument. Note that Lemma~\ref{lem:dist_adj}
implies that $\hei_2(x^{i+1}) >
\hei_2(x^{i}) > 0$, which implies $x^{i+1}_2 \le f_2(x^{i+1})$.
Since $x^{i+1} \not \in \Up(f)$ we must therefore have
$x^{i+1}_1 > f_1(x^{i+1})$. Consider the point $x' = x^i + e_1$, and note that
$x' \in \UC_1(x^{i+1})$. Hence Lemma~\ref{lem:contrdown} implies that $x'_1 >
f_1(x')$, which means that $\hei_1(x^i) = 1$ by definition.

\item In the case where $x^{i+1} = x^i + e_1$ we use the following argument. Note that since $x^i_2 \le f_2(x^i)$ we
can use Lemma~\ref{lem:monodown} to prove that $x^{i+1}_2 \le f_2(x^{i+1})$.
Since $x^{i+1} \not \in \Up(f)$ we must therefore have $x^{i+1}_1 >
f_1(x^{i+1})$. This means that $\hei_1(x^i) = 1$ by definition.
\end{itemize}
So we have shown that $\hei_2(x^i) \ge \hei_1(x^i) = 1 > 0$, so we can set $u =
i$, which completes the proof of the second claim.
\end{proof}

The previous lemma asserts the existence of points $x$ and $y$ with $\hei_1(x)
\ge \hei_2(x)$ and $\hei_1(y) \le \hei_2(y)$. The following lemma shows that we
can find two \emph{adjacent} points along the path defined by
$\mathcal{B}$ that satisfy this property. This is proved via applying the
intermediate value theorem to the two points that we obtained from the previous
lemma.  

Note that the lemma also permits us to produce only one point, but this point
$x$ must satisfy $\hei_1(x) = \hei_2(x)$. This clause is necessary since
$\Up(f) \cap \mathcal{B}$ may only contain one point, in which case we cannot
find two adjacent points.

\begin{lemma}\label{lem:order_dist_adj}
Let $x^1, x^2, \dots, x^k$ be the path defined by $\mathcal{B}$.
If $\Up(f) \cap s$ is non-empty, then there exists an index $i$ such
that either
\begin{itemize}
\item $x^i \in \Up(f)$ and $\hei_1(x^i) = \hei_2(x^i)$, or
\item $x^i, x^{i+1} \in \Up(f)$ and $\hei_1(x^i) \ge \hei_2(x^i)$ while $\hei_1(x^{i+1})
\le \hei_2(x^{i+1})$.
\end{itemize}
\end{lemma}
\begin{proof}
Let $l$ and $u$ be the indices whose existence is shown in 
Lemma~\ref{lem:order_dist}. If $l = u$, then we have that 
$\hei_1(x^l) = \hei_2(x^l)$, and so we are done.

Otherwise, we have $l < u$. Define the function $g(x) = \hei_1(x) - \hei_2(x)$
and observe that $g(x^l) \ge 0$ while $g(x^u) \le 0$. So by the intermediate
value theorem there must exist an index $i$ in the range $l \le i < u$ such
that $g(x^i) \ge 0$ while $g(x^{i+1}) \le 0$. Since $l \le i < u$, and since
$x^l, x^u \in \Up(f)$, Lemma~\ref{lem:bcontig} implies that both $x^i \in
\Up(f)$ and $x^{i+1} \in \Up(f)$, which completes the proof.
\end{proof}

We can now finish the proof of Lemma~\ref{lem:cb_shape}. We do this by taking
the points given to use by the previous lemma, and by showing that one of them
is the bottom-left point of an almost square critical box. In particular, 
if Lemma~\ref{lem:order_dist_adj} gives us a single point, then we will show
that that point defines an exactly square critical box. On the other hand, if 
Lemma~\ref{lem:order_dist_adj} gives us two points, then we will use the fact
that the maximum of $\hei_1$ and $\hei_2$ changes between these two adjacent
points, along with the properties from
Lemma~\ref{lem:dist_adj} that state that heights can differ by at most one
between adjacent points, to show that one of the two points defines an almost
square critical box. 

\begin{proof}[Proof of Lemma~\ref{lem:cb_shape}]
We consider the points whose existence is implied by
Lemma~\ref{lem:order_dist_adj}. If we are in the first case of that lemma then
we set $x = x^i$, we set $w = \hei_1(x)$ and $h = \hei_2(x)$. Observe that $x
\in \Up(f)$, so the first point of Definition~\ref{def:cb} is satisfied, and
that $x \in \mathcal{B}$, so the fourth point of Definition~\ref{def:cb} is
satisfied. Since we set $w = \hei_1(x)$ and $h = \hei_2(x)$ we have also
satisfied the other requirements of Definition~\ref{def:cb}. Thus, $(x, h, w)$
is a critical box, and moreover, since $\hei_1(x) = \hei_2(x)$, we have that
$|h - w| = 0$. 

If we are in the second case of Lemma~\ref{lem:order_dist_adj}, then we observe
that if $\hei_1(x^i) = \hei_2(x^i)$ then $(x^i, \hei_2(x^i), \hei_1(x^i))$ is a
square critical box, for the same reasons as above, while if $\hei_1(x^{i+1}) =
\hei_2(x^{i+1})$ then $(x^{i+1}, \hei_2(x^{i+1}), \hei_1(x^{i+1}))$ is a square
critical box, also for the same reasons.

This leaves the case in which $\hei_1(x^i) > \hei_2(x^i)$ and $\hei_1(x^{i+1}) <
\hei_2(x^{i+1})$. 
We consider two cases.
\begin{itemize}
\item In the case where $x^{i+1} = x^i + e_1$, we have 
\begin{align*}
0 &< \hei_1(x^i) - \hei_2(x^i) \\
&= \hei_1(x^{i+1}) + 1 - \hei_2(x^i) \\
&\le \hei_1(x^{i+1}) + 1 - (\hei_2(x^{i+1}) - 1) \\
& = 2 + \hei_1(x^{i+1}) - \hei_2(x^{i+1}) \\
& < 2,
\end{align*}
where the first inequality is by assumption, the first equality and second
inequality
follow from Lemma~\ref{lem:dist_adj}, and where the final inequality holds due
to the assumption that $\hei_1(x^{i+1}) < \hei_2(x^{i+1})$. 

Hence we have $0 < \hei_1(x^i) - \hei_2(x^i) < 2$, and since the $\hei$
function is integer valued, we therefore have 
$\hei_1(x^i) - \hei_2(x^i) = 1$. Thus if we set $h = \hei_2(x^i)$ and $w =
\hei_1(x^i)$, then
$(x^i, h, w)$ is a critical box with $| h - w | = 1$. 

\item The case where $x^{i+1} = x^i - e_2$ proceeds symmetrically. 
\begin{align*}
0 &< \hei_2(x^{i+1}) - \hei_1(x^{i+1}) \\
&= \hei_2(x^{i}) + 1 - \hei_1(x^{i+1}) \\
&\le \hei_2(x^{i}) + 1 - (\hei_1(x^{i}) - 1) \\
& = 2 + \hei_2(x^{i}) - \hei_1(x^{i}) \\
& < 2,
\end{align*}
where the first inequality is by assumption, the first equality and second
inequality
follow from Lemma~\ref{lem:dist_adj}, and where the final inequality holds due
to the assumption that $\hei_1(x^{i}) > \hei_2(x^{i})$. 

Hence we have $0 < \hei_1(x^{i+1}) - \hei_2(x^{i+1}) < 2$, and since the $\hei$
function is integer valued, we therefore have 
$\hei_1(x^{i+1}) - \hei_2(x^{i+1}) = 1$. Thus if we set $h = \hei_2(x^{i+1})$ and $w =
\hei_1(x^{i+1})$, then
$(x^{i+1}, h, w)$ is a critical box with $| h - w | = 1$. 

\end{itemize}
\end{proof}

\subsection{The States of the Algorithm}

\gridStates 

We are now ready to define the algorithm itself. The following definition gives
four types of states that will be considered by the algorithm. These states are
shown in Figure~\ref{fig:grid_states}.

\begin{definition}[Algorithm State] \label{def:state}
Our algorithm will always be in one of the four following states, schematically represented in Figure \ref{fig:grid_states}.
\begin{enumerate}
\item A \emph{left-lobe state} defined by two boxes: a main
lobe $(x, h, w)$, and a sub-lobe $(y, u, v)$ with the following restrictions.
\begin{itemize}
\item $y_1 \le x_1$ and $y_1 + v = x_1$.
\item $x_2 \le y_2$ and $y_2 + u \le x_2 + h$.
\item $h \leq \Floor{\frac{w}{4}}$.
\end{itemize}

\item A \emph{bottom-lobe state}
is defined by two boxes: a main
lobe $(x, h, w)$, and a sub-lobe $(y, u, v)$ with the following restrictions.
\begin{itemize}
\item $x_1 \le y_1$ and $y_1 + v \le x_1 + w$.
\item $y_2 \le x_2$ and $y_2 + u = x_2$.
\item $w \leq \Floor{\frac{h}{4}}$.
\end{itemize}

\item A \emph{diagonal-lobe state}
is defined by two diagonal boxes: a main lobe $(x, y, l)$, and a sub-lobe $(a,
b, m)$ with the following restriction.
\begin{itemize}
\item $(y_1-x_1) + (x_2-y_2) \leq \Floor{\frac{l}{8}}$.
\item $a_1 = y_1 + l$ or $a_2 = x_2 + l$
\item $b_1 = y_1 + l$ or $b_2 = x_2 + l$.
\end{itemize}

\item A \emph{full state} which consists of a central box $(x, h, w)$ along with all
three lobe states:
\begin{itemize}
\item A left lobe state $(x^l, h^l, w^l)$ and $(y^l, u^l, v^l)$ such that 
$x = x^l + w^l \cdot e_1$.
\item A bottom lobe state 
$(x^r, h^r, w^r)$ and $(y^r, u^r, v^r)$ such that $x = x^r + h^r \cdot e_2$.
\item A diagonal lobe state 
$(x^d, y^d, l^d)$ and $(a^d, b^d, m^d)$ such that $x^d = x + h \cdot e_2$ and
$y^d = x + w \cdot e_1$.
\end{itemize}
\end{enumerate}
\end{definition}

\paragraph{\bf State areas.}

Our algorithm will make progress by reducing the area of a state, which we now
define. 
Given a box $(x, h, w)$, we define the area of the box as $\area(x, h, w) = h
\cdot w$. Given a diagonal box $(x, y, l)$ we define the area of the box as
$\area(x, y, l) = (x_2-y_2+l)(y_1-x_1+l) - l^2$, which can be calculated by subtracting two right triangles of side length $l$ from the larger rectangle including $\DBox$.
We say that a lobe is \emph{empty} if its area is zero. In each of the states
of the algorithm, it is possible for some of the boxes to be empty. For
example, we may have a full state whose main left lobe is empty, but that state
might still have a non-empty left sub-lobe. 

\paragraph{\bf The algorithm invariant.}

Our algorithm will ensure that all states that are visited satisfy an
invariant, which we now define.
For each algorithm state $t$ we define $\points(t)$ to be the set of points
that are contained in all of the boxes and diagonal boxes used to define that
state. In particular, whenever $t$ is a full state, $\points(t)$ contains all
of the points of all of the sub-states used in the definition of $t$. 

\begin{definition}[Algorithm Invariant]
An algorithm state $t$ is said to satisfy the invariant for a slice $s$
if $\Up(f) \subseteq \points(t)$, and if $\Up(f)
\cap s \ne \emptyset$, then the following will be satisfied. 
\begin{itemize}
\item If $t$ is a left-lobe or bottom-lobe state, 
then all 
almost square critical boxes are entirely contained in the main lobe $\CBox(x, h, w)$.
\item If $t$ is a diagonal-lobe state, then 
there all almost square critical boxes are entirely contained in the main lobe $\DBox(x, y, l)$.
\item If $t$ is a full-state, then 
all almost square critical boxes are entirely contained in the central box $\CBox(x, h, w)$.
\end{itemize}
\end{definition}

Our goal is to produce an algorithm that, given a state $t$ that satisfies the
invariant, makes constantly many queries to $f$ and either finds a point in
$\Up(f)$, or produces a new state $t'$ that satisfies the invariant where
$\area(t') \le c \cdot \area(t)$ for some constant $c$. In
Section~\ref{sec:leftalg} we give such an algorithm for left- and bottom-lobe
states, in Section~\ref{sec:algdiag} we give such an algorithm for
diagonal-lobe states, and in Section~\ref{sec:fullalg} we give such an
algorithm for full states.

\subsection{Algorithms for Left and Bottom Lobe States}
\label{sec:leftalg}

In this section we will prove the following lemma.

\begin{lemma}
\label{lem:leftlobestate}
There is an algorithm that, given a left-lobe state $t$ that satisfies the
invariant for a slice $s$, makes constantly many queries and either
\begin{itemize}
\item finds a point in $\Up(f)$, or
\item produces a new left-lobe state or full state $t'$ that satisfies the
invariant for $s$ with $\area(t') \le 15/16 \cdot \area(t)$. 
\end{itemize}
Moreover each step of the algorithm runs in time that is polynomial in the
representation of $f$.
\end{lemma}

Throughout this section we consider a left-lobe state $t$ that consists of
a main lobe $(x, h, w)$ and a sub-lobe $(y, u, v)$. The algorithm consists of
several steps. 


\paragraph{\bf Step 1: Finding a new main lobe.}

\leftLobeMain

The algorithm first queries the point $q = x + \Floor{\frac{w}{2}} \cdot e_1$,
which is the point that lies half-way along the bottom side of the main lobe. It
then proceeds with the following case analysis, where each case is shown in
Figure~\ref{fig:leftLobeMain}.

\begin{itemize}
\item \textbf{Case 1:} $\mathbf{f_1(q) < q_1}$.
(Figure~\ref{fig:leftLobeMain} (a)). Here we note that a point 
$p \in \CBox(x, h, w)$ with $p_1 \ge q_1 + h$ lies in $\UC_1(q)$, and so by
Lemma~\ref{lem:contrdown} we have $p \not\in \Up(f)$. This means that
no critical box can lie in the shaded region in Figure~\ref{fig:leftLobeMain}
(a). So we can build a new main lobe $(x, h, \Floor{\frac{w}{2}} + h)$
that satisfies the algorithm invariant. 




\item \textbf{Case 2:} $\mathbf{f_1(q) \geq q_1}$. 
In this case note that any almost square critical box can have height at most $h$, and
width at most $h + 1$. Lemmas~\ref{lem:contrdown} and~\ref{lem:monodown}
imply that all points $p$ with $p_1 \le q_1$ and $p_2 \ge q_2$ satisfy
$p_1 \le f_1(p)$. On the other hand, if
$(z, a, b)$ is a critical box, then by definition we must have $f_1(z + (b+1)
\cdot e_1) < z_1 + b + 1$. Hence, we must have $z_1 \ge q_1 - h - 1$.

This means that we can define a new main lobe starting from point $x' = x + \left(\Floor{\frac{w}{2}} -
h - 1\right) \cdot e_1$ as the box $(x', h, w - \left(\Floor{\frac{w}{2}} -
h - 1\right))$, and any almost square critical box that lies in the 
main lobe of $t$ must also be contained in the newly defined main lobe, which satisfies the algorithm
invariant for the main lobe of a left-lobe state.

While we have reduced the size of the main lobe, we have not yet 
eliminated any points, so the area of the state has not been reduced.
To address this, we
make one further query at $q^{\text{mid}} = x + ( \Floor{\frac{w}{2}} -
h - 1) \cdot e_1 + \Floor{\frac{h}{2}} \cdot e_2$, and we consider the following
cases. 

\begin{enumerate}
\item $\mathbf{f_2(q^{\text{mid}}) < q^{\text{mid}}_2}$ (Figure~\ref{fig:leftLobeMain} (b)).
Here Lemmas~\ref{lem:contrdown} and~\ref{lem:monodown} imply that any point $p$
with $p_1 \le q^{\text{mid}}_1$ and $p_2 \ge q^{\text{mid}}_2$ must satisfy
$f_2(p) < p_2$, and so cannot lie in $\Up(f)$. So we can rule out these points.




\item $\mathbf{f_2(q^{\text{mid}}) \geq q^{\text{mid}}_2}$  \textbf{and}
$\mathbf{f_3(q^{\text{mid}}) <  q^{\text{mid}}_3}$.
(Figure~\ref{fig:leftLobeMain} (c)). Here Lemma~\ref{lem:monodown} implies that
any point $p$ with $p \le q^{\text{mid}}$ must satisfy $f_3(p) < p_3$, which
means that $p \not \in \Up(f)$, and so we can rule these points out.



\item $\mathbf{f_2(q^{\text{mid}}) < q^{\text{mid}}_2}$  \textbf{and}
$\mathbf{f_3(q^{\text{mid}}) <
q^{\text{mid}}_3}$. 
(Figure~\ref{fig:leftLobeMain} (d)). Here the arguments from both of the
previous cases can be combined to argue that any point $p$ with $p_1 \le
q^{\text{mid}}_1$ cannot lie in $\Up(f)$. 

\end{enumerate}
\end{itemize}

At the end of Step 1 of the algorithm we have defined a new smaller main lobe,
and we have ruled out some of the remaining points from the original main lobe:
in Case~1 all of these points are ruled out, while in Case 2 we made one extra
query in order to rule out half of the remaining space.

\paragraph{\bf Step 2: Building a sub-lobe.}

Following on from Step 1, we now need to build a new sub-lobe. In Case~1 of
Step 1 we do not alter the original sub-lobe, so we can proceed with that
sub-lobe and skip Step~2. On the other hand, in Case 2 of Step 1, we have the
remaining area of the original main lobe, which is non-empty in Case 2.1 and
Case 2.2, and we also have the original sub-lobe. We need to rationalize this
remaining space into a single sub-lobe for the new state.

If the original sub-lobe has been entirely ruled out in Step 1, as it has in 
examples in Figures~\ref{fig:leftLobeMain} (b) and~\ref{fig:leftLobeMain} (d), then we 
simply use the remaining space from the original main lobe as the new sub-lobe.
So for example, in Figure~\ref{fig:leftLobeMain} (b) the new sub-lobe would be
$(x, \Floor{\frac{h}{2}} , \Floor{\frac{w}{2}} - h - 1)$, while in
Figure~\ref{fig:leftLobeMain} (d) we can use an empty lobe as the new sub-lobe.
In both cases Step 2 of the algorithm terminates. 
If the original sub-lobe has only partially been ruled out in Step 1, then we
first reduce the sub-lobe by removing the eliminated points before continuing
with the procedure below.
 
The situation that we must handle is depicted in 
Figure~\ref{fig:leftLobeMain} (c): we have a non-empty remaining space on the
left of the old main lobe, and a sub-lobe that lies to the left of this space.
Our task is to rationalize these two spaces into a single sub-lobe for the new
state.

\leftLobeSub



Let $(\hat{y}, \hat{u}, \hat{v})$ be the remaining sub-lobe of state $t$ after
the original sub-lobe was reduced by removing the points that were ruled out in
Step 1. We query the points
\begin{align*}
q^{\text{bot}} &= x + \left( \Floor{\frac{w}{2}} - h - 1\right) \cdot e_1 +
(\hat{y}_2 - x_2) \cdot
e_2, \\
q^{\text{top}} &= x + \left( \Floor{\frac{w}{2}} - h - 1\right) \cdot e_1 +
(\hat{y}_2 + \hat{u} - x_2) \cdot e_2,
\end{align*}
as shown in Figure~\ref{fig:leftLobeSub}. We then consider the following cases.

\begin{itemize}
\item \textbf{Case 1.} $\mathbf{f_3(q^{\text{top}}) < q^{\text{top}}_3}$.
(Figure~\ref{fig:leftLobeSub} (a)). Here, as in Step 1, Case 2.2, all points
below and to the left of $q^{\text{top}}$ cannot lie in $\Up(f)$, and this
rules out the existing sub-lobe. So we can use the remaining space, which is
labelled sub' in 
Figure~\ref{fig:leftLobeSub}, as the new sub-lobe.



\item \textbf{Case 2.} 
$\mathbf{f_2(q^{\text{bot}}) < q^{\text{bot}}_2}$.
(Figure~\ref{fig:leftLobeSub} (b)). 
Here, as in Step 1, Case 2.1, all points
above and to the left of $q^{\text{bot}}$ cannot lie in $\Up(f)$, and this
rules out the existing sub-lobe, so we can use the remaining space as a new
sub-lobe.



\item \textbf{Case 3.
$\mathbf{f_3(q^{\text{top}}) \geq q^{\text{top}}_3}$ and 
$\mathbf{f_2(q^{\text{top}}) \geq q^{\text{top}}_2}$}. Here note that
$q^{\text{top}} \in \DC_1(q)$, where $q$ is the point that we initially queried
in Step 1, so we can use Lemma~\ref{lem:contrdown} to argue that 
$f_1(q^{\text{top}}) \geq q^{\text{top}}_1$, meaning that 
$q^{\text{top}} \in \Up(f)$. So the algorithm can terminate and return 
$q^{\text{top}}$. 

\item 
\textbf{Case 4.
$\mathbf{f_3(q^{\text{bot}}) \geq q^{\text{bot}}_3}$ and 
$\mathbf{f_2(q^{\text{bot}}) \geq q^{\text{bot}}_2}$}. Using the same reasoning
as the previous case we can conclude that
$q^{\text{bot}} \in \Up(f)$. So the algorithm can terminate and return 
$q^{\text{bot}}$.

\item \textbf{Case 5.} 
If none of the previous cases apply then we have 
$\mathbf{f_2(q^{top}) < q^{top}_2}$,
$\mathbf{f_3(q^{top}) \geq q^{top}_3}$, $\mathbf{f_2(q^{\text{bot}}) \geq
q^{\text{bot}}_2}$ \textbf{and}   $\mathbf{f_3(q^{\text{bot}}) <
q^{\text{bot}}_3}$, as shown in Figure~\ref{fig:leftLobeSub} (c). 
Here we can simply extend the sub-lobe $(y, u, v)$ to $(y, u, v + 
\Floor{\frac{w}{2}} - h - 1)$. 
\end{itemize}

At the end of Step 2 we now have a single sub-lobe that can be taken forward
into Step 3.

\paragraph{\bf Step 3: Halving the sub-lobe.}

\leftLobeSubSub

The final step of the algorithm is to halve the area of the sub-lobe. We need
to do this because a significant proportion of the area of $t$ may lie in the
sub-lobe, and this area may not have been reduced by the previous steps.

If there is no sub-lobe arising from the previous steps, then we skip Step 3.
Otherwise, 
let $(\bar{y}, \bar{u}, \bar{v})$ be the sub-lobe that has been defined in Steps 1 and 2. 
We query the point 
$$q^{\text{cen}} = \bar{y} + \bar{v} \cdot e_1 + \Floor{\frac{\bar{u}}{2}} \cdot e_2,$$
which is the point that is half-way along the right-hand side of the sub-lobe.
We then consider the following cases.
\begin{enumerate}
\item \textbf{Case 1.} $\mathbf{f_2(q^{\text{cen}}) < q^{\text{cen}}_2}$. (Figure~\ref{fig:leftLobeSubSub}
(a)). As in previous steps, this rules out all points weakly above and to the left of 
$q^{\text{cen}}$, which approximately halves the size of the sub-lobe. 
\item \textbf{Case 2.} $\mathbf{f_3(q^{\text{cen}}) < q^{\text{cen}}_3}$. 
(Figure~\ref{fig:leftLobeSubSub} (b)). 
As in previous cases, this rules out all points weakly below and to the left of 
$q^{\text{cen}}$, which approximately halves the size of the sub-lobe.
\item \textbf{Case 3.} If neither of the previous two cases apply then we have \textbf{
$\mathbf{f_2(q^{\text{cen}}) \geq q^{\text{cen}}_2}$ and
$\mathbf{f_3(q^{\text{cen}}) \geq q^{\text{cen}}_3}$}. We can use the same
reasoning as cases 3 and 4 of Step 2 to conclude that $q^{\text{cen}} \in
\Up(f)$, and so the algorithm can terminate. 
\end{enumerate}

\paragraph{\bf Step 4: Building the new state.}

At this point we have a new main lobe $m = (x', h', w')$ that was found in Step 1, and
a new sub-lobe $l = (y', u', v')$ that was found in Step 3. We now build a new state $t'$ in the following way.
\begin{itemize}
\item If $h' \le \Floor{\frac{w'}{4}}$, then we build a new left-lobe state using $m$
and $l$.
\item If $h' > \Floor{\frac{w'}{4}}$, then we build a new full state using $m$
as the central box, using $l$ as the main lobe of the left-lobe state, and
setting all other main and sub-lobes to be empty. 
\end{itemize}
The algorithm now returns $t'$ and terminates.

\paragraph{\bf Correctness.}

So far we have built a state $t'$, and we have shown that the algorithm
invariant holds for~$t'$. Specifically we have shown the following.
\begin{itemize}
\item In Step 1, we showed that if an almost square critical box lies in the main lobe
of $t$, then it must lie in the main lobe of $t'$. This is because we
specifically removed regions from the main lobe that could not contain an almost
square critical box. 
\item We have shown that $\Up(f)$ is contained within the main and sub lobes of
$t'$. This is because $\Up(f)$ was contained within $t$, and all steps of the algorithm only removed points that are provably not in $\Up(f)$. 
\end{itemize}

To finish the proof of Lemma~\ref{lem:leftlobestate} we need to show that
the area of $t'$ has been sufficiently reduced. 
\begin{itemize}
\item Step 1 of the algorithm removes at least $1/16$ of the area of the
original main lobe $(x, h, w)$. Specifically, Step 1 Case 1 removes a $(w -
\Floor{\frac{w}{2}} - h)/w$ fraction of the area, and since $h \le
\Floor{\frac{w}{4}}$,
this means that at least $\Ceil{\frac{w}{4}}/4 \ge 1/8$ fraction of the area has been removed.
The three different cases of Step 2 Case 2, remove less area, since only half of the
remaining area of the main lobe is eliminated, so these cases remove a $1/16$
fraction of the area.

\item 
The original sub-lobe is either entirely removed in Step 1 or Step 2, or it is
processed by Step 3, in which case we remove a $\Floor{\frac{u}{2}}/u$ fraction
of the area of the original sub-lobe $(y, u, v)$. 
\begin{itemize}
\item If $u = 0$ then the area of the sub-lobe is zero, and we are done.
\item If $u = 1$ then note that Step 4 Cases 1 and 2 both rule out at least
half of the points in the sub-lobe since weak inequalities are used when
determining what to rule out.
\item If $u \ge 2$ then note that $\Floor{\frac{u}{2}}/u \ge 1/3$.
\end{itemize}
\end{itemize}
So we have shown that the original main lobe and the original sub lobe have both
been reduced in area by a fraction of at least $1/16$, which implies $\area(t')
\le \frac{15}{16} \area(t)$.

\paragraph{\bf Bottom-lobe states.}

We can use the same algorithm for bottom-lobe states by simply exchanging
dimensions 1 and 2. So we have the following lemma.

\begin{lemma}
\label{lem:bottomlobestate}
There is an algorithm that, given a bottom-lobe state $t$ that satisfies the
invariant, makes constantly many queries and either
\begin{itemize}
\item finds a point in $\Up(f)$, or
\item produces a new bottom-lobe state or full state $t'$ that satisfies the invariant with $\area(t') \le 15/16 \cdot \area(t)$. 
\end{itemize}
Moreover each step of the algorithm runs in time that is polynomial in the
representation of $f$.
\end{lemma}

\subsection{The Algorithm for Diagonal States}
\label{sec:algdiag}

This section proves the equivalent of Lemma \ref{lem:leftlobestate} for
diagonal states. It follows exactly the same approach as we used for 
Lemma \ref{lem:leftlobestate}, but each step has to be adapted to work with
diagonal regions.

\begin{lemma}
	\label{lem:diaglobestate}
	There is an algorithm that, given a diagonal-lobe state $t$, makes constantly many
	queries and either
	\begin{itemize}
		\item finds a point in $\Up(f)$,
		\item declares that $\Up(f)$ is empty in the current slice, or
		\item produces a new diagonal state $t'$ that satisfies the invariant with
		$\area(t') \le 15/16 \cdot \area(t)$, or
		\item produces a new full state $t'$ that satisfies the invariant with
		$\area(t') \le 5 \cdot \area(t)$.
	\end{itemize}
	Moreover, each step of the algorithm runs in time that is polynomial in the
	representation of $f$.
\end{lemma}

Note that the last case here shows that the area may be \emph{increased}. This
is not a problem however, because as we shall see this case can occur at most
once per slice, and this increase can be cancelled out by taking constantly many
further steps. We will deal with this explicitly when we analyse the overall
running time of the algorithm in Section~\ref{sec:mainalgo}.

The algorithm receives as input a diagonal lobe state $t$ that consists of
a main lobe $(x,y,l)$ and a sub-lobe $(a,b,m)$. Similarly to the left lobe
case, the algorithm either finds a point in $\Up(f)$ or creates a new state
$t'$ that satisfies the invariant. 

We define the horizontal width $w$ of a diagonal lobe $(x,y,l)$ to be:
$$ w(x,y,l) = (x_2 - y_2) + (y_1-x_1)$$

\paragraph{\bf Step 1: Finding a new main lobe.}
As an initial query, the algorithm queries the point
$$q= \left(y_1 + \Floor{\frac{l}{2}}, x_2+\Floor{\frac{l}{2}}\right),$$
which is located halfway along the main lobe. We consider the following
cases.

\begin{itemize}
\item \textbf{Case 1:} $\mathbf{f_3(q) < q_3}$ (Figure
\ref{fig:diagLobeStep1} (a)). By Lemma \ref{lem:monodown}, no point $p$ in 
$\DBox(x,y,\Floor{\frac{l}{2}})$ can satisfy $p \in \Up(f)$, which in turn excludes
the existence of any critical box in this region. The new main lobe is defined
as 
$$\DBox\left(x+ \Floor{\frac{l}{2}} \cdot e_1 + \Floor{\frac{l}{2}} \cdot e_2,
\; y+
\Floor{\frac{l}{2}} \cdot e_1 +  \Floor{\frac{l}{2}} \cdot e_2, \;
\Floor{\frac{l}{2}}\right).$$
Since we have only eliminated points that cannot lie in
$\Up(f)$, any critical box contained in the old main must also be contained in
the new main lobe.

\item \textbf{Case 2:} $\mathbf{f_3(q) \geq q_3}$. In this case, any
critical box can have width at most $w(x,y,l)$ since it is contained within the
main lobe. Note that all points with $p_1 \geq q_1$ and $p_2 \geq q_2$ satisfy
$f_3(p) \geq p_3$, due to Lemma \ref{lem:monodown}. By definition, any critical
box $(z,h,w)$ must satisfy $f_3(z - (e_1 + e_2)) < z_3$, and since we cannot
satisfy this for any point $p \ge q$, we must have that
the bottom left corner of $(z, h, w)$ must have $z_1 \leq q_1 \le
y_1+ \Floor{\frac{l}{2}}$ and $z_2 \le q_2 \leq x_1+ \Floor{\frac{l}{2}}$.

So we can reduce the main lobe by eliminating any point that is not below
$(y_1+
\Floor{\frac{l}{2}} + 2w,x_2+ \Floor{\frac{l}{2}} + 2w)$, and therefore define
a new main lobe
$\DBox(x, y,\Floor{\frac{l}{2}} +2w)$. 
This is valid because any box that is contained within the original main lobe
can have height and width at most 
$(x_2-y_2)+(y_1-x_1) = w$.
So if the main lobe of state $t$ contains
an almost square critical box, then such a box is also contained within the new
main lobe, satisfying the invariant. 

This operation has reduced the size of the main lobe, but it has not yet
eliminated any points. To do this, we make one further query at

$$q^{\text{mid}} = \left(x_1 + \Floor{\frac{l}{2}} + 2w + \Floor{\frac{||x-y||_{\infty}}{2}}, \; x_2
+ \Floor{\frac{l}{2}} + 2w\right) $$

for which $f_3(q^{\text{mid}}) \geq q^{\text{mid}}_3$ holds. Given the assumption that $w \leq \Floor{\frac{l}{8}}$, $q^{\text{mid}}$ is well defined within the lobe. The orientation of $q^{\text{mid}}$ along dimensions 1 and 2 guarantees that the query point splits the remaining lobe in half, maintaining the invariant.

\begin{enumerate}
\item $\mathbf{f_1(q^{\text{mid}}) < q^{\text{mid}}_1}$  \textbf{and}
$\mathbf{f_2(q^{\text{mid}}) \geq   q^{\text{mid}}_2}$.
(Figure \ref{fig:diagLobeStep1} (b)). Here Lemmas~\ref{lem:contrdown} and \ref{lem:monodown} imply that any point $p$ with $p_1 \ge q^{\text{mid}}$ and $p_2 \le q^{\text{mid}}$ as well as all points in $\UC_1(q^{\text{mid}})$ must satisfy $f_1(p) < p_1$. Therefore for all points $p \not \in \Up(f)$, and so can be excluded.

\item $\mathbf{f_1(q^{\text{mid}}) \ge q^{\text{mid}}_1}$  \textbf{and}
$\mathbf{f_2(q^{\text{mid}}) <   q^{\text{mid}}_2}$. (Figure \ref{fig:diagLobeStep1} (c)).
Lemmas~\ref{lem:contrdown} and~\ref{lem:monodown} imply that any point $p$
with $p_1 \le q^{\text{mid}}_1$ and $p_2 \ge q^{\text{mid}}_2$ as well as all points in $\UC_2(q^{\text{mid}})$ must satisfy
$f_2(p) < p_2$, and is excluded from $\Up(f)$. These points can be ruled out.

\item $\mathbf{f_1(q^{\text{mid}}) < q^{\text{mid}}_1}$  \textbf{and}
$\mathbf{f_2(q^{\text{mid}}) <
q^{\text{mid}}_2}$. 
(Figure \ref{fig:diagLobeStep1} (d)). Here the arguments from both of the
previous cases can be combined to argue that no point in the union of the above regions can lie in $\Up(f)$, and the sub-lobe at case $t'$ is empty.

\end{enumerate}
\end{itemize}

Step 1 has reduced the main lobe for the next state, and ruled out regions of the original main and sub-lobe. Case 1 results in an empty sub-lobe, while case 2 introduces an additional query point that rules out half of the remaining space.

\diagLobeStepOne

\paragraph{\bf Step 2: Building a sub-lobe.}
This step deals with the construction of a new sub-lobe. Cases 1 and
2.3 of Step 1 result in an empty sub-lobe, so the algorithm directly jumps to
Step~4. Cases 2.1 and 2.2 maintain the original sub-lobe in parts or entirely,
and thus further action is required to 
merge the sub-lobe and the remaining space from the old main-lobe into a single
diagonal box.
In the case where the original sub-lobe is partially excluded 
by the constraints on $q^{\text{mid}}$ at Step 1, the sub-lobe is
first altered to remove those points before we proceed. 
Specifically, we denote the sub-lobe remaining after the processing of Step 1
as $(\hat{a}, \hat{b}, m)$. The algorithm then queries two additional points:

\begin{align*}
q^{\text{left}} &= \hat{a} - (\Floor{\frac{l}{2}} + 2w) \cdot e_1 - (\Floor{\frac{l}{2}} + 2w) \cdot e_2, \\
q^{\text{right}} &=  \hat{b} - (\Floor{\frac{l}{2}} + 2w) \cdot e_1 - (\Floor{\frac{l}{2}} + 2w) \cdot e_2
\end{align*}
as shown in Figure \ref{fig:diagLobeStep2}. We identify the following cases. 

\begin{itemize}
\item \textbf{Case 1.} $\mathbf{f_2(q^{\text{left}}) < q^{\text{left}}_2}$.
(Figure \ref{fig:diagLobeStep2} (a)). All points in the region
$UC_2(q^{\text{left}})$ located above and to the right of $q^{\text{left}}$ are restricted from $\Up(f)$ by Lemma~\ref{lem:contrdown}, which
rules out the existing sub-lobe. So we can use the remaining space as the new sub-lobe. 

\item \textbf{Case 2.} 
$\mathbf{f_1(q^{\text{right}}) < q^{\text{right}}_1}$. (Figure \ref{fig:diagLobeStep2} (b))
All points in $UC_1(q^{\text{right}})$, located to the left of
$q^{\text{right}}$ cannot lie in $\Up(f)$ by Lemma~\ref{lem:contrdown}, and the existing sub-lobe is excluded. The new sub-lobe is bounded by $q^{\text{right}}$ and $q^{\text{mid}}$.

\item \textbf{Case 3.}
$\mathbf{f_1(q^{\text{left}}) \geq q^{\text{left}}_1}$ and 
$\mathbf{f_2(q^{\text{left}}) \geq q^{\text{left}}_2}$. Since Step 2 is only reached when $f_3(q) \geq q_3$, Lemma \ref{lem:monodown} implies that $f_3(q^{\text{mid}}) \geq q^{\text{mid}}_3$, since $q^{\text{mid}}_1 \geq q_1$ and $q^{\text{mid}}_2 \geq q_2$. Then  
$q^{\text{left}} \in \Up(f)$, and the algorithm terminates returning
$q^{\text{left}}$. 

\item 
\textbf{Case 4.}
$\mathbf{f_1(q^{\text{right}}) \geq q^{\text{right}}_1}$ and 
$\mathbf{f_2(q^{\text{right}}) \geq q^{\text{right}}_2}$.  Applying the same monotonicity argument as above, the algorithm terminates, returning $q^{\text{right}} \in \Up(f)$.

\item \textbf{Case 5.} 
$\mathbf{f_1(q^{\text{left}}) \geq q^{\text{left}}_1}$,
$\mathbf{f_2(q^{\text{left}}) < q^{\text{left}}_2}$, $\mathbf{f_1(q^{\text{right}}) <
q^{\text{right}}_1}$ \textbf{and}   $\mathbf{f_2(q^{\text{right}}) \geq
q^{\text{right}}_2}$ (Figure \ref{fig:diagLobeStep2} (c)). 
In this case, the original sub-lobe remains intact and is extended as $(q^{\text{left}}, q^{\text{right}}, \Floor{\frac{l}{2}} - w - 1 + m)$.
\end{itemize}

Upon merging the above cases to a single sub-lobe, the algorithm moves
to Step 3 which performs a further area reduction.

\diagLobeStepTwo

\paragraph{\bf Step 3: Halving the sub-lobe.}

In this step, the algorithm makes one final query 
that reduces the remaining sub-lobe's area by a half. This is necessary,
because the sub-lobe area might constitute a significant
fraction of $\area(t)$.
This step is only executed when the algorithm so far has left a non empty sub-lobe.
If the sub-lobe is empty, then we can jump directly to Step 4.

Let $(\bar{a}, \bar{b}, \bar{m})$ be the sub-lobe that remains after Steps 1
and 2. 
We query the point 
$$q^{\text{cen}} = \Floor{\frac{\bar{a}+\bar{b}}{2}} - (\Floor{\frac{l}{2}} + 2w) \cdot e_1 - (\Floor{\frac{l}{2}} + 2w) \cdot e_2$$

which corresponds to the midpoint between $\bar{a}$ and $\bar{b}$ at a height of $(\Floor{\frac{l}{2}} + 2w)$. 
The following cases are considered.
\begin{enumerate}
\item \textbf{Case 1.} $\mathbf{f_2(q^{\text{cen}}) < q^{\text{cen}}_2}$. 
(Figure \ref{fig:diagLobeStep3} (a)). 
Lemmas~\ref{lem:monodown} and~\ref{lem:contrdown} rule out all points weakly above and to the right of 
$q^{\text{cen}}$, as well as points in $UC_2(q^{\text{cen}})$, which approximately halves the size of the sub-lobe.
\item \textbf{Case 2.} $\mathbf{f_1(q^{\text{cen}}) < q^{\text{cen}}_1}$. (Figure~\ref{fig:diagLobeStep3}
(b)). Similarly to the previous case, Lemmas~\ref{lem:monodown} and~\ref{lem:contrdown}
imply that all points weakly above and to the left of 
$q^{\text{cen}}$, , as well as points in $UC_1(q^{\text{cen}})$, can be excluded, and the resulting size of the sub-lobe is approximately halved. 
\item \textbf{Case 3.} If none of the above cases apply, then we have \textbf{
$\mathbf{f_1(q^{\text{cen}}) \geq q^{\text{cen}}_1}$ and
$\mathbf{f_2(q^{\text{cen}}) \geq q^{\text{cen}}_2}$}. By the same argument as cases 3 and 4 of Step 2, we can conclude that $q^{\text{cen}} \in
\Up(f)$, and the algorithm terminates, returning $q^{\text{cen}}$. 
\end{enumerate}

\diagLobeStepThree


\paragraph{\bf Step 4: Building the new state.}

The new main lobe, produced by Step 1, is defined as $(x', y', l')$, with width $w'$, and
the new sub-lobe $(a', b', m')$ is the result of Step 3. A new state $t'$ is build as follows.
\begin{itemize}
\item If $w' \le \Floor{\frac{l'}{8}}$, the next diagonal lobe state $t'$ is built using the lobes defined above.
\item If $w' > \Floor{\frac{l'}{8}}$, a new full state is built. The full state consists of a central box $\left( x, l+ (y_1-x_1), l + (y_2-x_2) \right)$, surrounding the main lobe $(x', y', l')$, and the new lobe of the diagonal lobe state becomes $(a', b', m')$. All other lobes are set to be empty. 
\end{itemize}

The new state $t'$ is returned and the algorithm terminates. 

\paragraph{\bf Correctness.}

Our analysis so far has shown that, if the algorithm invariant holds for state 
$t$, then it also holds for
state $t'$. Specifically, we have shown the following.
\begin{itemize}
\item Step 1 constructs the new main lobe. We have argued that if an almost
square critical box lies in the main lobe of $t$, it must also lie in the main
lobe of $t'$, since we have only excluded regions that cannot contain a
critical box.
\item We have shown that $\Up(f)$ is contained within the main and sub lobes of
$t'$. This is because $\Up(f)$ was contained within $t$, and all steps of the
algorithm only removed points that are provably not in~$\Up(f)$. 
\end{itemize}

As a last step of the proof of Lemma~\ref{lem:diaglobestate}, we show the area reduction at state $t'$ for the case where $w \leq \Floor{\frac{l}{8}}$. 

\begin{itemize}
\item Case 1 of Step 1 removes at least $1/2$ of the area of the original main
lobe $(x,y,l)$, under the weaker assumption $w < l$, as the remaining area is
$(y_1-x_1)(x_2-y_2) + \frac{l}{2}w$ compared to the original main lobe with
area $(y_1-x_1)(x_2-y_2) + lw$. Case 2 results in a new main lobe $(x,y,
\Floor{\frac{l}{2}} + 2w)$ which removes at least $1/8$ of the original lobe,
since $w \leq \Floor{\frac{l}{8}}$. The area of the main lobe in this case is defined as
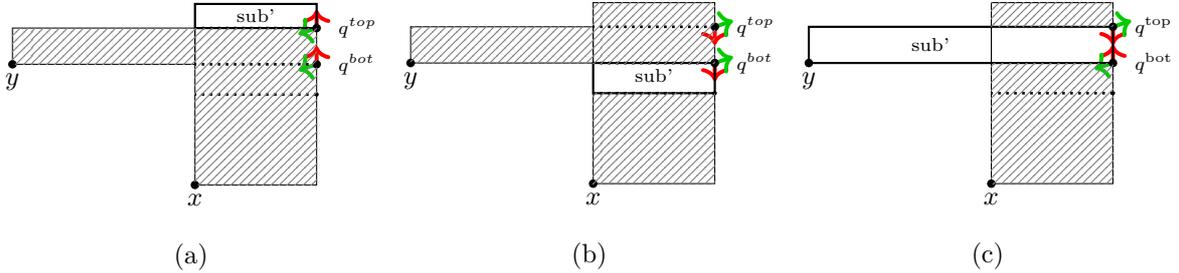
$(y_1-x_1)(x_2-y_2) + (\frac{l}{2} + 2w)w$.
\item Cases 1, 2 and 5 of Step 2 maintain parts of the original main lobe as the
new sub-lobe at $t'$. Since the remaining area of the original main lobe is at most $7/8$, the query at $q^{\text{mid}}$ and subsequent points results in at least $1/2$ reduction of the area, due to Lemma \ref{lem:contrdown} halving along the diagonal lobe. Combined, these cases result in an overall state area equal to $15/16$ of the original main lobe.

\item 
The sub-lobe of state $t$ is either entirely removed in Step 1 or Step 2, or processed at Step 3. In this case, half of the lobe is removed by the constraints on $q^{\text{cent}}$. 

\begin{itemize}
\item If $m = 0$, the entire sub-lobe is empty.
\item If $m \geq 1$, then Cases 1 and 2 of Step 3 rule out half of the diagonal sub-lobe, due to Lemma \ref{lem:contrdown}.
\end{itemize}
\end{itemize}

So we have shown that the original main lobe and the original sub lobe have both
been reduced in area by a fraction of at least $1/16$, which implies $\area(t')
\le \frac{15}{16} \area(t)$. So if we create a diagonal-lobe state then we are
done. 

If $w> \Floor{\frac{l}{8}}$, the area of the new full state is increased. The full state consists of a central box and diagonal lobe. The original sub-lobe at state $t$ becomes the main lobe of the diagonal lobe state, thus resulting in no increase in area. However, as the diagonal lobe must be attached to the central box, this is needs to be extended to match the new lobe, thus leading to an area increase. The new central box $\left( x, l + (x_2-y_2),l+ (y_1-x_1) \right)$ corresponds to the original main lobe with two added triangles. The area of the box is then $l^2 + l \cdot w + (y_1-x_1)(x_2-y_2)$ as opposed to $l \cdot w + (y_1-x_1)(x_2-y_2)$ of the main lobe. Analysing the two quantities, we conclude that the area of the central box is at most 5 times the area of the original main lobe of the diagonal state.

\subsection{The Algorithm for Full States}
\label{sec:fullalg}

In this section we will prove the following lemma.

\begin{lemma}
\label{lem:full-state}
Let $t$ be a full-state that satisfies the invariant for some slice $s$, and whose central box $(x, h, w)$ satisfies $h \le c \cdot
w$ and $w \le c \cdot h$ for some constant $c$. There is an algorithm that
makes constantly many queries and either
\begin{itemize}
\item finds a point in $\Up(f)$,
\item declares that $\Up(f)$ is empty in the current slice, or
\item produces a new full-state $t'$ that satisfies the invariant for $s$ with $\area(t') \le 2/3 \cdot \area(t)$. 
\end{itemize}
Moreover, each step of the algorithm runs in time that is polynomial in the
representation of $f$.
\end{lemma}

\paragraph{\bf Step 1. Grid search on the central box.}

Let $(x, h, w)$ be the central box of $t$. Step 1 of the algorithm 
considers a square grid with constantly many points covering this box. 

Without loss of generality we will assume that $w \le h$, and we define $k =
\lceil w / 100 \rceil$. For technical convenience, we will consider the box
$(x, h', w')$ where $h'$ and $w'$ are the smallest integers satisfying $h' \ge
h$ and $w' \ge w$ such that both $h'$ and $w'$ are divisible by $k$. Observe
that since $\CBox(x, h, w) \subseteq \CBox(x, h', w')$, we have that the
algorithm invariant is also satisfied by $(x, h', w')$, and moreover since $w
\le h$ we have 
\begin{align*}
\area(x, h', w') &= h \cdot w + \left\lceil \frac{w}{100} \right\rceil \cdot h
+ \left\lceil
\frac{w}{100} \right\rceil \cdot w
+ \left(\left\lceil \frac{w}{100} \right\rceil\right)^2 \\
&\le \left(1 + \frac{3}{100} \right) h \cdot w = 1.03 \cdot \area(x, h, w).
\end{align*}
We use the grid 
$$R = \{x_1 + 0, x_1 + k, x_1 + 2k, \dots, x_1 + w' \}
\times \{x_2 + 0, x_2 + k, x_2 + 2k, \dots, x_2 + h' \},$$ 
meaning that $R$ is a square grid that spans
the box defined by $(x, h', w')$, where each grid square has side-length
$k$. Since by assumption we have $w \le c \cdot h$ for some constant $c$, we
also have $w' \le c' \cdot h'$ for some constant $c'$, and so $R$ contains
constantly many points. 

\config

Hence we can query every point in $R$ using constantly many queries.
If some point in $R$ lies in $\Up(f)$, then the algorithm can terminate. If we do not
find a point in $\Up(f)$, we will then look for a \emph{CB-config} in $R$,
which is a $2 \times 2$ box in $R$ satisfies the following properties.
\begin{itemize}
\item We have that $f_1$ changes direction on either the bottom or top edge of
the box, meaning that $f_1$ moves weakly upwards on the left-hand side, while
$f_1$ moves strictly downwards on the right-hand side.
\item Likewise, we have that $f_2$ changes direction on either left or right
edges of the box. 
\item The bottom-left point of the box moves strictly downwards in the third
dimension, while the top-right point of the box moves weakly upwards.
\end{itemize}
The four possible CB-configs are shown in Figure~\ref{fig:config}, and we now
provide a formal definition.

\begin{definition}[CB-Config]
\label{def:config_square}
A box $(x, 2k, 2k)$ with $x \in R$ is a \emph{CB-config} if all of the
following properties are satisfied.

\begin{enumerate}
\item For dimension 1 we have

\noindent
\begin{minipage}{.5\textwidth}
\begin{align*}
f_1(x) &\ge x_1 \\
f_1(x + 2k \cdot e_1) &< x_1 + 2k
\end{align*}
\end{minipage}
or
\begin{minipage}{.5\textwidth}
\begin{align*}
f_1(x + 2 k \cdot e_2) &\ge x_1 \\ 
f_1(x + 2 k \cdot e_1 + 2 k \cdot e_2) &< x + 2k.
\end{align*}
\end{minipage}
\item Likewise for dimension 2 we have

\noindent
\begin{minipage}{.5\textwidth}
\begin{align*}
f_2(x) &\ge x_2 \\
f_2(x + 2k \cdot e_2) &< x_2 + 2k
\end{align*}
\end{minipage}
or
\begin{minipage}{.5\textwidth}
\begin{align*}
f_2(x + 2 k \cdot e_1) &\ge x_2 \\ 
f_2(x + 2 k \cdot e_1 + 2 k \cdot e_2) &< x + 2k.
\end{align*}
\end{minipage}

\item For dimension 3 we have
\begin{align*}
f_3(x) &< x_3 \\
f_3 (x + 2\eps \cdot e_1 + 2\eps \cdot e_2) &\ge x_3.
\end{align*}

\end{enumerate}
\end{definition}


We begin by showing that a CB-config will exist so long as $\Up(f) \cap s$ is
non-empty. Recall that algorithm invariant ensures that if $\Up(f) \cap s$ is
non-empty, then the central box of $t$ contains an almost square critical box.
Since our grid search using $R$ failed to find a point in $\Up(f)$, and since
all points in a critical box are in $\Up(f)$ by Lemma~\ref{lem:cb_up}, this
places an upper bound on the side-lengths of the almost square critical box.
Specifically, if the box was larger than $k \times k$, then some point in $R$
would lie in the box, and we would have found a point in $\Up(f)$, so we know
that the box is no larger than $k - 1 \times k$ or $k \times k - 1$. 

This implies that the almost square critical box, if it exists, must lie in a
$2k \times 2k$ square of $R$. The next lemma shows that if an almost square
critical box lies inside a $2k \times 2k$ square of $R$, then that square will
be a CB-config.





\begin{lemma}\label{lem:cb_config}
Let $S = (x, 2k, 2k)$ be a square with $x \in R$. If there exists an almost square
critical box inside $S$, then $S$ is a CB-config. 
\end{lemma}  
\begin{proof}


Suppose that the almost square critical box that lies inside $S$ is defined by
$C = (y, h, w)$, and note that we therefore have $x \le y$, and $h, w \le k$.


\cbDirections

\configFromCb

The definition of a critical box gives us information about the displacements
of $f$ around $C$, as shown in Figure~\ref{fig:cb_directions}, and these then
imply displacements at the corners of $S$, as shown in
Figure~\ref{fig:config_from_cb}. The properties for dimension 1, shown in blue
in the figures, can be proved formally as follows.

\begin{align*}
f_1 (y + w \cdot e_1) \geq y_1 + w & \Rightarrow f_1(x + (y_2 - x_2)
\cdot e_2) \geq x_1 & \text{(By Lemma~\ref{lem:contrdown})} \\
& \Rightarrow f_1 (x + 2k \cdot e_2) \geq x_1 & \text{(By Lemma~\ref{lem:monodown})}\\
\\
f_1 (y + (w+1) \cdot e_1) < y_1 + (w+1) & \Rightarrow f_1(x + (y_2 -
x_2) \cdot e_2 + 2k \cdot e_1) < x_1 + 2k & \text{(By Lemma~\ref{lem:contrdown})} \\
&\Rightarrow f_1 (x + 2k \cdot e_1) < x_1 + 2k & \text{(By
Lemma~\ref{lem:monodown})} 
\end{align*}
Symmetrically, the properties for dimension 2, shown in red in the figures, can
be proved formally as follows.
\begin{align*}
f_2 (y + h \cdot e_2) \geq y_2 + h & \Rightarrow f_2(x + (y_1 - x_1) \cdot e_1)
\ge x_2 & \text{(By Lemma~\ref{lem:contrdown})} \\
&\Rightarrow f_2 (x + 2k \cdot e_1) \geq x_2 & \text{(By Lemma~\ref{lem:monodown})}\\
\\
f_2 (y + (h+1) \cdot e_2) < y_2 + (h+1) &\Rightarrow f_2(x + (y_1 - x_1) \cdot e_1 + 2k
\cdot e_2) < x_2 + 2k & \text{(By Lemma~\ref{lem:contrdown})} \\
&\Rightarrow f_2 (x + 2k \cdot e_2) < x_2 + 2k & \text{(By Lemma~\ref{lem:monodown})}
\end{align*}
Finally, the properties for dimension 3 can be shown in the following way.
\begin{align*}
f_3 (y - e_1 - e_2) < y_3 & \Rightarrow f_3 (x) < x_3 & \text{(By Lemma~\ref{lem:monodown})}\\
f_3 (y + (w+1) \cdot e_1 + (h+1) \cdot e_2)) \geq y_3 & \Rightarrow f_3(x + 2k
\cdot e_1 + 2k \cdot e_2) \geq x_3 & \text{(By Lemma~\ref{lem:monodown})}
\end{align*}

So, we have shown that the requirements imposed by a CB-config for the third
dimensions are satisfied, but we must still prove that the requirements for
dimensions 1 and 2 hold.

For dimension 1, note that we must either have $x \in \DC_1(y + w \cdot e_1)$,
which would imply $f_1(x) \ge x_1$ by Lemma~\ref{lem:contrdown}, or we have
$x + 2k \cdot e_1 + 2k \cdot e_2 \in \UC_1(y + (w + 1) \cdot e_1)$, which would imply $f_1(x +
2k \cdot e_1 + 2k \cdot e_2) < x_1 + 2k$. So in either case the CB-config properties hold for
dimension 1.

Symmetrically, for dimension 2 we either have 
$x \in \DC_2(y + h \cdot e_2)$,
which would imply $f_2(x) \ge x_2$ by Lemma~\ref{lem:contrdown}, or we have
$x + 2k \cdot e_1 + 2k \cdot e_2 \in \UC_2(y + (h + 1) \cdot e_2)$, which would imply $f_2(x + 2k \cdot e_1 + 2k \cdot e_2) < x_2 + 2k$. So in either case the CB-config properties hold for
dimension 2.
\end{proof}

The following lemma states that if we find a CB-config, then all almost square
critical boxes must fall inside a $7k \times 7k$ grid square that contains the
CB-config. 

\begin{lemma}\label{lem:config_cb}
Let $(x, 2k, 2k)$ with $x \in R$ be a CB-config, and let $(y, h, w)$ be an
almost square critical box with $h, w \le k$. For each $i \in \{1, 2\}$ we have 
$$x_i - 3k \le y_i \le x_i + 4k.$$


\end{lemma}
\begin{proof}


\configExcluded

Figure~\ref{fig:config_excluded} shows the four possible CB-configs, and also
shows the excluded regions for $\Up(f)$ that arise from
Lemmas~\ref{lem:contrdown} and~\ref{lem:monodown}. Since by
Lemma~\ref{lem:cb_up} we have that critical boxes can only contain points from
$\Up(f)$, we know that no critical box can intersect the shaded regions. Note
that in Case (a) of Figure~\ref{fig:config_excluded}, the non-shaded region is
weakly larger than the non-shaded region in any other case, so we shall proceed
assuming that we are in Case (a). 

We start by proving that $y_1 \ge x_1 - 3k$ by assuming, for the sake of
contradiction, that $y_1 < x_1 - 3k$. As shown in
Figure~\ref{fig:config_excluded}, if $y_1 < x_1$, then to prevent $(y, h, w)$
intersecting the shaded regions, we must have $y_2 \ge x_2$ and
$h \le 2k$. This then implies $w \le 2k + 1$, since the critical box is almost
square, and so 
\begin{align*}
y_1 + w + 1 &< x_1 - 3k + w + 1 \\
&\le x_1 + k + 2 \\
&\le x_1 + 1,
\end{align*}
where we used the fact that $k \ge 1$.
In particular, this means that $y_1 + w + 1 \le x_1$.
Since $(y, h, w)$ is a critical box we have $f_1(y + (w+1) \cdot e_1) < y_1 + w +
1$, and so Lemma~\ref{lem:monodown} implies that $f_1(x) < x_1$, but this
contradicts the fact that $(x, 2k, 2k)$ is a CB-config: in Cases (a) and (c) this is
immediate, while in the other two we can apply Lemma~\ref{lem:monodown} to the
point $x + 2k \cdot e_2$ to 
argue that $f_1(x) \le x_1$.
The same proof can be applied symmetrically to argue that $y_2 \ge x_2 - 3k$.

Next we prove that $y_1 \le x_1 + 4k$ by assuming, for the sake of
contradiction, that $y_1 > x_1 + 4k$. Note that to avoid intersecting with
$\UC_1(x + 2k \cdot e_1)$ we must have $y_2 > 2k$. This means that $y - e_1 - e_2
\ge x + 2k \cdot e_1 + 2k \cdot e_2$. Since a CB-config insists that $f_3(x +2k \cdot e_1 + 2k \cdot e_2) \le
x_3$, we can apply Lemma~\ref{lem:monodown} to argue that $f_3(y - e_1 - e_2) \le
x_3$, which contradicts the fact that $(y, h, w)$ is a critical box.
The same proof can be applied symmetrically to argue that $y_2 \le x_1 + 4k$.
\end{proof}

With
Lemmas~\ref{lem:cb_config}
and~\ref{lem:config_cb} in hand, we can now specify Step 1 of the algorithm.

\begin{enumerate}
\item Query every point in $R$. This takes constantly many queries.
\item If there exists a point $x \in R \cap \Up(f)$, then terminate.
\item Search for a CB-config in $R$. Since $R$ contains constantly many points,
this operation takes constant time.
\item If there is no CB-config in $R$, then the contrapositive of
Lemma~\ref{lem:cb_config} implies that there are no almost square critical
boxes in $(x, h, w)$, and so the algorithm invariant can be used to conclude
that $\Up(f) \cap s$ is empty. Hence the algorithm can correctly declare
$\Up(f) \cap s$ to be empty, and terminate.
\item Otherwise, let $(a, 2k, 2k)$ be the CB-config in $R$. Let $b = a -
3k \cdot e_1 - 3k \cdot e_2$ and observe that by Lemmas~\ref{lem:config_cb} and~\ref{lem:cb_shape} we have that if $\Up(f)
\cap s$ is non-empty, then there exists an almost square critical box in the
square $(b, 7k, 7k)$. 
\end{enumerate}
So Step 1 of the algorithm either terminates, or gives us the square $S = (b, 7k,
7k)$, which we will use to define a new, smaller, full state in Step 2.

\paragraph{\bf Step 2. Defining the new full state.}

From Step 1, we are given the square $S$, and this square
satisfies the algorithm invariant for the central box of a full state. The task
now is to build a new full state $t'$ that uses $S$ as its central
box. 

If $\Up(f)$ is non-empty, then there exists an almost square critical box $b$ by
Lemma~\ref{lem:cb_shape}, and as we have argued, this critical box must lie in
$S$. Moreover, from Lemma~\ref{lem:cbbound} we know that $\Up(f)$ is contained
within $b$ and its three lobes. Observe that since $b$ is contained in $S$, we
therefore have 
$$\Up(f) \subseteq S \cup \leftl(S) \cup \downl(S) \cup \diagl(S).$$
So our intention is to define a new state using $\leftl(S)$ as the new left
lobe, $\downl(S)$ as the new bottom lobe, and $\diagl(S)$ as the new diagonal
lobe.

However, we must deal with the existing lobes of $t$, which may not have been
fully ruled out. 
Let $(x^l, h^l, w^l)$ and $(y^l, u^l, v^l)$ be the left lobe state
associated with $t$.
For $\leftl(S)$ we use the following procedure.
\begin{enumerate}
\item First compute the intersection $I$ of $\leftl(S)$ with the central box of
$t$.
\item If $(x^l, h^l, w^l) \cap \leftl(S) = \emptyset$, then all points in the
left-lobe state are ruled out by Lemma~\ref{lem:cbbound}, we can use $I$ as
the new left-lobe state, with an empty sub-lobe. 
\item Otherwise, we remove all points in 
$(x^l, h^l, w^l)$ that are not in $\leftl(S)$ giving a smaller box $(\bar{x}^l,
\bar{h}^l, \bar{w}^l)$, and we then apply Steps 2 and 3 from the algorithm in
Section~\ref{sec:leftalg} to $I$ and 
$(\bar{x}^l, \bar{h}^l, \bar{w}^l)$. These steps spend constantly many queries,
and return a single box $I'$ while only removing points that are guaranteed to not lie in $\Up(f)$. 
\item We then remove all points in 
$(y^l, u^l, v^l)$ that are not in $\leftl(S)$ giving a smaller box 
$(\bar{y}^l, \bar{u}^l, \bar{v}^l)$.
\item We then build the left-lobe state using $I'$ as the main lobe, and 
$(\bar{y}^l, \bar{u}^l, \bar{v}^l)$ as the sub-lobe.
\end{enumerate}
We can use the analogous approach to deal with $\downl(S)$. For $\diagl(S)$ we
also use the same approach, but here we use Steps 2 and 3 from Section~\ref{sec:algdiag}.

So we can spend constantly many queries to build a new full state $t'$. The
algorithm now terminates and returns $t'$.

\paragraph{\bf Correctness.}

We have shown that $t'$ satisfies the algorithm invariant whenever $t$
satisfies that invariant. Specifically, we have shown the following.
\begin{itemize}
\item Step 1 has shown that the central box of $t'$ contains an almost square
critical box whenever the central box of $t$ contains an almost square critical
box. 
\item As argued in Step 2, if $\Up(f)$ is non-empty, then $\Up(f)$ is contained
in $S$ and its lobes. Step 2 may remove points from the lobes, but it only
removes points that are provably not in $\Up(f)$. So we have that $\Up(f)$ is
contained in $t'$. 
\end{itemize}
To complete the proof, we must show that the area of $t'$ has been sufficiently
reduced. 

We start by computing the amount of area that has been removed from the central
box of $t$. Recall that we have divided this box into a $100 \times 100$ grid
of squares of side-length $k$.
\begin{itemize}
\item Our new central box has side lengths of $7k$, so we have retained $49
k^2$ area in the new central box.
\item The new main lobe of the left-lobe state has height $7k$ and width at
most $100k$, so its area is at most $700 k^2$.
\item Likewise, the new main lobe of the bottom-lobe state has width $7k$ and
height at most $100k$, so its area is at most $700 k^2$.
\item The new main lobe of the diagonal-lobe state $(x^d, y^d, l^d)$ has $x^d_2
- y^d_2 = 7k$, and 
$y^d_1 - x^d_1 = 7k$, and also $l \le 100 k$. So its area is at most $(107^2 -
100^2) \cdot k^2 = 1449 k^2$.
\end{itemize}
So in total we have kept at most $2898 k^2$ area from the central box, and its
total area was $(100 k)^2 = 10000 k^2$. We must also remember that we increased
the area of the central box by a factor of $1.03$ at the start of the
algorithm, so in total we have kept $(2898/10000) \cdot 1.03 < 0.5$ of
the area of the central box, and so we have eliminated more than half of its
area.

For the left, bottom, and diagonal states, observe that we have applied Step 3
to any area that was not already eliminated from these states, and so as we
argued in Sections~\ref{sec:leftalg} and~\ref{sec:algdiag}, we have removed at
least $1/3$ of the area from these states.

Hence we have removed at least $1/3$ of the area from each component of $t$, so we
have that $\area(t') \le 2/3 \cdot \area(t)$. This completes the proof of
Lemma~\ref{lem:full-state}.

\subsection{Moving Between Slices}
\label{sec:jump}

When we move to a new slice, we need to transpose the existing algorithm state
into the new slice, and then reestablish the algorithm invariant. This section
describes how this is achieved. All of the lemmas in this section show that we
can produce a new invariant-satisfying state whose size \emph{increases} by at most a
constant. We will show in our running time analysis that this is not a problem,
because we only increase the area in this way whenever we jump to a new slice,
and we only jump to a new slice $\log n$ times in total. 

\paragraph{\bf Transposing a state to a new slice.}

When we move to a new slice, we translate the existing state along the
principle diagonal vector $(1, 1, 1)$. The following lemma shows that if the
original state contained the up-set in the old slice, then the translated state
will contain the up-set in the new slice. 

\begin{lemma}
\label{lem:translate}
Let $t'$ be a state in slice $s'$ where $j'$ is the fixed coordinate, and suppose
that $\Up(f) \cap s' \subseteq t'$. Let $s$ be a slice where $j$ is the fixed
coordinate with $j > j'$. If we construct a new state $t$ for $s$ in which
each box and diagonal box of $t'$ is translated by the vector $(j - j') \cdot
(1, 1, 1)$, then $\Up(f) \cap s \subseteq t$.
\end{lemma}
\begin{proof}
Let $x$ be a point in $\Up(f) \cap s$. Since $x \in \Up(f)$ we
have $x_i \le f_i(x)$ for all indices $i$. 
Now we consider the point $y = x - (j - j') \cdot (1, 1, 1)$.
\begin{itemize}
\item If $y \in G$, then
observe that $y \in \DC_i(x)$ for all $i$, and therefore we can apply
Lemma~\ref{lem:contrdown} to argue that $y$ lies in
$\Up(f) \cap s'$. 
Therefore, if $\Up(f) \cap s \not\subseteq t$, then we can use the property
above to show that $\Up(f) \cap s' \not\subseteq t'$, which would be a
contradiction.
\item If $y \not \in G$ then let $z \in G$ be the point at which the vector between
$x$ and $y$ leaves $G$. Note that we must have $z_3 > 1$, so
Lemma~\ref{lem:updownboundary} implies that $z \not \in \Up(f)$. But since $z
\in \DC_i(x)$ for all $i$, we can apply Lemma~\ref{lem:contrdown} to prove that $x$ and $z$ witness a violation of contraction, which contradicts the fact that $f$ is violation-free. 
\end{itemize}
\end{proof}

For the rest of this section we will assume that we have moved to slice $s$,
and we have produced the translated state $t$ for this slice.
Lemma~\ref{lem:translate} 
partially reestablishes the algorithm invariant for $t$, since it shows that if
the state from the previous slice satisfied the invariant, then 
$\Up(f) \cap s \subseteq t$. However, we still need to show that if $\Up(f) \ne
\emptyset$, then an almost square critical box exists in the correct boxes in
$t$, and this requires further queries, which we will now describe.

\paragraph{\bf Handling left-lobe states.}

\llnewSlice

Suppose that $t$ is a left-lobe state, and let $(x, h, w)$ be the main lobe of
this state. We make a single query at point
$q = x + h \cdot e_2$, as shown in Figure~\ref{fig:llnewslice}. There are then
two cases to consider.

\begin{itemize}
\item \textbf{Case 1: $\mathbf{f_1(q) < q_1}$.} (Figure~\ref{fig:llnewslice}
(a)). In this case we can apply Lemmas~\ref{lem:contrdown}
and~\ref{lem:monodown} to argue that every point $p$ with $p_1 \ge q_1$ and
$p_2 \le q_2$ satisfies $f_1(p) < p_1$. So $p \not\in \Up(f)$, which means that
we can use the sub-lobe of $t$ as the main lobe of a new left lobe state that
satisfies the invariant. 

\item \textbf{Case 2: $\mathbf{q_1 \le f_1(q)}$.} (Figure~\ref{fig:llnewslice}
(b)). Here we can apply Lemma~\ref{lem:contrdown} to argue that every point $p
= x - h \cdot e_1 + c \cdot e_2$ with $0 \le c \le h$ satisfies $p_1 \le
f_1(p)$. Note that any critical box can have height at most $h$, and so width
at most $h + 1$. The definition of a critical box requires that there is a
point that points downward in dimension 1 immediately to the right of the box,
which means that no critical box $(x', h', w')$ can exist with $x'_1 \le x_1 -
2h - 1$. 
Therefore, we can define a new main lobe
$l = \left(x - (2h + 1) \cdot e_1, h, w + 2h + 1\right)$, and the argument
above implies that if an almost square critical box exists, then it must lie
in~$l$. If $(y, u, v)$ denotes the sub-lobe of $v$, if $y$ is not entirely
contained within the new main lobe $l$, then we define a new
sub-lobe $a = (y, u, v - 2h - 1)$, otherwise we set $a$ to be an empty box. We
then build a new left-lobe state $t'$ using $l$ and $a$. 
\end{itemize}

Note that Case 1 does not increase the area of the state, but Case 2 does.
Specifically, in case 2 the area is increased by a factor of $(w + h + 1 +
h)/w$ and since $h \le \Floor{\frac{w}{4}}$ we have 
$(w + h + 1 + h)/w < 2$. So we have proved the following lemma.

\begin{lemma}
\label{lem:leftjump}
If $t$ is a left-lobe state satisfying $\Up(f) \subseteq t$, then we can make a
single query and produce a new left-lobe state $t'$ that satisfies the
invariant such that $\area(t') \le 2 \cdot \area(t)$. 
\end{lemma}

\paragraph{\bf Handling bottom-lobe states.}

Bottom-lobe states can be handled in the same way as left-lobe states.
Specifically, we can simply exchange dimensions 1 and 2, and then apply the
algorithm given above. So we have the following lemma.

\begin{lemma}
\label{lem:bottomjump}
If $t$ is a bottom-lobe state satisfying $\Up(f) \subseteq t$, then we can make a
single query and produce a new bottom-lobe state $t'$ that satisfies the
invariant such that $\area(t') \le 2 \cdot \area(t)$. 
\end{lemma}

\paragraph{\bf Handling diagonal-lobe states.}

\dlnewslice

Suppose that $t$ is a diagonal-lobe state, and let $(x, y, l)$ be the main lobe of
this state. We define the width of the main lobe as $w = (y_1-x_1) +
(x_2-y_2)$, similar to the analysis used in Lemma~\ref{lem:diaglobestate}. We make a single query at point
$q = x + (l + (y_1 - x_1) \cdot e_1 + l \cdot e_2$, as shown in Figure \ref{fig:dl_new_slice}. We then consider the following two cases.

\begin{itemize}
\item \textbf{Case 1: $\mathbf{f_3(q) < q_3}$.} (Figure \ref{fig:dl_new_slice}
(a)). Applying Lemma~\ref{lem:monodown}, we observe that every point $p$ with $p_1 \leq q_1$ and $p_2 \leq q_2$ satisfies $f_3(p) < p_3$. So $p \not\in \Up(f)$, which rules out the entire main lobe, and the sub-lobe of $t$ is used as the main lobe of a new diagonal lobe state,
satisfying the invariant. 

\item \textbf{Case 2: $\mathbf{f_3(q) \geq q_3}$.} (Figure \ref{fig:dl_new_slice}
(b)). By Lemma \ref{lem:monodown}, we observe that every point $p$ with $p_1 \geq q_1$ and $p_2 \geq q_2$ satisfies $f_3(p) \geq p_3$. Since any any critical box can have width at most $w$, an almost square critical box has height bounded by $w + 1$. By definition of a critical box, the diagonally adjacent point to the bottom left corner must be pointing downwards in dimension 3, which means that no critical box can have its bottom left corner below $q$. The above arguments imply that no critical box can lie further than $2w$ from point $q$.

The new main lobe is extended and defined as 
$(x, y, l + 2w)$, and, by the above reasoning, if an almost square critical box exists, it must lie in the new main lobe. If the sub-lobe denoted as $(a, b, m)$ does not lie entirely within the new main lobe, we can define a new
sub-lobe $(a+ 2w \cdot e_1 + 2w \cdot e_2, a+ 2w \cdot e_1 + 2w \cdot e_2, m - 2w)$, otherwise the new sub-lobe is an empty box. We
then build a new left-lobe state $t'$ using the new main and sub-lobes. 
\end{itemize}

We observe that, in Case 1, the overall area is reduced, while Case 2 leads to an increase of the area of the main lobe.

Specifically, in case 2 the area is increased by a factor of $(l+2w)/l$ and since $w \le \Floor{\frac{l}{8}}$ we have 
$(l+2w)/l < 2$. Our analysis proves the following lemma.

\begin{lemma} 
\label{lem:diagjump}
If $t$ is a diagonal-lobe state satisfying $\Up(f) \subseteq t$, then we can make a
single query and produce a new diagonal-lobe state $t'$ that satisfies the
invariant such that $\area(t') \le 2 \cdot \area(t)$. 
\end{lemma}

\paragraph{\bf Handling full states.}

\fsnewslice

Finally, we consider the case where $t$ is a full state. Let $(x, h, w)$ be the
central box of $t$.
We query the following three points.
\begin{align*}
q^{\text{br}} &= x + w \cdot e_1 \\
q^{\text{tl}} &= x + h \cdot e_2 \\
q^{\text{tr}} &= x + w \cdot e_1 + h \cdot e_2.
\end{align*}
These are the bottom-right, top-left, and top-right points of the central box,
as depicted in Figure~\ref{fig:fsnewslice}. We then perform the following case
analysis.

\begin{itemize}
\item \textbf{Case 1: $\mathbf{f_1(q^{\text{tl}}) < q^{\text{tl}}_1}$.} Using the
same argument as we used for Case 1 of handling a left-lobe state, we can
conclude that all points $x \in \Up(f)$ satisfy $x_1 < q^{\text{tl}}_1$, which
rules out everything except the left-lobe state of $t$. So we can proceed using
Lemma~\ref{lem:leftjump} on the left-lobe of $t$ to produce a new left-lobe
state that satisfies the invariant. 

\item \textbf{Case 2: $\mathbf{f_2(q^{\text{br}}) < q^{\text{br}}_2}$.}
Likewise we can use the same argument as we used for Case 1 of handling a
bottom-lobe state, so we can 
proceed using
Lemma~\ref{lem:bottomjump} on the bottom-lobe of $t$ to produce a new bottom-lobe
state that satisfies the invariant. 

\item \textbf{Case 3: $\mathbf{f_3(q^{\text{tr}}) < q^{\text{tr}}_3}$.} Here we use the argument for Case 1 of handling a diagonal-lobe state, which states that any point with $x_1 \leq q^{\text{tr}}_1$ and $x_2 \leq q^{\text{tr}}_2$ has $x \not\in \Up(f)$. Applying Lemma~\ref{lem:diagjump} on the diagonal-lobe of $t$, we can produce a new diagonal lobe state satisfying the invariant. 

\item \textbf{Case 4.} If none of the previous cases apply then we have 
$\mathbf{f_1(q^{\text{tl}}) \ge q^{\text{tl}}_1}$, and
$\mathbf{f_2(q^{\text{br}}) \ge q^{\text{br}}_2}$, and
$\mathbf{f_3(q^{\text{tr}}) \ge q^{\text{tr}}_3}$, which is the situation shown
in Figure~\ref{fig:fsnewslice}.
Here we note the following properties.
\begin{itemize}
\item Since $f_1(q^{\text{tl}}) \ge q^{\text{tl}}_1$, we can use the same
argument as we used in Case 2 of handling a left-lobe state to argue that any
almost square critical box $(x', h', w')$ must satisfy $x'_1 \ge x_1 - 2h - 1$.

\item Since $f_2(q^{\text{br}}) \ge q^{\text{br}}_2$
we can use the same
argument as we used in Case 2 of handling a bottom-lobe state to argue that any
almost square critical box $(x', h', w')$ must satisfy $x'_2 \ge x_2 - 2w - 1$.

\item 
Since $f_3(q^{\text{tr}}) \ge q^{\text{tr}}_3$ we can use the same argument as
we used in Case 2 of handling a diagonal-lobe state to argue that any
almost square critical box $(x', h', w')$ must satisfy $x'_1 \le x_1$ or $x'_2
\le x_2$, and also $h' \le 2w$ and $w' \le 2h$. 
\end{itemize}

So we can define an expanded central box $b = (x - (2h + 1, 2w + 1), 5h +
1, 5w + 1 )$, and the arguments above imply that all
almost square critical boxes must lie in $b$. 
We then consider the left, bottom, and diagonal states in $t$. If they have
been entirely subsumed by $b$ then we replace them with an empty state.
Otherwise we shorten them so that they meet $b$, rather than the original
central box. 

We then build a new full state $t'$ from $b$ and the new left, bottom,
and diagonal states. Since $t'$ contains a superset of the points in $t$, the
algorithm invariant holds for $t'$. Moreover, the area of the central box of
$t'$ is at most $36$ times the area of the central box of $t$, while the
other states contained in $t'$ have not increased in area. So we have
$\area(t') \le 36 \cdot \area(t)$. 
\end{itemize}
Hence we have shown the following lemma.

\begin{lemma}
\label{lem:fulljump}
If $t$ is a full state satisfying $\Up(f) \subseteq t$, then we can make
constantly many queries to $f$ and produce a new state $t'$ that satisfies the
invariant such that $\area(t') \le 36 \cdot \area(t)$. Moreover each step of
the algorithm runs in time that is polynomial in the representation of $f$.
\end{lemma}

\subsection{The Main Algorithm}
\label{sec:mainalgo}

In this section we now piece together all of the properties that we have shown
so far into a single algorithm.

\paragraph{\bf The corresponding down-set definitions and algorithms.}

In Sections~\ref{sec:cb} through~\ref{sec:jump} we presented our definitions
and algorithms for the up-set. Everything that we presented in those sections
can also be applied to the down-set by simply flipping all dimensions. So
for example, a down-set critical box $(x, h, w)$ has $x \in \Down(f)$, and $x$
is the top-right corner of the box, while all other conditions are likewise
flipped. A down-set critical box has a right lobe, a
top lobe, and a down-left diagonal lobe. 

The algorithms can likewise be adapted by flipping all dimensions. The
algorithm invariant now insists that the state must contain $\Down(f) \cap s$,
and that every almost square down-set critical box is contained in the central
box or main lobe of the state. The algorithms take a down-set state and make
constantly many queries before either finding a point in $\Down(f)$, declaring
that $\Down(f) \cap s$ is empty in the slice $s$, or finding a new states whose
area has been reduced by a constant fraction.

\paragraph{\bf The algorithm. }



The algorithm performs a binary search on the third dimension of the instance.
Each step of the binary search considers a slice in which the third dimension
is fixed. The goal of this binary search is to find a slice $s$ that contains
a point in $\Up(f)$ and a point in $\Down(f)$, which is guaranteed to exist
since any fixed point lies in $\Up(f) \cap \Down(f)$. Once we have found such a
slice, we then move to the terminal phase of the algorithm, which finds a fixed
point in $s$.


Given an integer $i$, let $s_i$ be the slice that contains all points $x$ with
$x_3 = i$. The algorithm maintains two integers $l$ and $u$ with the invariant
that $\Up(f) \cap s_l \ne \emptyset$ and $\Down(f) \cap s_u \ne \emptyset$.
Initially we set $l = 1$ and $u = n$, so this invariant is satisfied
because $(1, 1, 1) \in \Up(f)$ and $(n, n, n) \in \Down(f)$ both hold
trivially.

The algorithm also maintains a \emph{current slice} defined by an index $i$,
where initially we set $i = \Floor{\frac{u - l}{2}}$.
For each slice $s_i$ visited, the algorithm maintains an up-set state
$t_\text{up}^i$ and a down-set state $t_\text{dn}^i$, with the property that
$t_\text{up}^i$ satisfies the up-set invariant for $s_i$, and $t_\text{dn}^i$
satisfies the down-set invariant for $s_i$. Initially we set $t_\text{up}^{i}$
to be a full state with the central box $((1, 1), n, n)$ and all other states
being empty. Observe that this trivially satisfies the up-set algorithm
invariant for slice $s_i$. The state $t_\text{dn}^{i}$ is likewise initially
set to a full state with central box $((n, n), n, n)$ with all other states
being empty, which also trivially satisfies the down-set algorithm invariant
for slice $s_i$. 

The main algorithm then proceeds as follows.

\begin{enumerate}
\item The algorithm performs a single step on $t_\text{up}^{i}$ by applying one of
Lemmas~\ref{lem:leftlobestate},~\ref{lem:bottomlobestate},~\ref{lem:diaglobestate},
or~\ref{lem:full-state}, depending on the type of $t_\text{up}^{i}$.
If $t_\text{up}^{i}$ contains constantly many points, we instead just search
these points to determine if there exists a point in $\Up(f) \cap s_i$ or not.


\item The algorithm performs a single step on $t_\text{dn}^{i}$ by applying one of
the flipped analogues of
Lemmas~\ref{lem:leftlobestate},~\ref{lem:bottomlobestate},~\ref{lem:diaglobestate},
or~\ref{lem:full-state}, depending on the type of $t_\text{dn}^{i}$.
If $t_\text{dn}^{i}$ contains constantly many points, we instead just search
these points to determine if there exists a point in $\Down(f) \cap s_i$ or not.


\item If both Step 1 and Step 2 produce a new state $t_\text{up}'$ and
$t_\text{dn}'$, respectively,
then we set $t_\text{up}^{i} = t_\text{up}'$ and $t_\text{dn}^{i} = t_\text{dn}'$. We then go back to Step
1 and repeat.






\item 
\label{itm:jump}
If Step 1 finds a point $x \in \Up(f) \cap s_i$ or if Step 2 declares
$\Down(f) \cap s_i = \emptyset$, then we set $l = i$. 
Likewise, if Step 2 finds a point $y \in \Down(f) \cap s_i$ or if Step 1
declares $\Up(f) \cap s_i = \emptyset$, then we set $u = i$.
Note that this covers all possible return values from the two algorithms.

\item 
If $u - l = 1$ then the main algorithm terminates, and we move to the two-slice
sub-algorithm, which will be described later. 

\item 
Otherwise we set $i =
\Floor{\frac{u - l}{2}}$, and carry out the following steps.
\begin{itemize}
\item Let $t_\text{up}$ be the result of taking the state
$t_\text{up}^{l}$, translating every box and diagonal box by $(i -
l) \cdot (1, 1, 1)$, and then applying the appropriate algorithm from
Lemmas~\ref{lem:leftjump},~\ref{lem:bottomjump},~\ref{lem:diagjump},
and~\ref{lem:fulljump}. Note that Lemma~\ref{lem:translate} implies that
$t_\text{up}$ satisfies the algorithm invariant for $s_i$, so we can set
$t_\text{up}^{i} = t_\text{up}$. 

\item Let $t_\text{dn}$ be the result of taking the state $t_\text{dn}^{u}$, translating every
box and diagonal box by $(i - u) \cdot (1, 1, 1)$, and then applying the
appropriate algorithm from the flipped analogues of
Lemmas~\ref{lem:leftjump},~\ref{lem:bottomjump},~\ref{lem:diagjump},
and~\ref{lem:fulljump}. Note that the flipped analogue of
Lemma~\ref{lem:translate} implies that $t_\text{dn}$ satisfies the algorithm invariant
for $s_i$, so we can set $t_\text{dn}^{i} = t_\text{dn}$. 

\item We then move back to Step 1 and continue in the new slice $i$.
\end{itemize}

\end{enumerate}

\paragraph{\bf Correctness of the main algorithm.}

We have already argued that each state visited by the algorithm satisfies the
appropriate algorithm invariant. We still need to show the invariant that
$\Up(f) \cap s_l \ne \emptyset$ and $\Down(f) \cap s_u \ne \emptyset$, where we
must inspect the decisions made in Step~\ref{itm:jump} of the algorithm, which
is the only step that modifies $l$ and $u$. Here if the up-set algorithm found
a point $x \in \Up(f) \cap s_i$, or if the down-set algorithm found point $y
\in \Down(f) \cap s_i$, then the invariant is trivially satisfied. For the
other two cases, we rely on the following lemma.


\begin{lemma}
\label{lem:upordown}
For every slice $s_i$, we have either $\Up(f) \cap s_i \ne \emptyset$ or $\Down(f)
\cap s_i \ne \emptyset$.
\end{lemma}
\begin{proof}
Consider the function $f' : \{1, 2, \dots, n\} ^2 \rightarrow 
\{1, 2, \dots, n\}^2$ defined so that $f'_j(x) = f(x \oplus i)$ for $j \in \{1, 2\}$,
where $x \oplus i$ places $i$ in the third dimension of $x$. 
The function $f'$ is monotone, and so possesses a fixed point $p'$ by Tarski's
theorem. If we let $p = p' \oplus i$, then since $p$ is fixed in dimensions 1 and 2, if
$p_3 \le f_3(p)$ then $p \in \Up(f)$, while if $p_3 \ge f_3(p)$ then $p \in
\Down(f)$. 
\end{proof}

Hence, when Step~\ref{itm:jump} sets $l = i$ because the down-set algorithm
declared that $\Down(f) \cap s_i = \emptyset$ Lemma~\ref{lem:upordown} tells us
that $\Up(f) \cap s_i \ne \emptyset$, so the invariant is maintained. The case
where Step~\ref{itm:jump} sets $u = i$ can be shown to be correct using the
same argument.

\paragraph{\bf The running time of the main algorithm.}

Each step of the algorithm reduces the size of the state by a constant fraction
except in the following two scenarios.
\begin{enumerate}
\item 
When a diagonal-lobe state is turned into a full state in the last case of
Lemma~\ref{lem:diaglobestate}, where the area of the state is increased by a
factor of $5$.

\item When we transition to a new slice, where the area is increased by a
factor of at most $36$.
\end{enumerate}
However, note that the first case can occur at most once per slice, because
once we have arrived at a full state Lemma~\ref{lem:full-state} ensures that
all further states that are considered in this slice are also full states.
Meanwhile, the second case also occurs at most once per slice, since it only
occurs when we transition to a new slice. So for each slice that we visit, we
may blow-up the area of the state by a factor of $180$.


On the other hand, each other step reduces the area of the state by a factor of
at least $15/16$. It can be verified that $(15/16)^{81} \cdot 180 < 1$, so
we can cancel out the increase from each slice by taking 81 
further steps. Furthermore, we visit exactly $\log n$ different slices in our
binary search, so we
pay $81 \cdot \log n$ steps in order to cancel out these increases.

Then, after $O(\log n)$ further steps the number of points contained in the state
must be a constant, after which the algorithm terminates after at most $O(\log
n)$ further steps to complete the binary search. So in total we have that the
main-algorithm terminates after making $O(\log n)$ queries to $f$. 
Moreover, every step of the algorithm runs in time that is polynomial in the
representation of $f$.

\subsection{The Two-Slice Sub-Algorithm}

When we reach the two-slice sub-algorithm, we have that $u - l = 1$. 
The following lemma shows that the main algorithm invariant ensures that either
$s_u$ contains a fixed point, or $s_l$ contains a fixed point. 

\begin{lemma}
\label{lem:fpexists}
Let $u, l$ be slice indices satisfying $u = l$ or $u = l + 1$. 
If $\Up(f) \cap s_l \ne \emptyset$ and $\Down(f) \cap s_u \ne \emptyset$
then there exists a fixed point in $s_u$ or $s_l$.  
\end{lemma}
\begin{proof}
Let $x \in \Up(f) \cap s_l$ whose existence is assumed by the lemma.
By Tarski's theorem we know that the greatest fixed point $g$ satisfies $g \ge
x$. If $g \in s_u$ or $g \in s_l$ then we are done, so we proceed assuming that $g_3
> u$. 

Symmetrically, let 
$y \in \Down(f) \cap s_u$ whose existence is assumed by the lemma.
By Tarski's theorem we know that the least fixed point $p$ satisfies $p \le
y$. If $p \in s_u$ or $p \in s_l$ then we are done, so we proceed assuming that $p_3
< l$.

By Corollary~\ref{cor:neighboring-lesser-fp} there exists a path of fixed points from $p$ to $g$
where each step of the path moves by distance at most $1$ in the
$\ell_\infty$-norm. Since $p$ and $g$ lie on either side of $s_l$ and $s_u$, the
path must pass through both slices, and so in this case there exists a fixed
point in both $s_l$ and $s_u$. 
\end{proof}

This means that there exists an index $i \in \{u, l\}$ such that $\Up(f) \cap
s_i \ne \emptyset$ and $\Down(f) \cap s_i \ne \emptyset$. The goal of the
sub-algorithm is to find this slice. 

The main algorithm
invariant ensures that $\Up(f) \cap s_l \ne \emptyset$ and that $\Down(f) \cap
s_u \ne \emptyset$.
The sub-algorithm simply keeps refining the down-set state in the slice $s_l$
until it either finds a down-set point, which means that $s_l$ is the slice that we
are looking for, or declares that $s_l$ does not contain any down-set points,
which then allows us to invoke Lemma~\ref{lem:fpexists} to conclude that $s_u$ is
the slice that we are looking for.

Formally, the sub-algorithm proceeds as follows.

\begin{enumerate}
\item Compute $t^l_\text{dn}$ by taking $t_\text{dn}^{u}$ from the main
algorithm, translating every box and diagonal box by $(-1, -1, -1)$, and then
applying the appropriate algorithm from the flipped analogues of
Lemmas~\ref{lem:leftjump},~\ref{lem:bottomjump},~\ref{lem:diagjump},
and~\ref{lem:fulljump}. 
If at any point the state $t$ contains a
constant number of points, then we simply search those points to determine if
there is a point in $\Up(f) \cap s_l$.

\item 
\label{itm:go}
Repeatedly apply 
the flipped analogues of
Lemmas~\ref{lem:leftlobestate},~\ref{lem:bottomlobestate},~\ref{lem:diaglobestate},
or~\ref{lem:full-state}, to 
$t^l_\text{dn}$ in slice $s_l$ until an algorithm either finds a point $x \in
\Down(f) \cap s_l$ or declares 
$\Down(f) \cap s_l = \emptyset$. If at any point the state $t$ contains a
constant number of points, then we simply search those points to determine if
there is a point in 
$\Down(f) \cap s_l$.

\item 
If the point $x$ is found in Step~\ref{itm:go}, then we have that 
$\Up(f) \cap s_l \ne \emptyset$ and $\Down(f) \cap s_l \ne \emptyset$, so the
sub-algorithm can terminate returning $l$

\item Otherwise, Lemma~\ref{lem:fpexists} ensures that 
$\Up(f) \cap s_u \ne \emptyset$ and $\Down(f) \cap s_u \ne \emptyset$, so the
sub-algorithm can terminate returning $u$.

\end{enumerate}

\paragraph{\bf Running time of the sub-algorithm.}

We can re-use the analysis of the main algorithm to show that the two-slice
sub-algorithm terminates after $O(\log n)$ steps. In particular, since we only
consider two different slices, we blow-up the area by a factor of 180 at most
twice, and so we spend at most $162$ steps cancelling out that
increase. Then after $O(\log n)$ further steps both states must contain
constantly many points, and so the algorithm will terminate. Moreover, every
step of the algorithm runs in time that is polynomial in the representation of
$f$.

\subsection{The Terminal Phase}

In the terminal phase of the algorithm we have a slice $s_i$ such that $\Up(f)
\cap s_i \ne \emptyset$ and $\Down(f) \cap s_i \ne \emptyset$.
Lemma~\ref{lem:fpexists} tells us that there is a fixed point in $s_i$, and the
goal of the terminal phase is to find this fixed point in $\log(n)$ steps.

\paragraph{\bf Paths of two-dimensional fixed points.}



We start by considering the two-dimensional fixed points of the slice $s_i$,
where a point $x$ is a two-dimensional fixed point if $x_1 = f_1(x)$ and $x_2 =
f_2(x)$. Recall that Corollary~\ref{cor:neighboring-lesser-fp} implies that, if $x$ is a
two-dimensional fixed point that is not the least two-dimensional fixed point,
then we can find a point $x' \le x$ with $\| x - x' \|_\infty  = 1$ such
that $x'$ is also a two-dimensional fixed point. We can use this fact to prove
the following lemma, which states that we can build a path of two-dimensional
fixed points from $x$ to the least two-dimensional fixed point of the slice, and also a second path from $x$
to the greatest two-dimensional fixed point of the slice.

\begin{lemma}
	\label{lem:paths}
	Let $x$ be a two-dimensional fixed point of a slice $s$, and let $p$ and $q$ be the least and
	greatest two-dimensional fixed points of $s$, respectively. 
	\begin{itemize}
		\item There exists a path of two-dimensional fixed points $p = v^1, v^2, \dots,
		v^k = x$ with $v^i \le v^{i+1}$ and $\| v^i - v^{i+1} \|_\infty = 1$ for all
		$i$.
		\item There exists a path of two-dimensional fixed points $x = u^1, u^2, \dots,
		u^k = q$ with $u^i \le u^{i+1}$ and $\| u^i - u^{i+1} \|_\infty = 1$ for all
		$i$.
	\end{itemize}
\end{lemma}
\begin{proof}
	The path from $p$ to $x$ can be obtained by repeatedly applying
	Corollary~\ref{cor:neighboring-lesser-fp} starting at $x$. Each application of the corollary to
	a point $v^{i+1}$ gives us a new two-dimensional fixed point $v^{i}$ with $v^i
	\le v^{i+1}$ and $\| v^i - v^{i+1} \|_\infty = 1$, and we can keep applying
	the corollary until we reach $p$. 
	
	The path from $x$ to $q$ can be likewise generated using
	Corollary~\ref{cor:neighboring-lesser-fp}, but here we must first flip dimensions $1$ and $2$
	before proceeding. In the flipped instance we start at $x$, and then each
	application of the corollary to a point $v^i$ gives us a new two-dimensional
	fixed point $v^{i+1}$
	with $v^i \ge v^{i+1}$ and $\| v^i - v^{i+1} \|_\infty = 1$ in the flipped
	instance, which means that 
	$v^i \le v^{i+1}$ and $\| v^i - v^{i+1} \|_\infty = 1$ in the original
	instance.
	We can keep applying the corollary until we reach $q$, which is the least
	two-dimensional fixed point in the flipped instance.
\end{proof}

The next lemma states that if the slice contains a global fixed point (that is,
a point that is fixed in all three dimensions), then the least and greatest
fixed points must obey certain inequalities in the third dimension.

\begin{lemma}
	\label{lem:signs}
	Let $s$ be a two-dimensional slice, and let $p$ and $q$ be the least and
	greatest two-dimensional fixed points of $s$, respectively. If $s$ contains a
	global fixed point, then we have 
	$p_3 \ge f_3(p)$ and $q_3 \le f_3(q)$.
\end{lemma}
\begin{proof}
	Note that every global fixed point in $s$ is also a two-dimensional fixed point
	of $s$. So if we had $p_3 < f_3(p)$ then Lemma~\ref{lem:monodown} would imply
	that every two-dimensional fixed point $v$ of $s$ satisfies $v_3 < f_3(v)$,
	contradicting the fact that $s$ contains a global fixed point. Likewise, if we
	had $q_3 > f_3(q)$ then Lemma~\ref{lem:monodown} would imply that every
	two-dimensional fixed point $v$ of $s$ satisfies $v_3 > f_3(v)$, again
	contradicting the fact that $s$ contains a global fixed point.
\end{proof}

Finally we prove the following lemma, which states that if we walk along any
path of two-dimensional fixed points like those whose existence is asserted in
Lemma~\ref{lem:paths}, and if the end points of the path satisfy conditions
like those proven in Lemma~\ref{lem:signs}, then the path must visit a global
fixed point.

\begin{lemma}
	\label{lem:pathfp}
	Let $v^1, v^2, \dots, v^k$ be a path of two-dimensional fixed points
	satisfying $\| v^i - v^{i+1} \|_\infty = 1$ for all $i$.
	If $v^1_3 \ge f_3(v^1)$ and $v^k_3 \le f_3(v^k)$, then there exists an $i$ such
	that $v^i$ is a global fixed point.
\end{lemma}
\begin{proof}
	First note that if $v^1_3 = f_3(v^1)$ or $v^k_3 = f_3(v^k)$ then one of
	$v^1$ or $v^k$ is a global fixed point and we are done. So we can assume that 
	$v^1_3 > f_3(v^1)$ and $v^k_3 < f_3(v^k)$.
	
	Since $v^1_3 - f_3(v^1) > 0$ and $v^k_3 - f_3(v^k) < 0$ there must be some edge
	in the path at which $v^i_3 - g^i_3(v)$ changes sign. Note that there cannot be
	two points $a$ and $b$ with $\| a - b \|_\infty = 1$ with $a_3 > f_3(a)$ and
	$b_3 < f_3(b)$ since this would imply $\| f(a) - f(b) \|_\infty > \| a-b \|_\infty$, giving
	us an immediate violation of strict contraction. Hence there must exist some point
	$v^i$ such that $v^i_3 = f_3(v^i)$. Since $v^i$ is a two-dimensional fixed
	point, we therefore have that $v^i$ is a global fixed point.
\end{proof}

Lemma \ref{lem:pathfp} implies that for every global fixed point $x$, there exists a
path from the least two-dimensional fixed point of the slice to the greatest
two-dimensional fixed point of the slice that passes through $x$. We will use
this fact crucially in our algorithm.

\paragraph{\bf Spaces ruled out by a single query.}

Our algorithm will maintain a region $R$ with the invariant that $R$ contains a
global fixed point. It will then make queries in order to reduce the size of
$R$. The following lemma shows that each query we make will rule out a
half-space of points.

\begin{lemma}
\label{lem:halfspace}
Let $R \subseteq s_i$ be a region that contains a global fixed point, and
suppose that we query point $q \in R$ that is not a global fixed point. 
There exists a half-space of one of the following forms such that, when we
remove the half-space from $R$ creating $R'$, we have that $R'$ contains a
global fixed point.
\begin{itemize}
\item A vertical half-space ruling out all points $p$ with either $p_1 \le
q_1$, $p_1 \ge q_1$, $p_1 \le q_1 - 1$, or $p_1 \ge q_1 + 1$. 
\item A horizontal half-space ruling out all points $p$ with either $p_2 \le
q_2$, $p_2 \ge q_2$, $p_2 \le q_2 - 1$, $p_2 \ge q_2 + 1$. 
\item An upper-left diagonal half-space ruling out all points $p$ with $p_2 -
q_2 \ge p_1 - q_1$. 
\item A lower-right diagonal half-space ruling out all points $p$ with $p_2 -
q_2 \le p_1 - q_1$.
\end{itemize}
\end{lemma}
\begin{proof}
To prove this lemma, we enumerate the possible responses to the query at $q$,
and in each case we show that a half-space can be ruled out. We begin with the
responses in which $f$ moves strictly in dimensions 1 and 2. In each case we
apply Lemmas~\ref{lem:contrdown} and~\ref{lem:monodown} to find a region of
points that must move strictly in either dimension 1 or 2, which allows us to
conclude that no point in those regions is a global fixed point. Hence, since
$R$ contains a global fixed point, it must lie outside the half-space, so we
can safely remove all points in the half-space.

\halfspaceFirst

\begin{itemize}
\item \textbf{Case 1: $\mathbf{q_1 < f_1(q)}$ and $\mathbf{q_2 < f_2(q)}$.}
(Figure~\ref{fig:case1-4} (a)). The figure shows the two excluded regions
implied by Lemmas~\ref{lem:contrdown} and~\ref{lem:monodown}. As we can see,
in this case we can rule out the vertical half-space defined by $p$ with $p_1
\le q_1$ and also the horizontal half-space defined by $p$ with $p_2 \le q_2$.

\item \textbf{Case 2: $\mathbf{q_1 > f_1(q)}$ and $\mathbf{q_2 < f_2(q)}$.}
(Figure~\ref{fig:case1-4} (b)). The figure shows the two excluded regions
implied by Lemmas~\ref{lem:contrdown} and~\ref{lem:monodown}. As we can see, in
this case we can rule out the lower-right diagonal half-space defined by $p$
with $p_2 - q_2 \le p_1 - q_1$.

\item \textbf{Case 3: $\mathbf{q_1 < f_1(q)}$ and $\mathbf{q_2 > f_2(q)}$.}
(Figure~\ref{fig:case1-4} (c)). The figure shows the two excluded regions
implied by Lemmas~\ref{lem:contrdown} and~\ref{lem:monodown}. As we can see,
in this case we can rule out the upper-left diagonal half-space defined by $p$
with $p_2 - q_2 \ge p_1 - q_1$.

\item \textbf{Case 4: $\mathbf{q_1 < f_1(q)}$ and $\mathbf{q_2 < f_2(q)}$.}
(Figure~\ref{fig:case1-4} (d)). The figure shows the two excluded regions
implied by Lemmas~\ref{lem:contrdown} and~\ref{lem:monodown}. As we can see,
in this case we can rule out the horizontal half-space defined by $p$ with $p_2
\ge q_2$ and the vertical half-space defined by $p$ with $p_1 \ge q_1$.
\end{itemize}

\halfspaceSecond

Now we consider the cases in which $f_1(q)$ moves strictly upward, while
$f_2(q) = q_2$. Here we rely on the third dimension to give us our half-space.
As before, in each case we apply Lemmas~\ref{lem:contrdown}
and~\ref{lem:monodown} to find a region of points that must move strictly in
either dimension 1 or 3, which allows us to conclude that no point in those
regions is a fixed point, so the regions can be safely removed from $R$.

\begin{itemize}
\item \textbf{Case 5: $\mathbf{q_1 < f_1(q)}$ and $\mathbf{q_3 < f_3(q)}$.}
(Figure~\ref{fig:case5-7} (a)). The figure shows the two excluded regions
implied by Lemmas~\ref{lem:contrdown} and~\ref{lem:monodown}. As we can see, in
this case we can rule out the horizontal half-space defined by $p$ with $p_2
\ge q_2$. 

\item \textbf{Case 6: $\mathbf{q_1 < f_1(q)}$ and $\mathbf{q_3 > f_3(q)}$.}
(Figure~\ref{fig:case5-7} (b)). The figure shows the two excluded regions
implied by Lemmas~\ref{lem:contrdown} and~\ref{lem:monodown}. As we can see, in
this case we can rule out the vertical half-space defined by $p$ with $p_1
\le q_1$. 

\item 
\textbf{Case 7: $\mathbf{q_1 < f_1(q)}$ and $\mathbf{q_2 = f_2(q)}$ and
$\mathbf{q_3 = f_3(q)}$.}
(Figure~\ref{fig:case5-7} (c)). The figure shows the two excluded regions
implied by Lemmas~\ref{lem:contrdown} and~\ref{lem:monodown}.
Here a half-space is not immediately ruled out, so a more involved proof is required.
We show that we can rule out the half-space defined by points $p$ with $p_1 \le
q_1 - 1$. We do this by considering the location of the global fixed point $x
\in R$, whose existence the lemma assumes. 
If $x_1 \ge q_1$ then we are done. 

So let us suppose that $x_1 < q_1$. We will show that there exists another
global fixed point $x'$ satisfying $x'_1 \ge q_1$. First we apply
Lemma~\ref{lem:paths} to obtain the path
$x = p^1, p^2, \dots, p^k = g$ from $x$ to the greatest two-dimensional fixed
point $g$
of the slice, and we observe 
that the path cannot pass through the
blue region in Figure~\ref{fig:case5-7} (c), which was obtained by applying
Lemmas~\ref{lem:contrdown} and~\ref{lem:monodown},  because all points in that region
move strictly upward in dimension 1, and all points on the path are two-dimensional
fixed points. 

Moreover, we note that $q$ is in the two-dimensional up-set of the slice, and
so by Tarski's theorem we have $q \le g$. As shown (Figure~\ref{fig:case5-7}
(c)), the only way to move from $x$ to $g$ while monotonically increasing in
every step and avoiding the already excluded region is to pass through some
point $x'$ with $x'_1 = q_1$ and $x'_2 \le q_2$. 

Note that we cannot have $x'_3 > f_3(x')$, since then Lemma~\ref{lem:monodown}
would imply that $x_3 > f_3(x)$ which would imply that $x$ is not a global fixed
point. We also cannot have $x'_3 < f_3(x')$, because then
Lemma~\ref{lem:monodown} would imply that $q_3 < f_3(q)$, which contradicts the
assumptions made in this case. Hence we must have 
$x'_3 = f_3(x')$, and so $x'$ is a global fixed point satisfying $x'_1 \ge
q_1$, as required.
\end{itemize}

The case where $f_1(q)$ moves strictly downward can be dealt with
symmetrically.

\halfspaceThird

\begin{itemize}
\item \textbf{Case 8: $\mathbf{q_1 > f_1(q)}$ and $\mathbf{q_3 < f_3(q)}$.}
(Figure~\ref{fig:case8-10} (a)). The figure shows the two excluded regions
implied by Lemmas~\ref{lem:contrdown} and~\ref{lem:monodown}. As we can see, in
this case we can rule out the vertical half-space defined by $p$ with $p_1
\ge q_1$. 

\item \textbf{Case 9: $\mathbf{q_1 > f_1(q)}$ and $\mathbf{q_3 > f_3(q)}$.}
(Figure~\ref{fig:case8-10} (b)). The figure shows the two excluded regions
implied by Lemmas~\ref{lem:contrdown} and~\ref{lem:monodown}. As we can see, in
this case we can rule out the horizontal half-space defined by $p$ with $p_2
\le q_2$. 

\item 
\textbf{Case 10: $\mathbf{q_1 > f_1(q)}$ and $\mathbf{q_2 = f_2(q)}$ and
$\mathbf{q_3 = f_3(q)}$.}
(Figure~\ref{fig:case8-10} (c)). The figure shows the two excluded regions
implied by Lemmas~\ref{lem:contrdown} and~\ref{lem:monodown}.
Here we can follow the flipped analogue of Case 7 to rule out the
vertical half-space defined by $p$ with $p_1 \ge q_1 + 1$.

Specifically, if the global fixed point $x$ satisfies $x_1 \le q_1$ then we are
done. Otherwise, as in Case 7 we can build a path of two-dimensional fixed
points from $x$ to the least fixed point $l$. Since $q$ is in the
two-dimensional down-set, this path must pass through a point $x'$ with $x'_1 =
q_1$ and $x'_2 \ge q_2$. Using Lemma~\ref{lem:monodown}, we have that if $x'_3 > f_3(x')$ then this contradicts the fact
that $x$ is a global fixed point, while if $x'_3 < f_3(x')$ then this
contradicts the fact that $q_3 = f_3(q)$. So must have $x'_3 = f_3(x')$,
meaning that $x'$ is a global fixed point satisfying $x'_1 \le q_1$. 
\end{itemize}

The cases where $f_2(q)$ moves strictly upwards or strictly downwards, while
$f_1(q)$ is fixed can be dealt with by exchanging dimensions 1 and 2, and then
applying one of Cases 5 through 10.

This leaves us with the cases in which both $f_1(q)$ and $f_2(q)$ are fixed. 
Since $q$ is not a global fixed point, we can use one of the following cases.

\halfspaceFourth

\begin{itemize}
\item \textbf{Case 11: $\mathbf{q_1 = f_1(q)}$ and $\mathbf{q_2 = f_2(q)}$ and
$\mathbf{q_3 < f_3(q)}$.} (Figure~\ref{fig:case11-12} (a)).
Here we argue that we can rule out the vertical\footnote{We could also rule out
the horizontal half-space defined by $p$ with $p_2 \ge q_2 + 1$ by exchanging
dimensions 1 and 2 and using the same argument, if we so desired.} half-space
defined by $p$ with $p_1 \ge q_1 + 1$. 

Let $x$ be the global fixed point that lies in $R$. If $x_1 \le q_1$ then we
are done. Note also that we cannot have $x \ge q$, since then
Lemma~\ref{lem:monodown} would imply that $x_3 < f_3(x)$, contradicting the
fact that $x$ is a global fixed point. 

So the only remaining case is when $x_1 > q_1$ and $x_2 < q_2$, as depicted in 
(Figure~\ref{fig:case11-12} (a)).
We apply Lemma~\ref{lem:paths} to obtain a path
$l = p^1, p^2, \dots, p^k = q$ from the
least two-dimensional fixed point of the slice $l$ to $q$. Note that we must have $l \le x$, because
$l$ is a least two-dimensional fixed point, so there must exist an index $i$
such that $p^i_2 = x_2$. Note also that $p^i_1 \le q_1 < x_1$. So we have
$p^i \le x$.

If $p^i_3 = f_3(p^i)$ then $p^i_3$ is a global fixed point satisfying $p^i_1
\le q_1$ and we are done. If $p^i_3 < f_3(p^i)$ then Lemma~\ref{lem:monodown}
would imply that $x_3 < f_3(x)$, contradicting the fact that $x$ is a global
fixed point. So the only remaining case is when $p^i_3 > f_3(p^i)$, where we
can apply Lemma~\ref{lem:pathfp} to argue that there exists a global fixed
point $p^j$ with $p^i \le p^j \le q$, which in particular means $p^j_1 \le
q_1$.

\item \textbf{Case 12: $\mathbf{q_1 = f_1(q)}$ and $\mathbf{q_2 = f_2(q)}$ and
$\mathbf{q_3 > f_3(q)}$.} (Figure~\ref{fig:case11-12} (b)).
This case is the symmetric analogue of Case 11, and so we can follow the same
approach to rule out the half-space defined by $p$ with $p_1 \le q_1 - 1$. 
We only note the differences between Case 11 and Case 12 below.

In this case, if the global fixed point $x$ satisfies $x_1 \le q_1$, then we
first note that we must have $x_2 > q_2$, since the case where $x_2 \le q_2$ is
ruled out by Lemma~\ref{lem:monodown}, and then we use 
Lemma~\ref{lem:paths} to obtain a path
$q = p^1, p^2, \dots, p^k = g$ from $q$ to the
greatest two-dimensional fixed point of the slice $g$. Since $g \ge x$,
there exists an index $i$ such that $p^i_2 = x_2$, and we have $x \le p^i_2$.

If $p^i_3 = f_3(p^i)$ then $p^i_3$ is a global fixed point satisfying $p^i_1
\ge q_1$. If $p^i_3 > f_3(p^i)$ then Lemma~\ref{lem:monodown} would imply that
$x_3 > f_3(x)$, contradicting the fact that $x$ is a global fixed point.
Finally, if $p^i_3 < f_3(p^i)$, then we can apply Lemma~\ref{lem:pathfp} to
find a a global fixed point $p^j$ with $p^i \ge p^j \ge q$, which in particular
means $p^j_1 \ge q_1$.
\end{itemize}

\end{proof}

\paragraph{\bf The algorithm.}

As mentioned, the algorithm maintains a region $R$ with the invariant that $R$
contains a global fixed point. Initially we set $R$ to be a square that
contains the entire slice, which satisfies the invariant due to
Lemma~\ref{lem:fpexists}. Then, in each iteration of the algorithm we make a
number of queries, and then for each query we eliminate the corresponding
half-space from Lemma~\ref{lem:halfspace}. 

We start with a box shaped region that contains the entire slice. Each query 
eliminates a vertical, horizontal, upper-left diagonal, or lower-right diagonal
half-space, which may again give us a box, or a box with a diagonal half-space
removed. 

We say that a region $R$ is \emph{valid} if it is defined by a box with at most
one upper-left half-space removed, and at most one lower-right half-space
removed. Note that if we query a valid region $R$, and then remove the
half-space given to us from Lemma~\ref{lem:halfspace} to get a reduced region
$R'$, then $R'$ is also valid.

\shapes

To aid our analysis, we further decompose each valid region into four
categories of \emph{basic shapes}. Specifically we consider spaces that are either
box shaped, a \emph{lower right-angle triangle} with width and height $k$ as
shown in Figure~\ref{fig:shapes} (b), a \emph{upper right-angle triangle} with
width and height $k$ as shown in Figure~\ref{fig:shapes} (c), a parallelogram
whose top and bottom sides are aligned with $(1, 0)$, and whose left and right
sides are aligned with $(1, 1)$ as shown in Figure~\ref{fig:shapes} (d).

The following lemma shows that each valid region can be decomposed into at most
5 basic shapes.

\begin{lemma}
\label{lem:shapedecomp}
Every valid region $R$ can be decomposed into the union of at most 5 basic
shapes.
\end{lemma}
\begin{proof}

\shapedecomp

First suppose that both and upper-left half-space and a lower-right half-space
have been removed to create $R$. If there is a value $x$ such that no point $p
\in R$ with $p_2 = x$ is eliminated by either half-space, then we can decompose
$R$ into two triangles and three boxes, as shown in Figure~\ref{fig:shapedecomp}
(a). If there is no such value, then we can decompose $R$ into two triangles,
two boxes, and one parallelogram, as 
shown in Figure~\ref{fig:shapedecomp} (b).

The cases where one or both of the diagonal half-spaces are missing can clearly
be done with even fewer basic shapes.
\end{proof}

The next lemma states that, if we have a region defined by a basic shape, then
we can make a single query and rule out at least a quarter of the area of that
shape.

\begin{lemma}
\label{lem:basicquery}
Given a region $R$ that is either a box of width and height at least 6, a lower
or upper right-angle triangle of width and height at least $4$, or a
parallelogram of height and width at most $6$,
we can make a single query and either find a global fixed point, or eliminate at least a $1/4$ fraction of the area in $R$.
\end{lemma}
\begin{proof}
We consider the different types of shape separately. In each case we make
a single query and then consider the possible half-spaces that
Lemma~\ref{lem:halfspace} rules out.
\begin{itemize}
\item To query a box $(x, h, w)$, we query the point 
$$q = x + \Floor{\frac{w}{2}} \cdot e_1 + \Floor{\frac{h}{2}} \cdot e_2,$$
which is the centre of the box as shown in Figure~\ref{fig:shapes} (a). 

\begin{enumerate}
\item If we eliminate a vertical half-space defined by points $p$ with $p_1 \le
q_1 - 1$ or $p_1 \le q_1$, then the fraction of area that we rule out is at
least
\begin{align*}
\frac{(\Floor{\frac{w}{2}} - 1) \cdot h}{w \cdot h} & =
\frac{\Floor{\frac{w}{2}} - 1}{w} \\
&\ge \frac{\frac{w}{4}}{w} \\
&= 1/4,
\end{align*}
where we have used the fact that $w \ge 6$ implies that $
\Floor{\frac{w}{2}} - 1 \ge \frac{w}{4}$.

\item If we eliminate a vertical half-space defined by points $p$ with $p_1 \ge
q_1$ or $p_1 \ge q_1 + 1$, then 
the fraction of area that we rule out is at
least
\begin{align*}
\frac{(w - (\Floor{\frac{w}{2}} + 1)) \cdot h}{w \cdot h} & =
1 - \frac{\Floor{\frac{w}{2}} + 1}{w} \\
& \ge 1 - \frac{3w/4}{w} \\
& = 1/4,
\end{align*}
where we have used the fact that $w \ge 4$ implies that $
\Floor{\frac{w}{2}} + 1 \le \frac{3w}{4}$.

\item If we eliminate a horizontal half-space, then we can use the same
argument as the first and second cases to argue that we rule out at least 1/4
of the area of $R$.

\item If we eliminate a diagonal half-space, then note that if $q$ was at the
centre of the box, then we clearly would have ruled out half the points.
However, due to
the rounding we used when computing $q$, we may not have ruled out exactly half
of the points. But the set of discrete points that should have been ruled out
but were not lie in a diagonal line along the vector $(1, 1)$ that spans the
box directly next to the half-space defined by $q$. If we assume, without loss
of generality, that $w \ge h$, this line contains $h$ points, and so the
fraction of the space
that we actually rule is at least
\begin{align*}
1/2 - \frac{h}{w \cdot h} &\le 1/2 - \frac{h}{h \cdot h} \\
&= 1/2 - \frac{1}{h},
\end{align*}
Since $h \ge 6$ we therefore rule
out at least $1/2 - \frac{1}{6} = 1/3$ of the space.
\end{enumerate}

\item To query a lower right-angle triangle of width and height $k$ with
lower-left corner $x$, we query
the point
$$q = x + \Floor{\frac{k}{4}} \cdot e_1 + \Floor{\frac{k}{4}} \cdot (1, 1),$$
as shown in Figure~\ref{fig:shapes} (b). 
\begin{enumerate}
\item If we eliminate a vertical half-space defined by points $p$ with $p_1
\ge q_1 + 1$, then we eliminate a right-angle triangle of width and height
$k - \Floor{\frac{k}{4}} - 1$, which has area $0.5 \cdot (k 
- \Floor{\frac{k}{4}} - 1)^2$. 
The overall area of the original triangle is $0.5
k^2$, so the fraction of area ruled out is at least
\begin{align*}
\frac{ 0.5 \cdot (k - \Floor{\frac{k}{4}} - 1)^2}{0.5k^2} &\ge \frac{(k -
k/2)^2}{k^2} \\
&= 1/4,
\end{align*}
where we used the fact that $k \ge 4$ to ensure that $\Floor{\frac{k}{4}} + 1 \le k/2$.

If we eliminate a vertical half-space defined by points $p$ with $p_1
\ge q_1$ then we rule out even more space, so we again rule out at least $1/4$
of the area.

\item If we eliminate a horizontal half-space defined by points 
$p$ with $p_2
\le q_2 - 1$, then we keep a right-angle triangle of width and height
$k - \Floor{\frac{k}{4}} + 1$, which has area $0.5 \cdot (k^2 - 
\Floor{\frac{k}{4}} + 1)^2$. 
The overall area of the original triangle is $0.5
k^2$, so the area we rule out has size $0.5 k^2 - 0.5 \cdot (k^2 -
\Floor{\frac{k}{4}} + 1)^2$, so the proportion of the area that is ruled out is
at least
\begin{align*}
\frac{0.5 k^2 - 0.5 \cdot (k - \Floor{\frac{k}{4}} + 1)^2}{0.5 k^2} 
& = 1 - \frac{(k - \Floor{\frac{k}{4}} + 1)^2}{k^2} \\
& \ge 1 - \frac{(k - k/4)^2}{k^2} \\
& = 1 - \frac{3^2}{4^2} \\
& = 7/16.
\end{align*}
where we used the fact that 
$\Floor{\frac{k}{4}} + 1 \ge \frac{k}{4}$ for all positive $k$.

If we eliminate a horizontal half-space defined by points $p$ with $p_2 \le
q_2$, then we rule out even more space, so we again rule out at least a $7/16$
fraction of the area.

\item The cases where we rule out a vertical half-space are entirely symmetric
to the cases for horizontal half-spaces.

\item If we rule out a lower-right diagonal half-space, then we rule out a
triangle with height and width $k - \frac{k}{4}$, which is strictly more than
the amount of area we considered in the first case, so we rule out at least 1/4
of the area.

\item If we rule out an upper-left diagonal half-space, then we keep a triangle
with height and width $k - \frac{k}{4}$, which is strictly less than the amount
of area we considered in the second case, so we rule out at least 7/16 of the
area.
\end{enumerate}

\item To query an upper right-angle triangle of width and height $k$, we can
flip both dimensions and follow the procedure for a lower right-angle triangle.

\item To query a parallelogram whose lower-left corner is $x$, and whose height
is $h$ and width is $w$, we query the point 
$$x + \Floor{\frac{w}{2}}
\cdot e_1 + \Floor{\frac{h}{2}} \cdot (1, 1),$$ 
as shown in Figure~\ref{fig:shapes} (d).

\pgrams

\begin{enumerate}
\item If we rule out a horizontal half-space, then note that each row of points
in the parallelogram has exactly $w$ points. So the number of points that we rule out
is exactly the same as the number of points we would rule out for a box of
height $h$ and width $w$. So we can reuse our proof for boxes to show that at
least $1/4$ of the area is ruled out.

\item

If we rule out a vertical half-space, then as shown in Figure~\ref{fig:pgrams}, note that if we remove the
triangles on either end of the parallelogram (which both have the same area),
then we are either left with a box-shaped region, or a parallelogram whose left
and right sides are aligned with $(0, 1)$, and whose top and bottom sides are
aligned with $(1, 1)$. In the first case we can use the analysis for box
regions, while in the second case we can exchange dimensions 1 and 2 and use
Case 1 of the parallelogram analysis. In both cases, at least $1/4$ of the area
is ruled out.

\item If we rule out a diagonal half-space, then note that each diagonal
contains $h$ points. So the number of points ruled out is exactly the same as
the number that would be ruled out by a vertical half-space for a box of height
$h$ and width $w$. So we have that at least $1/4$ of the area is ruled out.
\end{enumerate}

\end{itemize}
\end{proof}

Each iteration of the algorithm has a valid region $R$, and proceeds as
follows.
\begin{itemize}
\item Decompose $R$ into at most $5$ basic shapes using
Lemma~\ref{lem:shapedecomp}. 
\item Find the basic shape $S$ that has the largest area, and then
make a single query on $S$ using Lemma~\ref{lem:basicquery}, and rule out the
resulting half-space in $R$.
\item If query finds a global fixed point then we are done.
\item Otherwise, we have ruled out at least $1/4$ of the area of 
$S$, and $S$ made up at least $1/5$ of the area of $R$. So the total fraction of
area ruled out from $R$ is at least $1/20$. 
\end{itemize}

This process continues until we either find a global fixed point, or until we
no longer satisfy the preconditions of Lemma~\ref{lem:basicquery}. This can
only occur when the largest basic shape in the decomposition of $R$ has either
constant height or constant width. This either means that $R$ is contained in a
box of constant width or constant height (eg. in the example in
Figure~\ref{fig:shapedecomp} (a)), or that $R$ is contained in a $\DBox(x, y,
l)$ where $\|x
- y \|_\infty$ is a constant (eg. this could occur in the example in
Figure~\ref{fig:shapedecomp} (b)). We will deal with these two cases
separately.


\paragraph{\bf Handling diagonal regions.}

We are given a region $R = \DBox(x, y, l)$ where $\|x - y \|_\infty$ is
constant with the invariant that $R$ contains a global fixed point, and we must
find the global fixed point in $O(\log n)$ steps.






Consider a diagonal line defined by points $p^1, p^2, \dots, p^k$ such that
$p^{i+1} = p^i + (1, 1)$. If we query a point $p^i$ then we have the following
properties.
\begin{itemize}
\item If $p^i_1 < f_1(p^i)$ then all points $p^j$ with $j < i$ satisfy 
$p^j_1 < f_1(p^j)$ by Lemma~\ref{lem:contrdown}. On the other hand
$p^i_1 > f_1(p^i)$ then all points $p^j$ with $j > i$ satisfy 
$p^j_1 > f_1(p^j)$, for the same reason.

\item For the same reason, if $p^i_2 < f_2(p^i)$ then we can rule out all
points $p^j$ with $j < i$, and if $p^i_2 > f_2(p^i)$ then we can rule out all
points $p^j$ with $j > i$. 

\item If $p^i_3 < f_3(p^i)$ then we have
$p^j_3 < f_3(p^j)$ for all $j > i$ by Lemma~\ref{lem:monodown}, while if
$p^i_3 > f_3(p^i)$ then we have 
$p^j_3 > f_3(p^j)$ for all $j < i$, for the same reason.
\end{itemize}

So we can carry out a binary search on the line 
$p^1, p^2, \dots, p^k$ to determine if the line contains a global fixed point:
for each query that we make, if the queried point is not a global fixed point,
then one of the arguments above will rule out at least half of the points in
the line. 
This
process takes $O(\log n)$ steps.

Note that since $\|x - y \|_\infty$ is constant, the space
$\DBox(x, y, l)$ can be decomposed into constantly many diagonal lines. So we
can apply the procedure above to each of those lines to find the global fixed
point that lies in $R$ using $O(\log n)$ steps.

\paragraph{\bf Handling box shaped regions.}

In this case we are given a box $(x, h, w)$ that contains a global fixed point
where one of $h$ or $w$ is constant. We assume that $w$ is constant. This is
without loss of generality, because if $h$ is constant we can exchange
dimensions $1$ and $2$ and apply the same procedure.

Our approach is essentially the same as the diagonal case, with a slightly
different justification. Consider the vertical line of points $p^1, p^2, \dots, p^k$ 
$p^{i+1} = p^i + e_2$. If we query a point $p^i$ then we have the following
properties.
\begin{itemize}
\item If $p^i_1 < f_1(p^i)$ then all points $p^j$ with $j > i$ satisfy 
$p^j_1 < f_1(p^j)$ by Lemma~\ref{lem:monodown}. On the other hand
$p^i_1 > f_1(p^i)$ then all points $p^j$ with $j < i$ satisfy 
$p^j_1 > f_1(p^j)$, for the same reason.

\item If $p^i_2 < f_2(p^i)$ then all points $p^j$ with $j < i$ satisfy 
$p^j_2 < f_2(p^j)$ by Lemma~\ref{lem:contrdown}. On the other hand
$p^i_2 > f_2(p^i)$ then all points $p^j$ with $j > i$ satisfy 
$p^j_2 > f_2(p^j)$, for the same reason.

\item If $p^i_3 < f_3(p^i)$ then we have
$p^j_3 < f_3(p^j)$ for all $j > i$ by Lemma~\ref{lem:monodown}, while if
$p^i_3 > f_3(p^i)$ then we have 
$p^j_3 > f_3(p^j)$ for all $j < i$, for the same reason.
\end{itemize}
So as in the diagonal case, if the line contains a global fixed point, we can
find it using binary search in $O(\log n)$ steps. Since the box $(x, h, w)$
contains constantly many vertical lines, we can perform a binary search on each
of them independently to find the global fixed point. Overall this process
makes $O(\log n)$ queries to $f$, and each step runs in time that is polynomial
in the representation of $f$.

\subsection{The Result}

Finally, we can tie up everything in this section into our main algorithmic result.
We have shown that the algorithm ends after making $O(\log n)$ queries to $f$,
and note that all steps of our algorithm can be carried out in polynomial time
relative to the representation of $f$. So we have shown the following theorem.

\begin{theorem}
There is an algorithm that finds a fixed point in a violation-free
three-dimensional \DMAC instance $f$ using $O(\log n)$ queries, where each step
takes time that is polynomial in the representation of $f$.
\end{theorem}

Applying the decomposition theorem given in Theorem~\ref{thm:decomposition} gives
us the following theorem for $d$-dimensional \DMAC, where we note that
we must apply the reduction given in Lemma~\ref{lem:uniquefp} to each
three-dimensional instance that we consider while applying the decomposition
theorem to ensure that we always find the least fixed point of the instance.

\begin{theorem}
There is an algorithm that finds a fixed point in a violation-free
$d$-dimensional \DMAC instance $f$ using $O((c \cdot \log n)^{\lceil d/3
\rceil})$ queries, where each step
takes time that is polynomial in the representation of $f$.
\end{theorem}

Finally, we can state our results for monotone contractions, which are obtained
by applying 
Theorem~\ref{thm:mctodmac} to reduce \MonotoneContraction to \DMAC, and
then applying the previous theorem.

\begin{theorem}
There is an algorithm that finds an $\eps$-approximate fixed point of a
$d$-dimensional \MonotoneContraction instance $g$ using $O((c \cdot \log
(1/\eps))^{\lceil d/3 \rceil})$ queries, where each step
takes time that is polynomial in the representation of $g$.
\end{theorem}

\bibliographystyle{plain} 
\bibliography{refs} 

\end{document}